\documentclass[12pt]{article}

\usepackage[margin=1in]{geometry}

\usepackage{graphicx}
\usepackage{subcaption}

\usepackage{amsmath,amssymb}
\usepackage{amsthm}
\usepackage{bm}

\usepackage[colorlinks=true,
            linkcolor=blue,
            urlcolor=blue,
            citecolor=blue]{hyperref}
\usepackage{enumerate}
\usepackage{comment}
\usepackage{natbib}

\newtheorem{theorem}{Theorem}[section]
\newtheorem{corollary}[theorem]{Corollary}
\newtheorem{lemma}[theorem]{Lemma}
\newtheorem{proposition}[theorem]{Proposition}
\theoremstyle{remark}
\newtheorem{remark}{Remark}[section]

\newcommand{\sectionmath}[2]{\texorpdfstring{#1}{#2}}

\newcommand{\R}{\mathbb{R}}
\newcommand{\norm}[1]{\left\lVert #1 \right\rVert}
\newcommand{\ind}[1]{\mathbb{I}\left[ #1 \right]}
\newcommand{\sign}{\mathrm{sign}}
\newcommand{\col}{\mathrm{col}}
\newcommand{\dir}{\mathrm{dir}}

\newcommand{\E}{\mathrm{E}}
\newcommand{\Var}{\mathrm{Var}}
\newcommand{\Unif}{\mathrm{Unif}}
\newcommand{\Beta}{\mathrm{Beta}}

\newcommand{\Expo}{\mathrm{Expo}}

\newcommand{\suffstat}{\bm{S}^{(1 : k)}}
\newcommand{\suffstats}{\bm{S}^{(1)}}
\newcommand{\glasso}{\hat{\bm{\beta}}^\lambda}
\newcommand{\glassos}{\hat{\beta}^\lambda}
\newcommand{\glassoc}{\tilde{\bm{\beta}}^\lambda}
\newcommand{\betak}{\check{\bm{\beta}}^\lambda}
\newcommand{\ols}{\hat{\bm{\beta}}_{\mathrm{OLS}}}
\newcommand{\olsk}{\hat{\bm{\beta}}_{1:k, \mathrm{OLS}}}
\newcommand{\olsj}{\hat{\bm{\beta}}_{j, \mathrm{OLS}}}
\newcommand{\yhat}{\hat{\bm{y}}_{1:k}}
\newcommand{\sigmahat}{\hat{\sigma}_{1:k}}
\newcommand{\sphere}{\mathbb{S}}
\newcommand{\hyp}{H_{1 : k}}
\newcommand{\condsample}{\tilde{\bm{y}}}
\newcommand{\error}{\hat{\bm{e}}_{1:k}}

\DeclareMathOperator*{\argmin}{arg\,min}

\title{The $L$-test: Increasing the Linear Model $F$-test's Power Under Sparsity Without Sacrificing Validity}

\author{\begin{tabular}{ccc}
Danielle Paulson & Souhardya Sengupta & Lucas Janson
\end{tabular}}

\date{November 2025}

\begin{document}

\maketitle

\begin{center}
\textbf{Abstract}
\end{center}
\noindent
\begin{quote}
\small
We introduce a new procedure for testing the significance of a set of regression coefficients in a Gaussian linear model with $n \geq d$. Our method, the $L$-test, provides the same statistical validity guarantee as the classical $F$-test, while attaining higher power when the nuisance coefficients are sparse. Although the $L$-test requires Monte Carlo sampling, each sample's runtime is dominated by simple matrix--vector multiplications so that the overall test remains computationally efficient. Furthermore, we provide a Monte-Carlo-free variant that can be used for particularly large-scale multiple testing applications. We give intuition for the power of our approach, validate its advantages through extensive simulations, and illustrate its practical utility in both single- and multiple-testing contexts with an application to an HIV drug resistance dataset. In the concluding remarks, we also discuss how our methodology can be applied to a more general class of parametric models that admit asymptotically Gaussian estimators.
\end{quote}

\section{Introduction}
\label{sec:intro}

\subsection{Motivation}
\label{sec:motivation}
Consider the homoscedastic Gaussian linear model
\begin{align}
    \label{eq:linear_model}
    \bm{y} \sim \mathcal N(\bm{X}\bm{\beta}, \sigma^2 \bm{I}),
\end{align}
where $\bm{X} \in \R^{n \times d}$ is full-column rank with $n \geq d$. The design matrix could be random, but we treat it as fixed throughout to avoid having to explicitly condition on it. The parameters $\bm{\beta} \in \R^d$ and $\sigma \in \R_{>0}$ are treated as unknown. To assess the contribution of a specified subset of covariates, we focus without loss of generality on the first $k$ and test the hypothesis $$\hyp: \bm{\beta}_{1:k} = \mathbf 0.$$

The classic approach is the linear regression $F$-test proposed by Fisher a century ago \cite{Fisher1925}. While the $F$-test is widely used, it lacks flexibility in that it does not allow a user to incorporate known structural information. In particular, we are interested in settings where the coefficient vector is \textit{sparse}, meaning many of the true coefficients are zero (though we find our method also performs well under various forms of \emph{approximate} sparsity as well). Real-world systems in the sciences and engineering often exhibit sparsity because although numerous potential predictors may exist, typically only a small number meaningfully affects the outcome \citep[e.g.,][]{han2021, Belloni2011}. The goal of this work is to allow a practitioner to leverage a belief about sparsity outside the tested subset for increased power while maintaining exact Type I error control under the same assumptions as the $F$-test. The recent work \cite{Sengupta2024} addresses this goal in the case $k = 1$, but their approach does not extend to the multivariate case due to the greater complexity that quickly arises beyond a single dimension.

\subsection{Our contribution}
\label{sec:contribution}
We propose a new hypothesis test for $H_{1:k}$ called the $L$-test that uses a group LASSO estimate for $\bm \beta_{1:k}$ in its test statistic. Under the same assumptions as the $F$-test, it retains exact validity while achieving greater power when the nuisance parameters are sparse and, empirically, only suffers at most very minor power loss even when the nuisance parameters are dense. While the $L$-test is computationally efficient, we show how a slight modification to the $L$-test statistic yields an even faster companion test valuable in particularly large-scale multiple testing settings. We provide insight into our methodology's power, demonstrate its power gains and robustness through extensive simulations, and showcase its practical effectiveness on an HIV drug resistance dataset.

\subsection{Background and related work}
\label{sec:related_work}
Many works have tried to generalize the $F$-test to settings where its classical assumptions---particularly normality---may not hold by developing more robust tests \citep{Friedman1937, Pitman1937, Pitman1938, Kruskal1952, Tukey1958, Hajek1962, Adichie1967, Jaeckel1972, Efron1979, Freedman1981, Gutenbrunner1993, Lei2020}. Because these methods are focused on extending the $F$-test to more general settings, they typically do not surpass the $F$-test's performance within the standard Gaussian linear model. In contrast, the present work aims to make the \textit{same} assumptions as the $F$-test, while leveraging information about sparsity for increased power.

A wide range of methods that leverage sparsity in particular to improve testing in the linear model have been developed. A prominent example is fixed-X knockoffs \citep{Barber2015, Candès2018, Huang2020} and its many extensions \citep[e.g.,][]{Spector2022, Luo2022, Ren2024, Lee2024} as well as methods based on mirror statistics \citep{Dai2022, Xing2023}. All of these methods, however, are designed for multiple testing and cannot be used for testing a single hypothesis, which is the focus of this work.

A related approach to this paper's is Frequentist, assisted by Bayes (FAB), which incorporates auxiliary information through a prior and uses it to design frequentist procedures---confidence intervals \cite{Hoff2019} and hypothesis tests \citep{Hoff2022, Bryan2021, Ding2021}---that can outperform classical ones when the prior aligns with the truth. The $L$-test is similar in spirit, in that it leverages available information about the unknown parameters for improved inference, just not in the form of a Bayesian prior. The $L$-test is also tailored to exploit sparsity for higher power, and to our knowledge, FAB has not been used to target this structural feature in hypothesis testing.

Closer to our setting, a few works have studied the problem of leveraging sparsity for single multi-parameter hypothesis testing in the linear model. For example, \cite{Candès2011} develops procedures for testing the global null that are asymptotically powerful under alternatives with different sparsity levels. These procedures are calibrated for asymptotic (not finite-sample) validity, and because they are tailored to testing the global null, they rely on sparsity exclusively among the null coefficients. By contrast, our approach is exactly valid in finite samples and leverages sparsity in the nuisance parameters---an advantage when the tested group is small, making it most useful to exploit sparsity outside the null. A complementary line of work uses the asymptotic distribution of the debiased LASSO estimator \citep{vandeGeer2014, Zhang2014, Montanari2014} to construct confidence regions and p-values for low-dimensional subsets of $\bm \beta$; these tests are likewise only asymptotically valid and, moreover, their validity relies on a strict sparsity assumption. In our case, validity does not hinge on sparsity; sparsity is solely used to increase power. The recent work \cite{Sengupta2024} is the one that is closest to ours in the literature, as it achieves our precise objective when $k=1$. Unfortunately, key elements of that work do not extend beyond $k=1$, as we will discuss in more detail when we lay out the background for our approach in the next section.

\subsection{Notation}
\label{sec:notation}

Boldfaced symbols represent matrices and vectors. For a matrix $\bm{A} \in \R^{n \times m}$, $\col(\bm{A})$ refers to the column space of $\bm{A}$. We also use $\bm{A}_{1:i}$ to represent its first $i$ columns, $\bm{A}_{-1:i}$ to represent its submatrix with the first $i$ columns dropped, and $A_{ij}$ to represent the entry in the $i$-th row and $j$-th column. Similarly, for a vector $\bm{a} \in \R^n$, we use $\bm{a}_{1:i}$ to represent its first $i$ elements, $\bm{a}_{-1:i}$ to represent its subvector with the first $i$ elements dropped, and $a_{i}$ to represent its $i$-th entry. All vectors are assumed to be column vectors. We use the notation $\ind{\cdot}$ to refer to the indicator function that takes the value 1 if the condition inside the brackets is true and 0 otherwise. We also use the notation $\sphere^{n-1}$ to represent the unit sphere in $\R^n$. Finally, we use $\norm{\cdot}_2$ to denote the $\ell^2$-norm and $\norm{\cdot}_1$ to denote the $\ell^1$-norm. When no subscript is included on a normed expression, the $\ell^2$-norm is assumed.


\subsection{Code}
\label{sec:code}
\texttt{Python} code for implementing our methodology and reproducing all empirical results in this paper can be accessed at this GitHub repository: \href{https://github.com/dipaulson/L-test/}{https://github.com/dipaulson/L-test/}.

\section{Methodology}
\label{sec:method}

\subsection{Conceptual framework}
\label{sec:extend_l_test}
We start by following the idea of the $\ell$-test \citep{Sengupta2024} and defining an analogous test statistic for general $k$ whose null distribution we obtain by conditioning on a suitable sufficient statistic. The minimal sufficient statistic for the linear model in \eqref{eq:linear_model} under $\hyp$ is $\suffstat := (\bm{X}_{-1:k}^T\bm{y}, \bm{y}^T\bm{y})$. By definition of a sufficient statistic, the distribution of $\bm y \mid \suffstat$ under $\hyp$ is free of unknown parameters, and hence the distribution of any deterministic function of $\bm{y}$ conditional on $\suffstat$ under $\hyp$ will also be parameter-free. Define the following group LASSO estimate \citep{Bakin1999, Zhang2003, Yuan&Lin2006} with penalty $\lambda$ \footnote{The penalty in the original definition of the group LASSO by \cite{Bakin1999} is formulated as a sum of Mahalanobis norms over groups, with each group $j$ associated with a positive semidefinite matrix $\bm{K}_j$. Different specifications of $\bm{K}_j$ lead to different variants of the group LASSO. We use the common default choice of $\bm{K}_j = \bm{I}$, suggested in \cite{Yuan&Lin2006}.}:
\begin{align}
\label{eq:group_lasso}
\glasso  
    :=\argmin_{\bm{\beta} \in \R^d}\left(\frac{1}{2n}\norm{\bm{y} - \bm{X}\bm{\beta}}_2^2 + \lambda\left( \norm{\bm{\beta}_{1:k}}_2 + \norm{\bm{\beta}_{-1:k}}_1\right)\right).
\end{align}
Note that while \eqref{eq:group_lasso} seems natural for our setting, many other LASSO-based estimators have been proposed in the literature \citep[e.g., ][]{Tibshirani1996, Zou&Hastie2005, Zou2006, Tibshirani2013}, and we motivate this choice in the context of some of these alternatives in Appendix \ref{sec:power_plots_add}. Also let $\bm{V} \in \R^{n \times (n-d+k)}$ denote a matrix whose columns constitute an orthonormal basis for $\col(\bm{X}_{-1:k})^\perp$ with the first $k$ columns given by $\bm{V}_i = \frac{(\bm{I}-\bm{P}_{-1:i})\bm{X}_i}{\norm{(\bm{I}-\bm{P}_{-1:i})\bm{X}_i}}$ for $i \in \{1, \ldots, k\}$ and $\bm{P}_{-1:i}$ the projection matrix onto $\col(\bm{X}_{-1:i})$ (for an explicit construction, see Appendix \ref{sec:V_construction}). By the previous sufficiency argument, the test statistic
\begin{align}
\label{eq:L-test_stat_1}
T(\bm y) := \norm{\bm{V}_{1:k}^T\bm{X}_{1:k}\glasso_{1:k}}
\end{align}
conditional on $\suffstat$ under $\hyp$ has a distribution free of unknown parameters. To understand the form of \eqref{eq:L-test_stat_1}, note that the $F$-test can be shown to be exactly equivalent to rejecting large values of 
\begin{align}
\label{eq:F-test_stat_cond}
\norm{\bm{V}_{1:k}^T\bm{X}_{1:k}\olsk}
\end{align}
conditional on $\suffstat$, where $\ols$ denotes the OLS estimate (see Lemma \ref{lma:cond_F} for a proof). So, \eqref{eq:L-test_stat_1} is simply the result of replacing the OLS estimate $\olsk$ with the group LASSO estimate $\glasso_{1:k}$ in the $F$-test statistic defined in \eqref{eq:F-test_stat_cond}. The motivation is that since the group LASSO is a better estimator of $\bm{\beta}_{1:k}$ than OLS when $\bm{\beta}_{-1:k}$ is sparse, we expect using $\glasso_{1:k}$ in the test statistic will lead to superior power in such a setting.

One way to compute a p-value for the test that rejects for large values of $T(\bm{y})$ conditional on $\suffstat$ is via the co-sufficient sampling (CSS) framework \citep{Bartlett1937, Stephens2012, Barber2022}, using Monte Carlo (MC) draws:
\begin{align}
\label{eq:mc_p}
\frac{1}{M+1}\left(1 + \sum_{i = 1}^M\ind{T(\bm{y}^{(i)}) \geq T(\bm{y})}\right), \quad \bm{y}^{(i)} \overset{\text{i.i.d.}}{\sim} \bm{y} \mid \suffstat \text{ under $\hyp$}.
\end{align}
We can sample from $\bm{y} \mid \suffstat \text{ under $\hyp$}$ by letting $\yhat = \bm{P}_{-1:k}\bm{y}$ and $\sigmahat^2 = \norm{(\bm{I} - \bm{P}_{-1:k})\bm{y}}^2$ and noting that there exists a unique unit vector $\bm{u} \in \R^{n-d+k}$ such that 
\begin{align}
\label{eq:decomp}
\bm{y} = \yhat + \sigmahat \bm{V}\bm{u}, \quad \bm{u} \mid \suffstat \overset{\hyp}{\sim} \Unif(\sphere^{n-d+k-1}), 
\end{align}
a standard result that is derived as a multivariate extension of Proposition E.1 in \cite{Luo2022} in Lemma \ref{lma:decomp}. Because the data $\bm y$ and its sampled copies $\{\bm y^{(1)}, \ldots, \bm y^{(M)}\}$ are i.i.d. conditional on $\suffstat$ under $\hyp$, the data and its copies are exchangeable under $\hyp$. As a result, the MC p-value in \eqref{eq:mc_p} stochastically dominates $\Unif(0, 1)$ under $\hyp$ and is hence valid.

The difficulty with this MC approach, however, is that it requires running the group LASSO on each MC sample---a computationally expensive task, especially when large $M$ is needed for high-resolution p-values, as in multiple testing settings. Everything in this subsection thus far matches the preliminaries of the $\ell$-test \citep{Sengupta2024} when $k=1$, but that paper is then able to avoid MC by deriving the null distribution of $T(\bm{y}) \mid \suffstat$ (when $k = 1$). Their proof approach relies on insightfully characterizing $\glassos_1$ as a function of $u_1$ (Theorem 2.1 in \cite{Sengupta2024})---the first component of the unit vector $\bm u$ defined in \eqref{eq:decomp}---to convert the p-value from $\glassos_1$-space to $u_1$-space, but extending this idea to higher dimensions is intractable. In particular, for $k>1$, understanding $\glasso_{1:k}$ as a function of $\bm{u}_{1:k}$ is more complicated since the matrix derivative $\frac{\partial \glasso_{1:k}}{\partial \bm{u}_{1:k}}$ has no simplifying structure, and it is also unclear what analytic properties would even enable a transformation of the p-value into $\bm{u}_{1:k}$-space due to its greater complexity in $\glasso_{1:k}$-space; see Appendix \ref{sec:analytic_challenges} for a full discussion. In the remainder of this section, we will introduce a new approach to computationally tractably implement the conceptual idea introduced in this subsection for $k\ge 1$.

\subsection{Constructing the $L$-test}
\label{sec:construction}
To address the computational bottleneck of solving a group LASSO for every MC draw in \eqref{eq:mc_p}, we define an affine approximation for $\glasso_{1:k}(\bm{y})$ that depends only on $\suffstat$, thereby eliminating the need to run a group LASSO on each sample. We call the test that uses this approximation the $L$-test, which notably only requires two group LASSO runs, a substantial improvement over the original MC test, as Section \ref{sec:sim_studies} will show. The experiments in Section \ref{sec:sim_studies} also demonstrate that the $L$-test achieves statistical performance nearly identical to that of the original MC method.

We start by defining the function $f_{\suffstat}$ such that $\glasso_{1:k}=f_{\suffstat}(\bm u_{1:k})$. As we discuss in detail in Appendix \ref{sec:glasso_unit}, $f_{\suffstat}$ is challenging to analyze so we instead focus on its inverse $f^{-1}_{\suffstat}$. Although evaluating this inverse on an input requires running a LASSO and is therefore computationally demanding (Theorem \ref{thm:glasso_unit}), the gradient of $f^{-1}_{\suffstat}$ constrained to a hypersphere of radius $r$, $\sphere^{k-1}(r)$, reveals this inverse has a relatively simple structure. The gradient is shown in Lemma~\ref{lma:deriv} below and proved in Appendix \ref{sec:lma_deriv_pf}.
\begin{lemma}
\label{lma:deriv}
Fixing $r, \lambda \in \R_{>0}$, let $f_{\suffstat}^{-1}(\bm{b}; r) = f_{\suffstat}^{-1}(\bm{b})\biggr\vert_{\sphere^{k-1}(r)}$. Then,
\begin{equation}\label{eq:finverse}
    \nabla f_{\suffstat}^{-1}(\bm{b}; r) = \frac{1}{\sigmahat} (\bm{X}_{1:k}^T \bm{V}_{1:k})^{-1}\left[\bm{X}_{1:k}^T(\bm{I} - \bm{P}_{\mathcal{A}(\bm{b})})\bm{X}_{1:k} + \frac{n\lambda}{r} \bm{I}\right],
\end{equation}
where
\begin{align*}
    \mathcal{A}(\bm{b}) &= \left\{i \in \{k+1, \ldots, d\} : \glassos_i(\bm b) \neq 0\right\} \\
    \glasso_{-1:k}(\bm{b}) &= \argmin_{\bm{\beta}_{-1:k} \in \R^{d-k}}\left(\frac{1}{2n}\norm{\bm{y} - \bm{X}_{1:k}\bm{b} - \bm{X}_{-1:k}\bm{\beta}_{-1:k}}_2^2 + \lambda\norm{\bm{\beta}_{-1:k}}_1\right).
\end{align*}
\end{lemma}
Since the active set $\mathcal{A}(\bm{b})$ is constant in a neighborhood of $\bm{b}$, Lemma \ref{lma:deriv} implies that $f_{\suffstat}^{-1}(\bm{b}; r)$ is piecewise affine. Corollary \ref{cor:affine_piece} below, proved in Appendix \ref{sec:cor_affine_piece_pf}, explicitly specifies the functional form of one of these affine pieces.
\begin{corollary}
\label{cor:affine_piece}
    Given a point $\bm b^* \in \R^k$, for $\bm b$ in a sufficiently small neighborhood of $\bm{b}^*$ such that $\norm{\bm b} = \norm{\bm b^*}$,
    \begin{align}
        f_{\suffstat}^{-1}(\bm{b}; \norm{\bm b^*}) &= \nabla f_{\suffstat}^{-1}(\bm{b}^*; \norm{\bm b^*})\bm{b} + \bm{\nu}(\bm{b}^*) \nonumber \\
        \bm{\nu}(\bm{b}^*) &= -\frac{1}{\hat{\sigma}_{1:k}} \bm V_{1:k}^T\bm{X}_{1:k}\betak_{1:k}(\bm{b}^*) \nonumber \\
 \betak_{1:k}(\bm{b}^*) &= \bm{S}^{-1}\bm{X}_{1:k}^T\bm{P}_{-1:k}(\bm{y} - \bm{X}_{-1:k}\glasso_{-1:k}(\bm{b}^*) - P_{\mathcal{A}(\bm{b}^*)}\bm{X}_{1:k}\bm{b}^*) \label{eq:beta_est} \\
 \bm{S} &= \bm{X}_{1:k}^T\bm{V}_{1:k}\bm{V}_{1:k}^T\bm{X}_{1:k}. \nonumber
    \end{align}
\end{corollary}
While an exact characterization of $f_{\suffstat}^{-1}(\bm{b}; r)$ could be obtained by tracking how the active set $\mathcal{A}(\bm{b})$ evolves with its input, it is expensive to compute $\mathcal{A}(\bm{b})$ so we instead propose a simpler approach: selecting a single affine piece of this piecewise affine map and using it as a global approximation. This simplification is supported by simulations in which we visualized $f_{\suffstat}^{-1}(\bm{b}; r)$ and found it to be approximately ellipsoidal, suggesting that the affine pieces are similar enough for the map to be well-approximated as globally affine; see Appendix \ref{sec:f_viz} for further details.

Using the affine function corresponding to a point $\bm{b}^* \in \R^k$ specified in Corollary \ref{cor:affine_piece} as an approximation for $f_{\suffstat}^{-1}(\bm{b})$ gives 
$$\glasso_{1:k} = f_{\suffstat}(\bm u_{1:k}) \approx (\nabla f_{\suffstat}^{-1}(\bm{b}^*; \|\bm{b}^*\|))^{-1}(\bm u_{1:k} - \bm{\nu}(\bm{b}^*)),$$ 
which when plugged into the expression for $T(\bm{y})$ in Equation~\eqref{eq:L-test_stat_1} gives
$$T(\bm{y}) \approx \norm{\bm{V}_{1:k}^T\bm{X}_{1:k}(\nabla f_{\suffstat}^{-1}(\bm{b}^*; \|\bm{b}^*\|))^{-1}(\bm u_{1:k} - \bm{\nu}(\bm{b}^*))}.$$
To use this test statistic, we need to choose a $\bm{b}^*$. And, although we have until now treated $\lambda$ as fixed (and correspondingly suppressed it in our notation), in practice we need a data-driven way to choose it; we will follow the same strategy for choosing both $\bm{b}^*$ and $\lambda$. The obvious strategy of using cross-validation to choose $\lambda$ and (the first $k$ elements of) cross-validated group LASSO to choose $\bm{b}^*$ invalidates our MC p-value \eqref{eq:mc_p} because these choices still depend on the data after conditioning on $\suffstat$ under $\hyp$.\footnote{One could make these choices valid by refitting them to each MC sample, but this sacrifices the computational advantage of our affine approximation.} Instead, we propose sampling $\tilde{\bm u} \sim \Unif(\sphere^{n-d+k-1})$, substituting $\tilde{\bm u}$ into \eqref{eq:decomp} to obtain $\condsample$ independent of $\bm y$ conditional on $\suffstat$ under $\hyp$, and then applying this `obvious strategy' to the data $(\condsample, \bm X)$, denoting the corresponding choices simply (and without risk of confusion, as we fix this $\lambda$ choice for the rest of the paper) by $\lambda$ and $\bm b^* = \glassoc_{1:k}$; Appendix \ref{sec:penalty} discusses this approach further.

A last remaining issue is that if $\glassoc_{1:k}=\bm 0$, then Equation~\eqref{eq:finverse} for $\nabla f_{\suffstat}^{-1}$ diverges to $\infty$ with $r=0$ plugged in, resulting in our test statistic in Equation~\eqref{eq:L-test_stat_1} being set to zero for the data and all MC samples. This causes all test statistic comparisons in the MC p-value \eqref{eq:mc_p} to evaluate as ties, producing a trivial p-value of 1. We can avoid this pathology by instead taking a limit as $r\rightarrow 0$ to give nontrivial comparisons and potentially still a powerful test. Thus, we define our final test statistic for the $L$-test as
\begin{align*}
L(\bm y) := \norm{\bm{A}(\glassoc_{1:k}, \lambda)(\bm u_{1:k} - \bm{\nu}(\glassoc_{1:k}))},
\end{align*}
where 
\begin{equation*}
\bm{A}(\glassoc_{1:k}, \lambda) := \begin{array}{rl}
\bm{V}_{1:k}^T\bm{X}_{1:k}(\nabla f_{\suffstat}^{-1}(\glassoc_{1:k}; \|\glassoc_{1:k}\|))^{-1} & \text{ if }\glassoc_{1:k}\neq \bm 0 \\
\lim_{r\rightarrow 0}\bm{V}_{1:k}^T\bm{X}_{1:k}(\nabla f_{\suffstat}^{-1}(\glassoc_{1:k}; r))^{-1}/r = \frac{\sigmahat}{n\lambda} \bm{V}_{1:k}^T\bm{X}_{1:k}\bm{X}_{1:k}^T \bm{V}_{1:k} & \text{ if } \glassoc_{1:k} = \bm 0.
\end{array}
\end{equation*}



By construction, the MC $L$-test p-value is valid under the same assumptions as the $F$-test; we state this formally as a corollary below.

\begin{corollary}[Validity of the $L$-test]
For model \eqref{eq:linear_model} and all $\alpha \in [0, 1]$,
$$P_{\hyp}\left(\frac{1}{M+1}\left(1 + \sum_{i = 1}^M\ind{L(\bm{y}^{(i)}) \geq L(\bm{y})}\right) \leq \alpha \right) \leq \alpha.$$
\end{corollary}

\subsection{Monte-Carlo-free p-values}
\label{sec:precise_p_values}
Computing the MC p-value with $L$ requires only two group LASSO runs---one to obtain $\glassoc_{1:k}$ and another to obtain $\glasso_{-1:k}(\glassoc_{1:k})$---yielding a substantial computational gain over the original MC p-value that requires re-running the group LASSO for each MC sample. However, as can be seen by Equation~\eqref{eq:mc_p}, the MC p-value is lower-bounded by $1/(M+1)$ (recall that $M$ is the number of MC samples). This means that for settings where power requires very small p-values, such as in multiple testing settings with high multiplicities, $M$ may need to be so large that recomputing $L$ for each MC sample, which scales linearly in $M$, dominates the two group LASSOs' computation time. For such settings that would require an MC p-value to use an astronomically high $M$, we would prefer a method that does not require MC at all.

To this end, we propose omitting the premultiplying matrix $\bm{A}(\glassoc_{1:k}, \lambda)$ in $L$ because we can analytically characterize the density of $\norm{\bm{u}_{1:k} - \bm{\nu}} \mid \suffstat$ under $\hyp$, enabling p-value computation via efficient numerical integration of a two-dimensional integral that is not limited by a p-value lower-bound that scales inversely with computation time. The proof of the following proposition in Appendix~\ref{sec:thm_density_pf} relies on properties of Beta distributions and strategic changes of variables.
\begin{proposition}
\label{thm:density}
For $\bm u \sim \Unif(\sphere^{n-d+k-1})$, $k \geq 2$, and $\bm c \in \R^k$, the density of the random variable $Z = \norm{\bm u_{1:k} - \bm c}$ is given by 
\begin{align*}
f_Z(z) &= \ind{\norm{\bm{c}} \leq 1}\left\{zg(z)\ind{z \in [0, \norm{\bm{c}}+1]}\right\} \\
&+ \ind{\norm{\bm{c}} > 1}\left\{zg(z)\ind{z \in [\norm{\bm{c}}-1, \norm{\bm{c}}+1]}\right\}\\
g(z) &=
\begin{cases}
    \displaystyle\int_{\norm{\bm{c}}^2-2z\norm{\bm{c}} +z^2}^{\norm{\bm{c}}^2+2z\norm{\bm{c}} +z^2} h(z, t)~dt &\text{if $0 \leq z \leq \left|\norm{\bm{c}} - 1\right|$} \\
    \displaystyle\int_{\norm{\bm{c}}^2-2z\norm{\bm{c}} +z^2}^1 h(z, t)~dt &\text{if $\left|\norm{\bm{c}}-1\right| < z \leq \norm{\bm{c}}+1$}
\end{cases} \\
h(z, t) &= \frac{D}{\norm{\bm{c}}} \cdot (1-t)^{\frac{n-d-2}{2}}\left(t-\left(\frac{t + \norm{\bm{c}}^2 - z^2}{2\norm{\bm{c}}}\right)^2\right)^{\frac{k-3}{2}} \\
D &= \frac{\Gamma(\frac{k}{2})\Gamma(\frac{n-d+k}{2})}{\sqrt{\pi}\Gamma(\frac{k-1}{2})\Gamma(\frac{k}{2})\Gamma(\frac{n-d}{2})}.
\end{align*}
For the density in the $k = 1$ case, see Appendix \ref{sec:special_case}.
\end{proposition}

The experiments in Section \ref{sec:sim_studies} show that this test often performs similarly to the $L$-test, but can outperform or underperform in certain situations, and in general less closely matches the original MC test that uses the test statistic \eqref{eq:L-test_stat_1}. Nevertheless, it offers notable gains over the $F$-test and may be of interest to users who require very precise p-values---such as in multiple testing settings---since they can be obtained efficiently. In Section \ref{sec:application}, we show the computational speed up and power gain this test can offer in a real multiple testing setting. 

\begin{remark}
This test is similar to the recentered $\bm u_{1:k}$-test that is analogous to the recentered test in \cite{Sengupta2024} and discussed in Appendix \ref{sec:ell_recentered}, except that it uses a different recentering vector and has the key advantage that its p-value can be computed without MC sampling. While both recentering vectors are principled choices and give rise to similar test performance, we prefer the one presented in this work since it fits more neatly with the approach that gives rise to the $L$-test.
\end{remark}

\section{Intuition for the $L$-test's power gain}
\label{sec:power_analysis}


\subsection{An oracle test}
\label{sec:oracle_test}
To better understand the power gain of the $L$-test over the $F$-test under sparsity, it is helpful to compare it to an oracle test that outperforms the $F$-test. To wit, consider an oracle that knows the \emph{direction} of $\bm \beta_{1:k}$, i.e., $\bm\beta_{1:k}/\|\bm\beta_{1:k}\|$, and uses it to first construct the design matrix $\tilde{\bm{X}} := 
\begin{pmatrix}
\bm{X}_{1:k}\frac{\bm{\beta}_{1:k}}{\norm{\bm{\beta}_{1:k}}} & \bm{X}_{-1:k}
\end{pmatrix}$ and then perform a 1-sided $t$-test in the (well-specified) linear model
\begin{align}
\label{eq:linear_model_oracle}
\bm{y} \sim \mathcal{N}(\tilde{\bm{X}}\bm{\gamma}, \sigma^2\bm{I})
\end{align}
of $\tilde{H}_0: \gamma_1 = 0$ against $\tilde{H}_1: \gamma_1 > 0$. Since $\gamma_1=\|\bm\beta_{1:k}\|$, this is a valid test of $\hyp$; note $\tilde{\bm{X}}$ is ill-defined under $\hyp$, but as this section is about power, we will only consider the setting when $\hyp$ is false and hence the oracle is well-defined. When $k=1$, this oracle is exactly performing the 1-sided $t$-test on $\beta_1$ (in the correct direction), which \cite{Sengupta2024} shows the $\ell$-test mimics the performance of. We will argue that the $L$-test also aims to mimic the performance of the oracle test in the general $k\ge 1$ case, but allowing $k>1$ creates interesting dynamics not observed in \citep{Sengupta2024}. 
In particular, the oracle's privileged information, $\bm\beta_{1:k}/\|\bm\beta_{1:k}\|$, is a unit vector in $k$ dimensions, which means that as $k$ grows, the oracle's privileged information gets richer. On the one hand, this makes the oracle's power improvement over the $F$-test increasing in $k$. On the other hand, this also makes it increasingly hard for the $L$-test (or, one would expect, any test) to approach the power of the oracle as $k$ grows. On net, we find in Section~\ref{sec:sim_studies} that the $L$-test's power gains over the $F$-test can be substantially greater for $k>1$ than the $\ell$-test's ($k=1$) power gain over the $t$-test.


To compare the oracle with the $L$-test, it is useful to reframe the oracle procedure as a conditional test given $\suffstat$. In fact, it is equivalent to rejecting for large values of 
\begin{align}
\label{eq:oracle_stat}
O_m := \norm{\bm{u}_{1:k} + m \bm{V}_{1:k}^T\bm{X}_{1:k}\frac{\bm{\beta}_{1:k}}{\norm{\bm{\beta}_{1:k}}}}
\end{align}
conditional on $\suffstat$, in the limit as $m \to \infty$. We provide a formal statement and proof of this equivalence in Appendix \ref{sec:oracle_equiv}.

\subsection{$L$-test power}
\label{sec:L_test_power}
Lemma \ref{lma:cond_F} shows that the $F$-test is equivalent to the test that rejects for large $\norm{\bm{u}_{1:k}}$ conditional on $\suffstat$. Thus, the oracle test, $F$-test, and $L$-test statistics all take the following form
$$\norm{\bm{\phi}(\bm{u}_{1:k} - \bm{\omega})},$$
for $\bm{\omega} \in \R^k$ and positive definite $\bm{\phi}: \R^k \to \R^k$ that are conditionally independent of the data given $\suffstat$. Letting $\tilde{\bm{u}}_{1:k}$ denote the first $k$ elements of a $\Unif(\sphere^{n-d+k-1})$ drawn independently of the data (and continuing to denote by $\bm u$ the vector computed from the data via Equation~\eqref{eq:decomp}), this test statistic's conditional MC p-value in the limit as the number of MC samples, $M$, goes to infinity, is
\begin{align*}
    &P_{\hyp}\left( (\tilde{\bm{u}}_{1:k}-\bm{\omega})^T\bm{\phi}^T\bm\phi(\tilde{\bm{u}}_{1:k}-\bm{\omega})\geq ({\bm{u}}_{1:k}-\bm{\omega})^T\bm{\phi}^T\bm\phi({\bm{u}}_{1:k}-\bm{\omega}) \;\biggr\vert\; \suffstat\right) \\
    =\, &P_{\hyp}\left( (\tilde{\bm{u}}_{1:k}-\bm{\omega})^T\frac{1}{c}\bm{\phi}^T\bm\phi(\tilde{\bm{u}}_{1:k}-\bm{\omega})\geq 1 \;\biggr\vert\; \suffstat\right) \text{, where} \\
    c =\, &(\bm{u}_{1:k}-\bm{\omega})^T\bm{\phi}^T\bm\phi(\bm{u}_{1:k}-\bm{\omega}).
\end{align*}
The set of points $\bm{x} \in \R^k$ that satisfy $(\bm{x}-\bm{\omega})^T\frac{1}{c}\bm{\phi}^T\bm\phi(\bm{x}-\bm{\omega})=1$ defines a $k$-dimensional ellipsoid, so the p-value corresponds to the mass of $\tilde{\bm{u}}_{1:k}$ outside this ellipsoid. This ellipsoid's center is determined by the recentering vector $\bm{\omega}$, and its semi-axes are specified by the eigenvectors and eigenvalues of $\frac{1}{c}\bm{\phi}^T\bm\phi$. 
Below we analyze the recentering vector and premultiplier in the context of each of the three tests to shed light on when and how we expect the $L$-test to gain power over the $F$-test.

\subsubsection{The recentering}

The oracle's recentering vector is $\bm{\omega}^{O} = m\bm{V}_{1:k}^T\bm{X}_{1:k}\frac{\bm{\beta}_{1:k}}{\norm{\bm{\beta}_{1:k}}}$ for $m \gg 0$, and the $L$-test's recentering vector is $\bm \omega^{L} = \bm{\nu}(\glassoc_{1:k})$. The decomposition in \eqref{eq:decomp} implies that $\bm u_{1:k} = \frac{1}{\hat{\sigma}_{1:k}}\bm{V}_{1:k}^T\bm y$, from which we get that
\begin{align}
\label{eq:approx_dir}
\dir(\bm{u}_{1:k}) = \dir(\bm{V}_{1:k}^T\bm y) \approx \dir(\E[\bm{V}_{1:k}^T\bm y]) = \dir(\bm{V}_{1:k}^T\bm{X}_{1:k}\bm{\beta}_{1:k}).
\end{align}
The oracle essentially recenters with a vector that points with infinite magnitude in the direction exactly opposite \eqref{eq:approx_dir} because it has the benefit of knowing the exact direction of $\bm \beta_{1:k}$. In a similar manner, the $L$-test's recentering vector points approximately opposite \eqref{eq:approx_dir} by making a guess of the direction of $\bm{\beta}_{1:k}$. We can see this by thinking of $\betak_{1:k}(\glassoc_{1:k})$ from \eqref{eq:beta_est} as an estimator of
\begin{align*}
\hat{\bm{\beta}}^*_{1:k} &:= \bm{S}^{-1}\bm{X}_{1:k}^T\bm{P}_{-1:k}(\bm{y} - \bm{X}_{-1:k}\bm{\beta}_{-1:k}) \\
&\sim \mathcal N\left(\bm{S}^{-1}\bm{X}_{1:k}^T\bm{P}_{-1:k}\bm{X}_{1:k}\bm{\beta}_{1:k},\; \sigma^2\bm{S}^{-1}\bm{X}_{1:k}^T\bm{P}_{-1:k}\bm{X}_{1:k}\bm{S}^{-1}\right), 
\end{align*}
the same expression that defines $\betak_{1:k}(\glassoc_{1:k})$ but with $\glasso_{-1:k}(\glassoc_{1:k})$ replaced with $\bm \beta_{-1:k}$ and $\bm{P}_{\mathcal{A}(\glassoc_{1:k})}$ set to zero. Since the group LASSO is exactly designed to handle sparsity, we would expect $\glasso_{-1:k}(\glassoc_{1:k})$ to estimate $\bm\beta_{-1:k}$ best when $\bm\beta_{-1:k}$ is sparse. Also, when $\bm\beta_{-1:k}$ is sparse, we would expect the group LASSO estimate $\hat{\bm\beta}^\lambda_{-1:k}(\glassoc_{1:k})$ to be sparse as well, meaning it should have a small active set and hence $\bm{P}_{\mathcal{A}(\glassoc_{1:k})} \approx \bm{0}$.
Taken together, these heuristic arguments suggest that when $\bm\beta_{-1:k}$ is sparse, $\betak_{1:k}(\glassoc_{1:k})$ should approximately be distributed like $\hat{\bm{\beta}}_{1:k}^*$, which in turn allows us to say that
\begin{align*}
    \dir(\bm{\nu}(\glassoc_{1:k})) 
    = -\dir\left( \bm V_{1:k}^T\bm{X}_{1:k}\betak_{1:k}(\glassoc_{1:k})\right) 
    \approx -\dir(\bm{V}_{1:k}^T\bm{X}_{1:k}\hat{\bm{\beta}}^*_{1:k}).
\end{align*}
Further approximating by replacing $\hat{\bm{\beta}}^*_{1:k}$ by its expectation gives

\begin{equation*}
    \dir(\bm{\nu}(\glassoc_{1:k})) \approx -\dir((\bm{X}_{1:k}^T\bm{V}_{1:k})^{-1}\bm{X}_{1:k}^T\bm{P}_{-1:k}\bm{X}_{1:k}\bm{\beta}_{1:k}),
\end{equation*}
which forms an obtuse angle with \eqref{eq:approx_dir}. Note that there is no recentering effect when the blocks $\bm{X}_{1:k}$ and $\bm{X}_{-1:k}$ are orthogonal since $\betak_{1:k}(\glassoc_{1:k}) = \bm{0}$, a case we study empirically in Section \ref{sec:power_orthog}.

\subsubsection{The premultiplier}

The $L$-test is the only one of the three tests that takes $\bm{\phi} = \bm{A}(\glassoc_{1:k}, \lambda)$ to not be the identity. A detailed analysis of this premultiplier is performed in Appendix \ref{sec:premult_role} and reveals that its structure depends on how the variance is allocated among the first $k$ predictors. In particular, when the variance is spread evenly across the first $k$ predictors, $\bm{A}(\glassoc_{1:k}, \lambda) \approx \bm{I}$. By contrast, when nearly all the variance is concentrated in a single principal component (PC) and the true $\bm{\beta}_{1:k}$ is well aligned with that component, then $\bm{A}(\glassoc_{1:k}, \lambda)$ has one dominant eigenvalue with eigenvector approximately parallel to $\bm{u}_{1:k}$, while the remaining $k-1$ eigenvalues are close to zero. Geometrically, this essentially translates to the ellipse having a single finite principal axis along $\bm{u}_{1:k}$ and all other axes orthogonal to it extending infinitely.

\subsubsection{Putting the pieces together}
\begin{figure}[ht]
    \centering
    \begin{subfigure}[t]{0.25\linewidth}
        \centering
        \includegraphics[width=\linewidth]{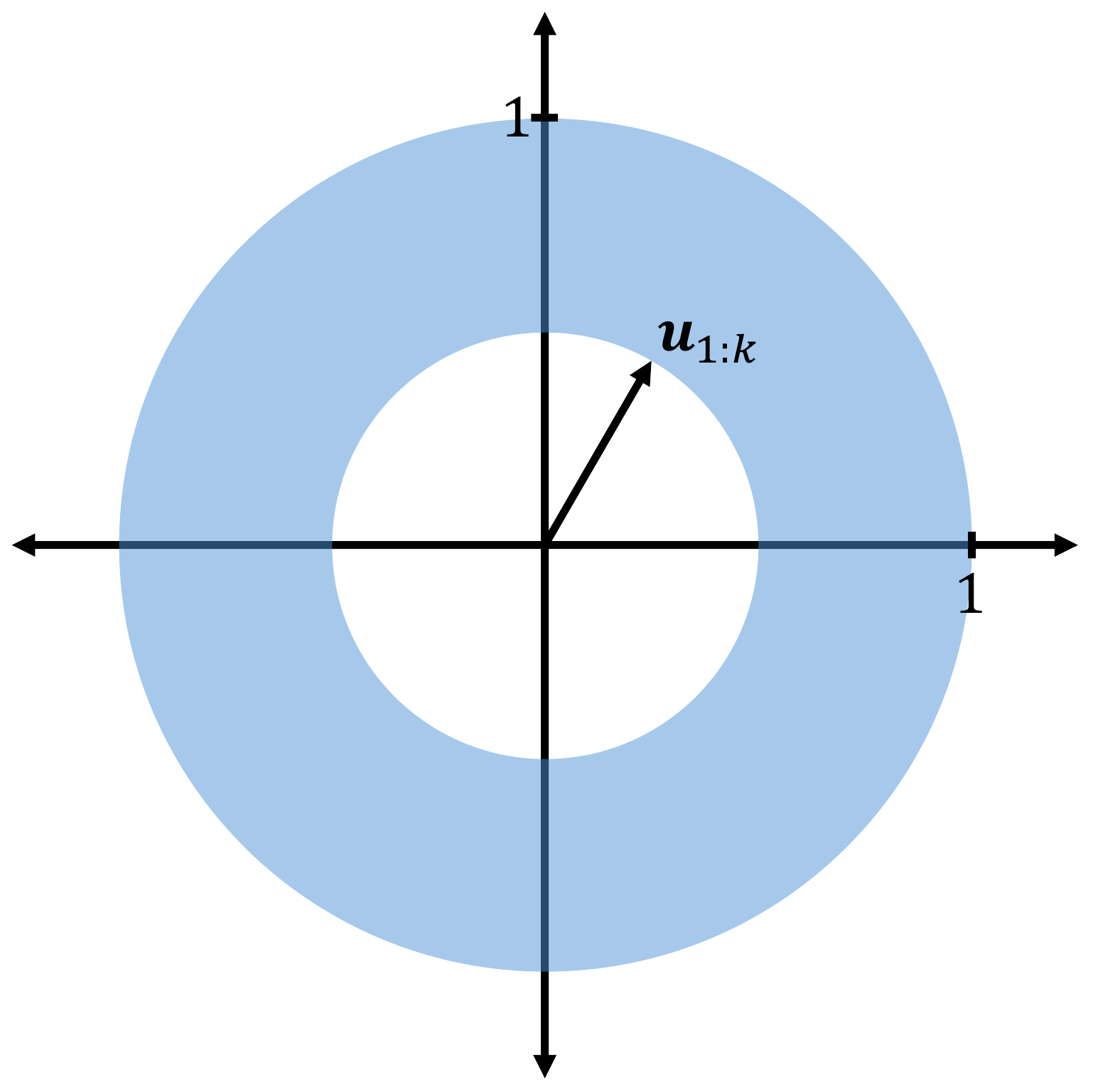}
        \caption{$F$-test geometry.}
        \label{fig:pvalue_F}
    \end{subfigure}
    \hspace{0.5cm}
    \begin{subfigure}[t]{0.25\linewidth}
        \centering
        \includegraphics[width=\linewidth]{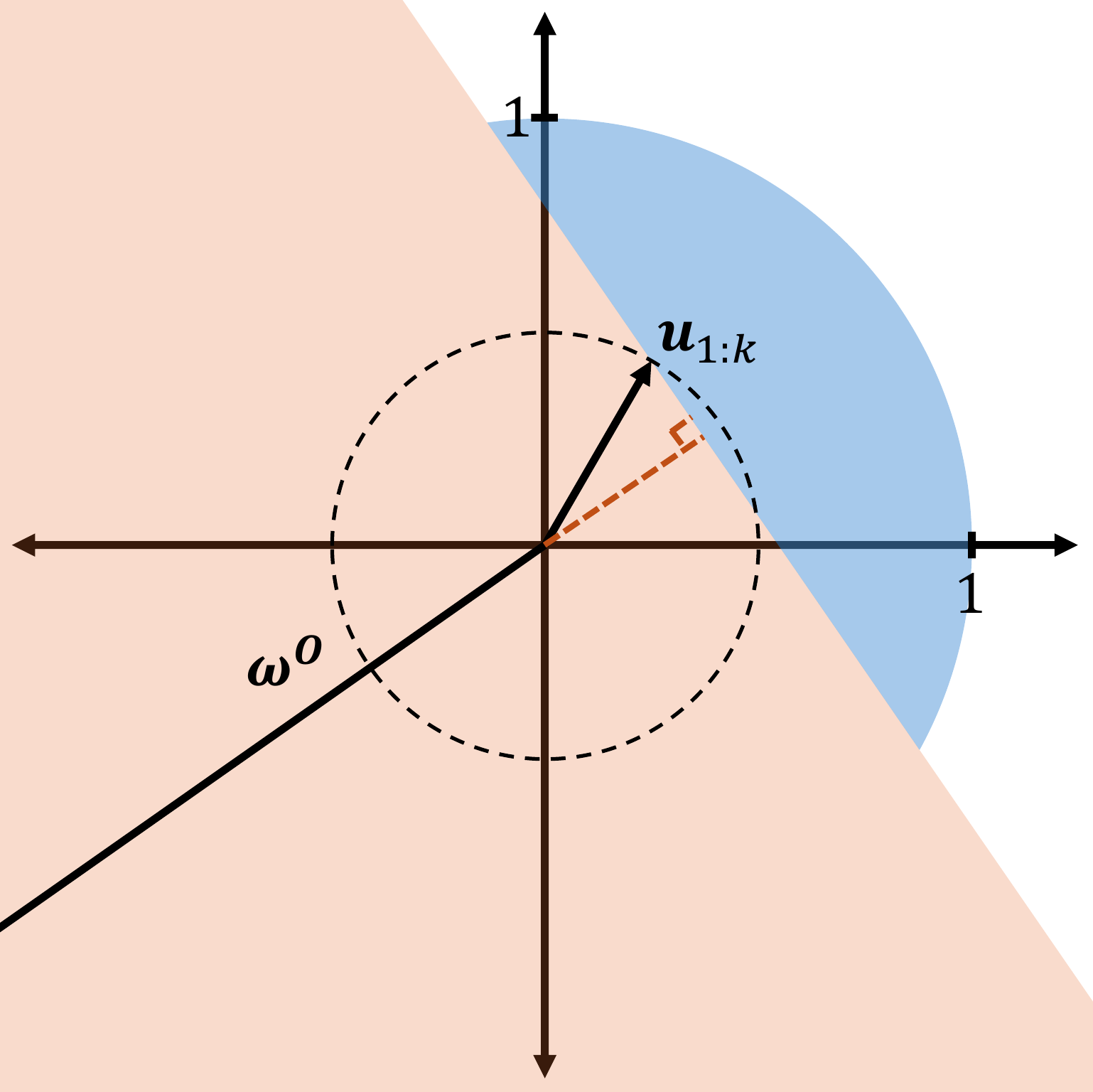}
        \caption{Oracle geometry.}
        \label{fig:pvalue_oracle}
    \end{subfigure}
    
    \begin{subfigure}[t]{0.25\linewidth}
        \centering
        \includegraphics[width=\linewidth]{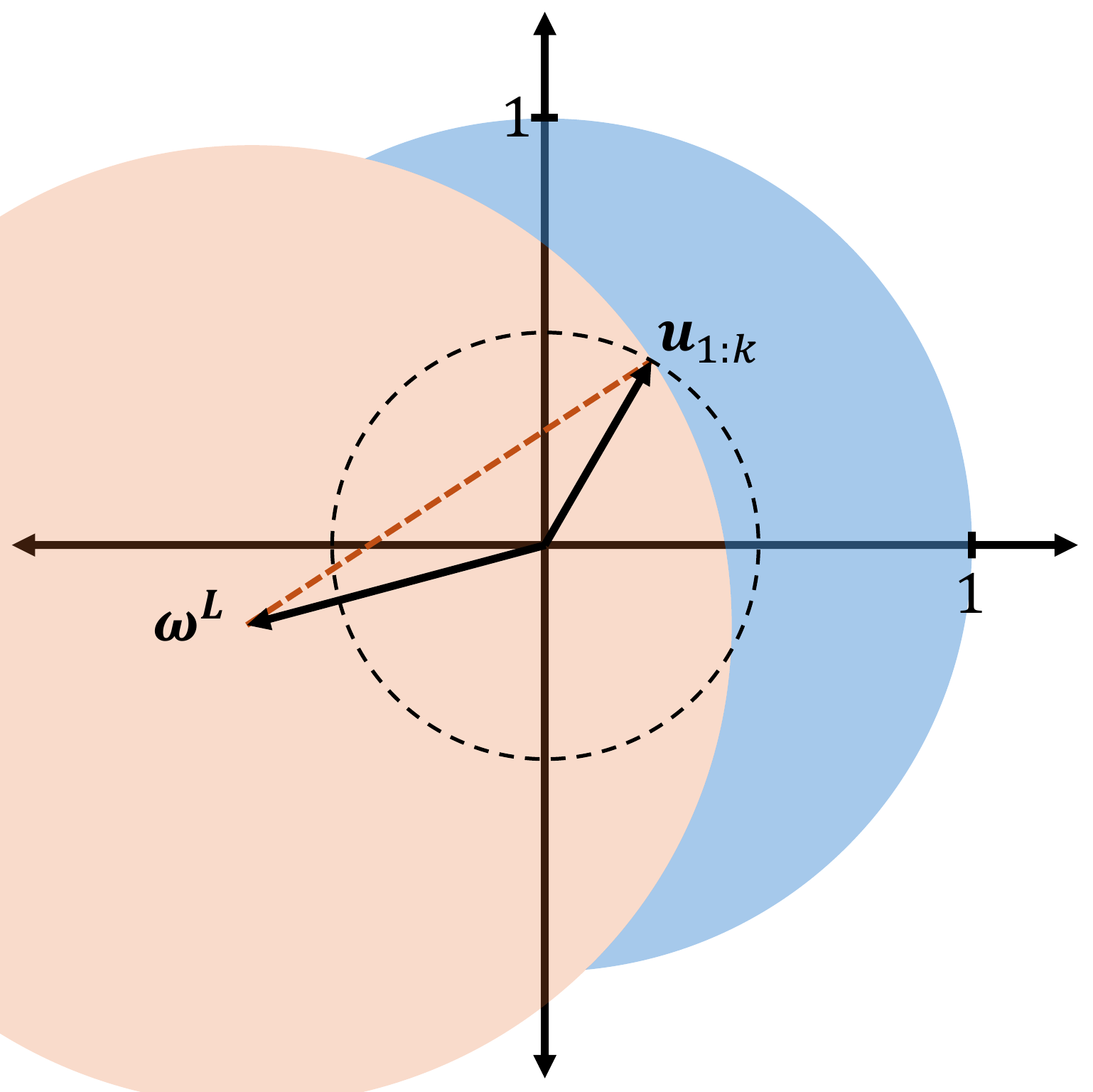}
        \caption{$L$-test geometry when $\bm{X}_{1:k}$ is orthonormal.}
        \label{fig:pvalue_L1}
    \end{subfigure}
    \hspace{0.5cm}
    \begin{subfigure}[t]{0.25\linewidth}
        \centering
        \includegraphics[width=\linewidth]{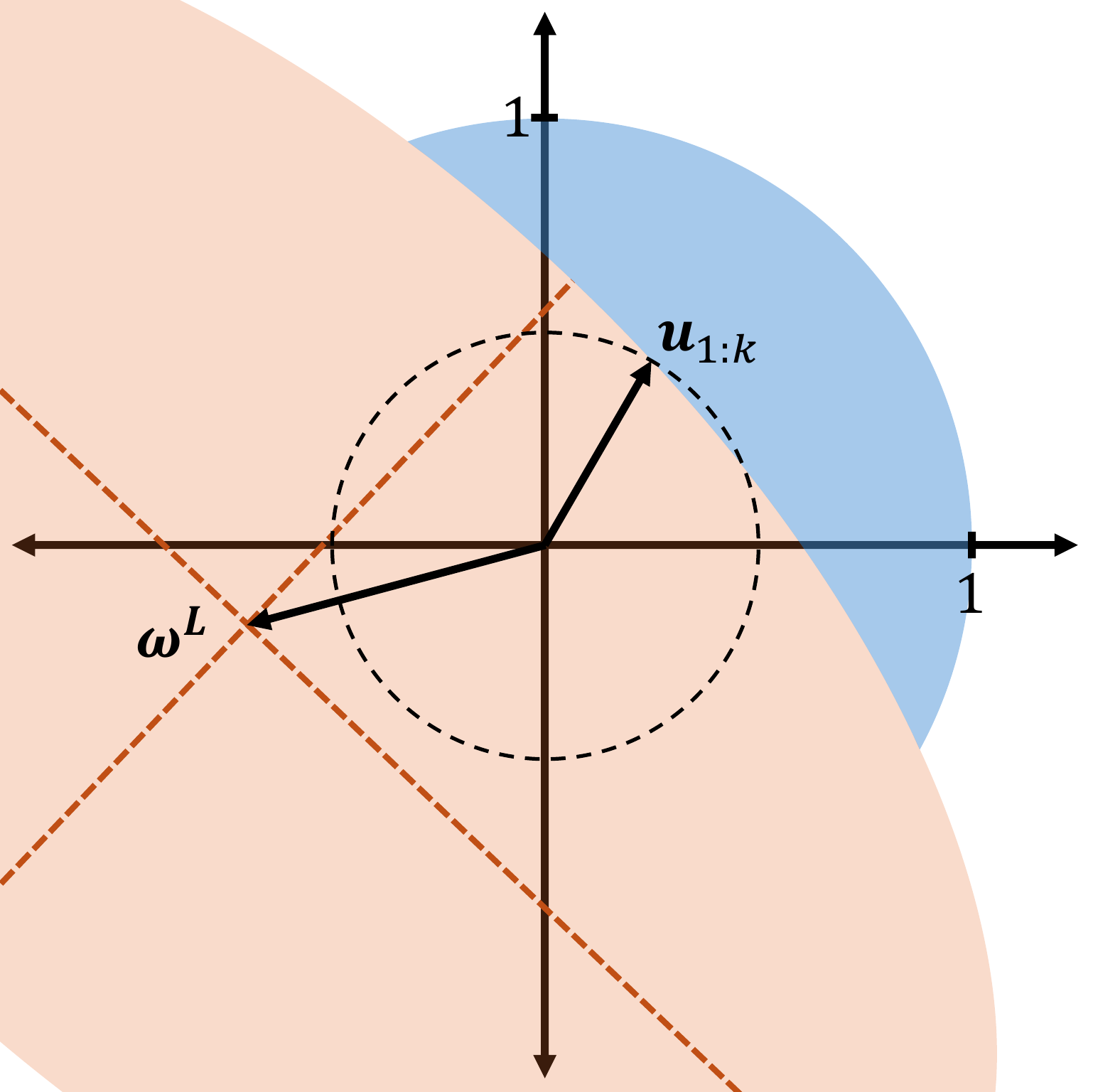}
        \caption{$L$-test geometry when $\bm{X}_{1:k}$ is near rank 1.}
        \label{fig:pvalue_L2}
    \end{subfigure}

    \caption{Schematics representative of the geometries of the oracle, $F$-, and $L$-tests. In all panels, the blue region corresponds to the region whose mass equals the p-value. 
    Note that the distribution of $\tilde{\bm u}_{1:k}$ is not uniform within the unit ball in $\R^k$.}
    \label{fig:p-value_masses}
\end{figure}

Since the $F$-test performs no recentering and takes $\bm{\phi} = \bm{I}$, its p-value corresponds to the mass in the disk region in Figure \ref{fig:pvalue_F}. Meanwhile, the oracle p-value corresponds to the mass of $\tilde{\bm{u}}_{1:k}$ outside the circle centered at $\bm\omega^O = m\bm{V}_{1:k}^T\bm{X}_{1:k}\frac{\bm{\beta}_{1:k}}{\norm{\bm{\beta}_{1:k}}}$ of radius $\norm{\bm{u}_{1:k} - m\bm{V}_{1:k}^T\bm{X}_{1:k}\frac{\bm{\beta}_{1:k}}{\norm{\bm{\beta}_{1:k}}}}$. As $m \to \infty$, this is equivalent to the region outside the half-space shown in Figure \ref{fig:pvalue_oracle}.

The $L$-test p-value corresponds to the mass of $\tilde{\bm{u}}_{1:k}$ lying outside an ellipse centered at $\bm\omega^L=\bm \nu(\glassoc_{1:k})$, which will typically be smaller in magnitude and less anti-aligned with $\bm u_{1:k}$ than $\bm\omega^O$. When variance is distributed relatively uniformly in the span of the first $k$ predictors, the premultiplier essentially reduces to the identity, so the ellipse is simply a circle of radius $\norm{\bm{u}_{1:k} - \bm{\nu}(\glassoc_{1:k})}$ as Figure \ref{fig:pvalue_L1} shows. This case isolates the role of the recentering vector: it enables the $L$-test to leverage auxiliary information about sparsity to infer the direction of $\bm{\beta}_{1:k}$---and, in turn, the opposite direction of $\bm{u}_{1:k}$---emulating the oracle. At the other extreme, when nearly all variance is captured by a single principal component and the true $\bm{\beta}_{1:k}$ aligns with it, the ellipse degenerates into one principal axis roughly aligned with $\bm{u}_{1:k}$ and $k-1$ orthogonal axes of very large magnitude, shown in Figure \ref{fig:pvalue_L2}. We can interpret the principal component direction as a second attempt at inferring the direction of $\bm{\beta}_{1:k}$ since this degenerate geometry only arises if the two are well aligned. Thus, this second case demonstrates how the premultiplier allows the $L$-test to sometimes exploit additional, unsupervised information encoded in the geometry of $\bm{X}_{1:k}$ to further infer the direction of $\bm{\beta}_{1:k}$ and come even closer to approximating the power of the oracle.



\section{Simulation studies}
\label{sec:sim_studies}
In this section, we conduct simulations to evaluate the $L$-test's performance. In all simulations, we use the linear model \eqref{eq:linear_model} and perform inference on the first $k$ coefficients $\bm{\beta}_{1:k}$. The rows of $\bm{X}$ are drawn i.i.d. from $\mathcal{N}(\bm{0}, \bm{\Sigma})$, where $\Sigma_{ij} = \rho^{|i-j|}$ for $\rho \in [0, 1)$, and then the columns are standardized. We randomly choose $k_1$ out of the first $k$ coefficients of $\bm{\beta}$ to be $\pm A/\sqrt{k_1}$ with equal probability and sample $k_2$ out of the remaining $d-k$ independently from a standard Gaussian. The rest of the coefficients are set to zero. All tests are performed at the $\alpha = 0.05$ level. For the implementation of tests that involve Monte Carlo sampling, $M = 200$, unless specified otherwise. Power estimates are computed as averages on $1000$ simulations, and error bars represent plus or minus two standard errors. We always take $\sigma = 1$, and in each of the experiments, we specify the parameters $n, d, k, A, k_1, k_2$, and $\rho$.

\subsection{Power of the $L$-test}
\label{sec:power_plots}
In Figure~\ref{fig:powers_standard}, we compare the powers of the $L$-test (``L-test"), MC-free test from Section \ref{sec:precise_p_values} (``MC-free"), group LASSO MC test from \eqref{eq:mc_p} (``gLASSO"), oracle test (``Oracle"), and $F$-test (``F-test") when we vary the signal of the tested coefficients, sparsity among null and non-null variables, and covariate correlations. For an additional benchmark, we also consider a naive extension of the $\ell$-test from \cite{Sengupta2024} that applies it separately to each of the tested coefficients with a Bonferroni correction (``Bonf-$\ell$").

The $L$-test significantly outperforms the $F$-test in sparse settings under various signal strengths and covariate correlations. As expected, its performance is insensitive to sparsity among null coefficients (because the group LASSO does not regularize towards sparsity within a group) but weakens when sparsity increases among non-null coefficients. As discussed in Section~\ref{sec:oracle_test}, since $k=10$, the $L$-test exhibits a large power gap with the oracle but outperforms the $F$-test by a substantial margin, and that margin is larger than when $k=1$. For instance, the maximal power gain of the $L$-test in the top left panel (13.8 percentage points) is about 30\% higher than that of the $\ell$-test in its paper's analogous $k=1$ simulations (10.7 percentage points) \citep{Sengupta2024}.
\begin{figure}[ht]
    \centering
    \begin{subfigure}[t]{0.4\linewidth}
        \centering
        \includegraphics[width=\linewidth]{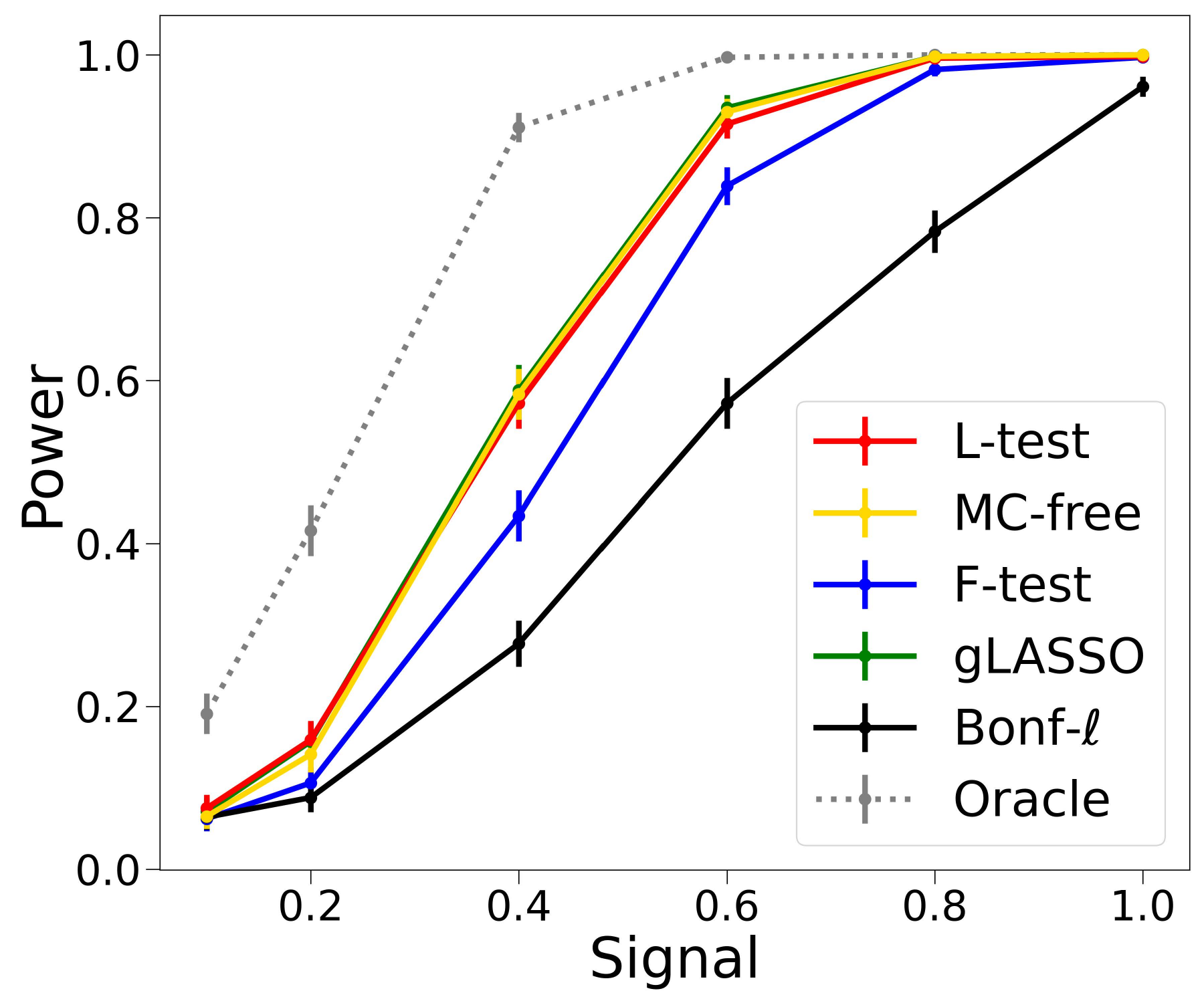}
    \end{subfigure}
    \hspace{0.2cm}
    \begin{subfigure}[t]{0.4\linewidth}
        \centering
        \includegraphics[width=\linewidth]{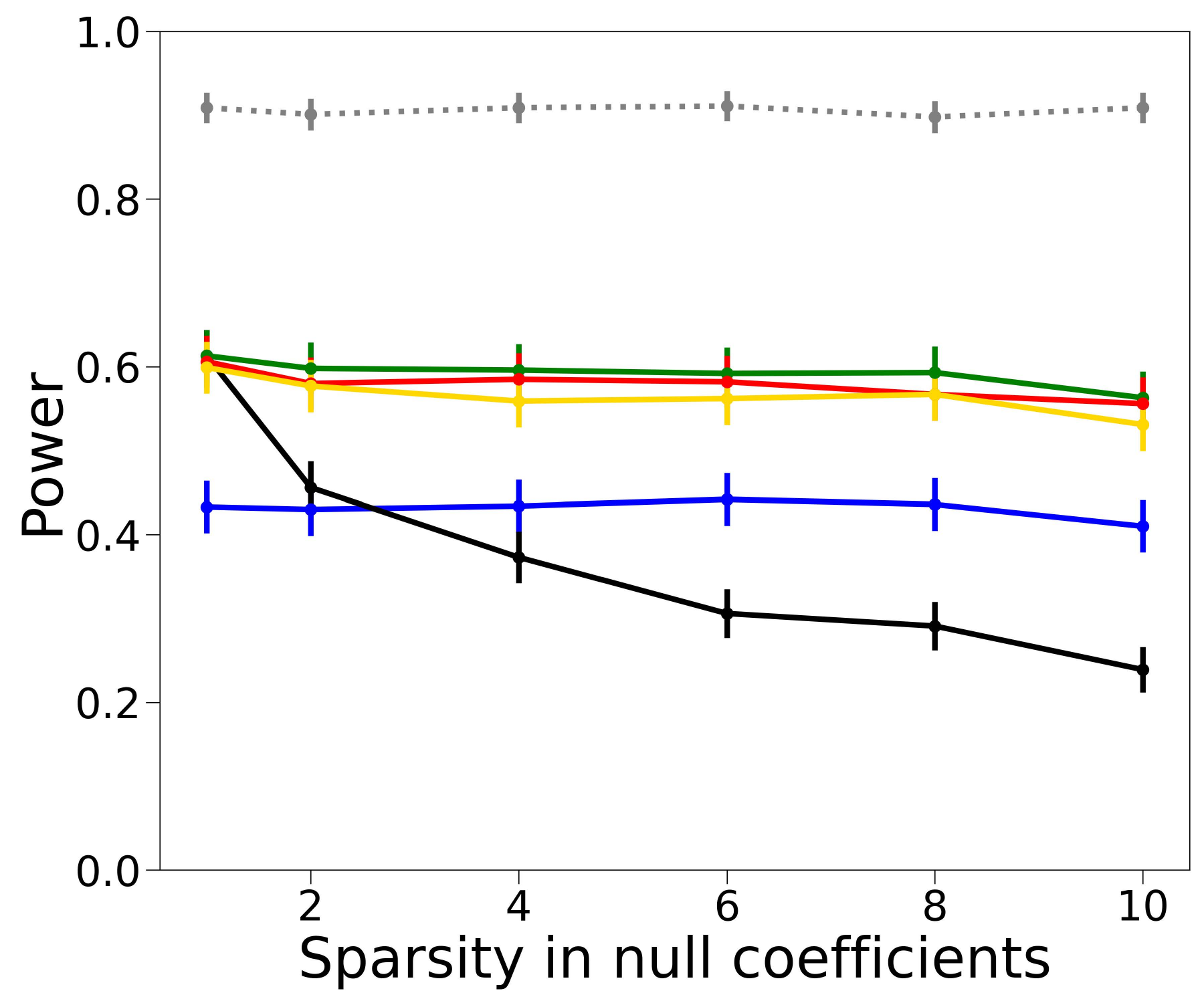}
    \end{subfigure} \\
    \begin{subfigure}[t]{0.4\linewidth}
        \centering
        \includegraphics[width=\linewidth]{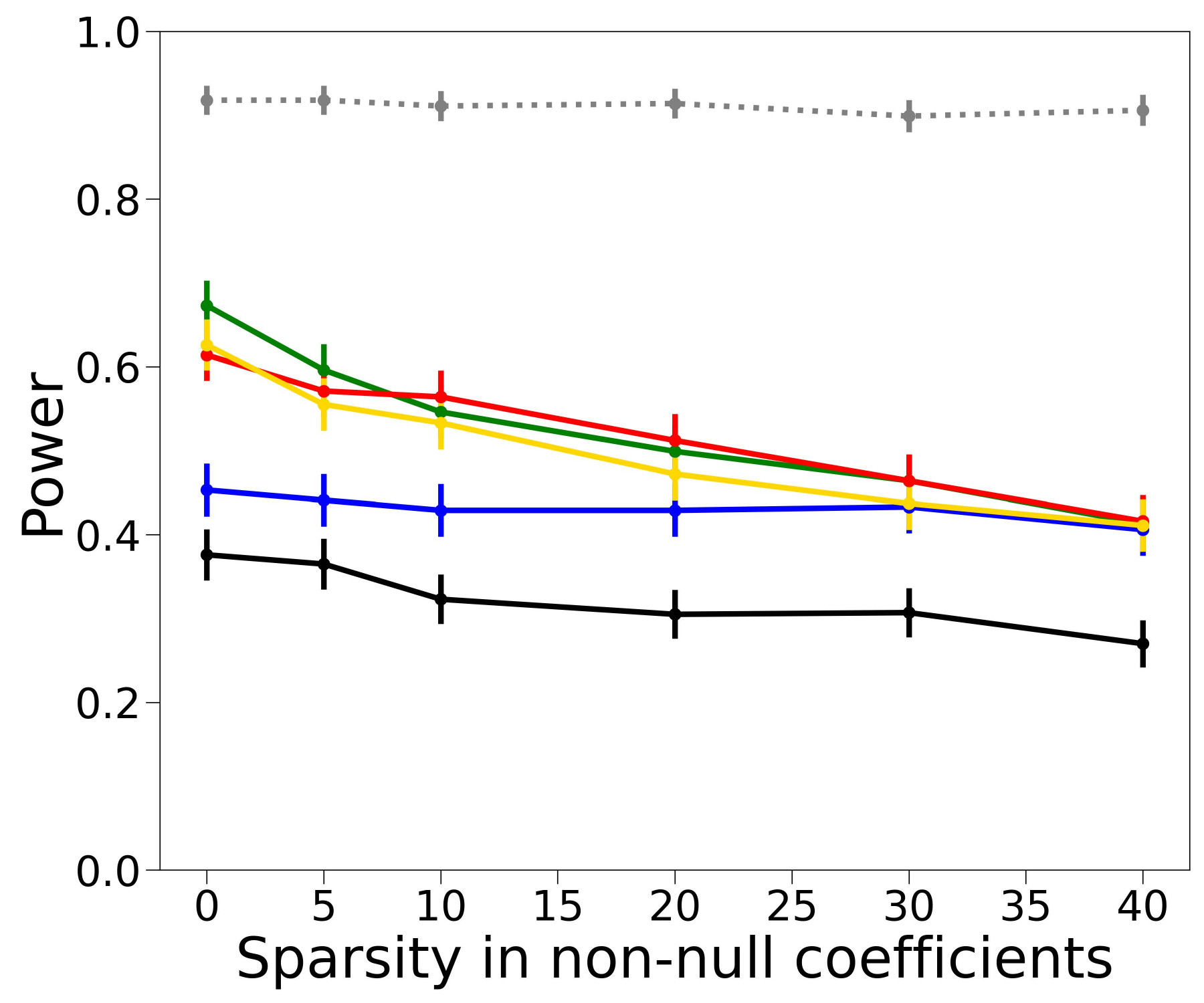}
    \end{subfigure}
    \hspace{0.2cm}
    \begin{subfigure}[t]{0.4\linewidth}
        \centering
        \includegraphics[width=\linewidth]{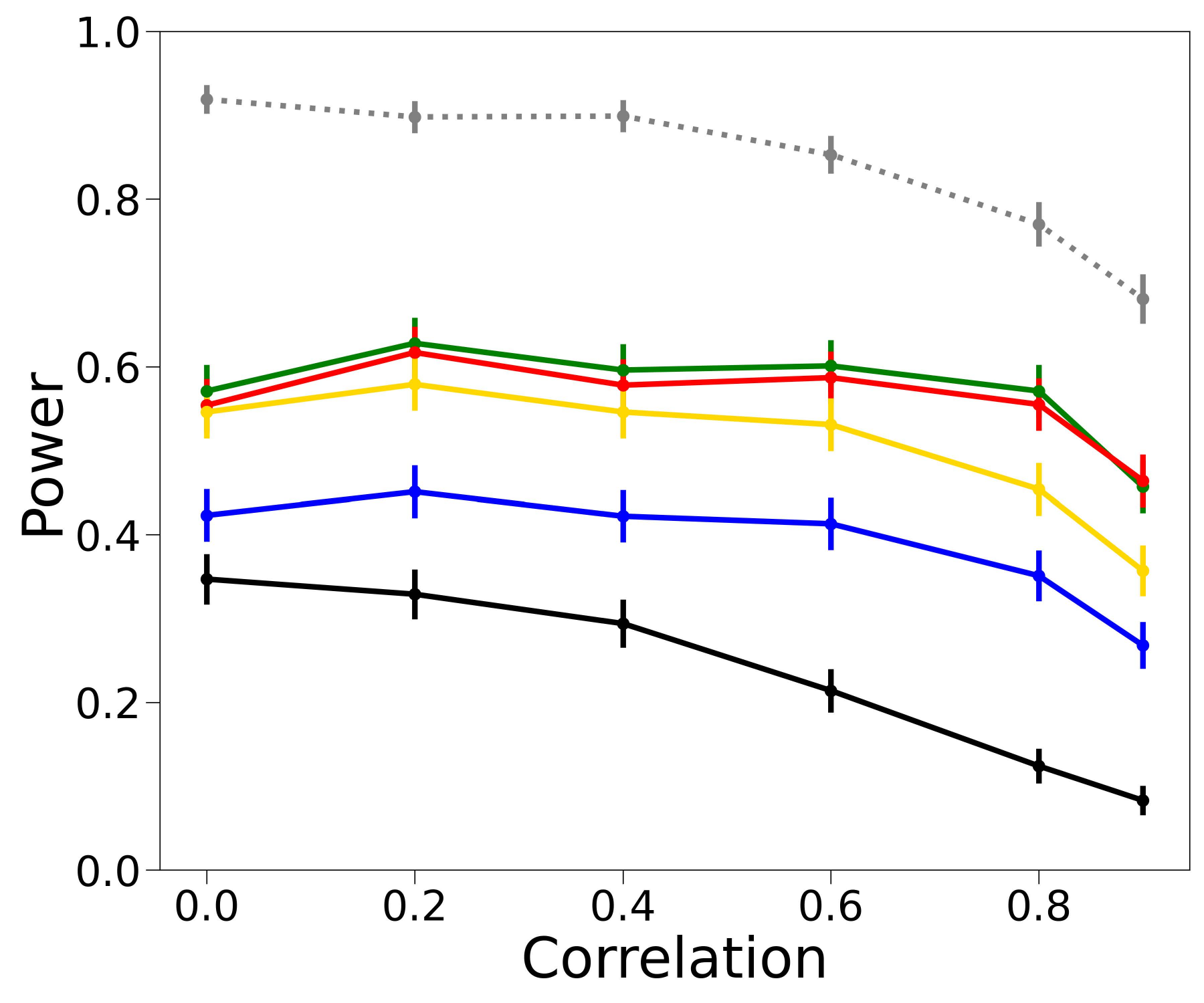}
    \end{subfigure}
    \caption{Comparison of test powers. In all panels, $n = 100$, $d = 50$, and $k = 10$. The top left takes $k_1 = k$, $k_2 = 4$, and $\rho = 0$ and varies $A$; top right takes $A = 0.4$, $k_2 = 4$, and $\rho = 0$ and varies $k_1$; bottom left takes $A = 0.4$, $k_1 = 4$, and $\rho = 0$ and varies $k_2$; bottom right takes $A = 0.4$, $k_1 = 4$, and $k_2 = 4$ and varies $\rho$.}
    \label{fig:powers_standard}
\end{figure}

In all settings, the $L$-test closely tracks the power of the group LASSO MC test (while, of course, being far more computationally efficient).
The MC-free test performs similarly to the $L$-test except when correlations are high, since variance concentrates in a few directions and the premultiplier $A(\glassoc_{1:k}, \lambda)$ exploits this geometry to give the $L$-test an additional boost; a similar effect occurs when dimensionality increases (see Figure \ref{fig:powers_high} in Appendix \ref{sec:power_plots_add}). The Bonferroni-$\ell$-test generally performs substantially worse than all other tests, approaching our methods' performance (and outperforming the $F$-test)  only in the special case where the null is 1-sparse. Similar trends are evident in larger models, shown in Figure \ref{fig:powers_large} of Appendix \ref{sec:power_plots_add}.

\subsection{Block-orthogonal design setting}
\label{sec:power_orthog}
To confirm our understanding of the premultiplier $\bm{A}(\glassoc_{1:k}, \lambda)$ in the $L$-test empirically, we constructed a related procedure, the PC-test, that leverages a similar geometry in $\bm{X}_{1:k}$ for power. This test selects the top $r$ PC directions explaining $85\%$ of the variance in $\bm{X}_{1:k}$ and then performs an $F$-test on the retained directions.

To isolate the impact of $\bm{A}(\glassoc_{1:k}, \lambda)$, we orthogonalized $\bm{X}_{1:k}$ and $\bm{X}_{-1:k}$ so the recentering $\bm \nu(\glassoc_{1:k})$ vanishes, and then assessed the $L$-test's similarity to the PC-test. As the left panel in Figure \ref{fig:PC_vary_corr} shows, both outperform the $F$-test as $\rho$ increases and variance concentrates in a few directions. Since both rely on alignment between $\bm{\beta}_{1:k}$ and these PCs, they underperform when $\bm{\beta}_{1:k}$ is anti-aligned (middle panel) but improve under maximum alignment (left vs. right panels). Further experiments in Appendix \ref{sec:power_plots_add} show that the MC-free test is more robust in the anti-aligned regime since it omits the premultiplier.
\begin{figure}[ht]
  \centering
  \begin{subfigure}[t]{0.3\textwidth}
    \centering
    \includegraphics[width=\textwidth]{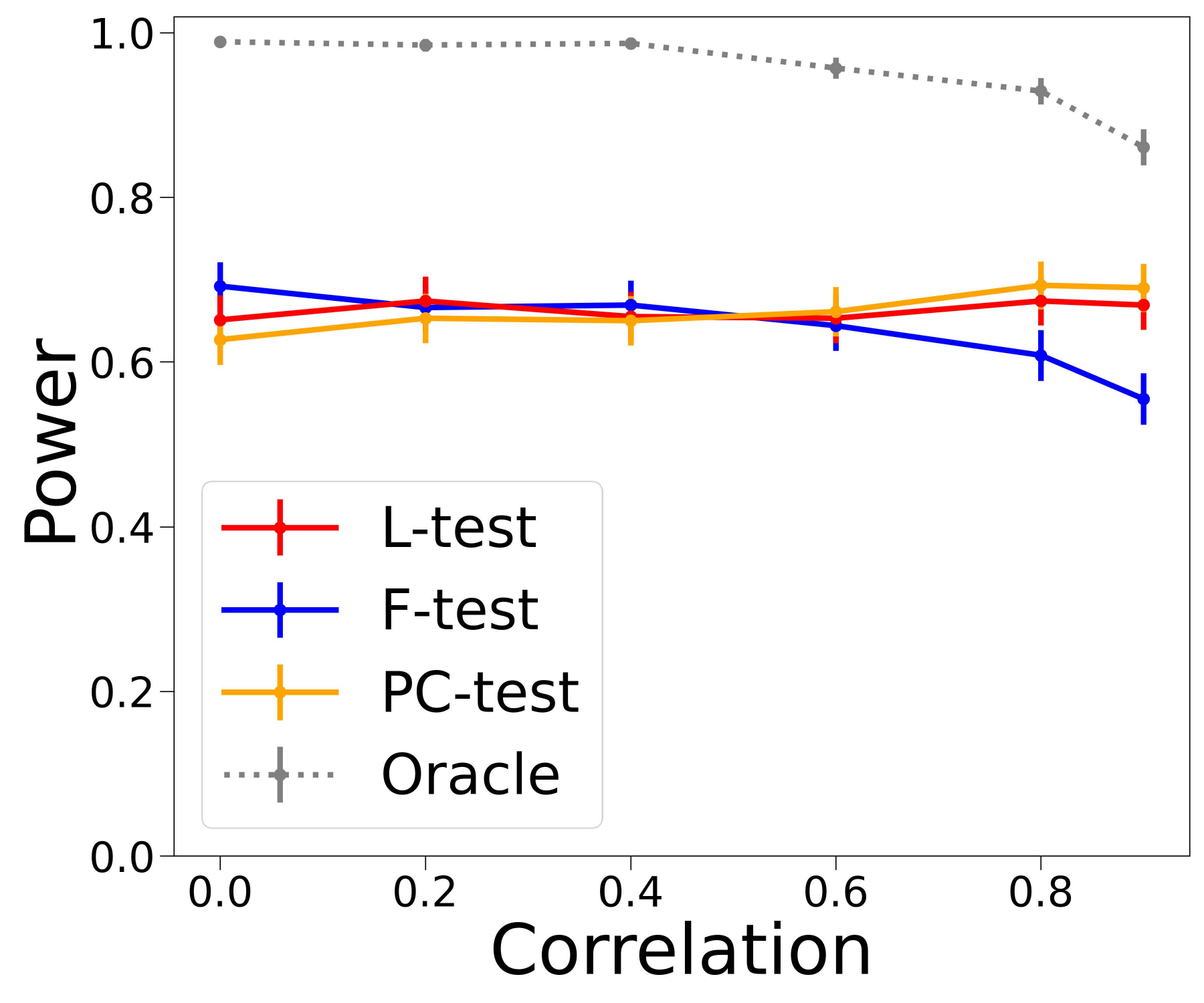}
  \end{subfigure}
  \hspace{0.2cm}
  \begin{subfigure}[t]{0.3\textwidth}
    \centering
    \includegraphics[width=\textwidth]{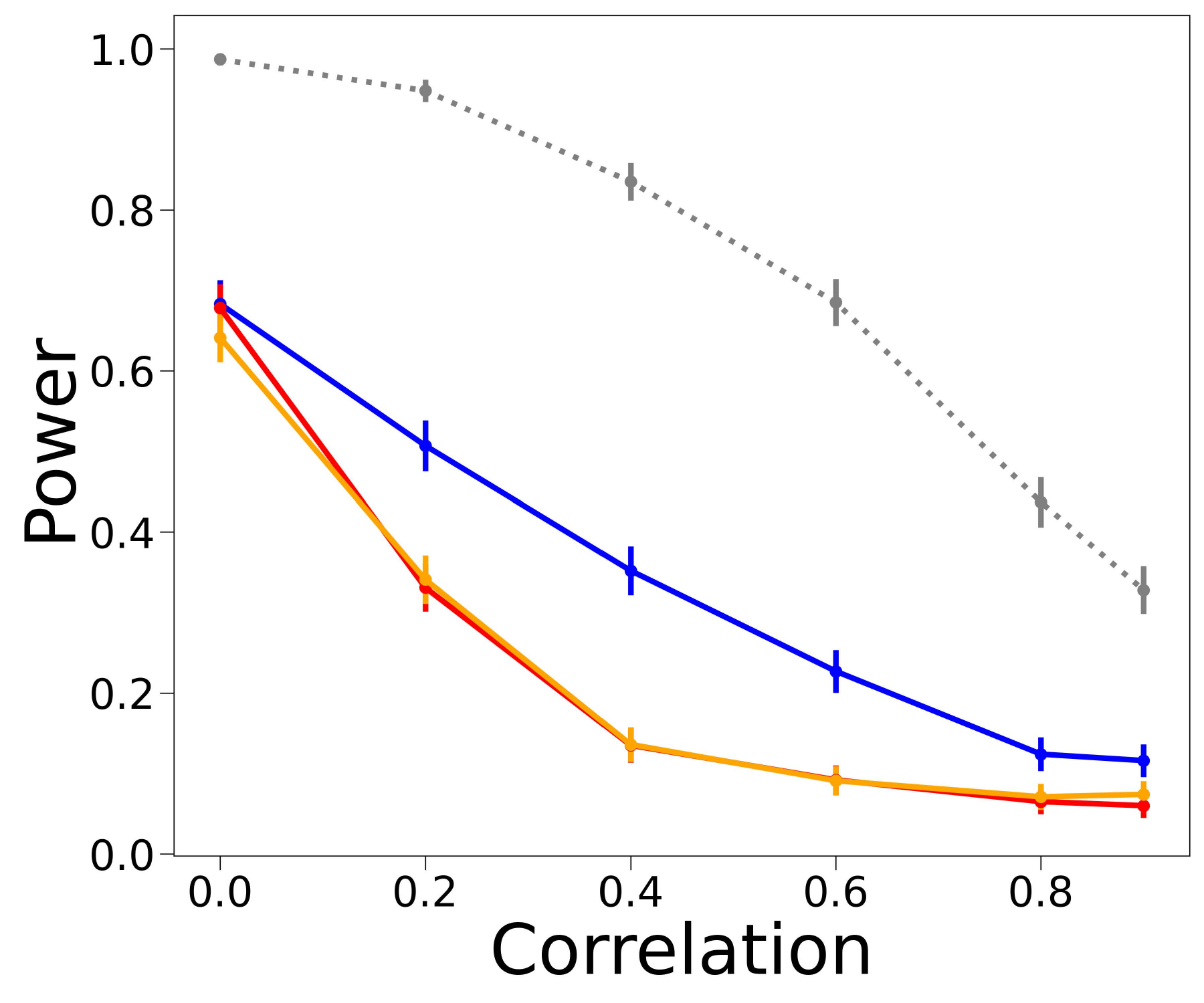}
  \end{subfigure}
  \hspace{0.2cm}
  \begin{subfigure}[t]{0.3\textwidth}
    \centering
    \includegraphics[width=\textwidth]{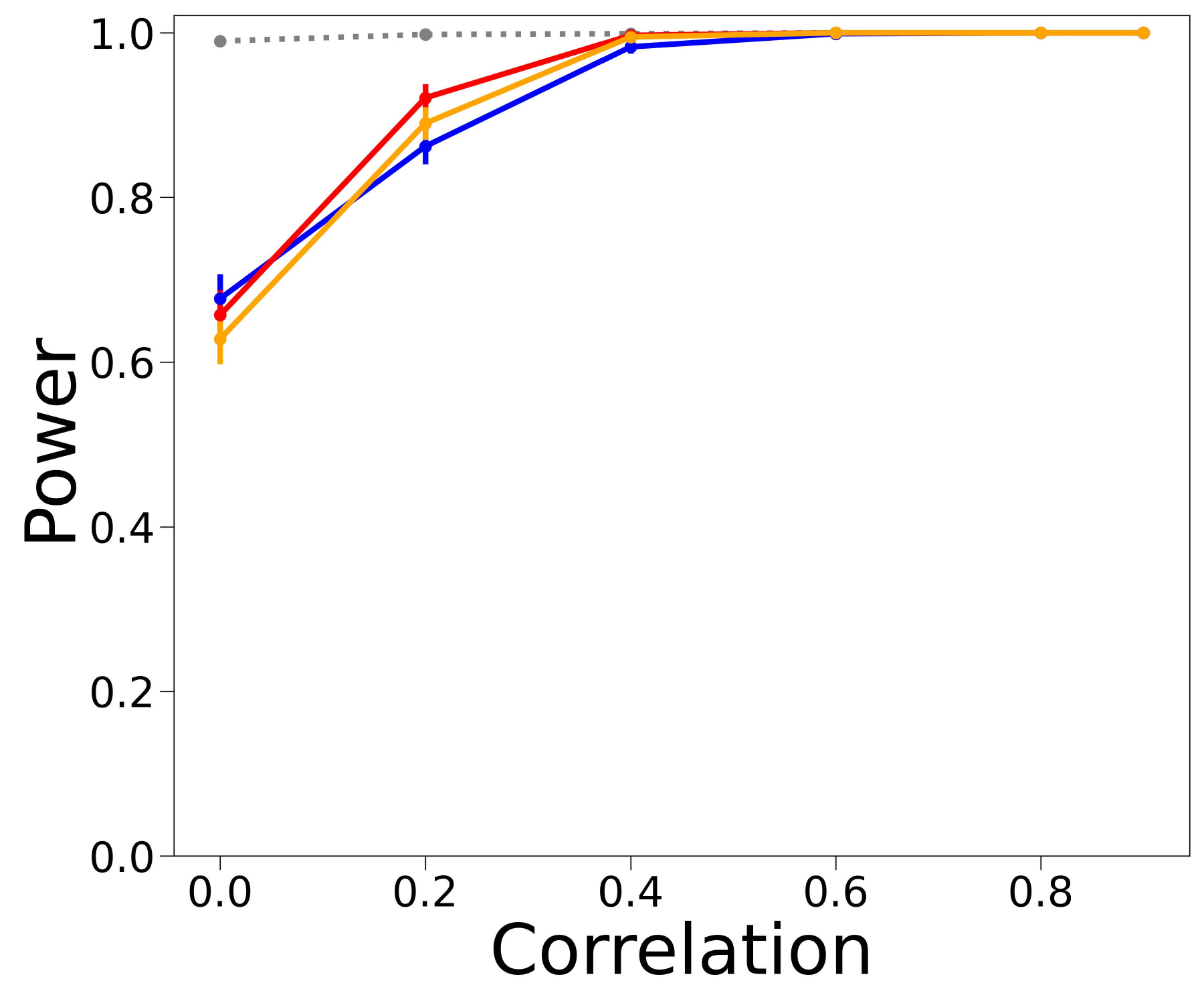}
  \end{subfigure}
  \caption{Similarity between $L$-test and PC-test under block-orthogonal design. In all panels, we orthogonalize $\bm{X}_{1:k}$ to $\bm{X}_{-1:k}$ by using the design $\begin{pmatrix}
      \bm{X}_{1:k} & (\bm{I} - \bm{P}_{1:k})\bm{X}_{-1:k}
  \end{pmatrix}$, followed by standardization. The left panel has the same setting as the bottom right panel of Figure \ref{fig:powers_standard}. The middle panel changes only $\bm \beta$ to be dense with alternating signs, giving maximum anti-alignment with $\bm{1}_k$, the top PC under maximum correlation; the right panel instead uses a dense, nonnegative $\bm \beta$, yielding maximum alignment.}
  \label{fig:PC_vary_corr}
\end{figure}


\subsection{Robustness}
\label{sec:robustness}
In this section, we assess the robustness of the $L$-test's validity to linear model violations, similar to the approach taken in \cite{Sengupta2024}. Four different types of violations are considered: heavy-tailed errors, skewed errors, heterosckedasticity, and non-linearity. In all of the experiments, we consider various sample sizes $n$ and take $d = n/2$, $k = 3$, $k_1 = 0$, $k_2 = 1$, and $\rho = 0$ or $\rho = 0.5$. Figure \ref{fig:robustness} shows that the $L$-test exhibits similar robustness to the $F$-test for each model violation considered when $\rho = 0$. Appendix \ref{sec:robustness_add} presents additional settings for each type of model violation---with Figure \ref{fig:robustness} showing the most extreme setting for each---and also performs the same experiments when $\rho = 0.5$.
\begin{figure}[ht]
    \centering
    \begin{subfigure}[t]{0.4\linewidth}
        \centering
        \includegraphics[width=\linewidth]{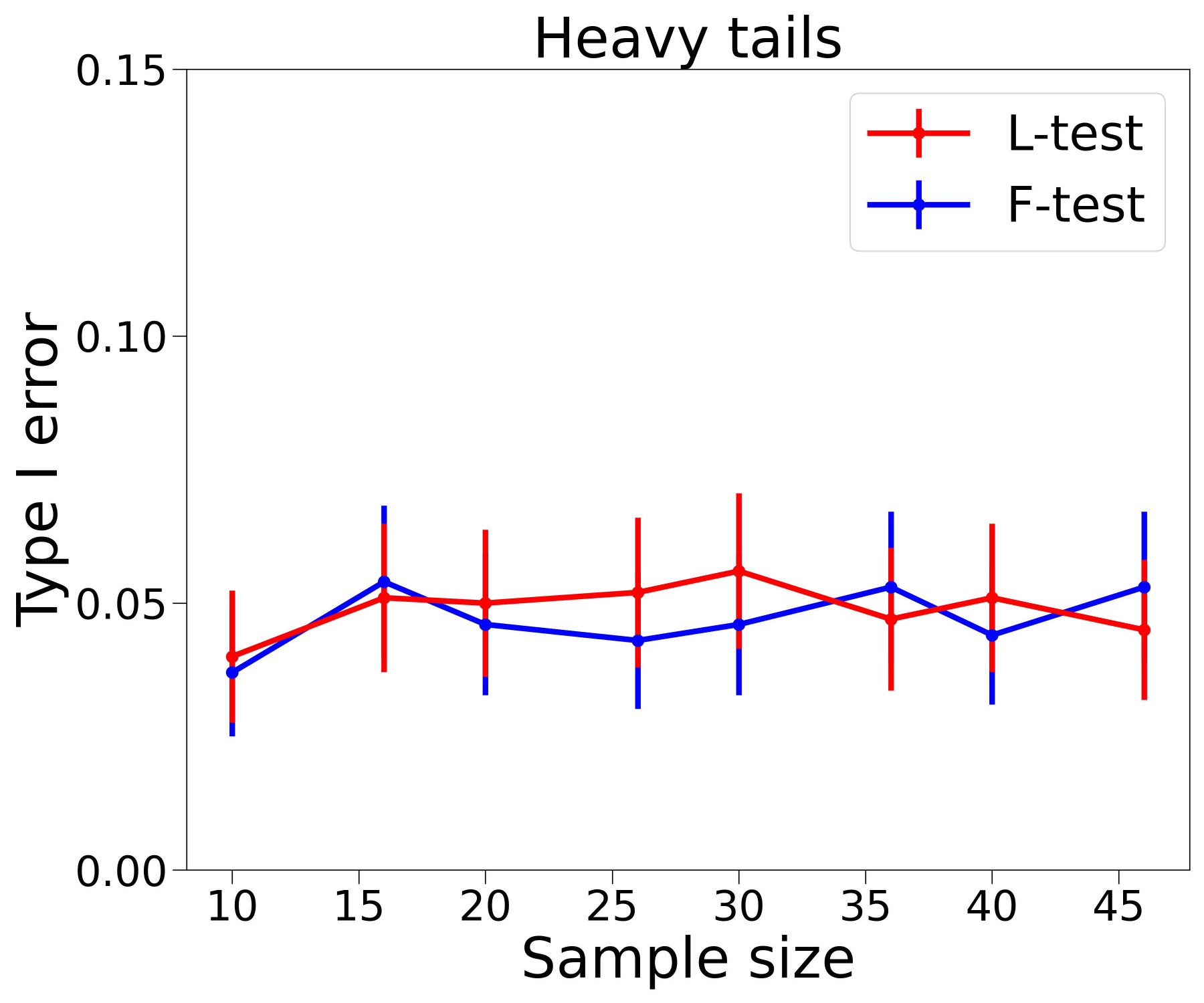}
    \end{subfigure}
    \hspace{0.2cm}
    \begin{subfigure}[t]{0.4\linewidth}
        \centering
        \includegraphics[width=\linewidth]{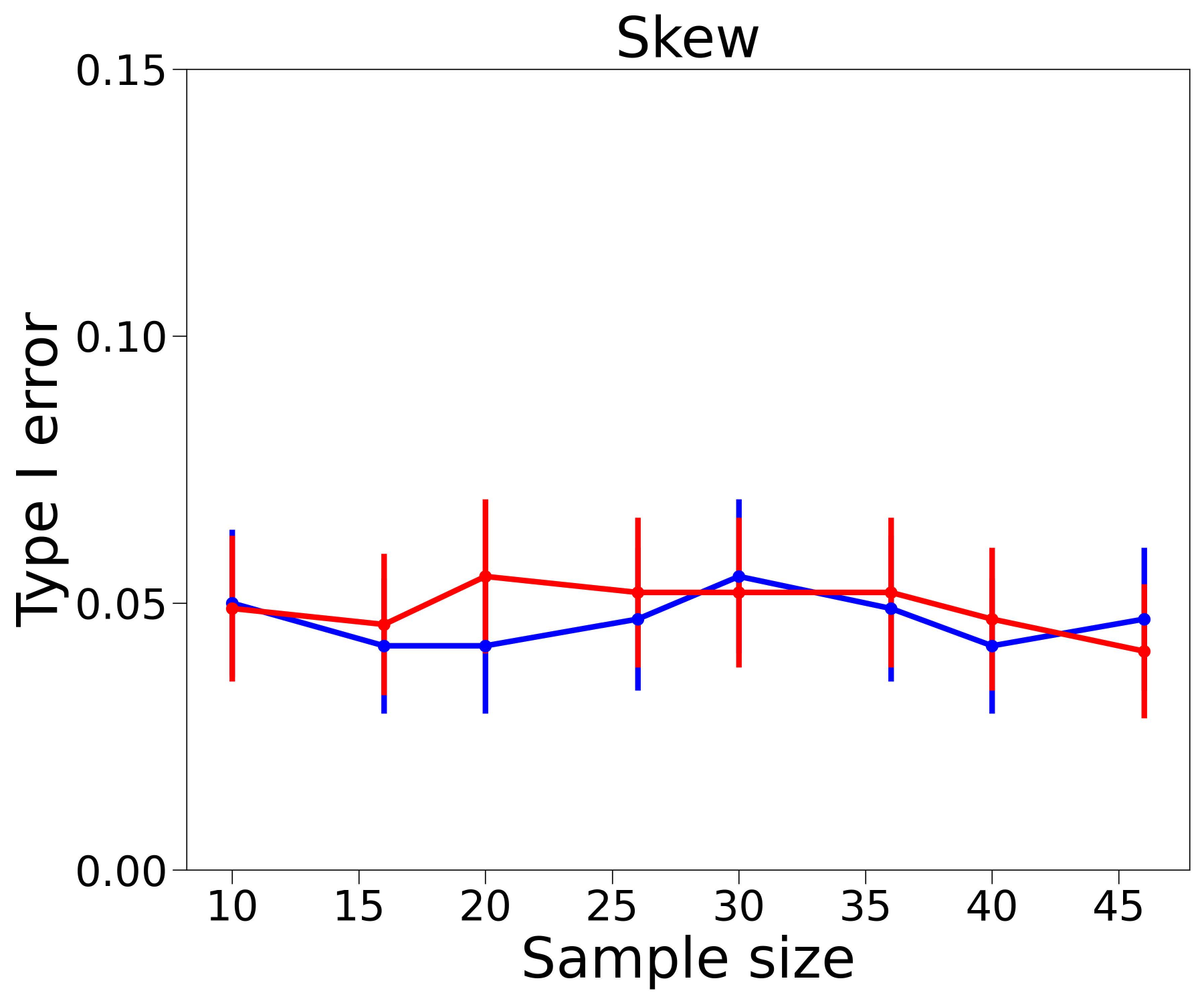}
    \end{subfigure} \\
    \begin{subfigure}[t]{0.4\linewidth}
        \centering
        \includegraphics[width=\linewidth]{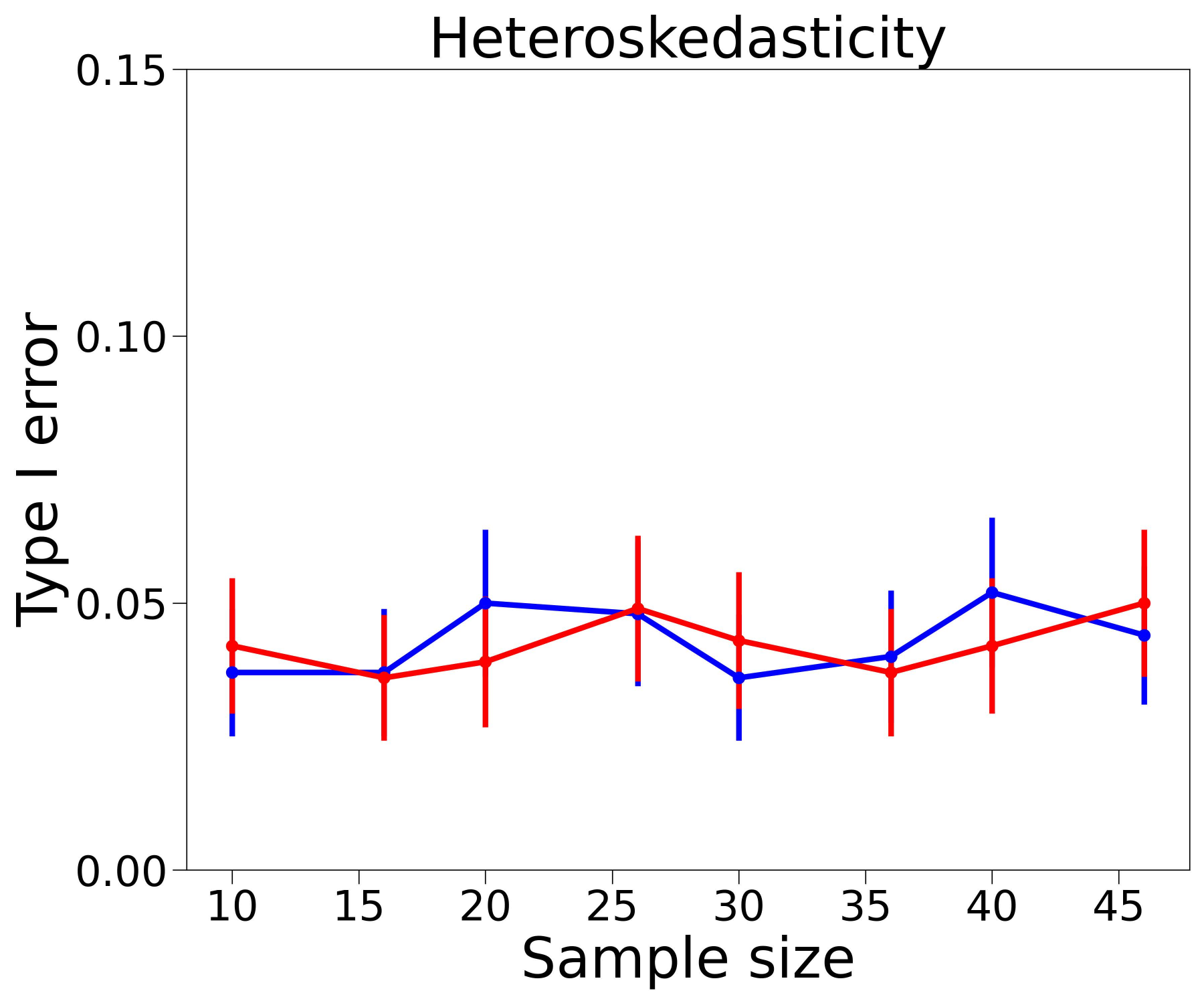}
    \end{subfigure}
    \hspace{0.2cm}
    \begin{subfigure}[t]{0.4\linewidth}
        \centering
        \includegraphics[width=\linewidth]{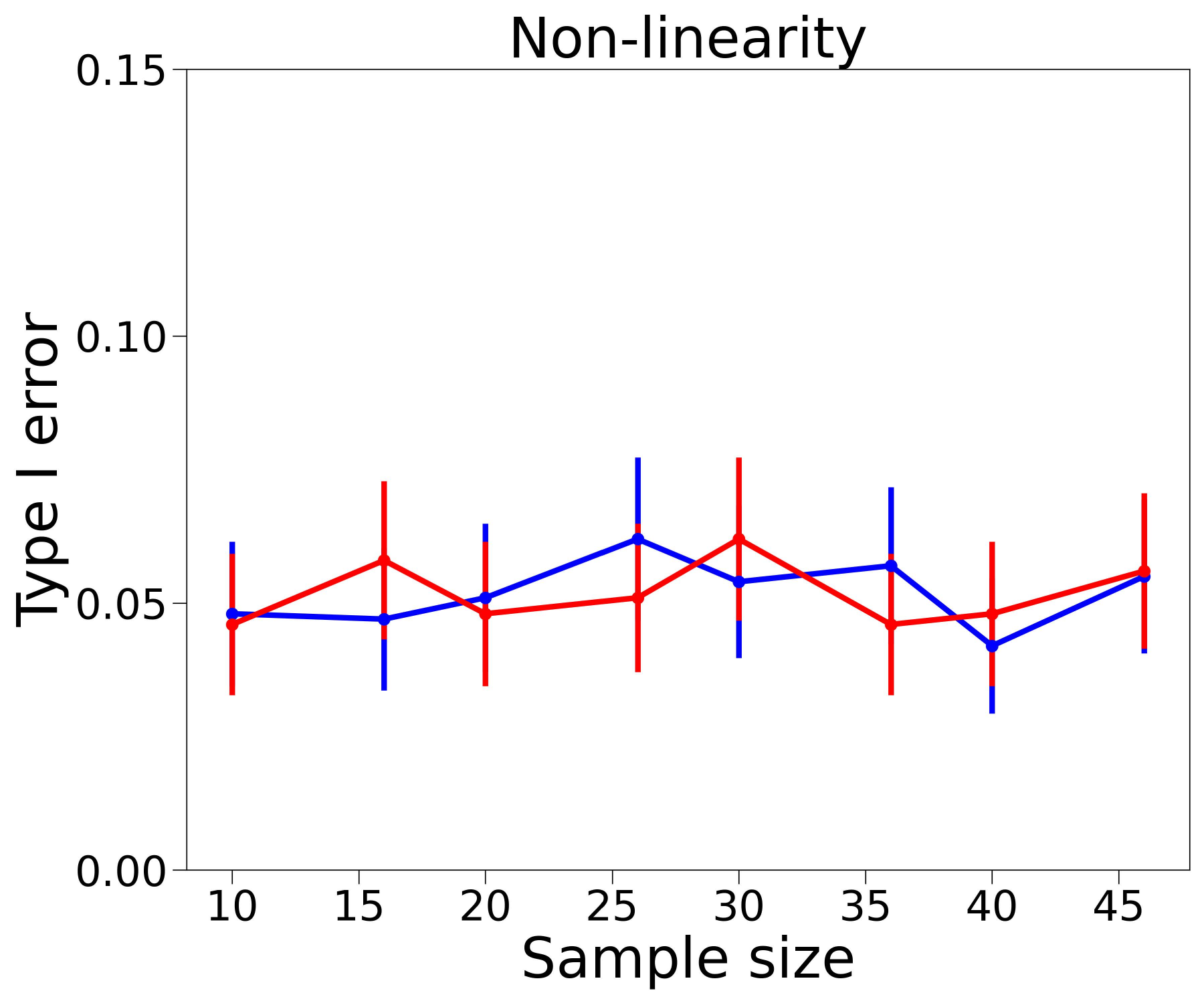}
    \end{subfigure}
    \caption{$F$-test and $L$-test Type I errors for $H_{1:k}$ under four different violations of model \eqref{eq:linear_model} when $\rho = 0$. In the top left panel, $\epsilon_i \overset{i.i.d.}{\sim} t_2$. In the top right, $\epsilon_i \overset{i.i.d.}{\sim} \Expo(1)$ standardized with its theoretical mean and std. In the bottom left, $\epsilon_i \overset{i.i.d.}{\sim} \mathcal N(0, 1)$ if the $i$-th row mean of $\bm{X}$ is less than 0 and $\epsilon_i \overset{i.i.d.}{\sim} \mathcal{N}(0, 8)$ otherwise. In the bottom right, the design matrix has $(i,j)$-th entry $\sign(X_{i, j})|X_{i, j}|^4$.}
    \label{fig:robustness}
\end{figure}

\section{Application to HIV drug resistance data}
\label{sec:application}

To demonstrate the power gains of our methodology on real data, we apply it to the task of discovering genes in the Human Immunodeficiency Virus Type 1 (HIV-1) that are associated with drug resistance. The data set from \cite{Rhee2006} consists of drug resistance measurements (the response) and genetic mutations (the covariates). It contains 16 data sets, each corresponding to a different drug.

We followed the same pre-processing steps outlined in \cite{Barber2015}, resulting in 16 regressions with $n$ ranging between 328 and 842 and $d$ ranging between 147 and 313. For each of the 16 regressions, we aggregated mutations to the gene level by grouping them based on their positions. Across the regressions, the group size $k$ ranged from 1 to 11, and the $F$-test, $L$-test, and MC-free test were performed for $\hyp$. In the single testing analysis, we also performed the group LASSO MC test in \eqref{eq:mc_p} to provide a benchmark for comparison.

\subsection{Single testing}
All tests were performed at the $\alpha = 0.05$ level, and the leftmost panel in Figure \ref{fig:single_testing} shows that the $F$-test and $L$-test had average powers of 24.30\% and 25.88\%, respectively, a 6.50\% improvement, with the largest power gains occurring on moderate to larger group sizes. Although the group LASSO MC test performed slightly better than the $L$-test overall with an average power of 26.45\%, this is mainly because it performed better on smaller group sizes (overrepresented in the middle panel); the right panel shows their performances were actually quite similar, indicating the $L$-test is a good approximation even in real data. The MC-free test demonstrated the best performance with an average power of 27.02\% (a 11.19\% improvement).

\begin{figure}[ht]
  \centering
  \begin{subfigure}[t]{0.3\textwidth}
    \centering
    \includegraphics[width=\textwidth]{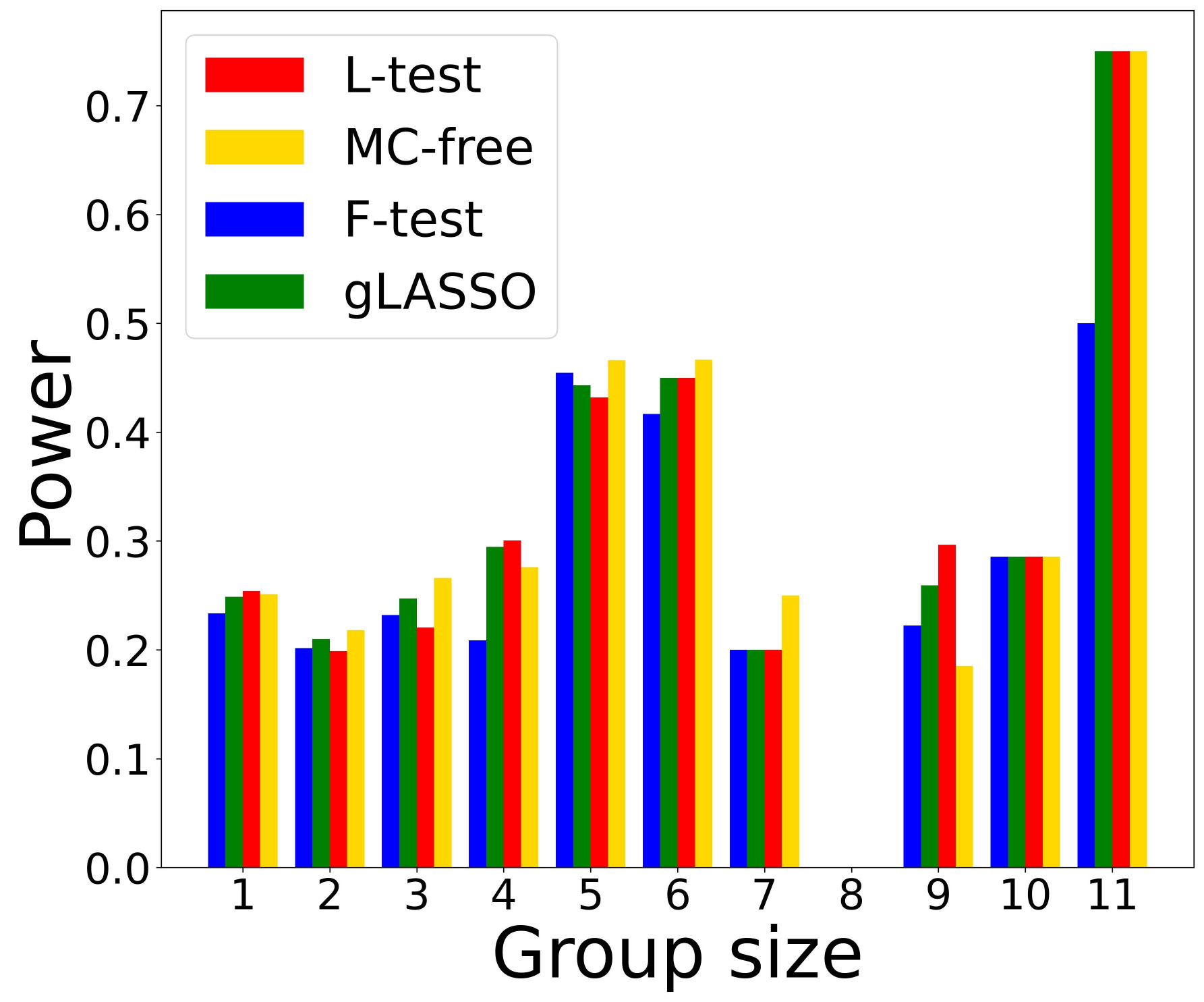}
  \end{subfigure}
  \hspace{0.2cm}
  \begin{subfigure}[t]{0.3\textwidth}
    \centering
    \includegraphics[width=\textwidth]{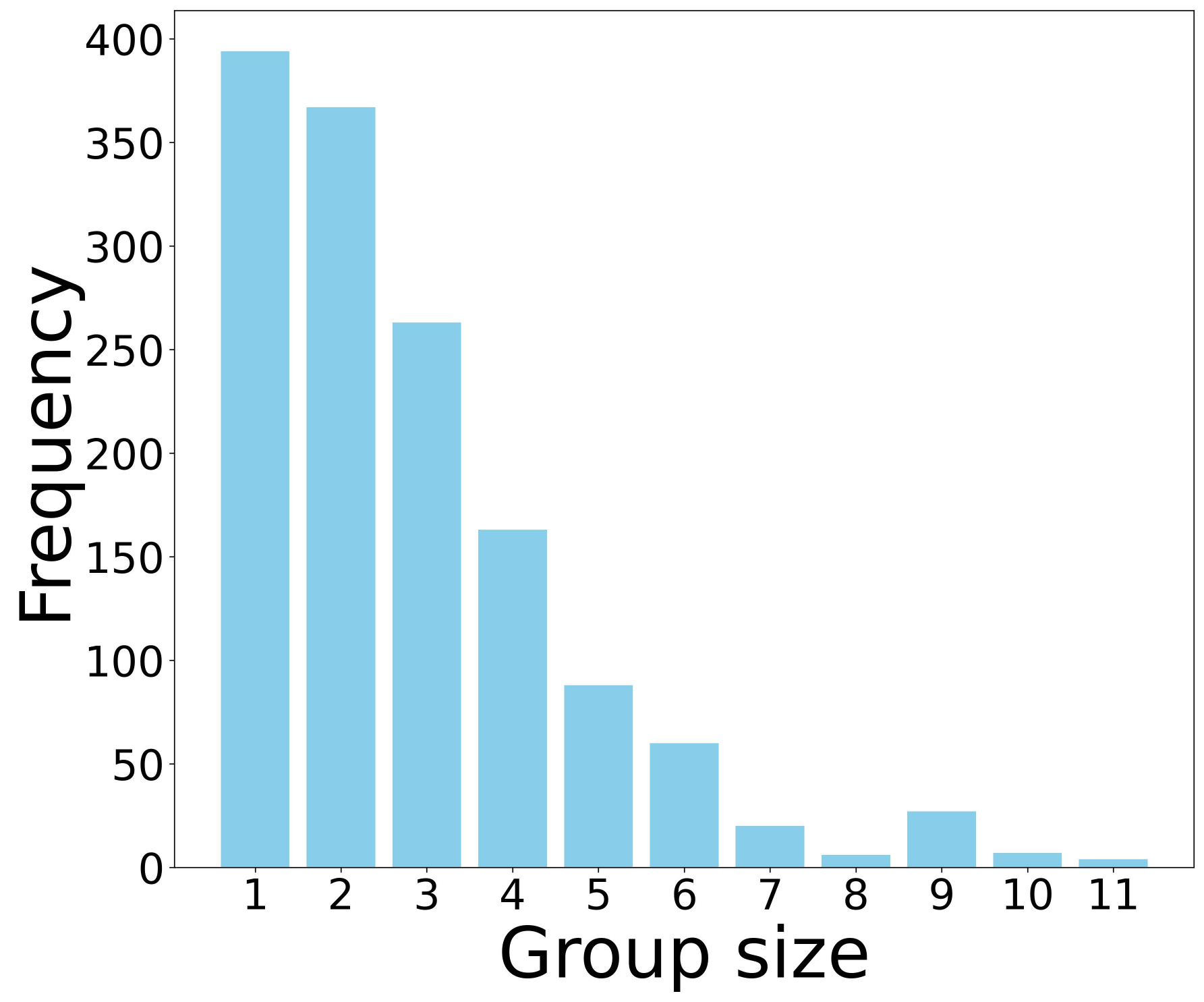}
  \end{subfigure}
  \hspace{0.2cm}
  \begin{subfigure}[t]{0.3\textwidth}
    \centering
    \includegraphics[width=\textwidth]{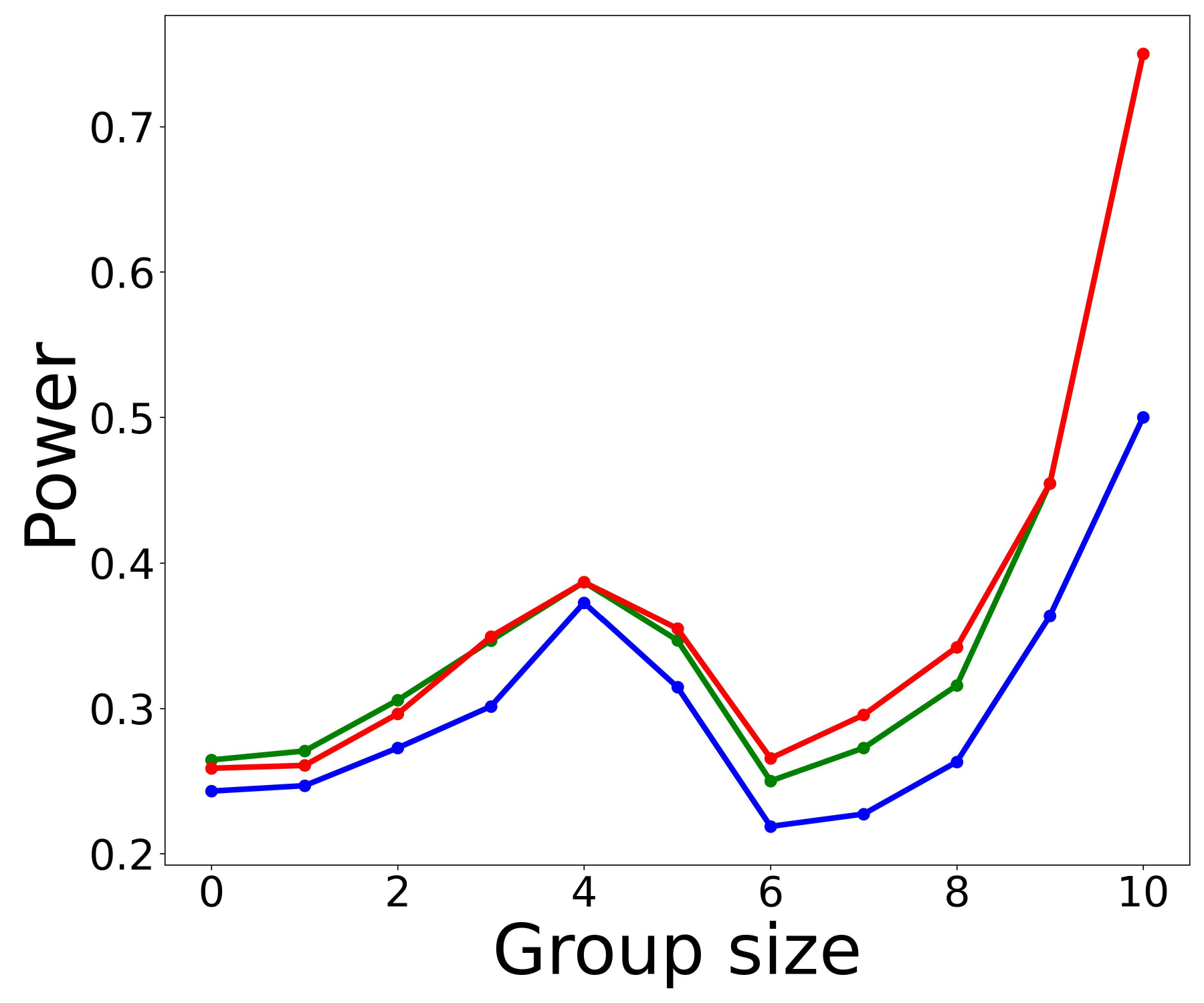}
  \end{subfigure}
  \caption{The left panel shows average powers by null group size across all 16 regressions. The middle panel shows the number of tests sorted by null group size across the regressions. The right panel shows average powers across all hypothesis tests that included subsets of coefficients of at least the size indicated on the $x$-axis.}
  \label{fig:single_testing}
\end{figure}

\subsection{Multiple testing}
To demonstrate the efficiency benefit of our MC-free test, we applied two multiple testing procedures: Holm at 5\% and Benjamini--Hochberg (BHq) at 10\%. As shown in the left panel of Figure \ref{fig:multiple_testing}, achieving a power gain with the $L$-test over the $F$-test with the Holm procedure requires more than $10^4$ MC samples. Fewer MC samples suffices to obtain more power than the $F$-test with BHq, but even a similar number of samples is needed for the $L$-test to reach maximum power. This added precision, however, comes at the cost of substantially higher computation time, as the right panel shows. In contrast, the MC-free test delivers the precision needed for high power at computation times an order of magnitude faster.

\begin{figure}[ht]
  \centering
  \begin{subfigure}[t]{0.4\textwidth}
    \centering
    \includegraphics[width=\textwidth]{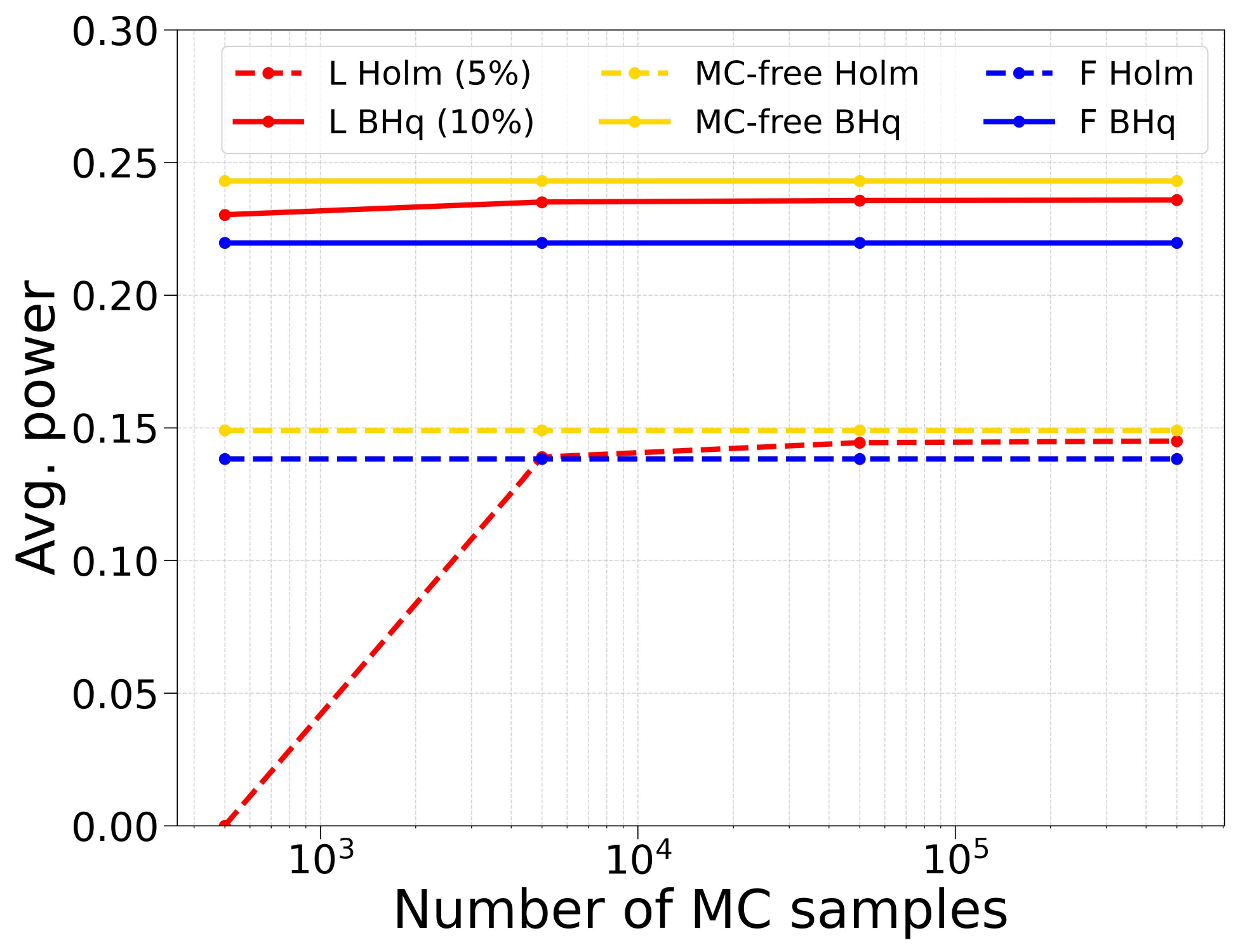}
  \end{subfigure}
  \hspace{0.2cm}
  \begin{subfigure}[t]{0.4\textwidth}
    \centering
    \includegraphics[width=\textwidth]{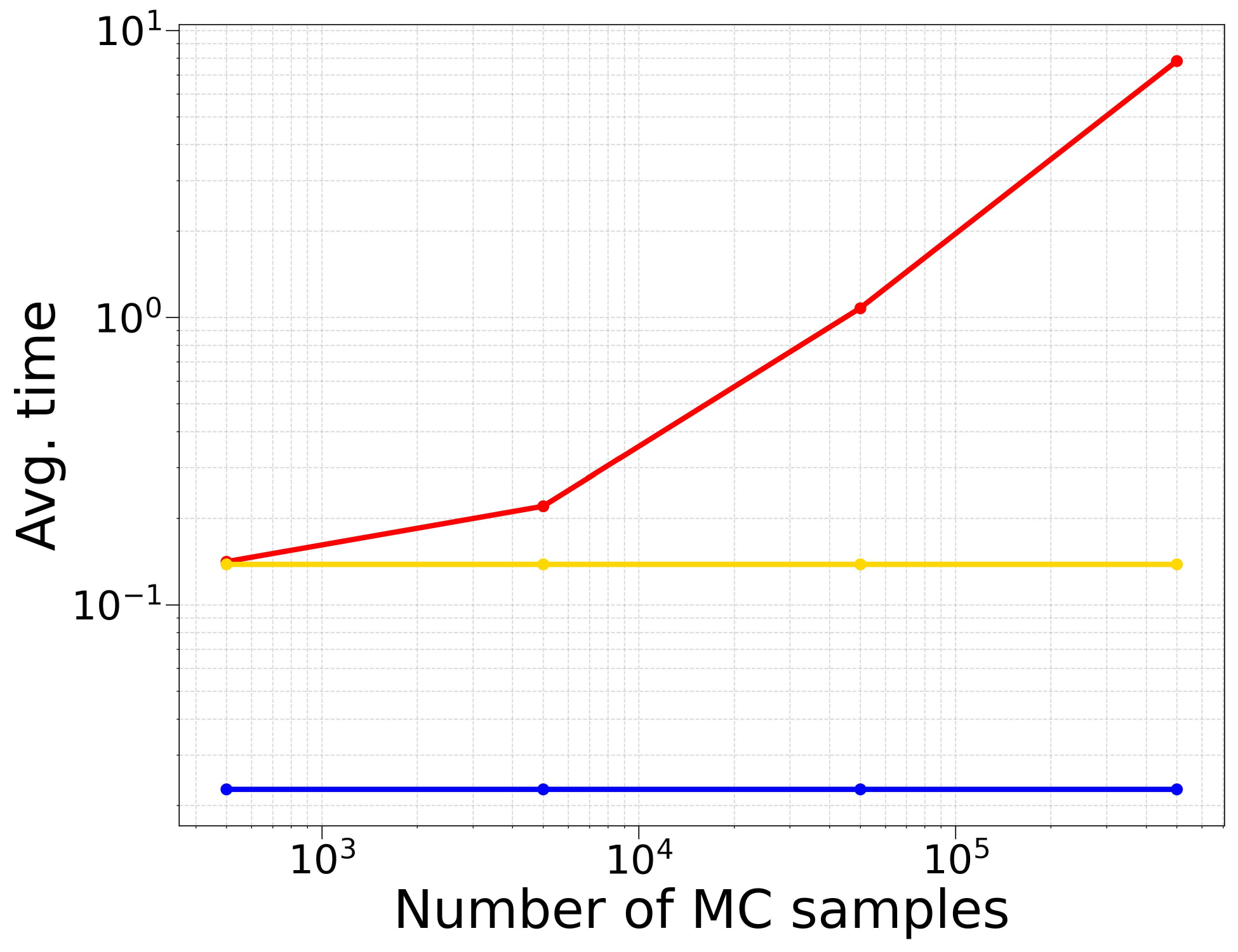}
  \end{subfigure}
  \caption{For each procedure (Holm and BHq), power was computed for each regression and then averaged across all 16. In order to be able to attribute the increased $L$-test power to the increased number of MC samples and not the randomness of the test, we performed this entire data analysis 50 times. The left panel shows the Holm and BHq powers averaged across these 50 iterations for each test. The right panel shows the average time associated with a single application of each test, computed over all tests from the 16 regressions and 50 iterations. Because the standard errors associated with the mean estimates in these plots were negligible, we chose to omit them.}
  \label{fig:multiple_testing}
\end{figure}

\section{Discussion}
\label{sec:discussion}
In closing, we identify some directions for future research:
\begin{enumerate}
    \item \textit{Generating confidence regions:} Note that we can extend the $L$-test to the more general hypothesis $H_{1:k}(\bm{a}): \bm{\beta}_{1:k} = \bm{a}$. This is because $\bm y$ satisfying $H_{1:k}(\bm{a})$ is equivalent to $\bm y - \bm{X}_{1:k}\bm{a}$ satisfying $\hyp(\bm 0) = \hyp$, so we can just apply the $L$-test to $(\bm y - \bm X_{1:k} \bm a, \bm X)$. 
    Letting $p^{(1:k)}(\bm a)$ denote the $L$-test p-value for $H_{1:k}(\bm{a})$, this generalization allows us to, in principle, generate a $100(1-\alpha)\%$ $L$-test confidence region given by $$\{\bm a \in \R^k: p^{(1:k)}(\bm a) > \alpha\}.$$
    First, naively generating the region in this way would be incredibly computationally expensive, so finding an efficient construction is an interesting direction for further work. Second, it would be useful to study whether these regions could be constructed to admit an interpretable geometry, which would be especially useful in higher dimensions.
    
    \item \textit{Exploiting additional structure:} In Section \ref{sec:extend_l_test} and Appendix \ref{sec:power_plots_add}, we consider replacing the OLS in \eqref{eq:F-test_stat_cond} with alternative estimators when the null group is sparse. It could be interesting to adapt our methodology to exploit additional structure---leveraging sparsity in the null group via the LASSO \cite{Tibshirani1996} or sparse group LASSO \cite{Tibshirani2013}, or encoding known covariance structure within the tested group by using the original Mahalanobis-norm formulation of the group LASSO \citep{Bakin1999, Yuan&Lin2006}, with $\bm{K}_j$ chosen to reflect that structure.

    \item \textit{Applying to Gaussian means and, asymptotically, to general parametric models}: In a Gaussian means problem, $\bm w \sim \mathcal N(\bm{\beta}, \bm{\Sigma})$ is observed for some known positive definite matrix $\bm \Sigma$ and $\bm{\beta}$ is unknown. This model can always be converted to a Gaussian linear regression setup by taking $\bm X = \bm \Sigma^{-1/2}$ and $\bm y = \bm{\Sigma}^{-1/2}\bm{w}$ so that $\bm y \sim \mathcal N(\bm X \bm{\beta}, \bm{I}_n)$. The standard method for testing $\hyp$ in this case of known variance is the likelihood ratio test, which uses the test statistic: $$\frac{1}{\sigma^2}(\norm{(\bm{I} - \bm{P}_{-1:k})\bm{y}}^2 - \norm{(\bm{I} - \bm{P})\bm{y}}^2) \sim \chi^2_k \text{ under $\hyp$}.$$ Following similar steps to the proof of Lemma \ref{lma:rewriteF} that rewrites the $F$-test statistic, it can be shown that the above is equivalent to $$\frac{1}{\sigma^2} \norm{\bm{V}_{1:k}^T\bm{X}_{1:k}\olsk}^2.$$ If we proceed in the same manner we obtained the $L$-test from the $F$-test by replacing $\olsk$ with $\glasso_{1:k}$ and conditioning on $\bm{X}_{-1:k}^T\bm{y}$, we will obtain the conditional test that rejects for large values of the $L$-test statistic. So, we might analogously expect the $L$-test to outperform the likelihood ratio test in this setting when the true $\bm \beta_{-1:k}$ is sparse. One promising implication of this is that the $L$-test may yield more powerful asymptotic inference (under sparsity) for a vast class of multivariate parametric models that admit asymptotically multivariate Gaussian estimators coupled with consistent estimators for the covariance matrix (see \citep{Sengupta2024} for a similar discussion and \citep{Fithian2021} for such an asymptotic result but for knockoffs).
\end{enumerate}

\section{Acknowledgments}
\label{sec:acknowledgments}
LJ and SS were partially supported by DMS-2045981.

\newpage
\bibliographystyle{plain}
\bibliography{references}

\newpage
\appendix

\section{Extending the $\ell$-test}
\label{sec:extending_ell}
\subsection{Challenges of analytic p-value characterization}
\label{sec:analytic_challenges}
Theorem \ref{thm:glasso_unit} in Appendix \ref{sec:proofs} establishes a functional relationship between $\glasso_{1:k}$ and $\bm{u}_{1:k}$ by extending the inverse characterization result of Theorem 2.1 in \cite{Sengupta2024}. After a change of variables, this result yields the density of $\glasso_{1:k} \mid \suffstat$ under $\hyp$ and therefore allows us to express the p-value of the conditional test that rejects for large values of \eqref{eq:L-test_stat_1} as follows. Letting $l > 0$ denote the observed test statistic,
\begin{align*}
P_{\hyp}\left(\norm{\bm{V}_{1:k}^T\bm{X}_{1:k}\glasso_{1:k}} \geq l \biggr\vert \suffstat\right) &= \int_{\{\bm{b} : \bm{b}^T\bm{M}\bm{b} \geq l^2\}}f_{\glasso_{1:k}}^{\hyp}(\bm b \mid \suffstat)~d\bm{b} \\
&=\int_{\{\bm{b} : \bm{b}^T\bm{M}\bm{b} \geq l^2\}}f_{\bm u_{1:k}}\left(f^{-1}_{\suffstat}(\bm{b})\right) \cdot \left|\frac{\partial f^{-1}_{\suffstat}(\bm{b})}{\partial \bm{b}}\right|~d\bm{b} \text{, where} \\
\bm{M} &= \bm{X}_{1:k}^T\bm{V}_{1:k}\bm{V}_{1:k}^T\bm{X}_{1:k}.
\end{align*}
In the second equality, note that although $f^{-1}_{\suffstat}$ is nondifferentiable on the measure-zero set where the LASSO estimate inside (see Theorem \ref{thm:glasso_unit}) changes its active set, this does not affect the validity of the change-of-variables formula for the integral, so we apply the Jacobian expression in the usual way. Evaluating this integral, however, is intractable. Even for moderately large $k$ and a well-behaved integrand, evaluating a $k$-dimensional integral over an ellipsoidal region is computationally demanding. In our setting, the problem is compounded by the fact that evaluating the integrand is an extremely computationally intensive task that requires solving a LASSO to obtain $\Lambda_{1:k}(\bm b, \frac{\bm b}{\norm{\bm b}})$ as well as computing the determinant of the matrix derivative $\frac{\partial f^{-1}_{\suffstat}(\bm{b})}{\partial \bm{b}}$.

When $k = 1$, \cite{Sengupta2024} circumvents this issue by showing that $\glassos_1$ is a monotone increasing function of $u_1$. This monotonicity allows the p-value to be mapped from $\glassos_1$-space to $u_1$-space as follows
\begin{align*}
P_{\hyp}\left(|\glassos_1| \geq l \biggr\vert \suffstats\right) &= P_{\hyp}(\glassos_1 \geq l \mid \suffstats) + P_{\hyp}(\glassos_1 \leq -l \mid \suffstats) \\
&= P(u_1 \geq f^{-1}_{\suffstats}(l) ) + P(u_1 \leq f^{-1}_{\suffstats}(-l)),
\end{align*}
which reduces to evaluating the CDF of $u_1$ via its one-to-one mapping to a $t$-distribution, $\frac{\sqrt{n-d}u_1}{\sqrt{1-u_1^2}} \sim t_{n-d}$. Crucially, this approach avoids the need to integrate a complex integrand that runs a LASSO.

When $k > 1$, the relationship between $\glasso_{1:k}$ and $\bm{u}_{1:k}$ becomes substantially harder to characterize. The matrix derivative $\nabla \glasso_{1:k}(\bm{u}_{1:k})$, evaluated explicitly in Proposition \ref{prop:glasso_matrix_deriv}, is quite complex and does not have any special structure (e.g. symmetry, diagonality, positive definiteness). Moreover, unlike in the $k = 1$ case where the monotonicity property of $\frac{\partial \glassos_1}{\partial u_1}$ enabled a straightforward transformation from $\glassos_1$-space to $u_1$-space, here it is unclear what kind of analytic information about $\glasso_{1:k}$ as a function of $\bm{u}_{1:k}$ would even facilitate such a transformation, given the additional complexity of the p-value expression in $\glasso_{1:k}$-space when $k > 1$.

\subsection{\sectionmath{Recentered $\bm{u}_{1:k}$-test}{Recentered u-test}}
\label{sec:ell_recentered}
While the exact characterization of \eqref{eq:L-test_stat_1} conditional on $\suffstat$ under $\hyp$ is restricted to $k=1$, \cite{Sengupta2024} also mentions a recentered $u_1$-test that approximates the $\ell$-test well and requires only a small, fixed number of LASSOs to perform. A multivariate analogue can be obtained by using Theorem \ref{thm:glasso_unit} and recentering $\bm{u}_{1:k}$ by the center of the ellipsoidal region $f^{-1}_{\suffstat}(\bm{0})$. This multivariate $\bm u_{1:k}$-test, however, has two shortcomings. First, while its performance in many settings is similar to the test that uses the p-value in \eqref{eq:mc_p}, it is not hard to find settings where it significantly underperforms. Second, even when it does perform well, if one wants to do multiple testing or for any other reason be able to obtain precise p-values, then one needs a large number of MC samples which, even without needing a LASSO for each, can become expensive. To make sparsity-informed coefficient group testing practical, we address both of these problems with new methods in this paper. 

\section{Power of the $L$-test}
\label{sec:premult_role}
In Section \ref{sec:L_test_power}, we claimed that the premultiplier in the $L$-test statistic takes on different structure depending on how variance is distributed among the first $k$ predictors $\bm X_{1:k}$. In this section, we formally establish those properties. Recall that the $L$-test premultiplier is given by 
$$\phi = \bm{A}(\glassoc_{1:k}, \lambda) = \bm{V}_{1:k}^T\bm{X}_{1:k} \left[\bm{X}_{1:k}^T(\bm{I} - \bm{P}_{\mathcal{A}(\glassoc_{1:k})})\bm{X}_{1:k} + \frac{n\lambda}{\norm{\glassoc_{1:k}}} \bm{I}\right]^{-1}\bm{X}_{1:k}^T \bm{V}_{1:k}.$$ We study it in two cases---when the block $\bm{X}_{1:k}$ is orthonormal and when it has rank 1. These two setups correspond to opposite extreme situations where the predictors provide maximally distinct vs. completely redundant information. 

Throughout the analysis below, we assume $\bm{\beta}_{-1:k}$ is sparse so that we can expect the active set $\mathcal{A}(\glassoc_{1:k})$ defined in Lemma \ref{lma:deriv} to be small. Because only a small fraction of the predictors are active in this case, the projection $\bm{P}_{\mathcal{A}(\glassoc_{1:k})}$ acts in a very low-dimensional space, capturing only about $|\mathcal{A}(\glassoc_{1:k})|/n$ of the variance, which is negligible when $|\mathcal{A}(\glassoc_{1:k})|$ is small. In this sense, 
\begin{align}
\label{eq:approx}
\bm{X}_{1:k}^T(\bm{I} - \bm{P}_{\mathcal{A}(\glassoc_{1:k})})\bm{X}_{1:k} \approx \bm{X}_{1:k}^T\bm{X}_{1:k},
\end{align}
an approximation we will make in both cases below. 

\subsection{Orthonormal case}
Since the columns of $\bm{X}_{1:k}$ are orthonormal, $\bm{X}_{1:k}^T\bm{X}_{1:k} = \bm{I}$. Moreover, for $i = 1, \ldots, k$,
\begin{align*}
    \frac{(\bm{I} - \bm{P}_{-1:i})\bm{X}_i}{\norm{(\bm{I} - \bm{P}_{-1:i})\bm{X}_i}} = \frac{\bm{X}_i}{\norm{\bm{X}_i}} \implies \bm{X}_{1:k}^T\bm{V}_{1:k} = \bm{I}.
\end{align*}
Making approximation \eqref{eq:approx} and substituting these expressions for $\bm{X}_{1:k}^T\bm{X}_{1:k}$ and $\bm{X}_{1:k}^T\bm{V}_{1:k}$, we get that 
$$\bm{A}(\glassoc_{1:k}, \lambda) \approx \left(1 + \frac{n\lambda}{\norm{\glassoc_{1:k}}}\right)^{-1} \bm{I}.$$
Thus, in the case that the variance is evenly distributed across predictors, the premultiplier has essentially no impact and the recentering vector plays the primary role.

\subsection{Rank-1 case}
Making approximation \eqref{eq:approx}, taking $\bm{X}_{1:k} = \bm{v}\bm{b}^T$, and applying the Sherman-Morrison formula, we have for some constant $C \in \R$
$$\bm{A}(\glassoc_{1:k}, \lambda) \approx C \cdot \bm{V}_{1:k}^T\bm{v}(\bm{V}_{1:k}^T\bm{v})^T.$$
The matrix $\bm{V}_{1:k}^T\bm{v}(\bm{V}_{1:k}^T\bm{v})^T$ is symmetric and has rank 1, so it has one non-zero eigenvalue corresponding to eigenvector $\bm{V}_{1:k}^T\bm{v}$ and $k-1$ zero eigenvalues. Writing out $\bm{u}_{1:k}$, we have that
\begin{align*}
\bm{u}_{1:k} &= \frac{1}{\sigmahat}\bm{V}_{1:k}^T(\bm{X}_{1:k}\bm{\beta}_{1:k} + \bm{\epsilon}) \text{, by Lemma \ref{lma:decomp}} \\
&= \frac{1}{\sigmahat}\bm{V}_{1:k}^T(\bm{v}\bm{b}^T\bm{\beta}_{1:k} + \bm{\epsilon}) \\
&= \frac{1}{\sigmahat}\left((\bm{b}^T\bm{\beta}_{1:k})\bm{V}_{1:k}^T\bm{v} + \bm{V}_{1:k}^T\bm{\epsilon}\right).
\end{align*}
If the true $\bm{\beta}_{1:k}$ and the PC direction $\bm{b}$ that explains all the variance in $\bm{X}_{1:k}$ are well alligned (ie. the inner product $\bm{b}^T\bm{\beta}_{1:k}$ is large in magnitude), we can expect the eigenvector $\bm{V}_{1:k}^T\bm{v}$ to be approximately parallel to $\bm{u}_{1:k}$.

\section{Proofs of results}
\label{sec:proofs}

\subsection{\sectionmath{Sampling from $\bm y \mid \suffstat$ under $\hyp$}{Sampling from y | S under H}}
\label{sec:lma_decomp_pf}
\begin{lemma}
\label{lma:decomp}
For the linear model in \eqref{eq:linear_model}, define $\yhat = \bm{P}_{-1:k}\bm{y}$ and $\sigmahat^2 = \norm{(\bm{I} - \bm{P}_{-1:k})\bm{y}}^2$, and recall the matrix $\bm{V} \in \R^{n \times (n-d+k)}$ described following Equation \eqref{eq:group_lasso}. There exists a unique vector $\bm{u} \in \R^{n-d+k}$, such that $\|\bm u\| = 1$ and the following decomposition holds:
\begin{align*}
\bm{y} = \yhat + \sigmahat \bm{V}\bm{u}.
\end{align*}
Moreover,
\begin{align*}
\bm{u} \mid \suffstat \overset{\hyp}{\sim} \Unif(\sphere^{n-d+k-1}).  
\end{align*}
\end{lemma}
\begin{proof}
Since $\bm{V}$ is full-column rank, there exists a unique vector $\bm{u}$ such that 
\begin{align*}
\bm{V}\bm{u} &= \frac{\bm{y} - \yhat}{\sigmahat}.
\end{align*}
In particular, $\bm{u} = \bm{V}^T\left(\frac{\bm{y} - \yhat}{\sigmahat}\right)$ with $\norm{\bm{u}} = \norm{\bm{V}\bm{u}} = \norm{\frac{\bm{y} - \yhat}{\sigmahat}} = \frac{1}{\sigmahat}\norm{\bm{y} - \yhat} = 1$ since $\bm{V}$ is an orthogonal transformation and hence preserves norms. This proves the decomposition. The conditional distribution of $\bm{u}$ follows from the fact that $$\bm{y} \mid \suffstat \overset{\hyp}{\sim} \yhat + \sigmahat \bm{V}\bm{u} \text{, where $\bm{u} \sim \Unif(\sphere^{n-d+k-1})$.}$$ 
The proof of this result is identical to that of Proposition E.1 in \cite{Luo2022}, except we use $\bm{P}_{-1:k}$, $(\bm{I} - \bm{P}_{-1:k})$, and $\bm{V}$ accordingly.
\end{proof}

\subsection{Conditional $F$-test statistic}
\begin{lemma}
\label{lma:cond_F}
Let $\olsk$ denote the first $k$ OLS estimate coefficients. The $F$-test for $\hyp$ in model \eqref{eq:linear_model} is equivalent to the test that rejects for large values of 
\begin{align*}
\norm{\bm{V}_{1:k}^T\bm{X}_{1:k}\olsk}
\end{align*}
conditional on $\suffstat$. This latter test is also equivalent to the conditional test rejecting for large values of $\norm{\bm{u}_{1:k}}$.
\end{lemma}
\begin{proof}
The proof of Lemma \ref{lma:cond_F} relies on a series of sub-lemmas. The first gives a representation for a subvector of the OLS that extends the convenient representation for one OLS coefficient $\olsj = \frac{\bm{X}_j^T(\bm{I} - \bm{P}_{-j})\bm{y}}{\norm{(\bm{I} - \bm{P}_{-j})\bm{X}_j}^2}$ used in \cite{Fithian2014}.
\begin{lemma}
\label{lma:rewrite_OLS}
The first $k$ OLS coefficients can be written as 
\begin{align*}
\olsk = \bm{S}^{-1}\bm{X}_{1:k}^T(\bm{I} - \bm{P}_{-1:k})\bm{y},
\end{align*}
where $\bm{S} = \bm{X}_{1:k}^T(\bm{I} - \bm{P}_{-1:k})\bm{X}_{1:k}$. 
\end{lemma}
\begin{proof}[Proof of Lemma \ref{lma:rewrite_OLS}]
We start by decomposing the OLS estimate into block matrices as follows
\begin{align*}
\ols &= (\bm{X}^T\bm{X})^{-1}\bm{X}^T\bm{y} \\
&= 
\left(
\begin{pmatrix}
    \bm{X}_{1:k}^T \\
    \bm{X}_{-1:k}^T
\end{pmatrix}
\begin{pmatrix}
    \bm{X}_{1:k} & \bm{X}_{-1:k}
\end{pmatrix}
\right)^{-1}
\begin{pmatrix}
    \bm{X}_{1:k}^T \\
    \bm{X}_{-1:k}^T
\end{pmatrix}
\bm{y} \\
&= \begin{pmatrix}
    \bm{X}_{1:k}^T\bm{X}_{1:k} & \bm{X}_{1:k}^T\bm{X}_{-1:k} \\
    \bm{X}_{-1:k}^T \bm{X}_{1:k} & \bm{X}_{-1:k}^T \bm{X}_{-1:k}
\end{pmatrix}^{-1}
\begin{pmatrix}
    \bm{X}_{1:k}^T \\
    \bm{X}_{-1:k}^T
\end{pmatrix}
\bm{y} \\
&= \begin{pmatrix}
    \bm{A} & \bm{B} \\
    \bm{C} & \bm{D}
\end{pmatrix}^{-1}
\begin{pmatrix}
    \bm{X}_{1:k}^T \\
    \bm{X}_{-1:k}^T
\end{pmatrix}
\bm{y} \textrm{ [where $\bm A, \bm B, \bm C$ and $\bm D$ represent the block-matrices they are replacing]}\\
&= \begin{pmatrix}
    \bm{E} & \bm{F} \\
    \bm{G} & \bm{H}
\end{pmatrix}
\begin{pmatrix}
    \bm{X}_{1:k}^T \\
    \bm{X}_{-1:k}^T
\end{pmatrix}
\bm{y} \textrm{ (say)}\\
&= \begin{pmatrix}
    \bm{E}\bm{X}_{1:k}^T + \bm{F}\bm{X}_{-1:k}^T \\
    \bm{G}\bm{X}_{1:k}^T + \bm{H}\bm{X}_{-1:k}^T
\end{pmatrix}
\bm{y}.
\end{align*}
In particular, letting $\bm{S} = \bm{A} - \bm{B}\bm{D}^{-1}\bm{C} = \bm{X}_{1:k}^T(\bm{I} - \bm{P}_{-1:k})\bm{X}_{1:k}$ denote the Schur complement of $\bm{D}$, then $\bm{E} = \bm{S}^{-1}$ and $\bm{F} = -\bm{S}^{-1}\bm{B}\bm{D}^{-1}$. Note that the block inversion is well-defined because $\bm{D}$ and $\bm S$ are both invertible. The invertibility of $\bm D$ follows from the fact that $\bm X_{-1:k}$ is also full column-rank, while that of $\bm S$ follows from the fact that so is $(\bm I - \bm P_{-1:k})\bm X_{1:k}$. To see the latter fact, note that for any vector $\bm v$,   $(\bm I - \bm P_{-1:k})\bm X_{1:k}\bm v = \bm 0$ implies that $\bm X_{1:k}\bm v$ lies in the column-space of $\bm X_{-1:k}$. Because the columns of $\bm X_{1:k}$ are linearly independent from that of $\bm X_{-1:k}$, this can only happen if $\bm v = 0$; thereby showing that $(\bm I - \bm P_{-1:k})\bm X_{1:k}$ is full column-rank.

Next, from the above decomposition of $\ols$, it follows that
\begin{align*}
\olsk &= \left(\bm{E}\bm{X}_{1:k}^T + \bm{F}\bm{X}_{-1:k}^T\right)\bm{y} \\
&= \left(\bm{S}^{-1}\bm{X}_{1:k}^T -\bm{S}^{-1}\bm{X}_{1:k}^T\bm{X}_{-1:k}(\bm{X}_{-1:k}^T \bm{X}_{-1:k})^{-1}\bm{X}_{-1:k}^T\right)\bm{y} \\
&= \bm{S}^{-1}\bm{X}_{1:k}^T\left(\bm{I} - \bm{P}_{-1:k}\right)\bm{y}.
\end{align*} 
\end{proof}

\begin{lemma}
\label{lma:rewriteF}
The $F$-test statistic for $\hyp$ in model \eqref{eq:linear_model} can be written as
\begin{align}
\label{eq:rewriteF}
F := \frac{\norm{\bm{V}_{1:k}^T\bm{X}_{1:k}\olsk}^2/k}{\left(\sigmahat^2 - \norm{\bm{V}_{1:k}^T\bm{X}_{1:k}\olsk}^2\right)/(n-d)}.
\end{align}
\end{lemma}
\begin{proof}[Proof of Lemma \ref{lma:rewriteF}]
We start by manipulating the classic $F$-test statistic as follows
\begin{align*}
    F &= \frac{(\norm{\bm{y} - \bm{P}_{-1:k}\bm{y}}^2 - \norm{\bm{y} - \bm{P}\bm{y}}^2)/k}{\norm{\bm{y} - \bm{P}\bm{y}}^2/(n-d)} \\
    &= \frac{\norm{\bm{P}\bm{y}-\bm{P}_{-1:k}\bm{y}}^2/k}{\norm{(\bm{I} - \bm{P})\bm{y}}^2/(n-d)} \\
    &= \frac{\norm{\bm{P}(\bm{I}-\bm{P}_{-1:k})\bm{y}}^2/k}{\norm{(\bm{I} - \bm{P})\bm{y}}^2/(n-d)} \\
    &= \frac{\norm{\bm{P}(\bm{I}-\bm{P}_{-1:k})\bm{y}}^2/k}{\norm{(\bm{I} - \bm{P})(\bm{I} - \bm{P}_{-1:k})\bm{y}}^2/(n-d)} \\
    &= \frac{\norm{\bm{P}(\bm{I}-\bm{P}_{-1:k})\bm{y}}^2/k}{\left(\sigmahat^2 - \norm{\bm{P}(\bm{I} - \bm{P}_{-1:k})\bm{y}}^2\right)/(n-d)}.
\end{align*}
By expanding the projection matrix $\bm{P}$, decomposing $\bm{X}^T$ into block matrices, and using Lemma \ref{lma:rewrite_OLS}, we can rewrite $\bm{P}(\bm{I}-\bm{P}_{-1:k})\bm{y}$ in terms of $\olsk$ as follows
\begin{align*}
\bm{P}(\bm{I}-\bm{P}_{-1:k})\bm{y} &= \bm{X}(\bm{X}^T\bm{X})^{-1}
\begin{pmatrix}
\bm{X}_{1:k}^T \\
\bm{X}_{-1:k}^T
\end{pmatrix}
(\bm{I}-\bm{P}_{-1:k})\bm{y} \\
&= \bm{X}(\bm{X}^T\bm{X})^{-1}
\begin{pmatrix}
\bm{S}\olsk \\
\bm{0}
\end{pmatrix} \\
&= (\bm{X}(\bm{X}^T\bm{X})^{-1})_{1:k}\bm{S}\olsk \\
&= (\bm{X}(\bm{X}^T\bm{X})^{-1})_{1:k}\bm{X}_{1:k}^T(\bm{I}-\bm{P}_{-1:k})\bm{X}_{1:k}\olsk \\
&= (\bm{X}(\bm{X}^T\bm{X})^{-1})_{1:k}\bm{X}_{1:k}^T\bm{V}\bm{V}^T\bm{X}_{1:k}\olsk \\
&= \bm{P}\bm{V}\bm{V}^T\bm{X}_{1:k}\olsk,
\end{align*}
where we use the fact that $\bm{V}\bm{V}^T = (\bm{I} - \bm{P}_{-1:k})$ since the columns of $\bm V$ form an orthonormal basis for the orthogonal complement of the column space of $\bm{X}_{-1:k}$. The definition of $\bm{V}$ implies that the columns of $\bm{V}_{-1:k}$ are orthogonal to the column space of $\bm{X}$, so $\bm{V}\bm{V}^T\bm{X}_{1:k} = \bm{V}_{1:k}\bm{V}_{1:k}^T\bm{X}_{1:k}$. We also have that $\bm{P}\bm{V}_{1:k} = \bm{V}_{1:k}$ since the columns of $\bm{V}_{1:k}$ lie in the column space of $\bm{X}$. Thus, $$\bm{P}\bm{V}\bm{V}^T\bm{X}_{1:k}\olsk = \bm{V}_{1:k}\bm{V}_{1:k}^T\bm{X}_{1:k}\olsk.$$ Substituting this expression into the $F$-test statistic and noting that $\bm{V}_{1:k}$ is orthogonal and hence norm-preserving, we find that
\begin{align*}
    F = \frac{\norm{\bm{V}_{1:k}^T\bm{X}_{1:k}\olsk}^2/k}{\left(\sigmahat^2 - \norm{\bm{V}_{1:k}^T\bm{X}_{1:k}\olsk}^2\right)/(n-d)}.
\end{align*}
\end{proof}

\begin{lemma}
\label{lma:OLS_unit}
The first $k$ OLS coefficients can be written in terms of the first $k$ coefficients of $\bm{u}$ from \eqref{eq:decomp} as follows
\begin{align*}
\olsk &= \sigmahat (\bm{V}_{1:k}^T\bm{X}_{1:k})^{-1}\bm{u}_{1:k}.
\end{align*}
\end{lemma}
\begin{proof}[Proof of Lemma \ref{lma:OLS_unit}]
By Lemma \ref{lma:rewrite_OLS}, we have that
\begin{align*}
\olsk &= \bm{S}^{-1}\bm{X}_{1:k}^T(\bm{I} - \bm{P}_{-1:k})\bm{y} \\
&= \bm{S}^{-1}\bm{X}_{1:k}^T(\bm{I} - \bm{P}_{-1:k})(\yhat + \sigmahat\bm{V}\bm{u}) \text{, by Lemma \ref{lma:decomp}} \\
&= \sigmahat\bm{S}^{-1}\bm{X}_{1:k}^T(\bm{I} - \bm{P}_{-1:k})\bm{V}\bm{u} \\
&= \sigmahat\bm{S}^{-1}\bm{X}_{1:k}^T\bm{V}\bm{u},
\end{align*}
because the columns of $\bm{V}$ lie in the space orthogonal to the columns of $\bm{X}_{-1:k}$. Since the columns of $\bm{V}_{-1:k}$ are orthogonal to the column space of $\bm{X}$, $\bm{X}_{1:k}^T\bm{V}\bm{u} = \bm{X}_{1:k}^T\bm{V}_{1:k}\bm{u}_{1:k}$. Moreover, $\bm{S}^{-1} = (\bm{X}_{1:k}^T(\bm{I}-\bm{P}_{-1:k})\bm{X}_{1:k})^{-1} = (\bm{X}_{1:k}^T\bm{V}\bm{V}^T\bm{X}_{1:k})^{-1} = (\bm{X}_{1:k}^T\bm{V}_{1:k}\bm{V}_{1:k}^T\bm{X}_{1:k})^{-1}$. Thus,
\begin{align*}
\olsk &= \sigmahat (\bm{V}_{1:k}^T\bm{X}_{1:k})^{-1}(\bm{X}_{1:k}^T\bm{V}_{1:k})^{-1}\bm{X}_{1:k}^T\bm{V}_{1:k}\bm{u}_{1:k} = \sigmahat(\bm{V}_{1:k}^T\bm{X}_{1:k})^{-1}\bm{u}_{1:k}.
\end{align*}
\end{proof}

\begin{lemma}
\label{lma:indep}
    The $F$-test statistic is independent of $\suffstat$ under $\hyp$.
\end{lemma}
\begin{proof}[Proof of Lemma \ref{lma:indep}]
By Lemma \ref{lma:rewriteF}, the $F$-test statistic can be written as
\begin{align*}
F &= \frac{\norm{\bm{V}_{1:k}^T\bm{X}_{1:k}\olsk}^2/k}{\left(\sigmahat^2 - \norm{\bm{V}_{1:k}^T\bm{X}_{1:k}\olsk}^2\right)/(n-d)} \\
&= \frac{\frac{\norm{\bm{V}_{1:k}^T\bm{X}_{1:k}\olsk}^2}{\sigmahat^2}/k}{\left(1 - \frac{\norm{\bm{V}_{1:k}^T\bm{X}_{1:k}\olsk}^2}{\sigmahat^2}\right)/(n-d)}. 
\end{align*}
It suffices to show that $\frac{\norm{\bm{V}_{1:k}^T\bm{X}_{1:k}\olsk}^2}{\sigmahat^2}$ is independent of $\suffstat$ under $\hyp$, which we will demonstrate by showing that this expression's conditional distribution is the same as its marginal distribution under $\hyp$. By Lemma \ref{lma:OLS_unit} and \ref{lma:decomp}, under $\hyp$,
\begin{align*}
\frac{\norm{\bm{V}_{1:k}^T\bm{X}_{1:k}\olsk}^2}{\sigmahat^2} \biggr\vert \suffstat \sim \norm{\bm{u}_{1:k}}^2 \mid \suffstat \sim \norm{\tilde{\bm{u}}_{1:k}}^2,
\end{align*}
where $\tilde{\bm{u}}\sim \Unif(\sphere^{n-d+k-1})$.

Next, under $\hyp$, we can write
\begin{align*}
    \bm{y} = \bm{X}_{-1:k} \bm{\beta}_{-1:k} + \bm{ \epsilon} \text{, for some $\bm{ \epsilon} \sim \mathcal{N}(\bm{0}, \sigma^2 \bm{I})$}.
\end{align*}
Recall $(\bm{I} - \bm{P}_{-1:k}) = \bm{V}\bm{V}^T$ and $\bm{S}^{-1} =(\bm{X}_{1:k}^T\bm{V}_{1:k}\bm{V}_{1:k}^T\bm{X}_{1:k})^{-1}$. Since the columns of $\bm{V}_{-1:k}$ are orthogonal to the column space of $\bm{X}$, we also have $\bm{X}_{1:k}^T\bm{V}\bm{V}^T = \bm{X}_{1:k}^T\bm{V}_{1:k}\bm{V}_{1:k}^T$. Thus,
\begin{align*}
\frac{\norm{\bm{V}_{1:k}^T\bm{X}_{1:k}\olsk}^2}{\sigmahat^2} &= \frac{\norm{\bm{V}_{1:k}^T\bm{X}_{1:k}\bm{S}^{-1}\bm{X}_{1:k}^T(\bm{I}-\bm{P}_{-1:k})\bm{y}}^2}{\sigmahat^2} \text{, by Lemma \ref{lma:rewrite_OLS}} \\
&= \frac{\norm{\bm{V}_{1:k}^T\bm{X}_{1:k}\bm{S}^{-1}\bm{X}_{1:k}^T\bm{V}\bm{V}^T\bm{y}}^2}{\norm{\bm{V}\bm{V}^T\bm{y}}^2} \\
&= \frac{\norm{\bm{V}_{1:k}^T\bm{X}_{1:k}\bm{S}^{-1}\bm{X}_{1:k}^T\bm{V}_{1:k}\bm{V}_{1:k}^T\bm{y}}^2}{\norm{\bm{V}^T\bm{y}}^2} \\
&= \frac{\norm{\bm{V}_{1:k}^T\bm{y}}^2}{\norm{\bm{V}^T\bm{y}}^2} \\
&= \norm{\bm{u}^*_{1:k}}^2,
\end{align*}
where $\bm{u}^* = \frac{\bm{V}^T\bm{y}}{\norm{\bm{V}^T\bm{y}}} \sim \Unif(\sphere^{n-d+k-1})$ since $\bm{V}^T\bm{y} = \bm{V}^T\bm{\epsilon} \sim \mathcal{N}(\bm{0}, \sigma^2 \bm{I}_{n- d + k})$. Thus, $$ F \perp \suffstat \text{ under $\hyp$}.$$
\end{proof}

By Lemmas \ref{lma:rewriteF} and \ref{lma:indep}, the $F$-test is equivalent to the test that rejects for large values of \eqref{eq:rewriteF} conditional on $\suffstat$. Conditional on $\suffstat$, however, \eqref{eq:rewriteF} is a monotonically increasing function of \eqref{eq:F-test_stat_cond}. Thus, the $F$-test is equivalent to the test that rejects for large values of \eqref{eq:F-test_stat_cond} conditional on $\suffstat$, giving the result. The last part of the lemma can easily be seen by substituting the expression for $\olsk$ from Lemma \ref{lma:OLS_unit} into \eqref{eq:F-test_stat_cond}.
\end{proof}

\subsection{\sectionmath{$\glasso_{1:k}$ as a function of $\bm{u}_{1:k}$}{Group LASSO as a function of u}}
\label{sec:glasso_unit}
In an attempt to demonstrate why it is difficult to analytically study $\glasso_{1:k}$ as a function of $\bm{u}_{1:k}$, we state and prove a result below that specifies the matrix derivative of $\glasso_{1:k}$ with respect to $\bm{u}_{1:k}$. Different elements of the matrix derivative shown in Proposition \ref{prop:glasso_matrix_deriv} exhibit nice properties. For example, $\bm{X}_{1:k}^T(\bm{I} - \bm{P}_{\mathcal{A}(\bm{u}_{1:k})})\bm{X}_{1:k}$ is positive definite, the inverted matrix is symmetric, and $\bm{X}_{1:k}^T\bm{V}_{1:k}$ is a diagonal matrix with positive diagonal entries. Nevertheless, a study of the composite matrix reveals that it does not exhibit any clear useful properties (ie. symmetric, positive definite, sparsity, low-rank structure) that would lend insight into $\glasso_{1:k}$'s behavior as a function of $\bm{u}_{1:k}$.
\begin{proposition}
\label{prop:glasso_matrix_deriv}
Suppose that $\bm{u}_{1:k} \in \mathbb{R}^k$ is such that $\glasso_{1:k}(\bm{u}_{1:k}) \neq \bm{0}$. Then, 
\begin{align*}
    \nabla \glasso_{1:k}(\bm{u}_{1:k}) &= \sigmahat \left(\bm{X}_{1:k}^T(\bm{I} - \bm{P}_{\mathcal{A}(\bm{u}_{1:k})})\bm{X}_{1:k} + \frac{n\lambda}{\norm{\glasso_{1:k}}} \left(\bm{I} - \frac{\glasso_{1:k}(\glasso_{1:k})^T}{\norm{\glasso_{1:k}}^2}\right)\right)^{-1}\bm{X}_{1:k}^T\bm{V}_{1:k} \text{, where} \\
    \mathcal{A}(\bm{u}_{1:k}) &= \left\{ i \in \{k + 1, \ldots, d \} : \glassos_i(\bm{u}_{1:k}) \neq 0\right\}
\end{align*}
\end{proposition}
\begin{proof}
To show how $\glasso_{1:k}$ can be written as a function of $\bm{u}_{1:k}$, we start by first rewriting the objective function associated with the group LASSO in \eqref{eq:group_lasso} using a similar approach to that in the proof of Theorem 2.1 in \cite{Sengupta2024}. Recalling that $\bm y = \yhat + \sigmahat \bm{V}\bm{u}$ from \eqref{eq:decomp} and letting $\error = \sigmahat \bm{V}\bm{u}$, we can decompose the error term $\error$ as $\error = \hat{\bm{e}}_{\parallel} + \hat{\bm{e}}_{\perp}$, where
\begin{align*}
\hat{\bm{e}}_{\parallel} &= (\bm{I} - \bm{P}_{-1:k})\bm{X}_{1:k}(\bm{X}_{1:k}^T(\bm{I} - \bm{P}_{-1:k})\bm{X}_{1:k})^{-1}\bm{X}_{1:k}^T(\bm{I} - \bm{P}_{-1:k})\error \\
&= (\bm{I} - \bm{P}_{-1:k})\bm{X}_{1:k}\olsk \text{, by Lemma \ref{lma:rewrite_OLS}} \\
\hat{\bm{e}}_{\perp} &= \left(\bm{I} - (\bm{I} - \bm{P}_{-1:k})\bm{X}_{1:k}(\bm{X}_{1:k}^T(\bm{I} - \bm{P}_{-1:k})\bm{X}_{1:k})^{-1}\bm{X}_{1:k}^T(\bm{I} - \bm{P}_{-1:k})\right)\error
\end{align*}
Note that $\hat{\bm{e}}_{\parallel}$ is the projection of $\error$ onto the column span of $(\bm{I} - \bm{P}_{-1:k})\bm{X}_{1:k}$ and $\hat{\bm{e}}_{\perp}$ is perpendicular to the column space of $\bm{X}$. The objective function in \eqref{eq:group_lasso} can thus be written as 
\begin{align*}
    \norm{\bm y - \bm x \bm{\beta}}^2 + 2n\lambda \norm{\bm{\beta}_{1:k}}_2 + 2n\lambda \norm{\bm{\beta}_{-1:k}}_1 &= \norm{\yhat + \hat{\bm{e}}_{\parallel} + \hat{\bm{e}}_{\perp} - \bm X \bm{\beta}}^2 + 2n\lambda \norm{\bm{\beta}_{1:k}}_2 + 2n\lambda \norm{\bm{\beta}_{-1:k}}_1 \\
    &= \norm{\yhat  - \bm{P}_{-1:k}\bm{X}_{1:k}\bm{\beta}_{1:k} -\bm{X}_{-1:k}\bm{\beta}_{-1:k}}^2  + \norm{\hat{\bm{e}}_{\perp}}^2 \\
    &+ \norm{(\bm{I} - \bm{P}_{-1:k})\bm{X}_{1:k}(\olsk - \bm{\beta}_{1:k})}^2 + 2n\lambda \norm{\bm{\beta}_{1:k}}_2 \\
    &+ 2n\lambda \norm{\bm{\beta}_{-1:k}}_1 \\
    &= \tilde{f}(\bm{\beta}; \bm y , \bm X) + \norm{(\bm{I} - \bm{P}_{-1:k})\bm{X}_{1:k}(\olsk - \bm{\beta}_{1:k})}^2,
\end{align*}
where $\tilde{f}$ denotes the expression it is replacing. Defining the function 
\begin{align*}
\glasso(\bm{a}) &= \argmin_{\bm{\beta} \in \R^d}\left(\tilde{f}(\bm{\beta}; \bm y , \bm X) + \norm{(\bm{I} - \bm{P}_{-1:k})\bm{X}_{1:k}(\bm{a} - \bm{\beta}_{1:k})}^2\right) \\
&= \argmin_{\bm{\beta} \in \R^d} U(\bm{a}, \bm{\beta}),
\end{align*}
we see that $\glasso_{1:k} = \glasso_{1:k}(\olsk) = \glasso_{1:k}(\sigmahat (\bm{V}_{1:k}^T \bm{X}_{1:k})^{-1}\bm{u}_{1:k})$ by Lemma \ref{lma:OLS_unit}. Then, by the chain rule, we have that 
$$\nabla \glasso_{1:k}(\bm{u}_{1:k}) = \nabla\glasso_{1:k}(\olsk)\cdot \nabla \olsk(\bm{u}_{1:k}).$$
$\nabla \olsk(\bm{u}_{1:k}) = \sigmahat (\bm{V}_{1:k}^T \bm{X}_{1:k})^{-1}$, so all we need is $\nabla\glasso_{1:k}(\bm{a})$. 
We start by establishing that $\glasso(\bm{a})$ is continuous in $\bm a$ by invoking Lemma 3 from \cite{Sengupta2024}, the proof of which relies on an application of Berge's maximum theorem. The first assumption of this lemma is satisfied since $U(\bm{a}, \bm{\beta})$ is a continuous function of its arguments. For the second assumption, let $\bm{a}^* \in \R^k$, $d > 0$, and $\epsilon = \sup_{\bm{a} \in B(\bm{a}^*, d)} U(\bm{a}, \bm{0})$, and note that $U(\bm{a}, \bm{\beta}) \geq \tilde{f} \geq 2n\lambda(\norm{\bm{\beta}_{1:k}}_2 + \norm{\bm{\beta}_{-1:k}}_1) \geq 2n \lambda \norm{\bm{\beta}}$. Take $M = \frac{\epsilon}{2n\lambda}$ and note that if $\norm{\bm \beta} > M$, then $\inf_{\bm{a} \in B(\bm{a}^*, d)} U(\bm{a}, \bm{\beta}) \geq 2n\lambda \norm{\bm{\beta}} > \epsilon$, as desired. It follows that $\glasso(\bm{a})$ is continuous in $\bm a$.

Assume that $\bm u_{1:k}$ is such that $\glasso_{1:k} \neq \bm 0$, and let $\mathcal{A}(\bm{a}) = \left\{i \in \{k + 1, \ldots, d\} : \glassos_i(\bm{a}) \neq 0\right\}$. Using $[k]$ as shorthand for $\{1, \ldots, k\}$, since $\glasso(\bm{a})$ minimizes $U$, it follows that $\glasso_{[k] \cup \mathcal{A}(\bm{a})}$ is a minimizer of $U(\bm{a}, \bm{\gamma})$, where $\bm{\gamma}$ is a vector of length $k + |\mathcal{A}(\bm{a})|$. Here we abuse notation slightly: $U(\bm{a}, \bm{\beta})$ is originally defined with $\bm \beta \in \R^p$ as its second argument, while $\bm \gamma$ denotes the subvector of parameters indexed by $[k] \cup \mathcal{A}(\bm{a})$. Thus, $U(\bm{a}, \bm{\gamma})$ is to be interpreted as evaluating $U$ at the vector in $\R^p$ that agrees with $\bm \gamma$ on $[k] \cup \mathcal{A}(\bm{a})$ and is zero elsewhere.

Because all of the entries of $\glasso_{[k] \cup \mathcal{A}(\bm{a})}$ are non-zero, $U(\bm{a}, \bm{\gamma})$ is differentiable in $\bm \gamma$ at $\glasso_{[k] \cup \mathcal{A}(\bm{a})}$. The first-order conditions indicate that for $i \in \mathcal{A}(\bm{a})$, 
\begin{align*}
 \frac{\partial U}{\partial \gamma_i}\biggr\vert_{\bm{\gamma} = \glasso_{[k] \cup \mathcal{A}(\bm{a})}}   &= 0 \\
\implies -\bm{X}_i^T\left(\yhat - \bm{P}_{-1:k} \bm{X}_{1:k}\glasso_{1:k}(\bm{a}) - \bm{X}_{\mathcal{A}(\bm{a})}\glasso_{\mathcal{A}(\bm{a})}(\bm{a})\right) + n\lambda \sign(\glassos_i(\bm{a})) &= 0.
\end{align*}
Thus,
\begin{align*}
-\bm{X}_{\mathcal{A}(\bm{a})}^T\left(\yhat - \bm{P}_{-1:k} \bm{X}_{1:k}\glasso_{1:k}(\bm{a}) - \bm{X}_{\mathcal{A}(\bm{a})}\glasso_{\mathcal{A}(\bm{a})}(\bm{a})\right) + n\lambda \sign(\glasso_{\mathcal{A}(\bm{a})}(\bm{a})) = \bm{0}.
\end{align*}
Note that $\glasso(\bm a)$ is continuous in the neighborhood of $\bm a$ in which the active set $\mathcal A(\bm a)$ does not change and hence the signs of the active variables also do not change. So, differentiating both sides of the equation above with respect to $\bm{a}$ yields
\begin{align}
\bm{X}_{\mathcal{A}(\bm{a})}^T\left(\bm{P}_{-1:k} \bm{X}_{1:k} \nabla \glasso_{1:k}(\bm{a}) + \bm{X}_{\mathcal{A}(\bm{a})} \nabla \glasso_{\mathcal{A}(\bm{a})}(\bm{a})\right)    &= \bm{0} \nonumber \\
\implies  \nabla \glasso_{\mathcal{A}(\bm{a})}(\bm{a}) &= -(\bm{X}_{\mathcal{A}(\bm{a})}^T\bm{X}_{\mathcal{A}(\bm{a})})^{-1}\bm{X}_{\mathcal{A}(\bm{a})}^T\bm{P}_{-1:k} \bm{X}_{1:k} \nabla \glasso_{1:k}(\bm{a})
\label{eq:sub_1}
\end{align}
Similarly, for $i \in [k]$,
\begin{align*}
\frac{\partial U}{\partial \gamma_i}
\biggr\vert_{\bm{\gamma} = \glasso_{[k] \cup \mathcal{A}(\bm{a})}} &= 0 \\
\implies
 -\bm{X}_i^T\left(
    \yhat
    - \bm{P}_{-1:k} \bm{X}_{1:k}\glasso_{1:k}(\bm{a})
    - \bm{X}_{\mathcal{A}(\bm{a})}\glasso_{\mathcal{A}(\bm{a})}(\bm{a})
  \right) \\
  + \; n\lambda
    \frac{\glassos_i(\bm{a})}{\norm{\glasso_{1:k}(\bm{a})}}
  - \bm{X}_i^T(\bm{I} - \bm{P}_{-1:k})\bm{X}_{1:k}
    \left(\bm{a} - \glasso_{1:k}(\bm{a})\right) &= 0 \\
\implies 
 -\bm{X}_{1:k}^T\left(
    \yhat
    - \bm{P}_{-1:k} \bm{X}_{1:k}\glasso_{1:k}(\bm{a})
    - \bm{X}_{\mathcal{A}(\bm{a})}\glasso_{\mathcal{A}(\bm{a})}(\bm{a})
  \right)& \\
 + \; n\lambda
    \frac{\glasso_{1:k}(\bm{a})}{\norm{\glasso_{1:k}(\bm{a})}}
  - \bm{X}_{1:k}^T(\bm{I} - \bm{P}_{-1:k})\bm{X}_{1:k}
    \left(\bm{a} - \glasso_{1:k}(\bm{a})\right) &= 0
\end{align*}
Differentiating again with respect to $\bm{a}$ yields 
\begin{align}
 \bm{X}_{1:k}^T\bm{P}_{-1:k} \bm{X}_{1:k}\nabla \glasso_{1:k}(\bm{a})
    +  \bm{X}_{1:k}^T\bm{X}_{\mathcal{A}(\bm{a})}\nabla \glasso_{\mathcal{A}(\bm{a})}(\bm{a})& \nonumber \\
 + \; n\lambda
    \frac{1}{\norm{\glasso_{1:k}(\bm{a})}}\left(\bm{I} - \frac{\glasso_{1:k}(\bm{a})\glasso_{1:k}(\bm{a})^T}{\norm{\glasso_{1:k}(\bm{a})}^2}\right)\nabla \glasso_{1:k}(\bm{a})& \nonumber \\
  - \bm{X}_{1:k}^T(\bm{I} - \bm{P}_{-1:k})\bm{X}_{1:k} + \bm{X}_{1:k}^T(\bm{I} - \bm{P}_{-1:k})\bm{X}_{1:k}\nabla \glasso_{1:k}(\bm{a}) &= 0
  \label{eq:sub_2}
\end{align}
Substituting \eqref{eq:sub_1} into \eqref{eq:sub_2} and rearranging, we get
\begin{align*}
\nabla \glasso_{1:k}(\bm{a}) &= \left(\bm{X}_{1:k}^T(\bm{I} - \bm{P}_{\mathcal{A}(\bm{a})})\bm{X}_{1:k} + \frac{n\lambda}{\norm{\glasso_{1:k}}} \left(\bm{I} - \frac{\glasso_{1:k}(\glasso_{1:k})^T}{\norm{\glasso_{1:k}}^2}\right)\right)^{-1}\bm{X}_{1:k}^T(\bm{I} - \bm{P}_{-1:k})\bm{X}_{1:k}
\end{align*}
Note that $\bm{X}_{1:k}^T(\bm{I} - \bm{P}_{-1:k})\bm{X}_{1:k} = \bm{X}_{1:k}^T\bm{V}\bm{V}^T\bm{X}_{1:k} = \bm{X}_{1:k}^T\bm{V}_{1:k}\bm{V}_{1:k}^T\bm{X}_{1:k}$ since $\bm{V}\bm{V}^T = (\bm{I} - \bm{P}_{-1:k})$ and the columns of $\bm{V}_{-1:k}$ are orthogonal to the column space of $\bm{X}$ by definition of $\bm{V}$. So, combining the above with $\nabla \olsk(\bm{u}_{1:k})$ delivers the result. 
\end{proof}

\subsection{\sectionmath{Inverse of $\glasso_{1:k}$ as a function of $\bm{u}_{1:k}$}{Inverse of group LASSO as a function of u}}
\label{sec:thm_glasso_unit_pf}
\begin{theorem}
\label{thm:glasso_unit}
For any $\bm{b} \in \mathbb{R}^k$ and $\bm{\epsilon} \in B_k(\bm 0, 1)$, the unit ball centered at the origin in $\mathbb R^k$, define the functions
\begin{align}
\label{eq:obj}
\glasso_{-1:k}(\bm{b}) &= \argmin_{\bm{\beta}_{-1:k} \in \R^{d-k}}\left(\frac{1}{2n}\norm{\bm{y} - \bm{X}_{1:k}\bm{b} - \bm{X}_{-1:k}\bm{\beta}_{-1:k}}_2^2 + \lambda\norm{\bm{\beta}_{-1:k}}_1\right) \\
\Lambda_{1:k}(\bm{b}, \bm{\epsilon}) &= \frac{1}{\sigmahat} (\bm{X}_{1:k}^T \bm{V}_{1:k})^{-1}\left(-\bm{X}_{1:k}^T(\yhat - \bm{X}_{1:k}\bm{b} -\bm{X}_{-1:k}\glasso_{-1:k}(\bm{b})) + n \lambda \bm{\epsilon}\right) \nonumber
\end{align}
Then, the function $f_{\bm{S}^{(1:k)}}: \mathbb{R}^k \to \mathbb{R}^k$ defined via its inverse as 
\begin{align*}
f_{\suffstat}^{-1}(\bm{b}) &=
\begin{cases}
    \Lambda_{1:k}(\bm{b}, \frac{\bm{b}}{\|\bm{b}\|}) \text{, if $\bm{b} \neq \bm{0}$} \\
    \left\{\Lambda_{1:k}(\bm{0}, \bm{\epsilon}):\|\bm{\epsilon}\| \leq 1\right\} \text{, if $\bm{b} = \bm{0}$}
\end{cases}
\end{align*}
satisfies $\glasso_{1:k} = f_{\suffstat}(\bm{u}_{1:k})$. 
\end{theorem}
\begin{proof}
In order to prove Theorem \ref{thm:glasso_unit}, we first need the intermediate result below, following the same approach taken in the proof of Theorem 2.1 in \cite{Sengupta2024}. 
\begin{lemma}
\label{lma:equiv_statements}
For data $(\bm{y}, \bm{X})$ with $\bm{X}$ full column-rank and $\lambda > 0$, define $\glasso_*(\bm{b})$ to be $\bm{b}$ for the first $k$ coordinates and $\glasso_{-1:k}(\bm{b})$ on the rest, where $\glasso_{-1:k}(\bm{b})$ is defined as in the statement of Theorem \ref{thm:glasso_unit}. Also let $$f_\lambda(\bm{y}; \bm{\beta}) = \frac{1}{2n}\norm{\bm{y} - \bm{X}\bm{\beta}}_2^2 + \lambda(\norm{\bm{\beta}_{1:k}}_2 + \norm{\bm{\beta}_{-1:k}}_1)$$
be the group LASSO objective function. Then, for any $\bm{b}\in \R^k$, the following are equivalent statements:
\begin{enumerate}
    \item $\bm{0} \in \partial_{\bm{\beta}_{1:k}}f_\lambda(\bm{y}; \bm{\beta})\vert_{\bm{\beta} = \glasso_*(\bm{b})}$
    \item $\glasso = \glasso_*(\bm{b})$
    \item $\glasso_{1:k} = \bm{b}$
\end{enumerate}
\end{lemma}
\begin{proof}[Proof of Lemma \ref{lma:equiv_statements}]
First, we show equivalence between (1) and (2), so assume $\bm{0} \in \partial_{\bm{\beta}_{1:k}}f_\lambda(\bm{y}; \bm{\beta})\vert_{\bm{\beta} = \glasso_*(\bm{b})}$. The function $f_{\lambda}(\bm{y}; (\bm{\beta}_{-1:k} = \glasso_{-1:k}(\bm{b}), \bm{\beta}_{1:k}))$ is convex in $\bm{\beta}_{1:k}$, so since $\bm{0} \in \partial_{\bm{\beta}_{1:k}}f_\lambda(\bm{y}; \bm{\beta})\vert_{\bm{\beta} = \glasso_*(\bm{b})} \iff \bm{0} \in\partial_{\bm{\beta}_{1:k}}f_{\lambda}(\bm{y}; (\bm{\beta}_{-1:k} = \glasso_{-1:k}(\bm{b}), \bm{\beta}_{1:k}))\vert_{\bm{\beta}_{1:k} = \bm{b}}$,
\begin{align*}
    \bm{b} &= \argmin_{\bm{\beta}_{1:k}} \left(f_\lambda(\bm{y}; (\bm{\beta}_{-1:k} = \glasso_{-1:k}(\bm{b}), \bm{\beta}_{1:k}))\right).
\end{align*}
By definition of $\glasso_{-1:k}(\bm{b})$, we also have that
\begin{align*}
\glasso_{-1:k}(\bm{b}) &= \argmin_{\bm{\beta}_{-1:k}}\left(f_\lambda(\bm{y}; (\bm{\beta}_{-1:k}, \bm{\beta}_{1:k} = \bm{b}))\right).
\end{align*}
The above implies that running a blockwise coordinate descent with the blocks $[1:k]$ and $\{[j]\}_{(k+1) \leq j \leq d}$ starting at $\glasso_*(\bm{b})$ will result in iterates being constant at $\glasso_*(\bm{b})$. The algorithm for the optimization of the sparse group LASSO taking $\lambda_2 = 0$ in \cite{Tibshirani2010} implies that blockwise coordinate descent over groups will result in the minimizer. Thus, $\glasso = \glasso_*(\bm{b})$. Assuming (2), since $\glasso$ is the minimizer of the objective $f_\lambda(\bm{y}; \bm{\beta})$, the sub-gradient of the objective with respect to each coordinate, in particular the first $k$, evaluated at the minimizer must contain 0. Since $\glasso = \glasso_*(\bm{b})$, condition (1) follows.

We now show equivalence between (2) and (3). That (2) implies (3) follows from the definition of $\glasso_*(\bm{b})$. So, assuming (3), the blockwise coordinate descent argument from before implies 
\begin{align*}
    \glasso_{-1:k} = \argmin_{\bm{\beta}_{-1:k}}\left(f_\lambda(\bm{y}; (\bm{\beta}_{-1:k}, \bm{\beta}_{1:k} = \glasso_{1:k}))\right)
\end{align*}
Then since $\glasso_{1:k} = \bm{b}$ by assumption, we have that $\glasso_{-1:k} = \glasso_{-1:k}(\bm{b})$. So, $\glasso = \glasso_*(\bm{b})$.
\end{proof}
We now proceed with the proof of Theorem \ref{thm:glasso_unit} by using the decomposition in \eqref{eq:decomp} from Lemma \ref{lma:decomp} to decompose the group LASSO objective function as follows
\begin{align*}
f_{\lambda}(\bm{y}; \bm{\beta}) &= \frac{1}{2n}\norm{\bm{y} - \bm{X} \bm{\beta}}_2^2 + \lambda \norm{\bm{\beta}_{1:k}}_2 + \lambda\norm{\bm{\beta}_{-1:k}}_1 \\
&= \frac{1}{2n}\norm{\yhat - \bm{X}_{-1:k} \bm{\beta}_{-1:k} - \bm{P}_{-1:k}\bm{X}_{1:k}\bm{\beta}_{1:k}}_2^2 + \frac{1}{2n}\norm{\sigmahat\bm{V}\bm{u} - (\bm{I} - \bm{P}_{-1:k})\bm{X}_{1:k}\bm{\beta}_{1:k}}_2^2 \\
&+\lambda \norm{\bm{\beta}_{1:k}}_2 + \lambda\norm{\bm{\beta}_{-1:k}}_1.
\end{align*}
The subdifferential of $f_{\lambda}(\bm{y}; \bm{\beta})$ with respect to $\bm{\beta}_{1:k}$ is given by
\begin{align*}
\partial_{\bm{\beta}_{1:k}}f_{\lambda}(\bm{y}; \bm{\beta}) &= 
-\frac{1}{n}(\bm{P}_{-1:k}\bm{X}_{1:k})^T\left(\yhat - \bm{X}_{-1:k} \bm{\beta}_{-1:k} - \bm{P}_{-1:k}\bm{X}_{1:k}\bm{\beta}_{1:k}\right) \\
&- \frac{1}{n}((\bm I - \bm{P}_{-1:k})\bm{X}_{1:k})^T\left(\sigmahat\bm{V}\bm{u} - (\bm{I} - \bm{P}_{-1:k})\bm{X}_{1:k}\bm{\beta}_{1:k}\right) \\
&+ \lambda s(\bm{\beta}_{1:k}) \text{, where}\\
s(\bm{\beta}_{1:k}) &= 
\begin{cases}
\frac{\bm{\beta}_{1:k}}{\norm{\bm{\beta}_{1:k}}} \text{, if $\bm{\beta}_{1:k} \neq \bm{0}$} \\
\{\bm{v} \in \mathbb{R}^k: \norm{\bm{v}} \leq 1\} \text{, otherwise}
\end{cases}
\end{align*}
The above simplifies as follows
\begin{align*}
\partial_{\bm{\beta}_{1:k}}f_{\lambda}(\bm{y}; \bm{\beta}) &= -\frac{1}{n}\bm{X}_{1:k}^T\yhat - \frac{\sigmahat}{n}\bm{X}_{1:k}^T\bm{V} \bm{u} + \frac{\bm{X}_{1:k}^T \bm{X}_{-1:k}\bm{\beta}_{-1:k} + \bm{X}^T_{1:k} \bm{X}_{1:k}\bm{\beta}_{1:k}}{n} + \lambda s(\bm{\beta}_{1:k}) \\
&= -\frac{1}{n}\bm{X}_{1:k}^T(\yhat - \bm{X}\bm{\beta}) - \frac{\sigmahat}{n}\bm{X}_{1:k}^T\bm{V} \bm{u} + \lambda s(\bm{\beta}_{1:k}). 
\end{align*}
We can write $\bm{X}_{1:k}^T\bm{V} \bm{u}$ as $\bm{X}_{1:k}^T\bm{V}_{1:k} \bm{u}_{1:k}$ since the columns of $\bm{V}_{-1:k}$ are orthogonal to the column space of $\bm{X}$. We then have that
\begin{align*}
\partial_{\bm{\beta}_{1:k}}f_{\lambda}(\bm{y}; \bm{\beta}) &= -\frac{1}{n}\bm{X}_{1:k}^T(\yhat - \bm{X}\bm{\beta}) - \frac{\sigmahat}{n}\bm{X}_{1:k}^T\bm{V}_{1:k} \bm{u}_{1:k} + \lambda s(\bm{\beta}_{1:k}). 
\end{align*}
Setting $\bm{\beta} = \glasso*(\bm{b})$ for $\bm{b} \in \R^k$ and applying Lemma \ref{lma:equiv_statements}, we have the following equivalences
\begin{align*}
\glasso_{1:k} = \bm{b} &\iff \bm{0} \in \partial_{\bm{\beta}_{1:k}} f_{\lambda}(\bm{y}; \bm{\beta})|_{\bm{\beta} = \glasso_*(\bm{b})} \\
&\iff \bm{u}_{1:k} \in \frac{1}{\sigmahat} (\bm{X}_{1:k}^T \bm{V}_{1:k})^{-1}\left(-\bm{X}_{1:k}^T(\yhat - \bm{X}_{1:k}\bm{b} -\bm{X}_{-1:k}\glasso_{-1:k}(\bm{b})) + n \lambda s(\bm{b})\right).
\end{align*}
Equivalently, using $\Lambda_{1:k}(\bm{b}, \bm{\epsilon})$ as defined in the statement of Theorem \ref{thm:glasso_unit}, we find that $\glasso_{1:k} = f_{\suffstat}(\bm{u}_{1:k})$, where
\begin{align*}
f_{\suffstat}^{-1}(\bm{b}) &=
\begin{cases}
    \Lambda_{1:k}(\bm{b}, \frac{\bm{b}}{\norm{\bm{b}}}) \text{, if $\bm{b} \neq \bm{0}$} \\
    \left\{\Lambda_{1:k}(\bm{0}, \bm{\epsilon}):\norm{\bm{\epsilon}} \leq 1\right\} \text{, if $\bm{b} = \bm{0}$}
\end{cases}.
\end{align*}
\end{proof}

\subsection{Proof of Lemma \ref{lma:deriv}}
\label{sec:lma_deriv_pf}
\begin{proof}
From Theorem \ref{thm:glasso_unit}, the function $f_{\suffstat}^{-1}(\bm{b}; r)$ specified in the lemma's statement is given by $$f_{\suffstat}^{-1}(\bm{b}; r) = \frac{1}{\sigmahat} (\bm{X}_{1:k}^T \bm{V}_{1:k})^{-1}\left(-\bm{X}_{1:k}^T(\yhat - \bm{X}_{1:k}\bm{b} -\bm{X}_{-1:k}\glasso_{-1:k}(\bm{b})) + n \lambda \frac{\bm{b}}{r}\right).$$
We start by establishing that $\glasso_{-1:k}(\bm b)$ is continuous in $\bm b$. Let $U(\bm{b}, \bm{\beta}_{-1:k})$ denote the objective function in \eqref{eq:obj}, and note that it is a continuous function of its arguments. Also, $U(\bm{b}, \bm{\beta}_{-1:k}) \geq \lambda \norm{\bm{\beta}_{-1:k}}$. So letting $\bm{b}^* \in \R^k$, $ d > 0$, and $\epsilon = \sup_{\bm{b} \in B(\bm{b}^*, d)}U(\bm{b}, \bm{0})$, if we take $M = \frac{\epsilon}{\lambda}$, then $\norm{\bm{\beta}_{-1:k}} > M$ implies that $\inf_{\bm{b} \in B(\bm{b}^*, d)} U(\bm{b}, \bm{\beta}_{-1:k}) \geq \lambda \norm{\bm{\beta}_{-1:k}} > \epsilon$. By Lemma 3 from \cite{Sengupta2024}, $\glasso_{-1:k}(\bm{b})$ is continuous in $\bm b$.

We will have that $\nabla \glasso_{-1:k}(\bm{b}) = \bm{0}$ if $\glasso_{-1:k}(\bm{b}) = \bm{0}$ so consider a value of $\bm b \in \R^k$ such that $\glasso_{-1:k}(\bm{b}) \neq \bm{0}$, and define the active set $\mathcal{A}(\bm{b}) = \left\{i \in \{k+1, \ldots, d\} : \glassos_{i}(\bm b) \neq 0\right\}$. Since $\glasso_{-1:k}(\bm{b})$ minimizes $U$, it follows that $\glasso_{\mathcal{A}(\bm{b})}(\bm{b})$ is a minimizer of $U(\bm{b}, \bm{\gamma})$, where $\bm{\gamma}$ is a vector of length $|\mathcal{A}(\bm{b})|$. Because all of the entries of $\glasso_{\mathcal{A}(\bm{b})}(\bm{b})$ are non-zero, $U(\bm{b}, \bm{\gamma})$ is differentiable in $\bm \gamma$ at $\glasso_{\mathcal{A}(\bm{b})}(\bm{b})$. Thus, the first-order conditions indicate that for $i \in \mathcal{A}(\bm{b})$, 
\begin{align*}
    \frac{\partial}{\partial \gamma_i} U(\bm{b}, \bm{\gamma})\biggr\vert_{\bm{\gamma} = \glasso_{\mathcal{A}(\bm{b})}(\bm{b})} &= 0 \\
    \implies -\bm{X}_i^T(\bm{y} - \bm{X}_{1:k}\bm{b} - X_{\mathcal{A}(\bm{b})}\glasso_{\mathcal{A}(b)}(\bm{b})) + 2n\lambda \sign(\glassos_i(\bm{b})) &= 0
\end{align*}
Note that $\glasso_{-1:k}(\bm{b})$ is continuous in the neighborhood of $\bm{b}$ in which the active set $\mathcal{A}(\bm{b})$ does not change and hence the signs of the active variables also do not change. So, differentiating both sides of the equation above with respect to $\bm{b}$ yields,
\begin{align*}
 \bm{X}_i^T \bm{X}_{1:k} + \bm{X}_i^T\bm{X}_{\mathcal{A}(\bm{b})} \nabla \glasso_{\mathcal{A}(\bm{b})}(\bm{b}) &= 0 \text{, for all $i \in \mathcal{A}(\bm{b})$} \\
 \implies  \bm{X}_{\mathcal{A}(\bm{b})}^T \bm{X}_{1:k} + \bm{X}_{\mathcal{A}(\bm{b})}^T\bm{X}_{\mathcal{A}(\bm{b})} \nabla \glasso_{\mathcal{A}(\bm{b})}(\bm{b}) &= 0 \\
 \implies \nabla \glasso_{\mathcal{A}(\bm{b})}(\bm{b}) &= -(\bm{X}_{\mathcal{A}(\bm{b})}^T\bm{X}_{\mathcal A(\bm{b})})^{-1}\bm{X}_{\mathcal A(\bm{b})}^T\bm{X}_{1:k}.
\end{align*}
With $\nabla \glasso_{-1:k}(\bm{b})$ determined, the result follows from simple manipulations.
\end{proof}

\subsection{Proof of Corollary \ref{cor:affine_piece}}
\label{sec:cor_affine_piece_pf}
\begin{proof}
By Lemma \ref{lma:deriv}, for $\bm b$ in a sufficiently small neighborhood of $\bm{b}^*$ such that $\norm{\bm{b}} = \norm{\bm{b}^*}$, 
$$f_{\suffstat}^{-1}(\bm{b}; \norm{\bm{b}^*}) = \nabla f_{\suffstat}^{-1}(\bm{b}^*; \norm{\bm{b}^*})\bm{b} + \bm{\nu}(\bm{b}^*)\text{, where } \bm{\nu}(\bm{b}^*) = f_{\suffstat}^{-1}(\bm{b}^*; \norm{\bm{b}^*}) - \nabla f_{\suffstat}^{-1}(\bm{b}^*; \norm{\bm{b}^*})\bm{b}^*.$$
All we need to show is that $\bm{\nu}(\bm{b}^*)$ simplifies to the form specified in the corollary's statement. Observe that
\begin{align*}
f_{\suffstat}^{-1}(\bm{b}^*; \norm{\bm{b}^*}) - \nabla f_{\suffstat}^{-1}(\bm{b}^*; \norm{\bm{b}^*})\bm{b}^* &= \frac{1}{\sigmahat} (\bm{X}_{1:k}^T \bm{V}_{1:k})^{-1} \\
&\cdot \left(-\bm{X}_{1:k}^T(\yhat - \bm{X}_{1:k}\bm{b}^* -\bm{X}_{-1:k}\glasso_{-1:k}(\bm{b}^*)) + n \lambda \frac{\bm{b}^*}{\norm{\bm{b}^*}}\right) \\
&- \frac{1}{\sigmahat} (\bm{X}_{1:k}^T \bm{V}_{1:k})^{-1}\left[\bm{X}_{1:k}^T(\bm{I} - \bm{P}_{\mathcal{A}(\bm{b}^*)})\bm{X}_{1:k} + \frac{n\lambda}{\norm{\bm{b}^*}} \bm{I}\right]\bm{b}^* \\
&= -\frac{1}{\sigmahat} (\bm{X}_{1:k}^T \bm{V}_{1:k})^{-1}\bm{X}_{1:k}^T\yhat + \frac{1}{\sigmahat} (\bm{X}_{1:k}^T \bm{V}_{1:k})^{-1}\bm{X}_{1:k}^T\bm{X}_{-1:k}\glasso_{-1:k}(\bm{b}^*) \\
&+ \frac{1}{\sigmahat} (\bm{X}_{1:k}^T \bm{V}_{1:k})^{-1}\bm{X}_{1:k}^T\bm{P}_{\mathcal{A}(\bm{b}^*)}\bm{X}_{1:k}\bm{b}^* \\
&= -\frac{1}{\sigmahat} (\bm{X}_{1:k}^T \bm{V}_{1:k})^{-1}\bm{X}_{1:k}^T(\yhat-\bm{X}_{-1:k}\glasso_{-1:k}(\bm{b}^*)-\bm{P}_{\mathcal{A}(\bm{b}^*)}\bm{X}_{1:k}\bm{b}^*) \\
&= -\frac{1}{\sigmahat} (\bm{V}_{1:k}^T\bm{X}_{1:k})\bm{S}^{-1}\bm{X}_{1:k}^T(\yhat-\bm{X}_{-1:k}\glasso_{-1:k}(\bm{b}^*)-\bm{P}_{\mathcal{A}(\bm{b}^*)}\bm{X}_{1:k}\bm{b}^*).
\end{align*}
Using the notation in the corollary statement, the result follows from the line above.
\end{proof}

\subsection{Proof of Proposition \ref{thm:density}}
\label{sec:thm_density_pf}
\begin{proof}
Since $\bm{u}_{1:k}$ is rotationally invariant, 
\begin{align*}
\norm{\bm{u}_{1:k} - \bm{c}} \sim \norm{\bm{u}_{1:k} - \norm{\bm{c}}\bm{e}_1},
\end{align*}
where $\bm e_1 = (1, 0, \ldots, 0) \in \R^{k}$. Let $\bm{u}_{1:k} = \frac{\bm{Z}_{1:k}}{\norm{\bm{Z}}}$, where $\bm{Z} \sim \mathcal N(\bm{0}, \bm{I}_{n-d+k})$. Then, 
\begin{align*}
\norm{\bm{u}_{1:k} - \norm{\bm{c}}\bm{e}_1}^2 &= (\bm{u}_{1:k}^T - \norm{\bm{c}}\bm{e}_1^T)(\bm{u}_{1:k} - \norm{\bm{c}}\bm{e}_1) \\
&= \norm{\bm{u}_{1:k}}^2 - 2\norm{\bm{c}}u_1 + \norm{\bm{c}}^2 \\
&= \frac{Z_1^2 + \ldots + Z_k^2}{Z_1^2 + \ldots + Z_{n-d+k}^2} - 2\norm{\bm{c}}\frac{Z_1}{\sqrt{Z_1^2 + \ldots + Z_{n-d+k}^2}} + \norm{\bm{c}}^2.
\end{align*}
Let $Z_1 = S X$, where $S$ is a random sign and $X \sim \chi_1$, so that $\{S, X, Z_2, \ldots, Z_{n-d+k}\}$ are jointly independent. Also, let $G_1 = Z_1^2$, $G_2 = Z_2^2 +\ldots + Z_k^2$, and $G_3 = Z_{k+1}^2 + \ldots + Z_{n-d+k}^2$. Then, we have that
\begin{align*}
\frac{Z_1^2 + \ldots + Z_k^2}{Z_1^2 + \ldots + Z^2_{n-d+k}} &= \frac{G_1 + G_2}{G_1 + G_2 + G_3} \\
&=B_2 \sim \Beta\left(\frac{k}{2}, \frac{n-d}{2}\right) \\
\frac{Z_1}{\sqrt{Z_1^2 + \ldots + Z^2_{n-d+k}}} &= \frac{SX}{\sqrt{Z_1^2 + \ldots + Z^2_{n-d+k}}} \\
&= S \sqrt{\frac{Z_1^2}{Z_1^2 + \ldots + Z_{n-d+k}^2}} \\
&= S \sqrt{\frac{G_1}{G_1 + G_2 + G_3}} \\
&= S\sqrt{B_1 B_2} \text{, where $B_1 = \frac{G_1}{G_1+G_2}\sim \Beta\left(\frac{1}{2}, \frac{k-1}{2}\right)$} \\
\implies \norm{\bm{u}_{1:k} - \norm{\bm{c}}\bm{e}_1}^2 &= B_2 - 2 \norm{\bm{c}}S \sqrt{B_1 B_2} + \norm{\bm{c}}^2.
\end{align*}
Note that $B_1 \perp B_2$ is a standard result, and $S \perp \{B_1, B_2\}$ since $B_1$ and $B_2$ are functions of $\{X, Z_2, \ldots, Z_{n-d+k}\}$ and $S$ was chosen to be independent of $\{X, Z_2, \ldots, Z_{n-d+k}\}$. Thus, $\{S, B_1, B_2\}$ are jointly independent. Letting $W = \norm{\bm{u}_{1:k} -\norm{\bm{c}} \bm{e}_1}^2$, the density $f_W(w)$ is
\begin{align*}
f_W(w) &= f_{W\mid S = 1}(w)\frac{1}{2} + f_{W\mid S = -1}(w)\frac{1}{2} \\
&=\frac{1}{2}(f_{W_1}(w) + f_{W_2}(w)) \text{, where} \\
W_1 &= B_2 + 2 \norm{\bm{c}}\sqrt{B_1B_2} +  \norm{\bm{c}}^2  \\
W_2 &= B_2 - 2 \norm{\bm{c}}\sqrt{B_1B_2} +  \norm{\bm{c}}^2.
\end{align*}
We will now determine the densities of $W_1$ and $W_2$. Define the invertible function
\begin{align*}
    g_1(B_1, B_2) &= (W_1, B_2) \\
    \implies g_1^{-1}(W_1, B_2) &= \left(\frac{1}{B_2}\left(\frac{W_1 -B_2 - \norm{\bm{c}}^2}{2\norm{\bm{c}}}\right)^2, B_2\right). 
\end{align*}
By change of variables, 
\begin{align*}
f_{W_1, B_2}(w, t) &= f_{B_1, B_2}(g_1^{-1}(w, t)) \left|\frac{w-t-\norm{\bm{c}}^2}{2t\norm{\bm{c}}^2}\right| \\
&= h(\sqrt{w}, t),
\end{align*}
where $h(w, t)$ is defined in the statement of Proposition \ref{thm:density}. It follows that
\begin{align*}
f_{W_1}(w)&=
\begin{cases}
    \displaystyle\int_{\norm{\bm{c}}^2-2\sqrt{w}\norm{\bm{c}}+w}^{w-\norm{\bm{c}}^2}h(\sqrt{w}, t)~dt &\text{if $\norm{\bm{c}}^2 \leq w \leq \norm{\bm{c}}^2 + 1$} \\
    \displaystyle\int_{\norm{\bm{c}}^2-2\sqrt{w}\norm{\bm{c}}+w}^{1}h(\sqrt{w}, t)~dt &\text{if $\norm{\bm{c}}^2 + 1 < w \leq (\norm{\bm{c}}+1)^2$} \\
    0 &\text{otherwise}
\end{cases}.
\end{align*}
Proceeding in a similar manner, define the invertible function
\begin{align*}
    g_2(B_1, B_2) &= (W_2, B_2) \\
    \implies g_{2}^{-1}(W_2, B_2) &= \left(\frac{1}{B_2}\left(\frac{\norm{\bm{c}}^2 + B_2 - W_2}{2\norm{\bm{c}}}\right)^2, B_2\right).
\end{align*}
By change of variables, 
\begin{align*}
f_{W_2, B_2}(w, t) &= f_{B_1, B_2}(g_2^{-1}(w, t))\left|\frac{\norm{\bm{c}}^2 + t - w}{2t\norm{\bm{c}}^2}\right| \\
&= h(\sqrt{w}, t).
\end{align*}
It follows that 
\begin{align*}
f_{W_2}^{\norm{\bm{c}} \leq 1}(w) &=
\begin{cases}
\displaystyle\int_{\norm{\bm{c}}^2-2\sqrt{w}\norm{\bm{c}} +w}^{\norm{\bm{c}}^2+2\sqrt{w}\norm{\bm{c}} +w} h(\sqrt{w}, t)~dt &\text{if $0 \leq w \leq (\norm{\bm{c}} - 1)^2$} \\
\displaystyle\int_{\norm{\bm{c}}^2 -2\sqrt{w}\norm{\bm{c}}+w}^1 h(\sqrt{w}, t)~dt &\text{if $(\norm{\bm{c}}-1)^2 < w \leq \norm{\bm{c}}^2$} \\
    \displaystyle\int_{w-\norm{\bm{c}}^2}^1 h(\sqrt{w}, t)~dt &\text{if $\norm{\bm{c}}^2 < w \leq \norm{\bm{c}}^2 + 1$} \\
    0 &\text{otherwise}
\end{cases} \\
f_{W_2}^{\norm{\bm{c}} > 1}(w) &=
\begin{cases}
\displaystyle\int_{\norm{\bm{c}}^2 -2\sqrt{w}\norm{\bm{c}}+w}^1 h(\sqrt{w}, t)~dt &\text{if $(\norm{\bm{c}}-1)^2 < w \leq \norm{\bm{c}}^2$} \\
    \displaystyle\int_{w-\norm{\bm{c}}^2}^1 h(\sqrt{w}, t)~dt &\text{if $\norm{\bm{c}}^2 < w \leq \norm{\bm{c}}^2 + 1$} \\
    0 &\text{otherwise}
\end{cases}.
\end{align*}
Putting these pieces together, we can obtain the density of $W$, and after applying change of variables, we find that the density of the random variable $Z = \sqrt{W} = \norm{\bm u_{1:k} - \norm{\bm{c}} \bm{e}_1}$ is
\begin{align*}
f_Z(z) &= \ind{\norm{\bm{c}} \leq 1}\left\{zg(z)\ind{z \in [0, \norm{\bm{c}}+1]}\right\} \\
&+ \ind{\norm{\bm{c}} > 1}\left\{zg(z)\ind{z \in [\norm{\bm{c}}-1, \norm{\bm{c}}+1]}\right\}\\
g(z) &=
\begin{cases}
    \displaystyle\int_{\norm{\bm{c}}^2-2z\norm{\bm{c}} +z^2}^{\norm{\bm{c}}^2+2z\norm{\bm{c}} +z^2} h(z, t)~dt &\text{if $0 \leq z \leq \left|\norm{\bm{c}} - 1\right|$} \\
    \displaystyle\int_{\norm{\bm{c}}^2-2z\norm{\bm{c}} +z^2}^1 h(z, t)~dt &\text{if $\left|\norm{\bm{c}}-1\right| < z \leq \norm{\bm{c}}+1$}
\end{cases}
\end{align*}
\end{proof}

\subsection{\sectionmath{Density of $\norm{\bm{u}_{1:k} - \bm{\nu}} \mid \suffstat$ under $\hyp$ in the special case $k = 1$}{Density of recentered test statistic conditional on S under H in the special case k = 1}}
\label{sec:special_case}
\begin{proof}
The approach taken in the proof of Proposition \ref{thm:density} does not apply to the special case $k = 1$, so this section addresses this case. For any $z \in \R$, with some rewriting, we have that
\begin{align}
\label{eq:interm_step}
P(|u_1 -c| \geq z) = 1-P(u_1 \leq z+c) + P(u_1 \leq -z+c).
\end{align}
By similar reasoning from the proof of Proposition \ref{thm:density}, $u_1 \sim S\sqrt{B}$ where $S$ is a random sign and $B \sim \Beta\left(\frac{1}{2}, \frac{n-d}{2}\right)$ with $S \perp B$. It follows that for $a \in \R$, 
\begin{align*}
    P(u_1 \leq a) &= \frac{1}{2}\left(P(\sqrt{B} \leq a) + P(-\sqrt{B} \leq a)\right) \\
    &= \frac{1}{2}\left(1 + \sign(a)F_{B}(a^2)\right), 
\end{align*}
where $F_{B}(\cdot)$ denotes the CDF of $B$. Taking $a = c\pm z$ and substituting into \eqref{eq:interm_step}, after some simplification, we find that $$P(|u_1 - c| \geq z) = 1 \mp \frac{1}{2}\sign(c \pm z)F_{B}\left((c \pm z)^2\right).$$
\end{proof}

\subsection{Oracle test equivalence}
\label{sec:oracle_equiv}
Denote the cumulative distribution function (CDF) for the distribution of $O_m$ from \eqref{eq:oracle_stat} conditional on $\suffstat$ under $\hyp$ evaluated at some $x \in \mathbb{R}$ by $$F_{O_m}(x \mid \suffstat) := P_{\hyp}(O_m \leq x \mid \suffstat).$$
We can then use $\bar{F}_{O_m}(O_m \mid \suffstat)$ as a valid $p$-value for $H_{1:k}$ that rejects for large $O_m$ conditional on $\suffstat$, where $$\bar{F}_{O_m}(x \mid \suffstat) := 1-F_{O_m}(x \mid \suffstat).$$
The following result formalizes how we can think of the oracle test as being equivalent in some sense to the test that uses the p-value $\bar{F}_{O_m}(O_m \mid \suffstat)$ in the limit as $m \to \infty$.
\begin{proposition}
\label{prop:equiv}
Let $T$ denote the one-sided $t$-test statistic for $H_1$ in model \eqref{eq:linear_model_oracle}, and let $F_T(\cdot)$ denote the CDF of the $t$-distribution with $n-d+k-1$ degrees of freedom so that $\bar{F}_{T}(T):=1-F_{T}(T)$ represents the p-value of the test. Also let $\bar{F}_{O_m}(O_m \mid \suffstat)$ denote the p-value associated with the conditional test for $\hyp$ in model \eqref{eq:linear_model} that rejects for large $O_m$. Then, $$\bar{F}_{O_m}(O_m \mid \suffstat)\overset{a.s.}{\to} \bar{F}_{T}(T) \text{ as $m \to \infty$}.$$
\end{proposition}
\begin{proof}
We start by recharacterizing the one-sided $t$-test for $H_1: \gamma_1 = 0$ in the model $(\bm{y}, \tilde{\bm{X}})$ specified in \eqref{eq:linear_model_oracle}. Note that the statistic $\suffstat$ is also the minimal sufficient statistic for this model under the null $H_1: \gamma_1 = 0$. In the proof of Lemma 2.2 of \cite{Sengupta2024}, it is shown that the one-sided $t$-test statistic $T$ is a monotone increasing function of $u_1$ conditional on $\suffstat$. So, the conditional one-sided $t$-test is equivalent to the conditional test using $u_1$. Since $u_1 = \frac{1}{\sigmahat} \frac{\tilde{\bm{X}}_1^T(\bm{I} - \bm{P}_{-1:k})}{\norm{(\bm{I} - \bm{P}_{-1:k})\tilde{\bm{X}}_1}}\bm{y}$, however, the conditional one-sided test using $u_1$ is the same as the conditional one-sided test using $\tilde{\bm{X}}_1^T(\bm{I} - \bm{P}_{-1:k}) \bm{y}$. Moreover, since $T$ is independent of $\bm{S}^{(1:k)}$, which follows from theory on ancillary statistics and is also proved in Appendix D.4 of \cite{Sengupta2024}, the standard one-sided $t$-test is equivalent to the conditional one-sided $t$-test and hence the conditional test using $\tilde{\bm{X}}_1^T(\bm{I} - \bm{P}_{-1:k}) \bm{y}$.

The square of the test statistic $O_m$ defined in \eqref{eq:oracle_stat} can be written as
\begin{align*}
\norm{\bm{u}_{1:k}}^2 + 2m\left(\frac{\bm{\beta}_{1:k}^T}{\norm{\bm{\beta}_{1:k}}}\right)\left(\bm{X}_{1:k}^T\bm{V}_{1:k}\right)\bm{u}_{1:k} + m^2 \norm{\left(\bm{V}_{1:k}^T\bm{X}_{1:k}\right)\left(\frac{\bm{\beta}_{1:k}}{\norm{\bm{\beta}_{1:k}}}\right)}^2
\end{align*}
The third term does not depend on the data $\bm{y}$, so the test that rejects for large values of $O_m$ conditional on $\suffstat$ is equivalent to the conditional test that rejects for large values of the test statistic
\begin{align*}
\left(\bm{u}_{1:k}^T + 2m\left(\frac{\bm{\beta}_{1:k}^T}{\norm{\bm{\beta}_{1:k}}}\right)\left(\bm{X}_{1:k}^T\bm{V}_{1:k}\right)\right)\bm{u}_{1:k}.
\end{align*}
Recall $\bm{u}_{1:k} = \frac{1}{\sigmahat}\bm{V}_{1:k}^T\bm{y}$ and $\bm{X}_{1:k}^T\bm{V}_{1:k}\bm{V}_{1:k}^T = \bm{X}_{1:k}^T\bm{V}\bm{V}^T = \bm{X}_{1:k}^T(\bm{I} - \bm{P}_{-1:k})$. The realized p-value for the test that rejects for large values of $O_m$ conditional on $\suffstat$ is
\begin{align*}
&P_{\hyp}\left(\left(\tilde{\bm{u}}_{1:k}^T + 2m\left(\frac{\bm{\beta}_{1:k}^T}{\norm{\bm{\beta}_{1:k}}}\right)\left(\bm{X}_{1:k}^T\bm{V}_{1:k}\right)\right)\tilde{\bm{u}}_{1:k} \geq \left(\bm{u}_{1:k}^T + 2m\left(\frac{\bm{\beta}_{1:k}^T}{\norm{\bm{\beta}_{1:k}}}\right)\left(\bm{X}_{1:k}^T\bm{V}_{1:k}\right)\right)\bm{u}_{1:k}\biggr\vert \suffstat\right) \\
&= P_{\hyp}\left(\frac{\bm{\beta}_{1:k}^T}{\norm{\bm{\beta}_{1:k}}}\left(\bm{X}_{1:k}^T\bm{V}_{1:k}\right)\tilde{\bm{u}}_{1:k} \geq \frac{\bm{\beta}_{1:k}^T}{\norm{\bm{\beta}_{1:k}}}\left(\bm{X}_{1:k}^T\bm{V}_{1:k}\right)\bm{u}_{1:k} + \frac{D}{m}\biggr\vert \suffstat \right) \\
&= P_{\hyp}\left(\frac{1}{\sigmahat}\frac{\bm{\beta}_{1:k}^T}{\norm{\bm{\beta}_{1:k}}}\bm{X}_{1:k}^T(\bm{I}-\bm{P}_{-1:k})\tilde{\bm{y}} \geq \frac{1}{\sigmahat}\frac{\bm{\beta}_{1:k}^T}{\norm{\bm{\beta}_{1:k}}}\bm{X}_{1:k}^T(\bm{I}-\bm{P}_{-1:k})\bm{y} + \frac{D}{m}\biggr\vert \suffstat \right) \text{, where} \\
D &= \frac{1}{2}\left(\norm{\bm{u}_{1:k}}^2 - \norm{\tilde{\bm{u}}_{1:k}}^2\right).
\end{align*}
Letting $A_m$ denote the event corresponding to the last probability expression above, we see that the sequence of events $A_1, A_2, \ldots$ are nested. So, by the continuity of probability, as $m \to \infty$ the probability above converges to 
\begin{align*}
P_{\hyp}\left(\frac{\bm{\beta}_{1:k}^T}{\norm{\bm{\beta}_{1:k}}}\bm{X}_{1:k}^T(\bm{I}-\bm{P}_{-1:k})\tilde{\bm{y}} \geq \frac{\bm{\beta}_{1:k}^T}{\norm{\bm{\beta}_{1:k}}}\bm{X}_{1:k}^T(\bm{I}-\bm{P}_{-1:k})\bm{y}\biggr\vert \suffstat\right).
\end{align*}
Since $\tilde{\bm{X}}_1 = \frac{\bm{X}_{1:k}\bm{\beta}_{1:k}}{\|\bm{\beta}_{1:k}\|}$, this is the p-value associated with the conditional test using $\tilde{\bm{X}}_1^T(\bm{I} - \bm{P}_{-1:k}) \bm{y}$ that we previously justified is equivalent to the one-sided $t$-test for $H_1:\gamma_1 = 0$ in the model $(\bm{y}, \tilde{\bm{X}})$ from \eqref{eq:linear_model_oracle}. Using $\bar{F}_{O_m}(\cdot \mid \suffstat)$ and $\bar{F}_{T}(\cdot)$ as defined in the statement of Proposition \ref{prop:equiv}, it follows that $$\bar{F}_{O_m}(O_m \mid \suffstat)\overset{a.s.}{\to} \bar{F}_{T}(T) \text{ as $m \to \infty$}.$$
\end{proof}

\section{\sectionmath{Choice of $\lambda$ and $\bm b^*$}{Choice of l and b}}
\label{sec:penalty}
In Section \ref{sec:construction}, we recommended choosing the tuning parameter $\lambda$ and $\bm b^*$ for the $L$-test by cross-validating on $(\tilde{\bm y}, \bm X)$, where $\tilde{\bm y}$ is a draw from the distribution $\bm y \mid \suffstat$ under $\hyp$. Taking a similar approach to \cite{Sengupta2024}, to motivate this choice, in the first section below, we consider two conventions used to pick the optimal $\lambda$ in CV (with $\bm b^*$ then taken to be the group LASSO estimate at that chosen $\lambda$) and conclude that the min rule delivers the best results. In the following section, we quantify the randomness that selecting $\lambda$ and $\bm b^*$ in this manner introduces into the $L$-test p-value. We also compare the $L$-test that performs CV on $(\tilde{\bm y}, \bm X)$ to the $L$-test that performs CV on $(\yhat, \bm{X})$, which would introduce no variability into the test. While cross-validating on the full dataset $(\bm y, \bm X)$ would invalidate the p-value, we still consider how this invalid test would do since it serves as a good benchmark for the performance we might hope to achieve with our valid $L$-test.

\subsection{Min rule vs. 1se rule}
In this section, we compare two conventions used to select the optimal hyperparameter in CV: the min rule, which selects the $\lambda$ corresponding to the minimum cross-validated error, and the 1se rule, which selects the largest $\lambda$ resulting in a cross-validated error within one standard deviation of the minimum cross-validated error. Figure \ref{fig:rules} compares the performance of the $L$-test that performs CV on $(\tilde{\bm y}, \bm X)$ using these two rules and shows that the min rule performs noticeably better in all settings considered.
\begin{figure}[ht]
    \centering
    \begin{subfigure}[t]{0.35\linewidth}
        \centering
        \includegraphics[width=\linewidth]{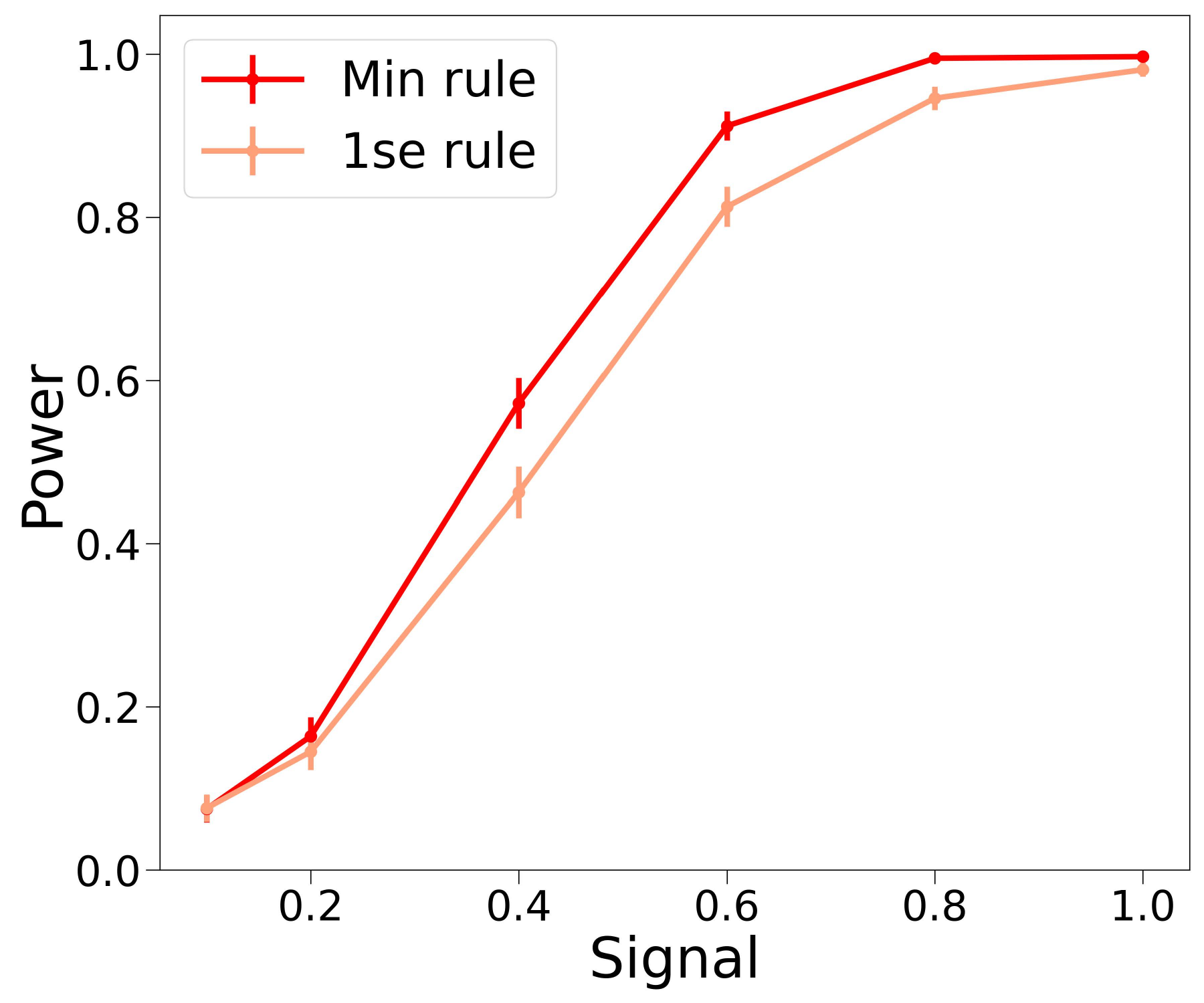}
    \end{subfigure}
    \hspace{0.2cm}
    \begin{subfigure}[t]{0.35\linewidth}
        \centering
        \includegraphics[width=\linewidth]{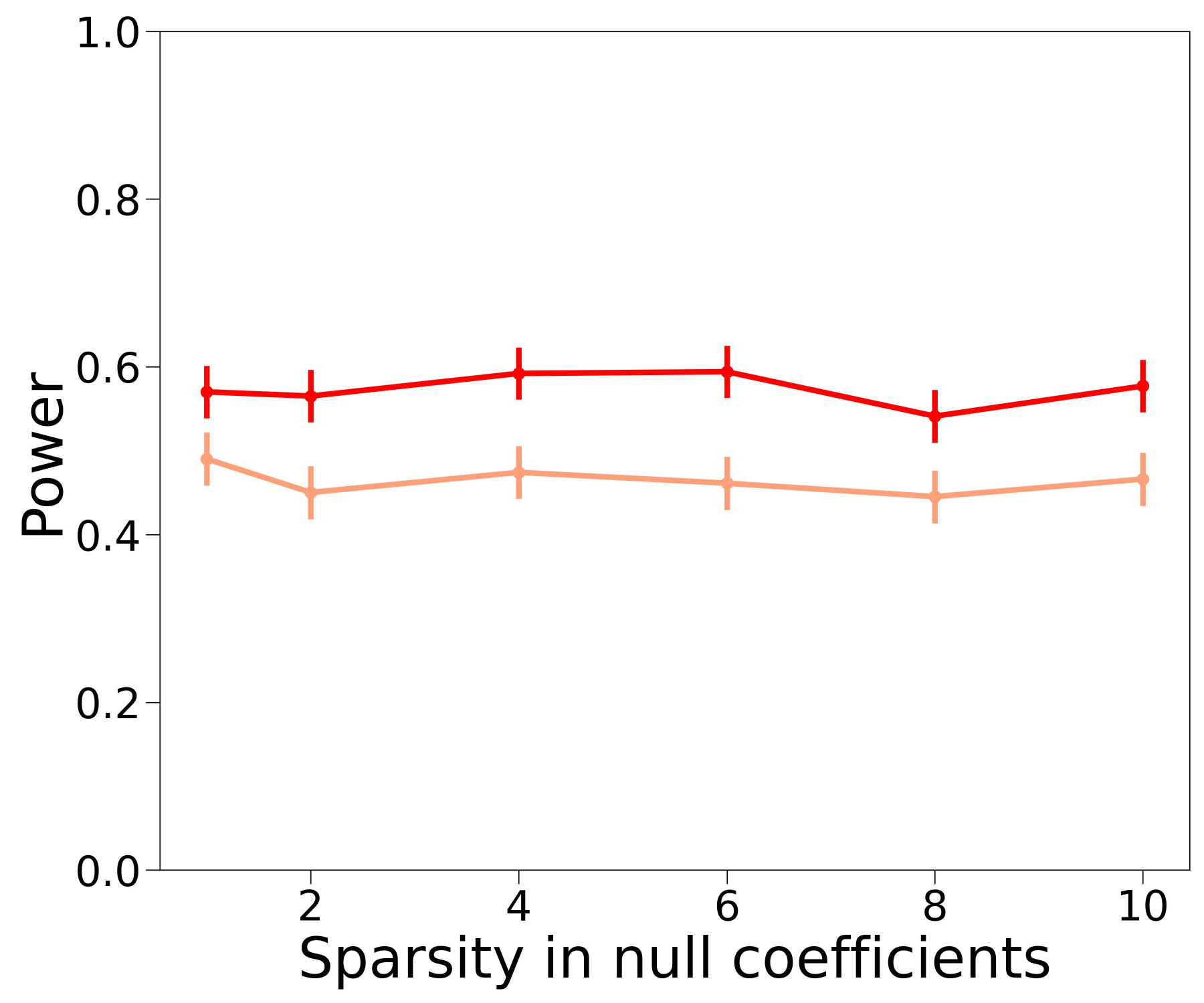}
    \end{subfigure} \\
    \begin{subfigure}[t]{0.35\linewidth}
        \centering
        \includegraphics[width=\linewidth]{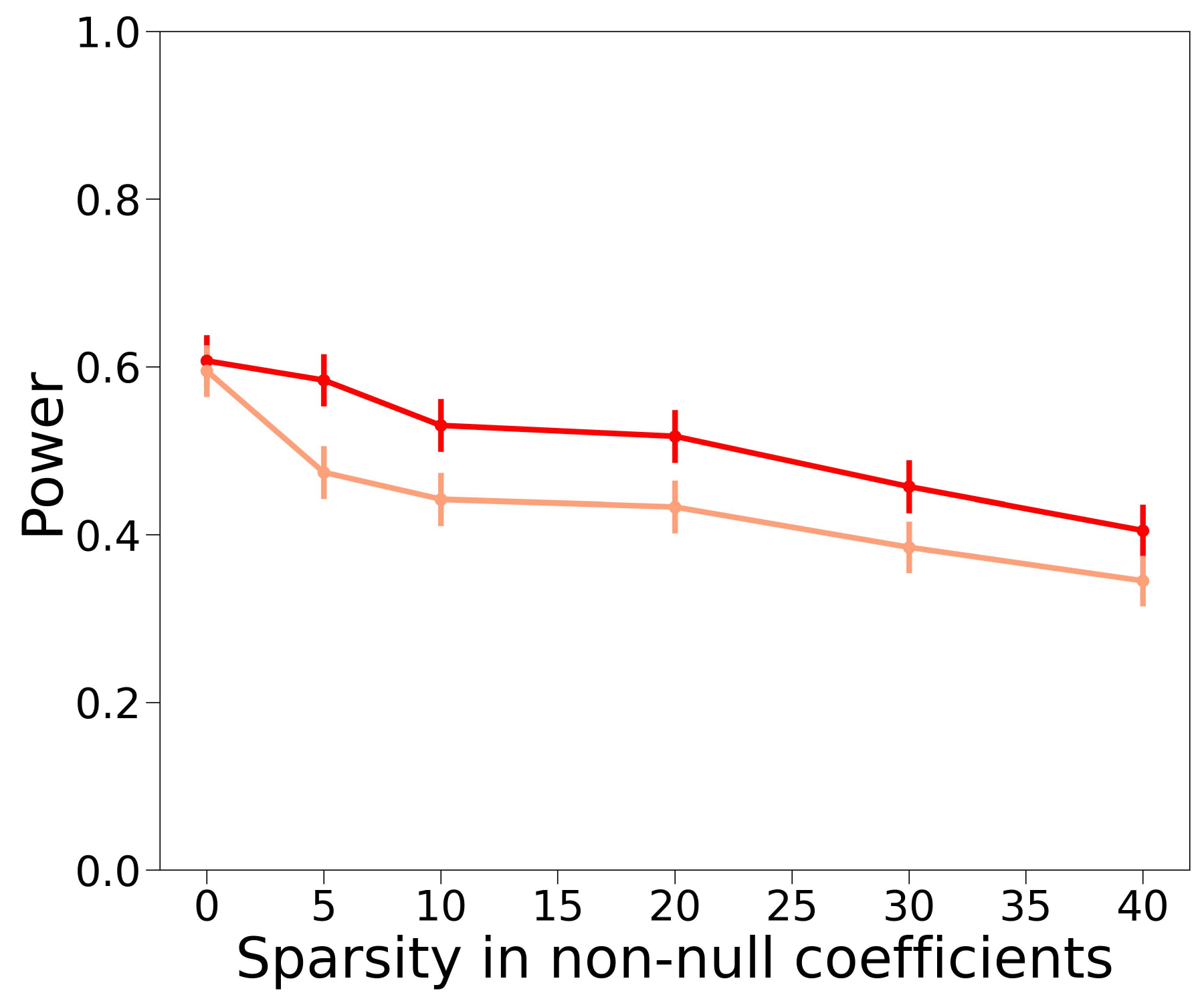}
    \end{subfigure}
    \hspace{0.2cm}
    \begin{subfigure}[t]{0.35\linewidth}
        \centering
        \includegraphics[width=\linewidth]{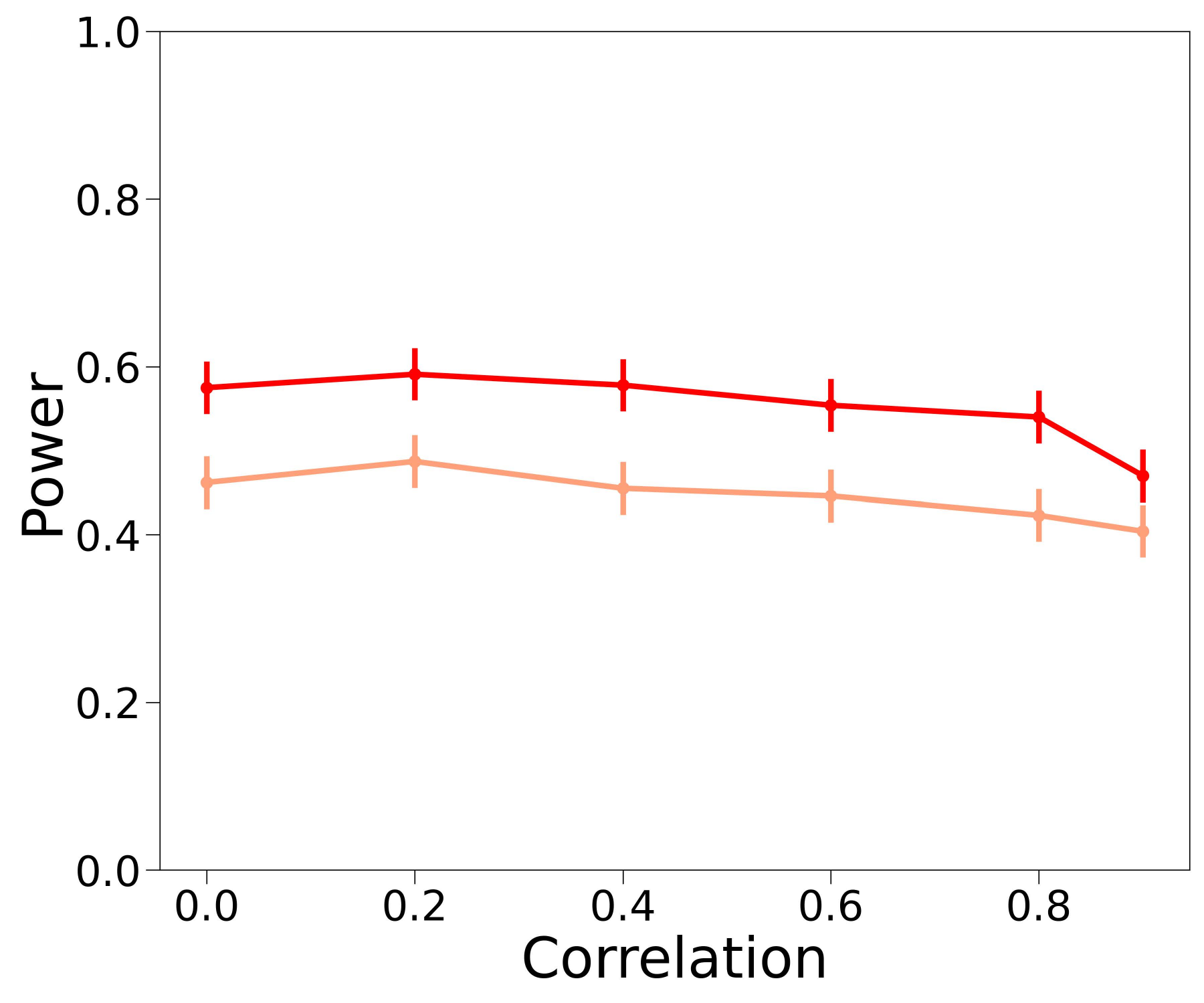}
    \end{subfigure}
    \caption{Min vs. 1se rule $L$-tests. Panel settings are the same as those in Figure \ref{fig:powers_standard}.}
    \label{fig:rules}
\end{figure}

\subsection{Randomness in $L$-test}
To assess the randomness due to cross-validating on $(\tilde{\bm y}, \bm X)$, we run $m = 100$ replications for each of the first five signals in the top-left panel of Figure \ref{fig:powers_standard}, denoting the $i$-th replication's dataset by $(\bm y^{(i)}, \bm X^{(i)})$. For each replication, we draw 100 samples of $\tilde{\bm y}$, yielding p-values $p_{ij}$ corresponding to the $j$-th sample for dataset $i$.

The overall standard deviation of the p-values given by $\sqrt{\Var(p_{i, j})}$ and the standard deviation conditioned on a particular dataset given by $\sqrt{\E(\Var(p_{ij} \mid (\bm{y}^{(i)}, \bm X^{(i)})))}$ are estimated using 
\begin{align*}
\sqrt{\frac{1}{m^2-1}\sum_{i, j}^{m} (p_{ij} - \bar p)^2} \quad \text{and} \quad  \sqrt{\frac{1}{m(m-1)}\sum_{i, j=1}^m (p_{ij} - \bar{p}_i)^2},
\end{align*}
    respectively, where $\bar p = \frac{1}{m^2} \sum_{i, j = 1}^{m} p_{ij}$ and $\bar p_i = \frac{1}{m}\sum_{j = 1}^m p_{ij}$. This entire procedure is repeated 200 times and the means and associated standard errors of the overall and within standard deviation estimates are plotted in Figure \ref{fig:var_exp}. We conduct this experiment for both the $L$-test and MC-free test from Section \ref{sec:precise_p_values}, as the small p-values in the highest signal setting would typically be computed using the MC-free test per our recommendation.

As shown in Figure \ref{fig:var_exp}, sampling $\tilde{\bm{y}}$ accounts for at most 33\% and 16\% of the overall standard deviation (11\% and 2.7\% of the overall variance) for the $L$-test and MC-free test, respectively. This variance contribution comes down even further to less than 3\% and 1\% for the $L$-test and MC-free test, respectively, when we take a mean of 5 repeated $\lambda$'s and $\bm b^*$'s. Repeating the experiment with $k = 1$ reproduced the results of \cite{Sengupta2024} ($\approx 1\%$ of overall variance), so increasing $k$ raised the variability slightly, though it did not worsen further when we tried $k = 20$.
\begin{figure}[ht]
    \centering
    \begin{subfigure}[t]{0.4\linewidth}
        \centering
        \includegraphics[width=\linewidth]{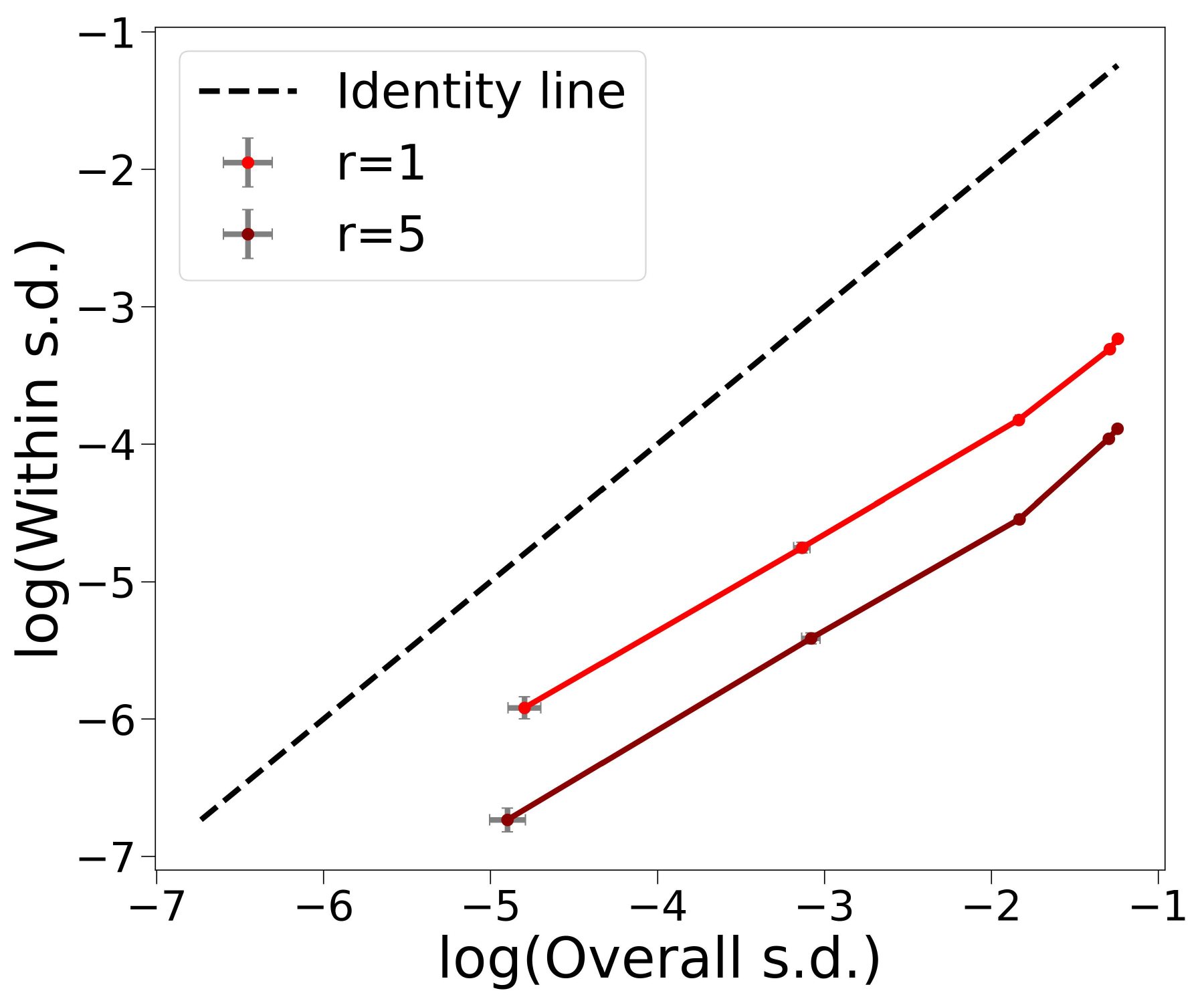}
    \end{subfigure}
    \hspace{0.2cm}
    \begin{subfigure}[t]{0.4\linewidth}
        \centering
        \includegraphics[width=\linewidth]{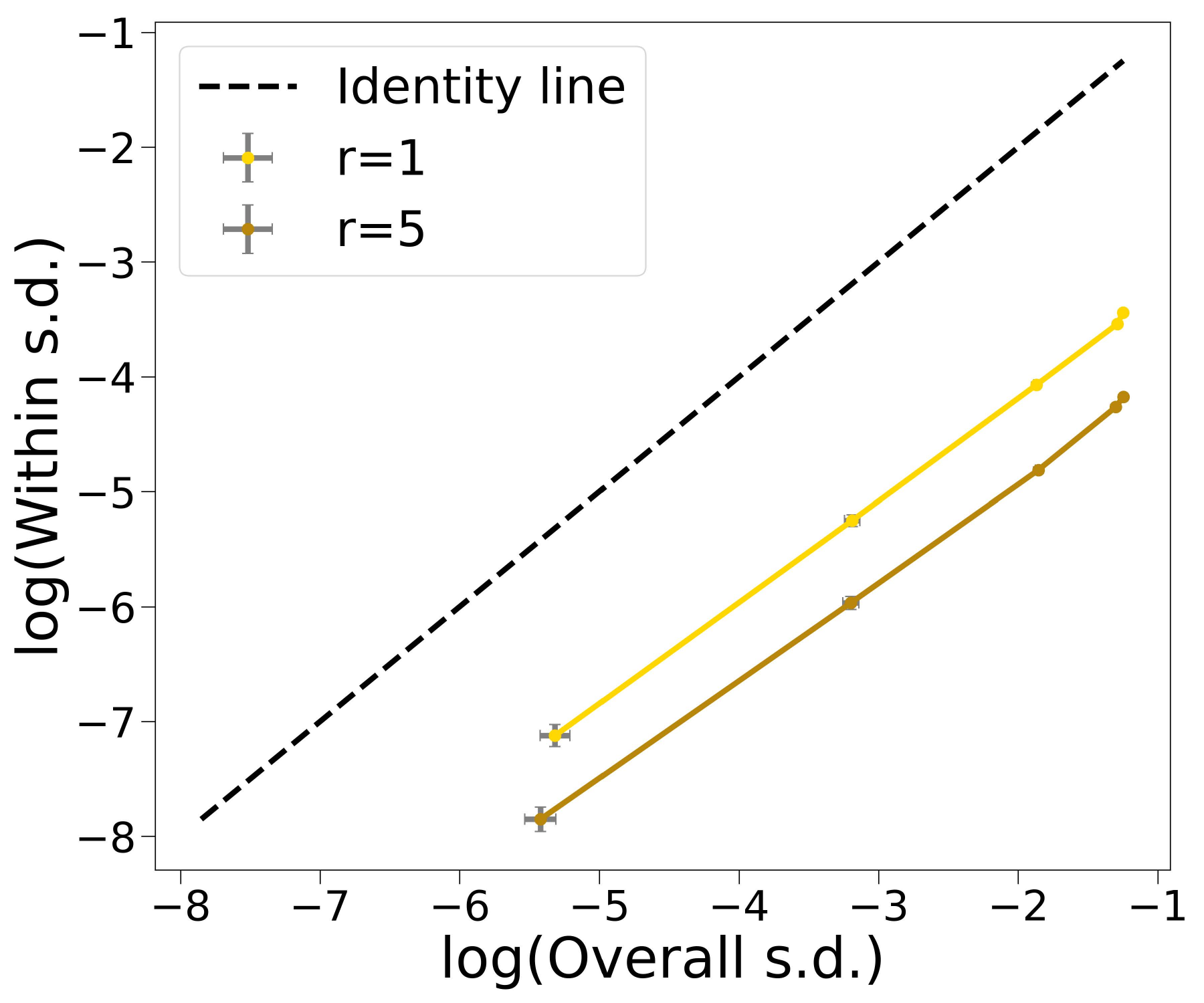}
    \end{subfigure}
    \caption{Plot showing the randomness introduced by conditional sampling of $\tilde{\bm{y}}$ in the $L$-test (left) and MC-free test (right) for each of the first five signal settings in Figure \ref{fig:powers_standard}. For each test, $r$ repetitions of $\lambda$ and $\bm b^*$ are averaged to compute the penalty and point estimate.}
    \label{fig:var_exp}
\end{figure}

Even though this randomness is negligible, in pursuit of an approach that introduces no randomness into the performance of the $L$-test, we considered the $L$-test that does CV on $(\yhat, \bm X)$. In Figure \ref{fig:diff_data}, we compare the performance of the $L$-tests that do CV on $(\tilde{\bm y}, \bm X)$ and $(\yhat, \bm X)$ to the invalid test that does CV on $(\bm y, \bm X)$. We see that the $L$-test that does CV on $(\tilde{\bm y}, \bm X)$ does better than the one that cross-validates on $(\yhat, \bm X)$ and in fact does pretty much just as well as the invalid test, demonstrating that our recommended method does very close to the best we could hope to do.
\begin{figure}[ht]
    \centering
    \begin{subfigure}[t]{0.4\linewidth}
        \centering
        \includegraphics[width=\linewidth]{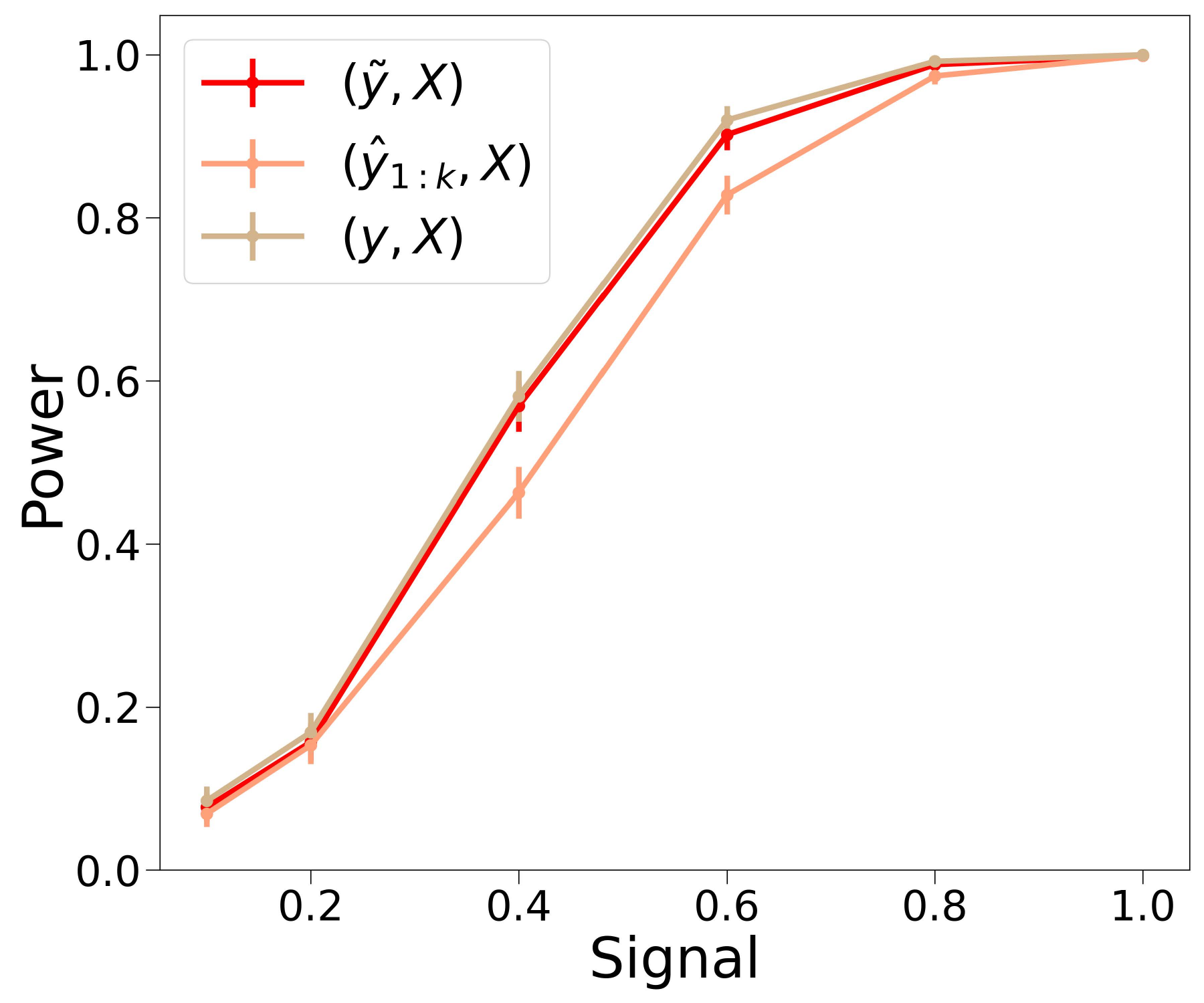}
    \end{subfigure}
    \hspace{0.2cm}
    \begin{subfigure}[t]{0.4\linewidth}
        \centering
        \includegraphics[width=\linewidth]{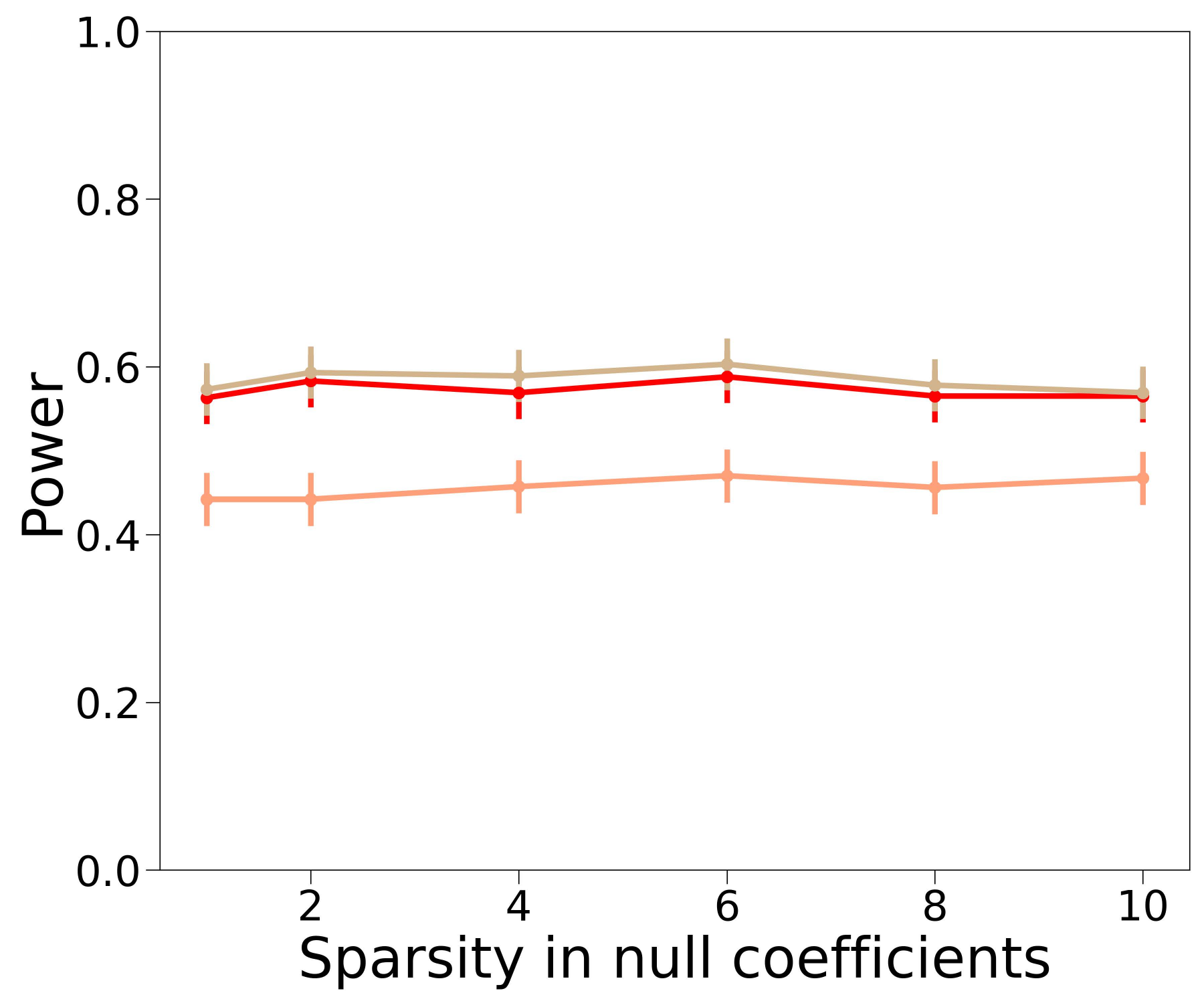}
    \end{subfigure} \\
    \begin{subfigure}[t]{0.4\linewidth}
        \centering
        \includegraphics[width=\linewidth]{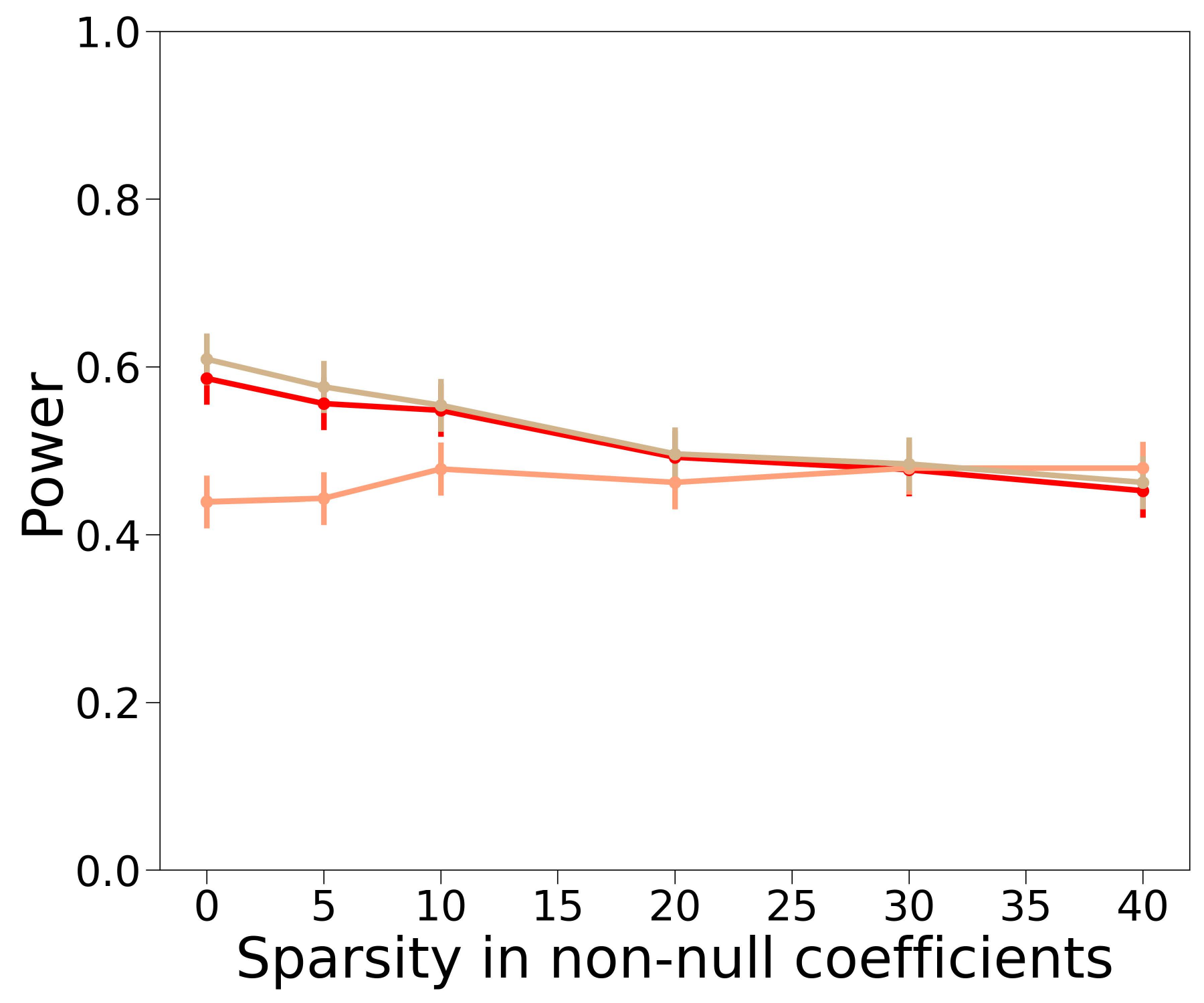}
    \end{subfigure}
    \hspace{0.2cm}
    \begin{subfigure}[t]{0.4\linewidth}
        \centering
        \includegraphics[width=\linewidth]{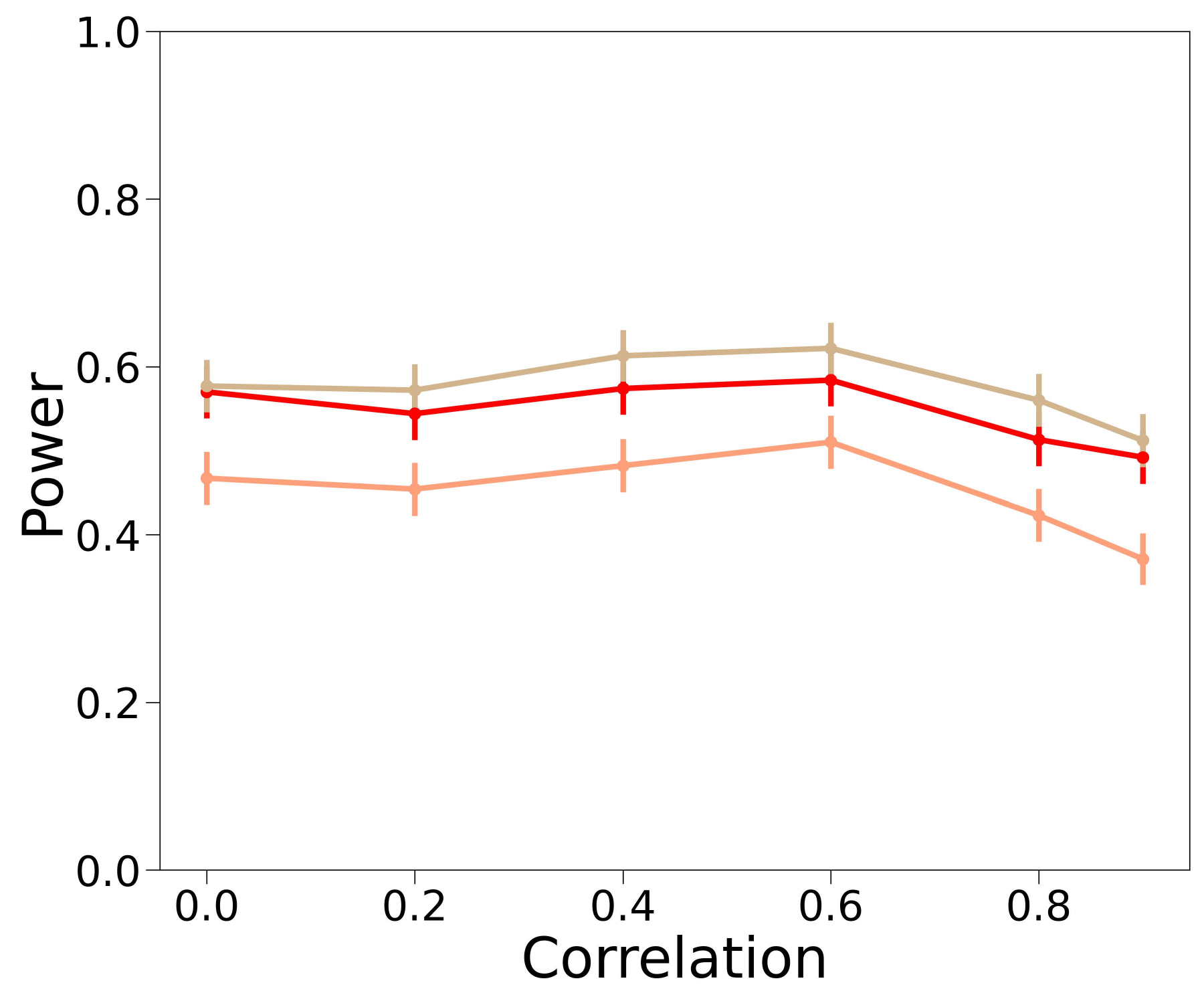}
    \end{subfigure}
    \caption{$L$-tests with $\lambda$ chosen by cross-validating on $(\bm{y}, \bm{X})$, $(\tilde{\bm{y}}, \bm{X})$, and $(\hat{\bm{y}}_{1:k}, \bm{X})$. Panel settings are the same as those in Figure \ref{fig:powers_standard}.}
    \label{fig:diff_data}
\end{figure}

\section{Additional simulation studies}
\label{sec:add_sims}

\subsection{\sectionmath{Explicit construction for $\bm V$}{Explicit construction for V}}
\label{sec:V_construction}
Recall the matrix $\bm{V} \in \R^{n \times (n-d+k)}$ described following Equation \eqref{eq:group_lasso}. We can construct such an orthonormal matrix by taking $\bm{V}_i = \frac{(\bm{I}-\bm{P}_{-1:i})\bm{X}_i}{\|(\bm{I}-\bm{P}_{-1:i})\bm{X}_i\|}$ for $i \in \{1, \ldots, k\}$. Then, letting $\bm X = \bm U \bm \Sigma \tilde{\bm V}^T$ denote the singular value decomposition of $\bm X$, define $\bm V_{(k+1):n-d+k} = \bm U_{(d+1):n}$. It is easy to check that the resulting columns of $\bm V$ are an orthonormal basis that spans a space orthogonal to the column space of $\bm{X}_{-1:k}$. But, since $\bm V$ has $n - d + k$ columns, its column space is the orthogonal complement of the column space of $\bm X_{-1:k}$. This particular construction of $\bm V$ is used in all simulation experiments.

\subsection{\sectionmath{Visualization of $f^{-1}_{\suffstat}$ from Theorem \ref{thm:glasso_unit}}{Visualization of f from Theorem \ref{thm:glasso_unit}}}
\label{sec:f_viz}
In Section \ref{sec:construction}, we explained that our approximation for the piecewise affine function $f^{-1}_{\suffstat}(\bm{b}; r)$ was motivated by visualizations in which we found it to appear globally affine. Figure \ref{fig:f_inv} shows examples of such visualizations in both low and high correlation regimes, where we can see that the contours appear ellipsoidal.
\begin{figure}[ht]
    \centering
    \begin{subfigure}[t]{0.4\linewidth}
        \centering
        \includegraphics[width=\linewidth]{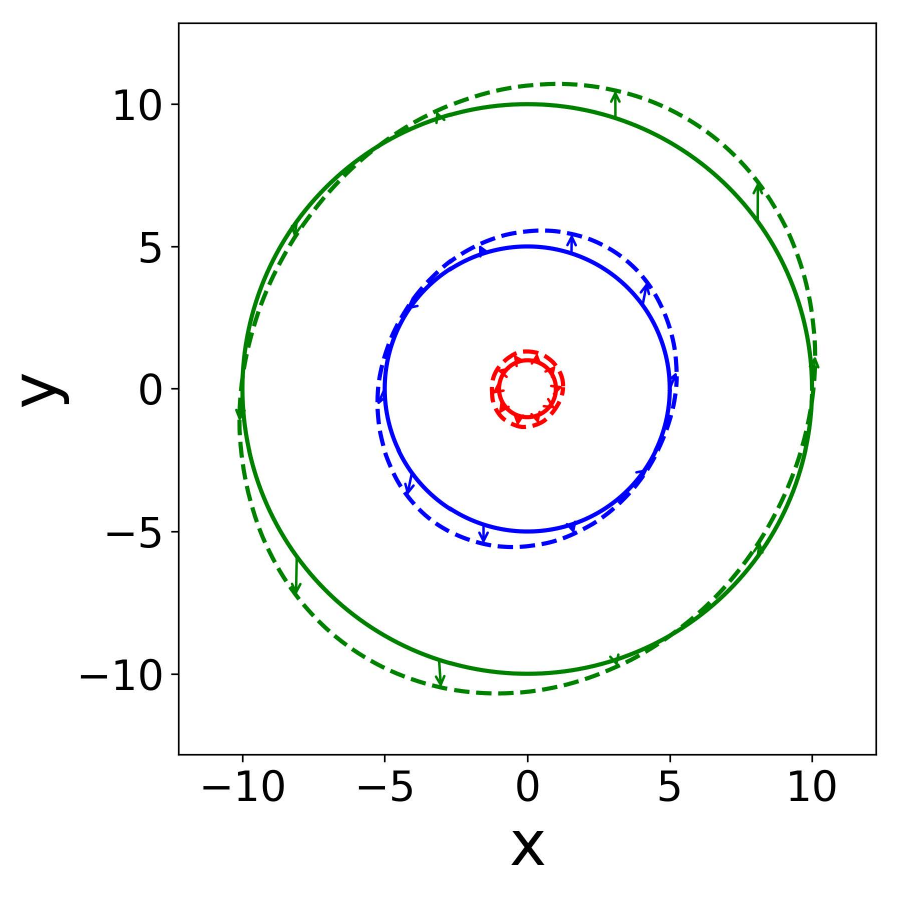}
    \end{subfigure}
    \hspace{0.2cm}
    \begin{subfigure}[t]{0.4\linewidth}
        \centering
        \includegraphics[width=\linewidth]{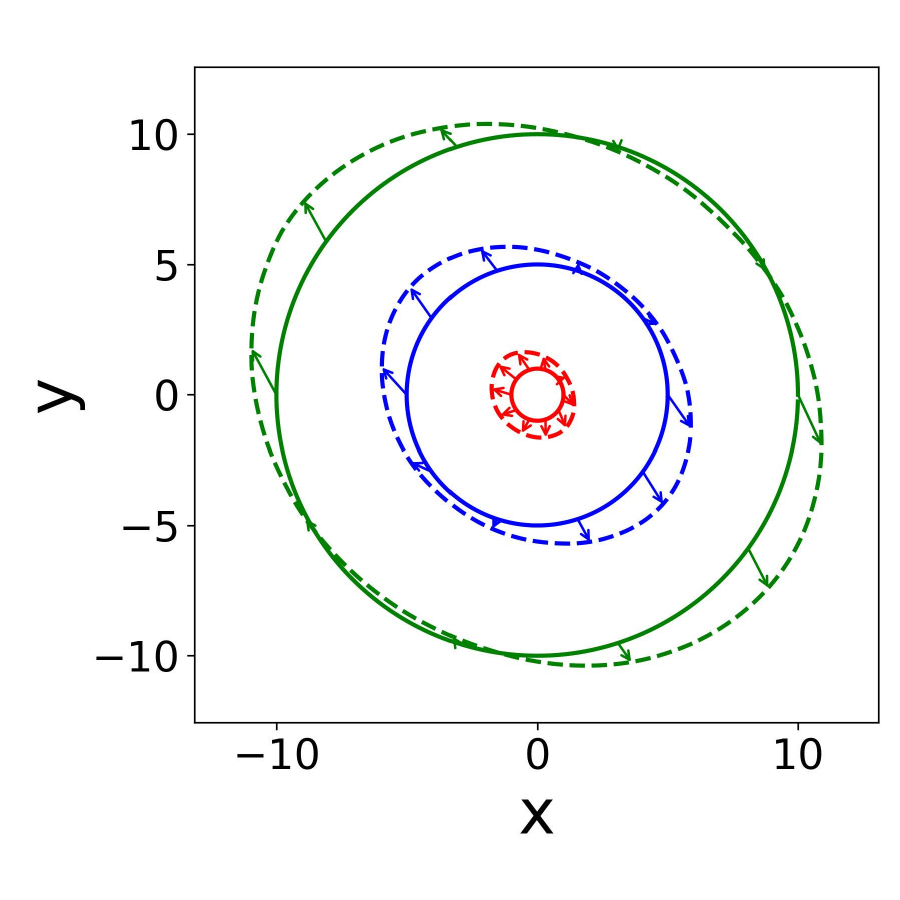}
    \end{subfigure} \\
    \begin{subfigure}[t]{0.4\linewidth}
        \centering
        \includegraphics[width=\linewidth]{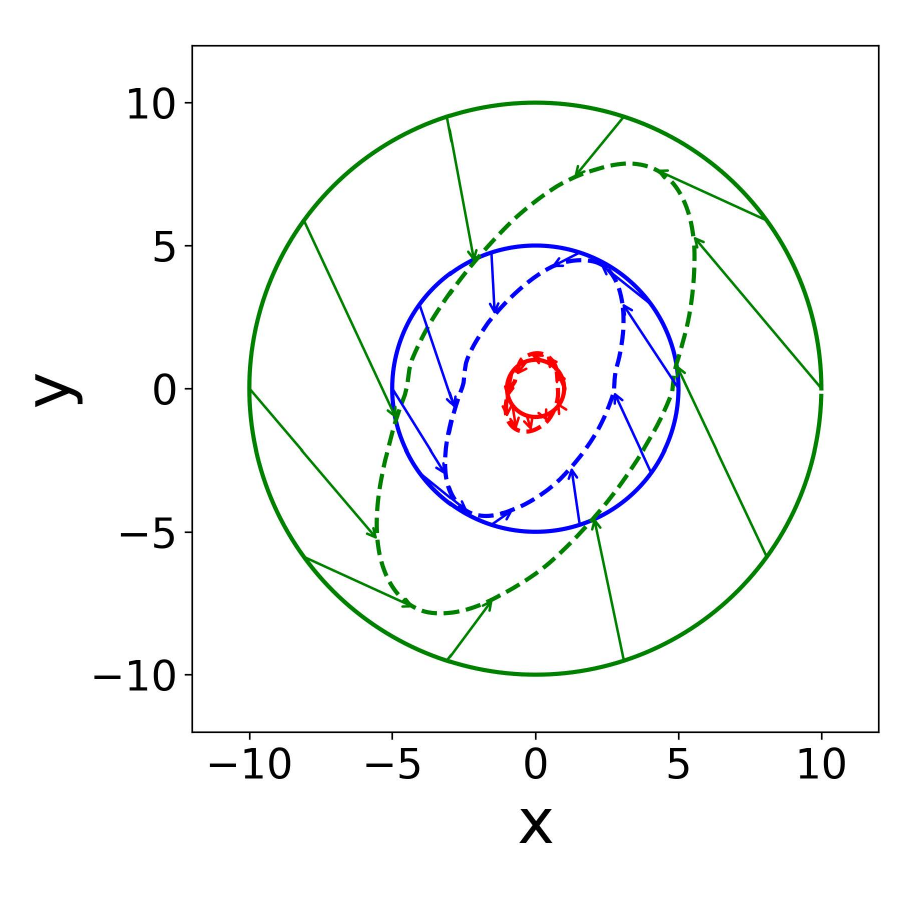}
    \end{subfigure}
    \hspace{0.2cm}
    \begin{subfigure}[t]{0.4\linewidth}
        \centering
        \includegraphics[width=\linewidth]{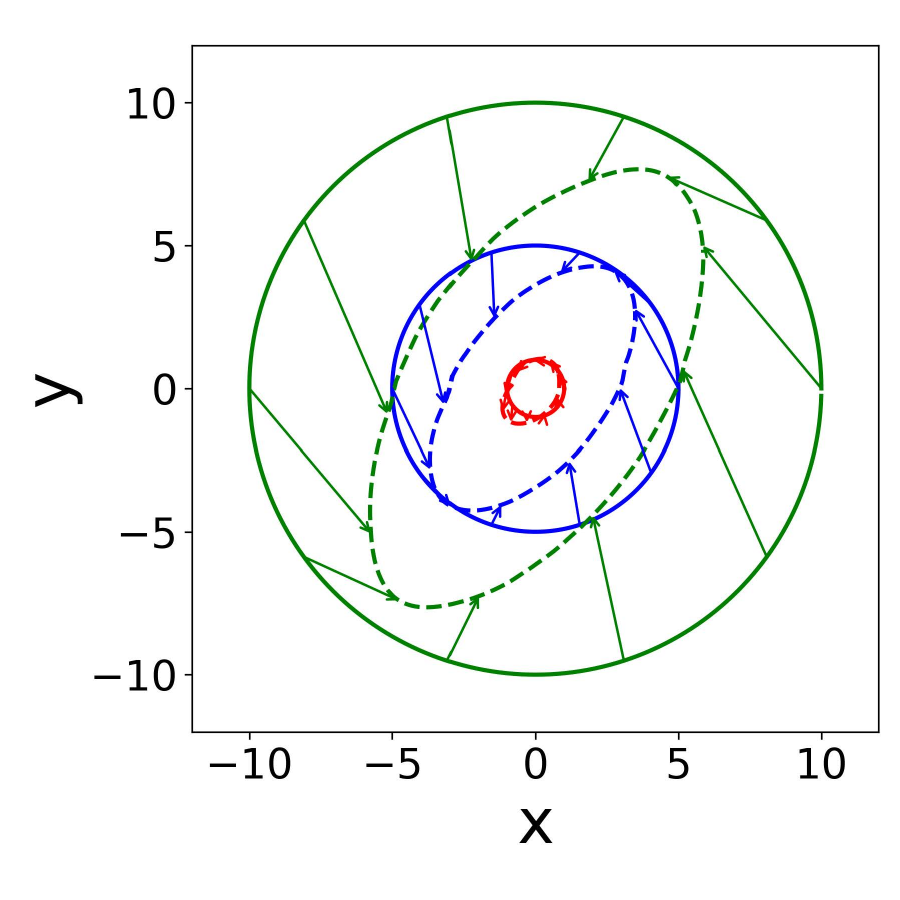}
    \end{subfigure}
    \caption{Plots show $f^{-1}_{\suffstat}$ from Theorem \ref{thm:glasso_unit} applied to three concentric circles. The arrows point to the corresponding outputs. The parameters are $n = 100$, $d = 50$, $k = 2$, $A = 0.3$, $k_1 = k$, $k_2 = 5$, and $\rho = 0$ in the top row while $\rho = 0.9$ in the bottom row. Each row shows two instantiations of $f^{-1}_{\suffstat}$, corresponding to two different data generations.}
    \label{fig:f_inv}
\end{figure}

\subsection{P-value visualizations}
Building on the theoretical intuition in Section \ref{sec:power_analysis} and p-value geometry shown in Figure \ref{fig:p-value_masses}, we empirically visualize the p-value mass of the $L$-, oracle, and $F$-tests in Figure~\ref{fig:heatmaps}. At low correlation, the $L$-test’s nearly circular ellipsoid gains power mainly from recentering; at high correlation, variance concentrates along few PCs, producing an eccentric ellipsoid nearly orthogonal to $\bm{u}_{1:k}$ and yielding near-oracle power.
\begin{figure}[ht]
    \centering
    \includegraphics[width=0.9\linewidth]{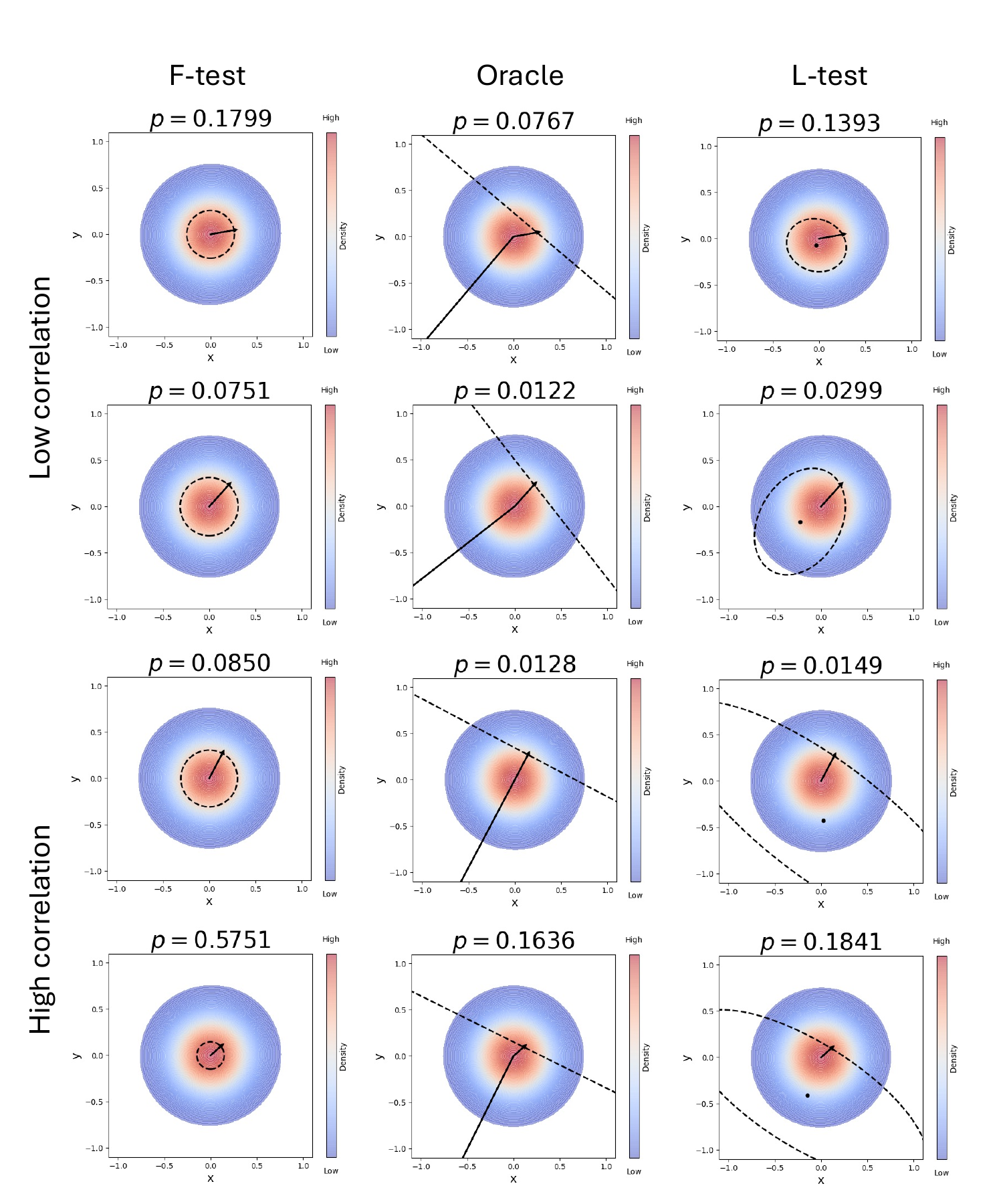}
    \caption{Heatmaps show the density of $\tilde{\bm{u}}_{1:k}$ for $\tilde{\bm{u}} \sim \Unif(\sphere^{n-d+k})$. The black arrows correspond to the observed $\bm{u}_{1:k}$. In the left column, the $F$-test p-value mass is the region outside the dashed circle. In the middle column, the oracle test p-value mass is the region outside the dashed hyperplane drawn perpendicular to the oracle's recentering vector in black and through $\bm{u}_{1:k}$. In the right column, the $L$-test's p-value mass is the region outside the dashed ellipsoid. The parameters are $n = 100$, $d = 50$, $k = 2$, $A = 0.3$, $k_2 = 5$, and $\bm{\beta}_{1:k}$ dense and nonnegative. The first two rows take $\rho = 0$ while the last two take $\rho = 0.9$ and show the p-value geometry on two data generations.}
    \label{fig:heatmaps}
\end{figure}

\subsection{Power plots}
\label{sec:power_plots_add}
In Figure \ref{fig:powers_large}, we extend the experiments from Section \ref{sec:power_plots} by further comparing the tests' powers in a large model setting with $n = 1000$ and $d = 500$ and models that are close to unidentifiability (ie. $n \approx d$).
\begin{figure}[ht]
    \centering
    \begin{subfigure}[t]{0.4\linewidth}
        \centering
        \includegraphics[width=\linewidth]{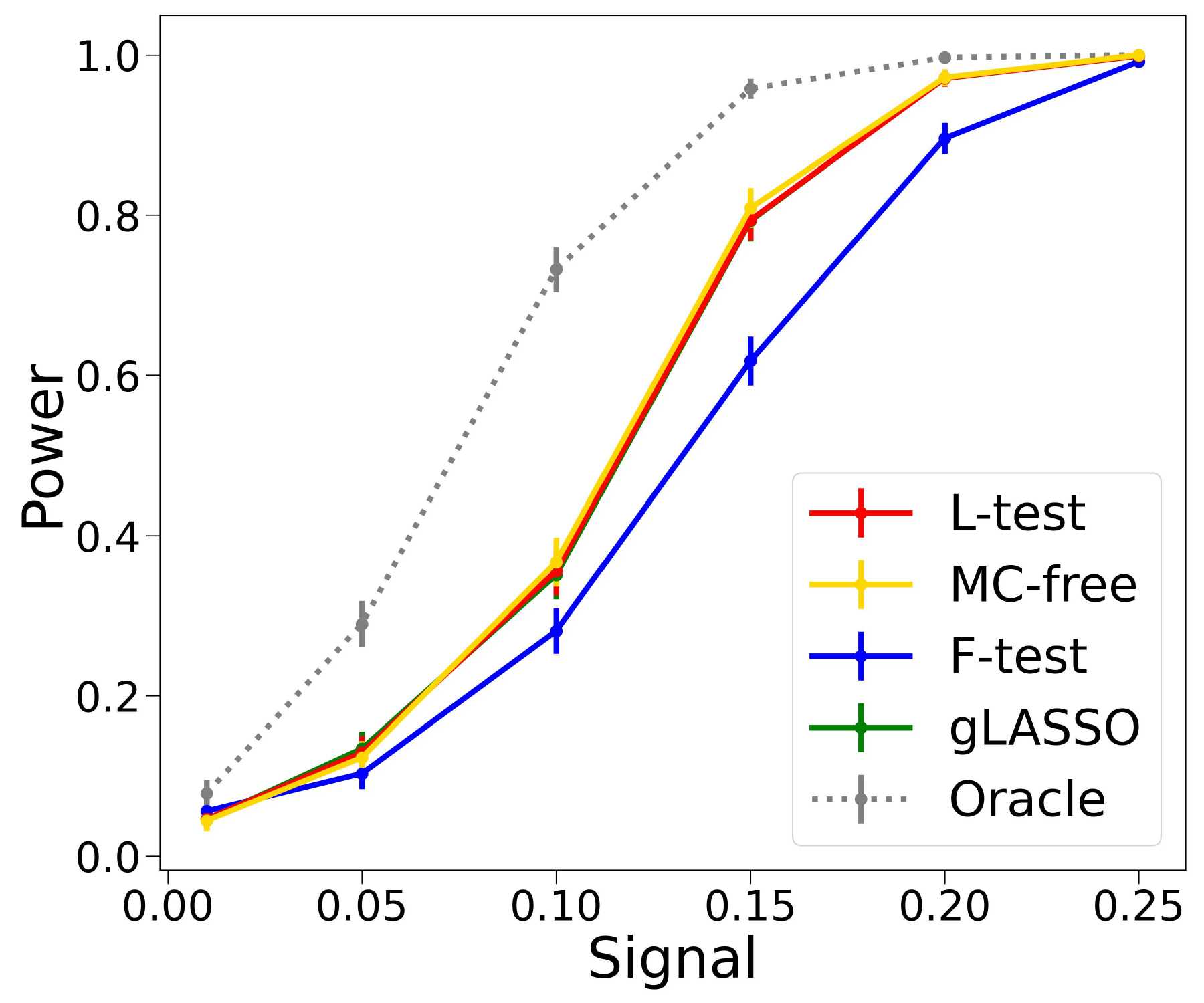}
    \end{subfigure}
    \hspace{0.2cm}
    \begin{subfigure}[t]{0.4\linewidth}
        \centering
        \includegraphics[width=\linewidth]{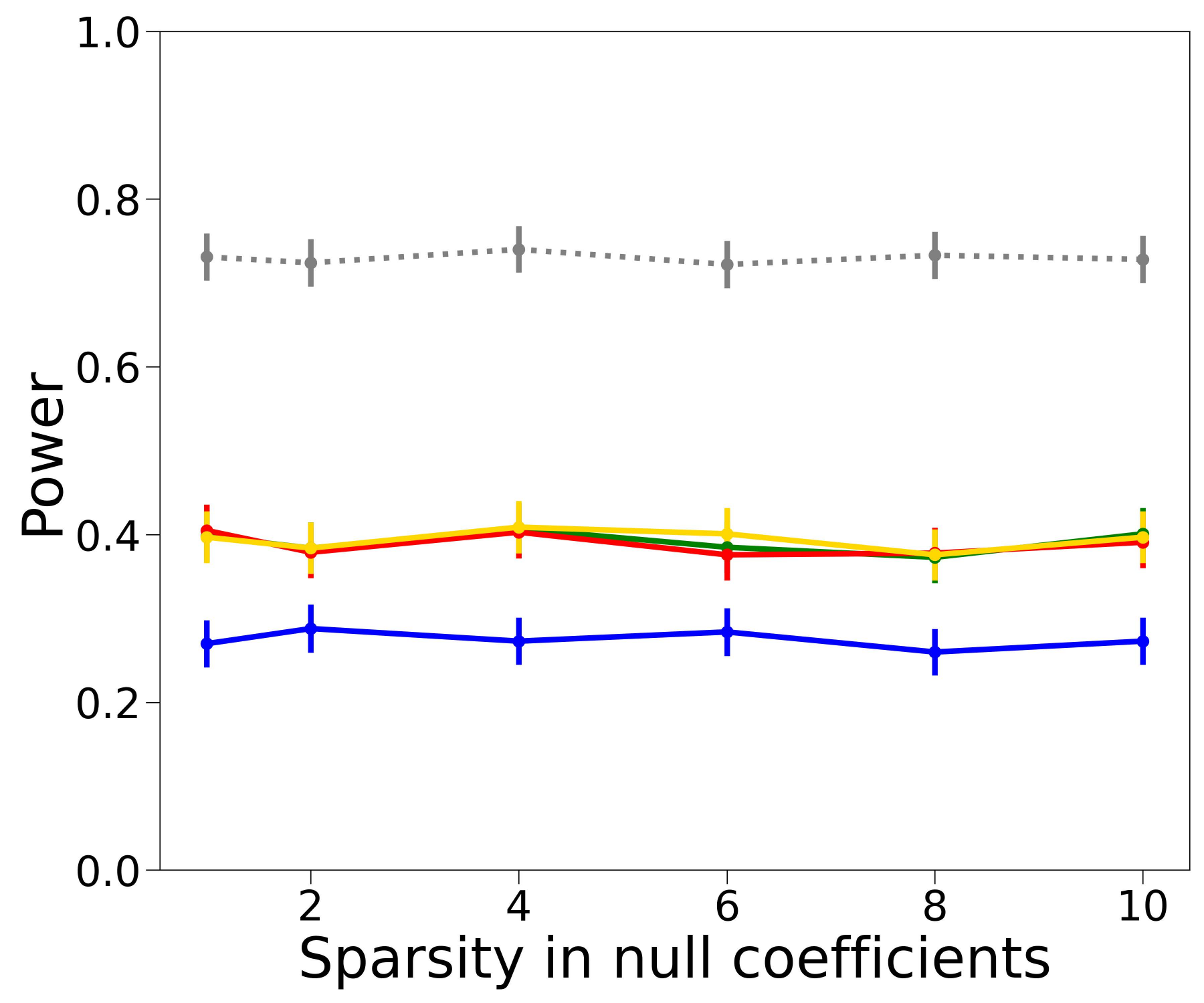}
    \end{subfigure} \\
    \begin{subfigure}[t]{0.4\linewidth}
        \centering
        \includegraphics[width=\linewidth]{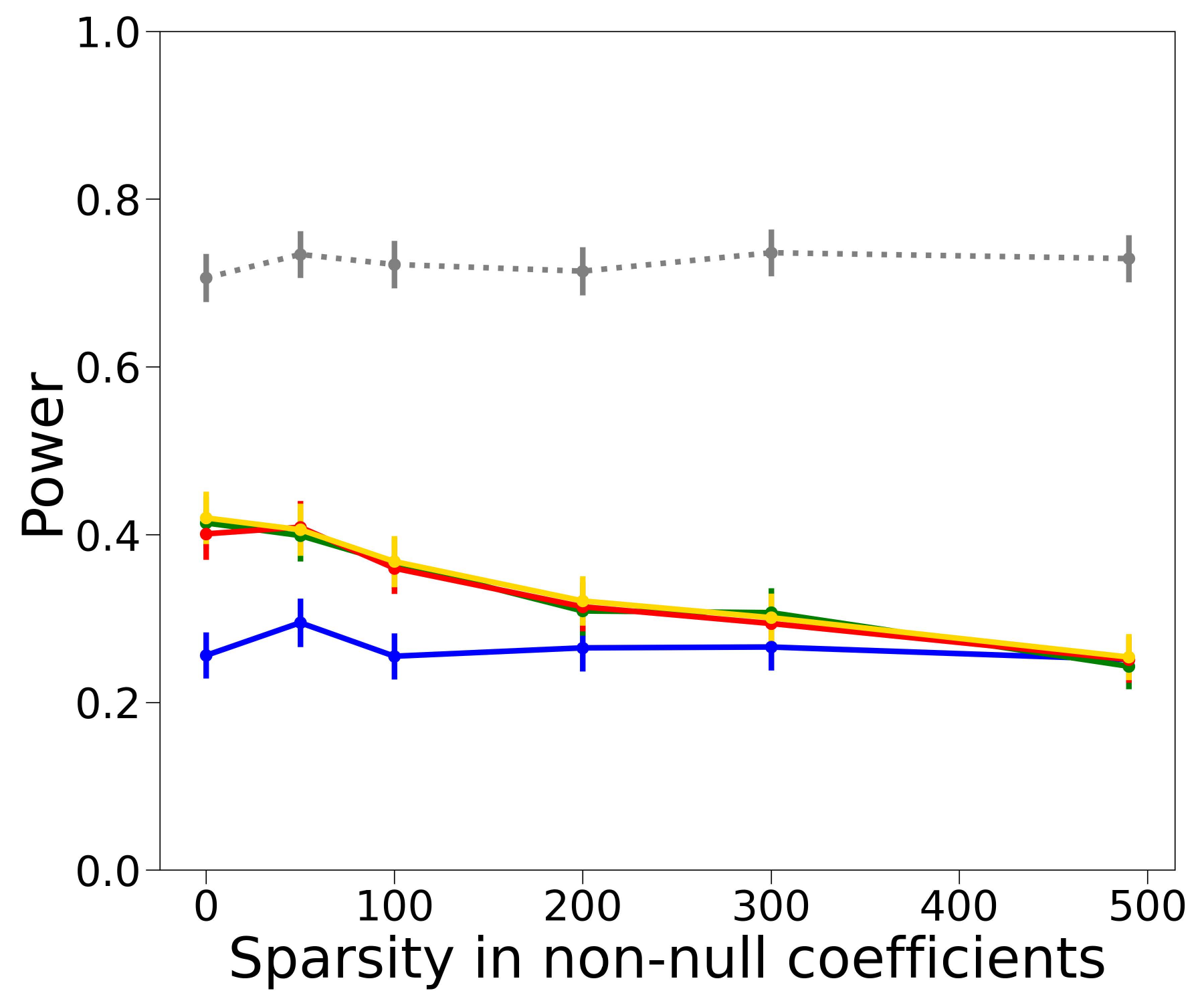}
    \end{subfigure}
    \hspace{0.2cm}
    \begin{subfigure}[t]{0.4\linewidth}
        \centering
        \includegraphics[width=\linewidth]{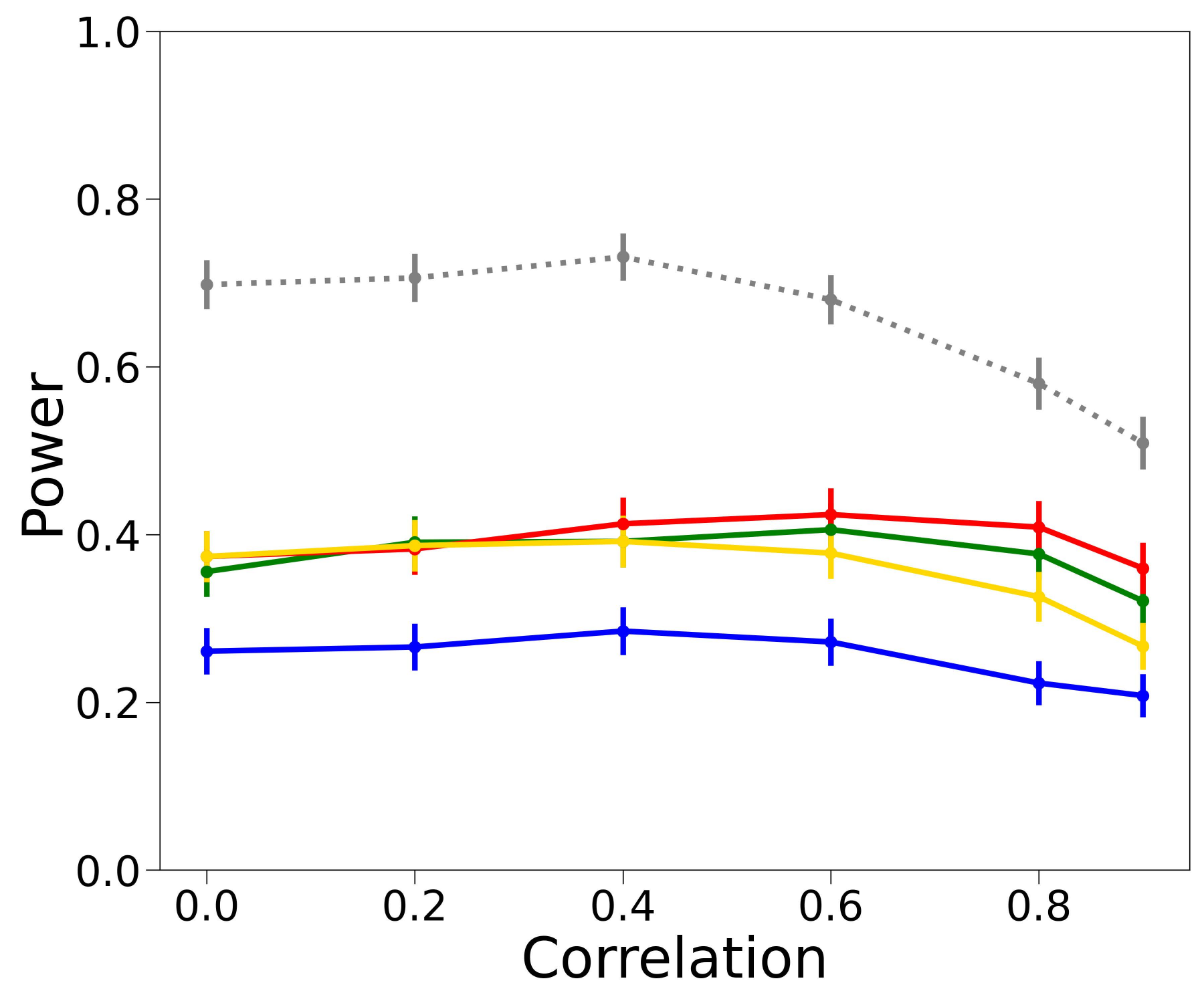}
    \end{subfigure}
    \caption{Comparison of test powers in large model. In all panels, $n = 1000$, $d = 500$, and $k = 10$. The top left panel takes $k_1 = k$, $k_2 = 49$, $\rho = 0$ and varies $A$; top right takes $A = 0.1$, $k_2 = 49$, and $\rho = 0$ and varies $k_1$; bottom left takes $A = 0.1$, $k_1 = 4$, $\rho = 0$ and varies $k_2$; bottom right takes $A = 0.1$, $k_1 = 4$, $k_2 = 49$ and varies $\rho$.}
    \label{fig:powers_large}
\end{figure}

\begin{figure}[ht]
    \centering
    \begin{subfigure}[t]{0.4\linewidth}
        \centering
        \includegraphics[width=\linewidth]{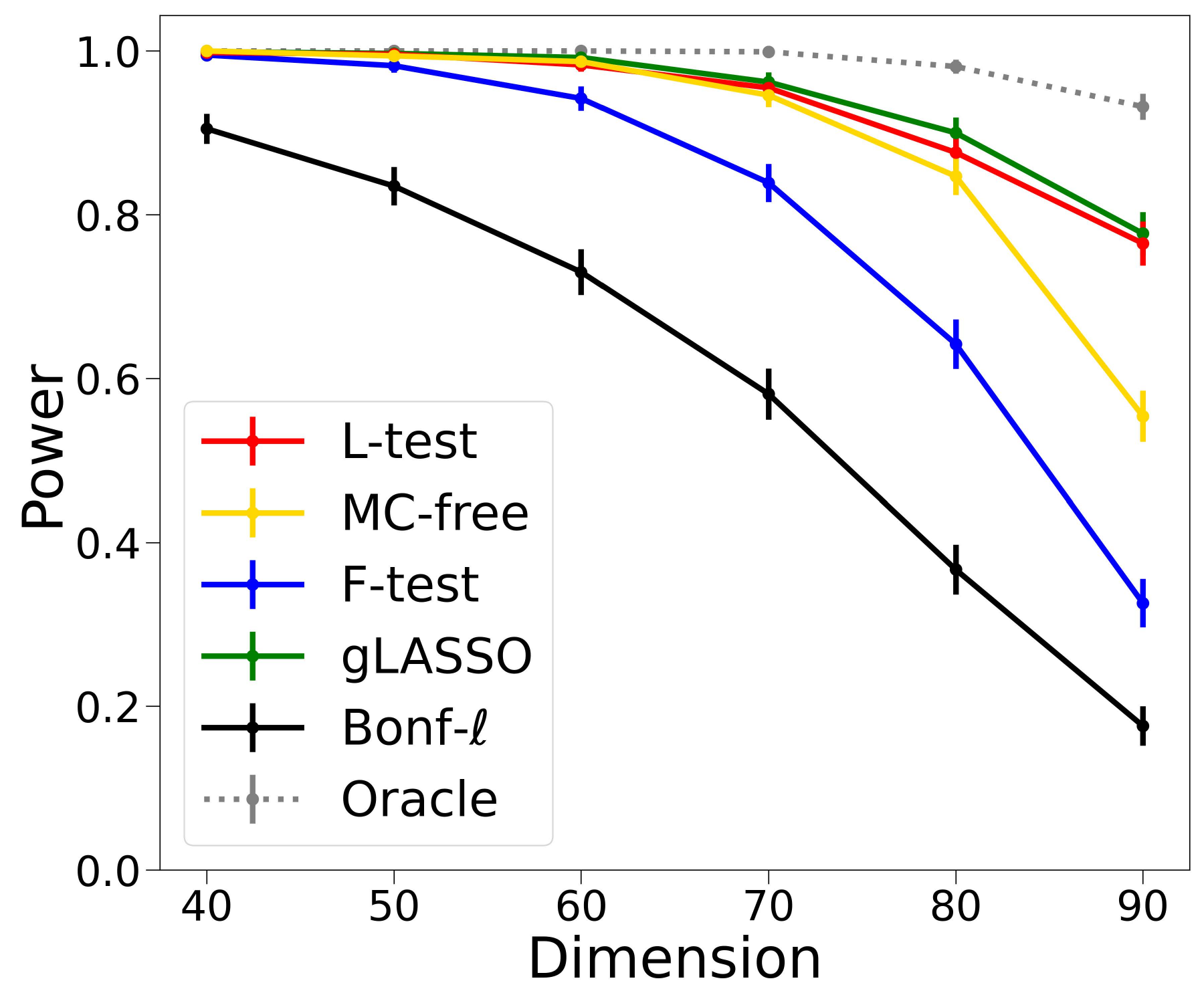}
    \end{subfigure}
    \hspace{0.2cm}
    \begin{subfigure}[t]{0.4\linewidth}
        \centering
        \includegraphics[width=\linewidth]{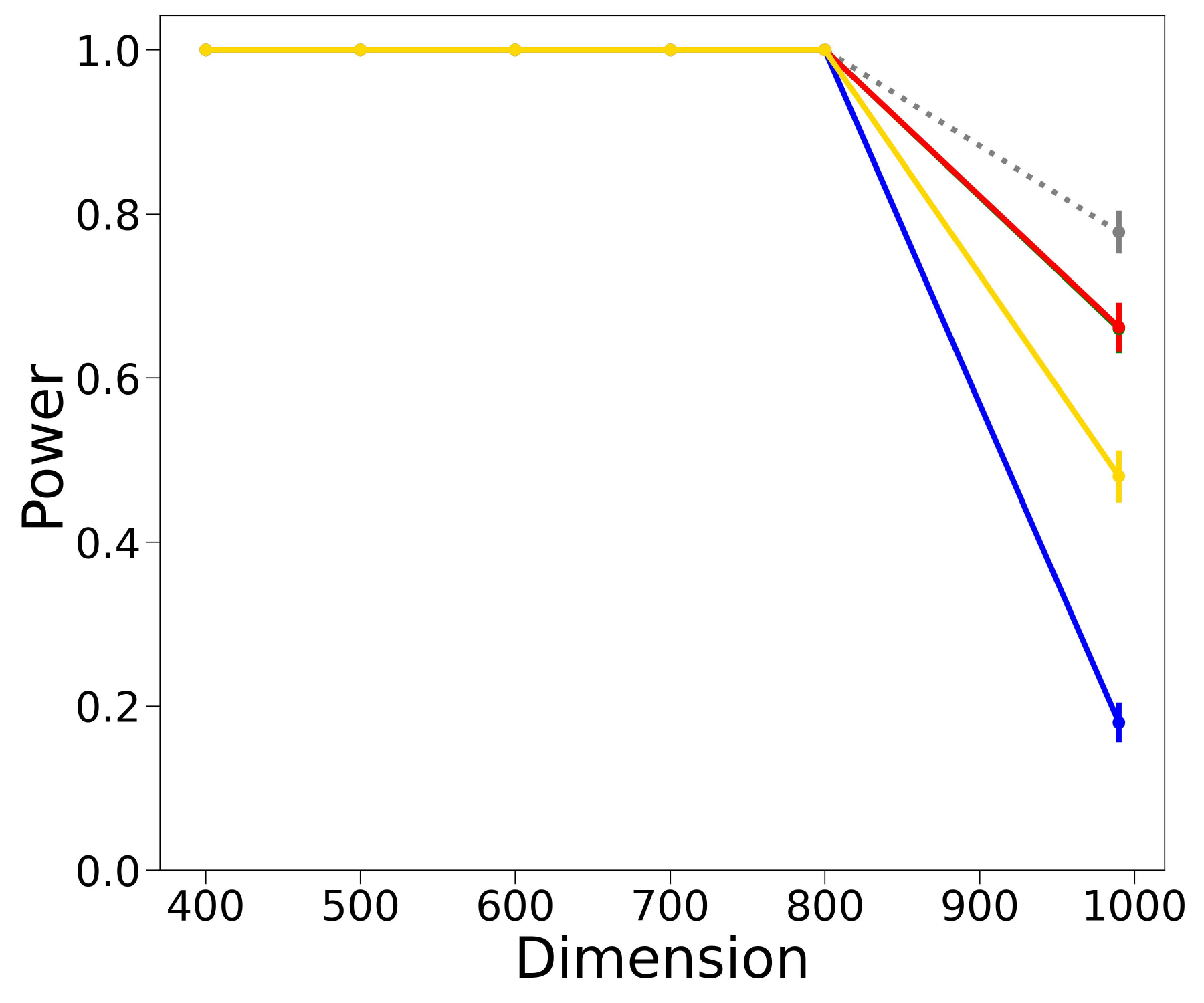}
    \end{subfigure}
    \caption{Comparison of test powers in high-dimensional models. The left panel takes $n = 100$, $k = 10$, $A = 0.8$, $k_1 = k$, and $\rho = 0$ and varies $k_2$ and $d$ to maintain 10\% non-null sparsity; the right does the same but takes $n = 1000$, $k = 10$, $A = 0.6$, $k_1 = k$, and $\rho = 0$.}
    \label{fig:powers_high}
\end{figure}

Figure \ref{fig:PC_orthog} further demonstrates the similarity between the $L$-test and PC-test in the block-orthogonal design setting studied in Section \ref{sec:power_orthog}. To show that the way both the PC-test and premultiplier $\bm{A}(\glassoc_{1:k}, \lambda)$ exploit the geometry of $\bm{X}_{1:k}$ for power gains is not specific to block orthogonality, Figure \ref{fig:PC_stand} compares the PC-test with a test that applies $\bm{A}(\glassoc_{1:k}, \lambda)$ directly to $\bm{u}_{1:k}$ (without recentering)---we call this the $\phi$-test---in a standard design setting. Interestingly, the bottom panel shows a setting where the $L$-test underperforms the group LASSO MC test slightly, suggesting greater robustness of the latter to anti-alignment between $\bm \beta_{1:k}$ and the top PCs of $\bm{X}_{1:k}$. Nevertheless, the gap is small, and the $L$-test closely matches the group LASSO MC test's performance in nearly all other simulations and our real data analysis. 
\begin{figure}[ht]
    \centering
    \begin{subfigure}[t]{0.4\linewidth}
        \centering
        \includegraphics[width=\linewidth]{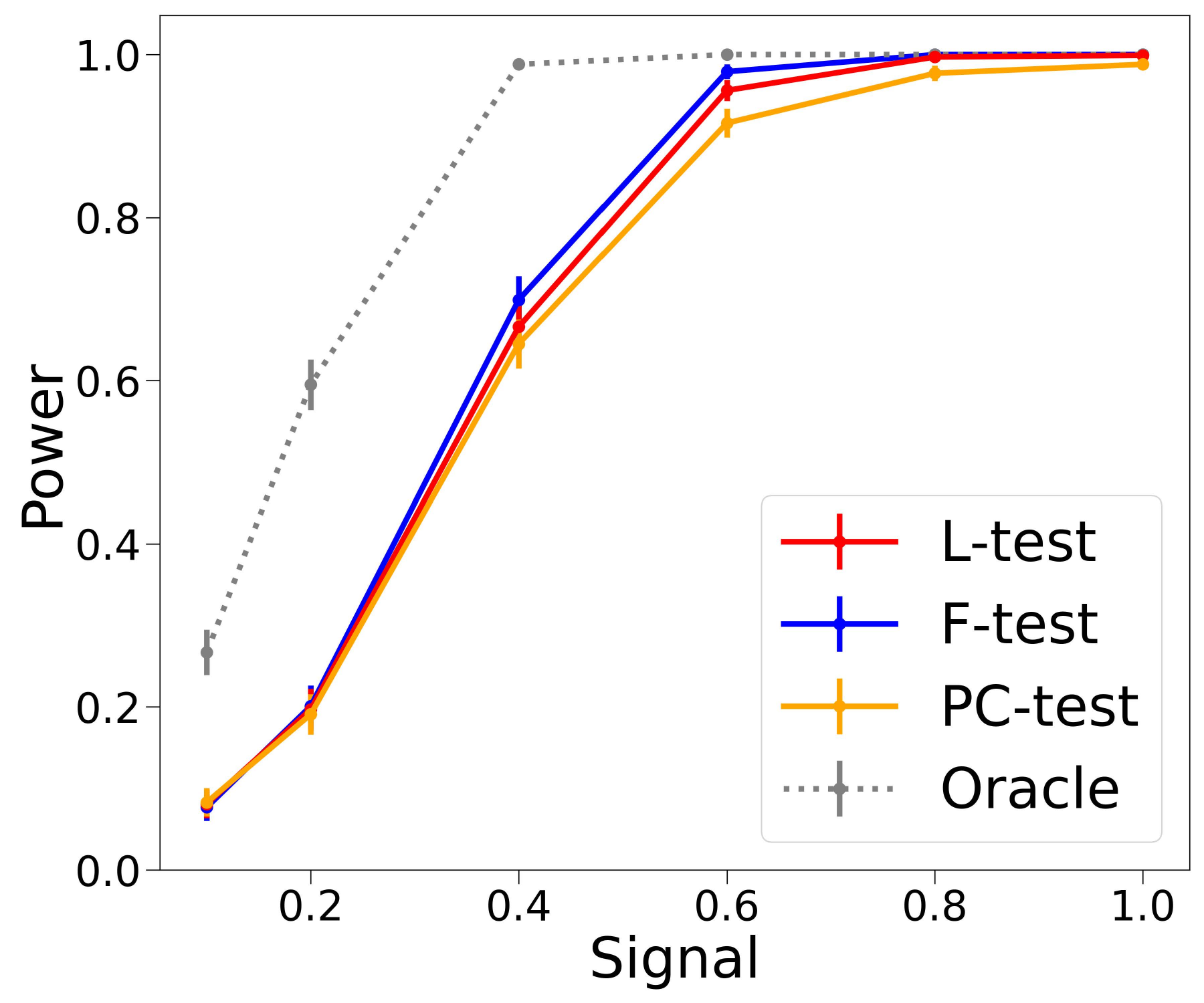}
    \end{subfigure}
    \hspace{0.2cm}
    \begin{subfigure}[t]{0.4\linewidth}
        \centering
        \includegraphics[width=\linewidth]{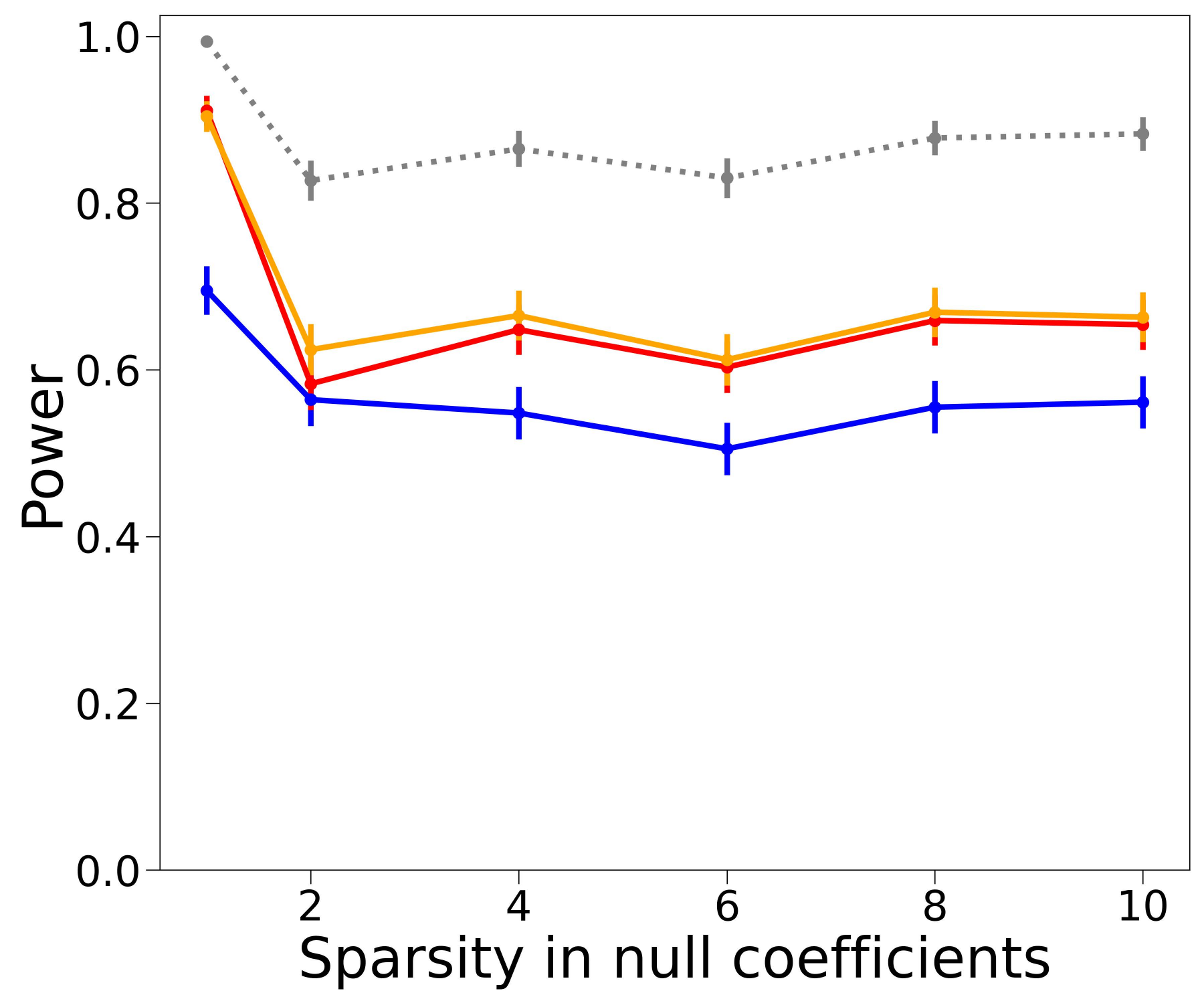}
    \end{subfigure} \\
    \begin{subfigure}[t]{0.4\linewidth}
        \centering
        \includegraphics[width=\linewidth]{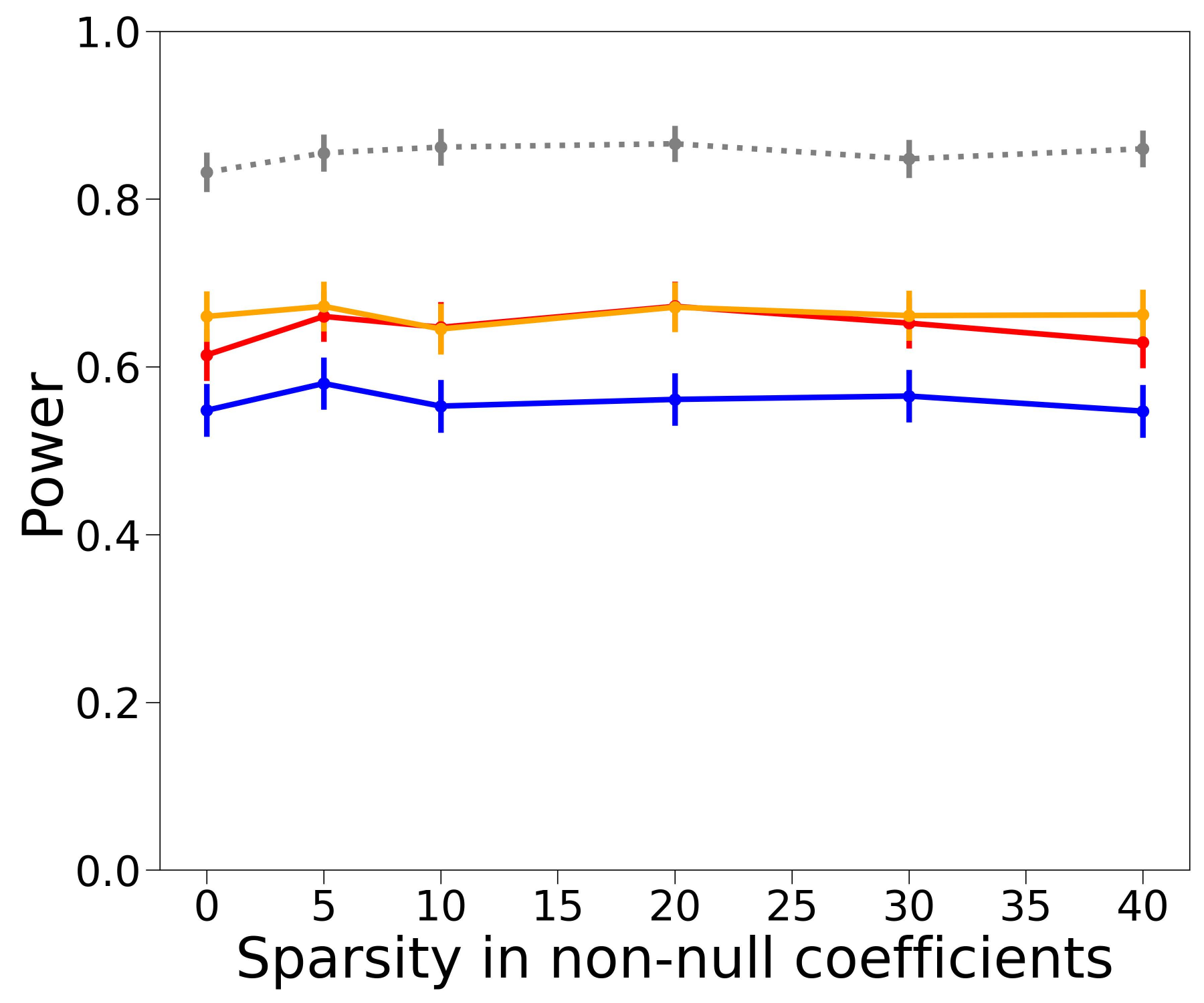}
    \end{subfigure}
    \caption{In all panels, we orthogonalize $\bm{X}_{1:k}$ to $\bm{X}_{-1:k}$ by using the design $\begin{pmatrix}
      \bm{X}_{1:k} & (\bm{I} - \bm{P}_{1:k})\bm{X}_{-1:k}
  \end{pmatrix}$, followed by standardization. Panel settings are the same as those in Figure \ref{fig:powers_standard}, except top right and bottom take $\rho = 0.9$.}
    \label{fig:PC_orthog}
\end{figure}

\begin{figure}[ht]
    \centering
    \begin{subfigure}[t]{0.4\linewidth}
        \centering
        \includegraphics[width=\linewidth]{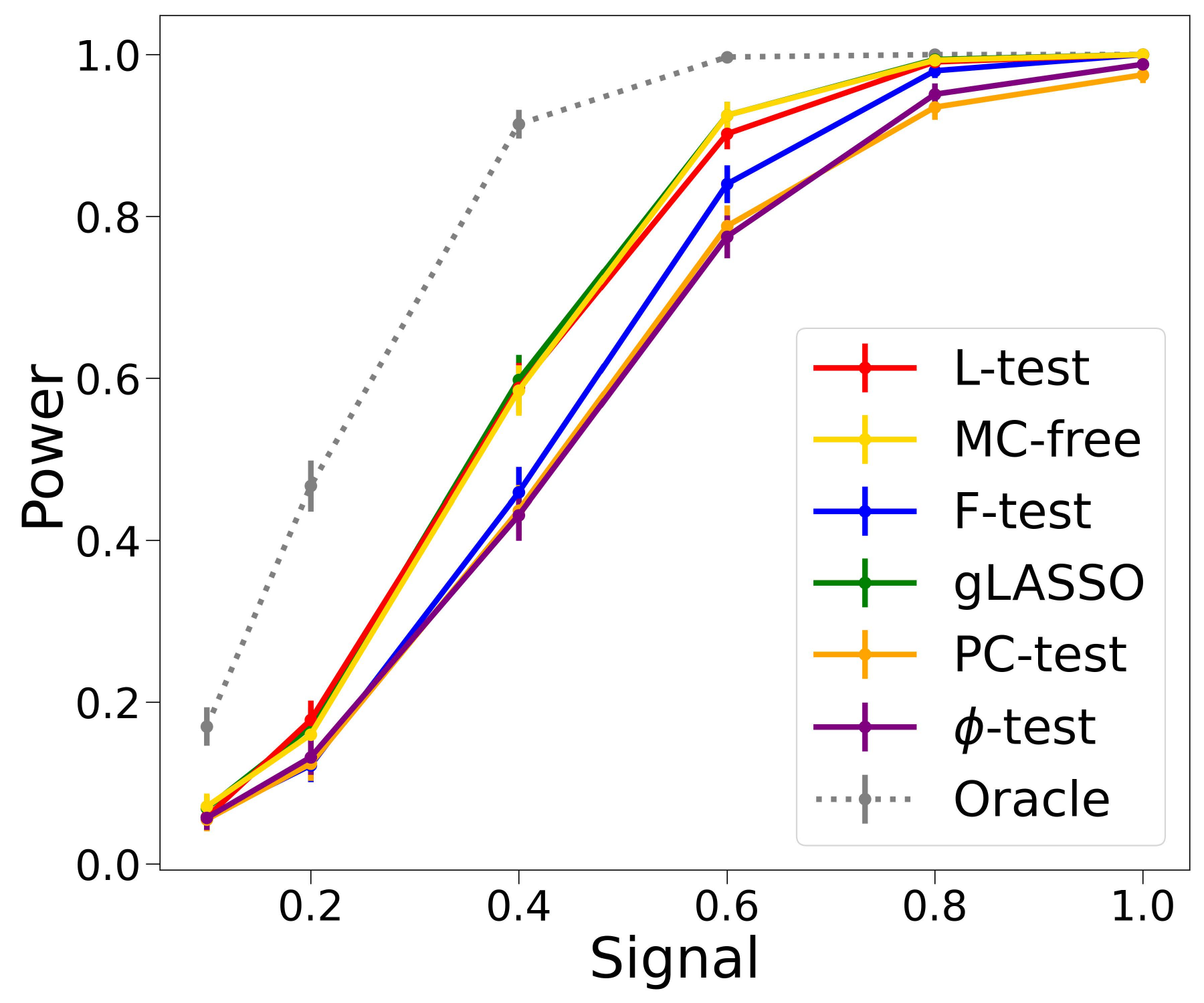}
    \end{subfigure}
    \hspace{0.2cm}
    \begin{subfigure}[t]{0.4\linewidth}
        \centering
        \includegraphics[width=\linewidth]{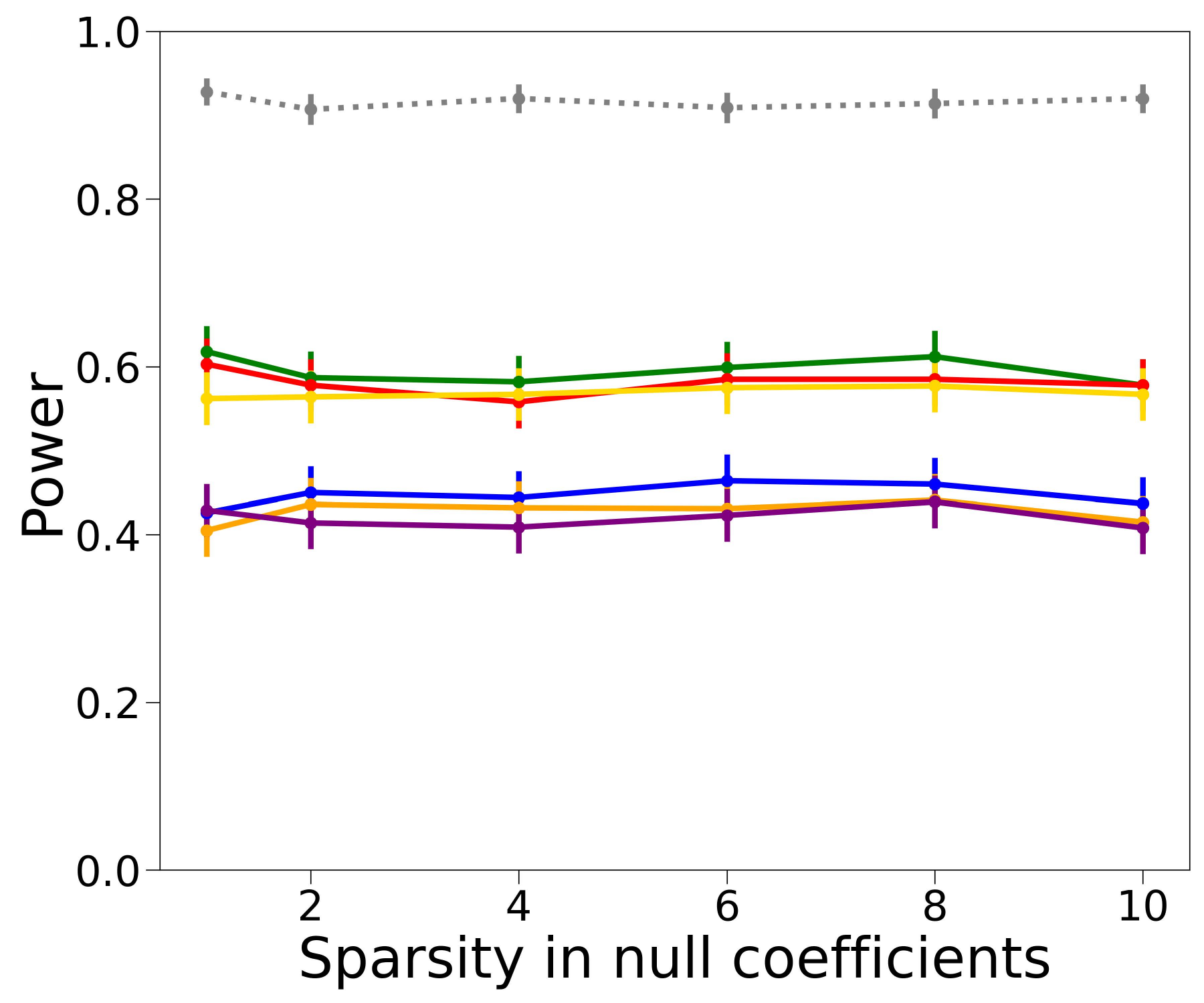}
    \end{subfigure} \\
    \begin{subfigure}[t]{0.4\linewidth}
        \centering
        \includegraphics[width=\linewidth]{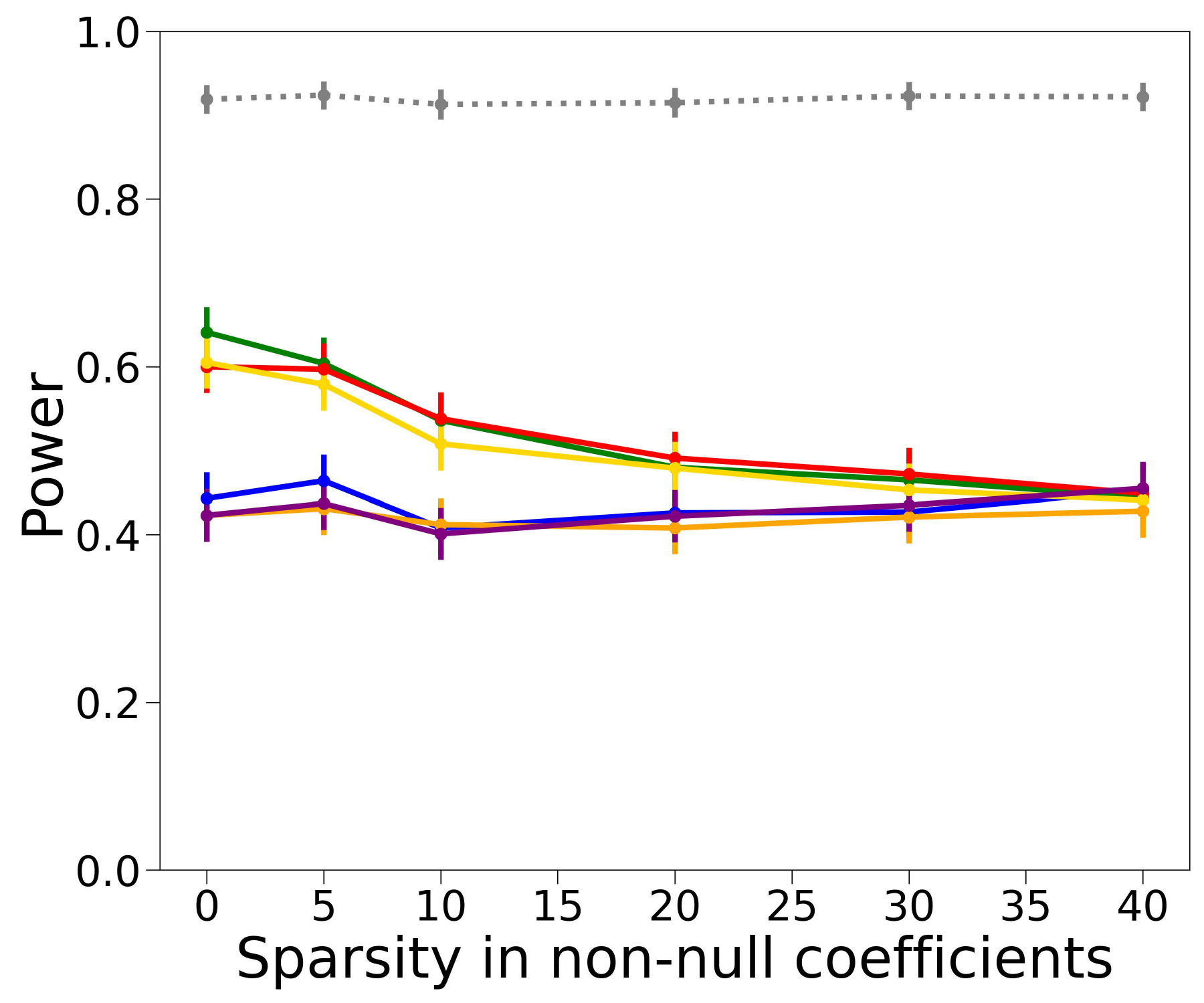}
    \end{subfigure}
    \hspace{0.2cm}
    \begin{subfigure}[t]{0.4\linewidth}
        \centering
        \includegraphics[width=\linewidth]{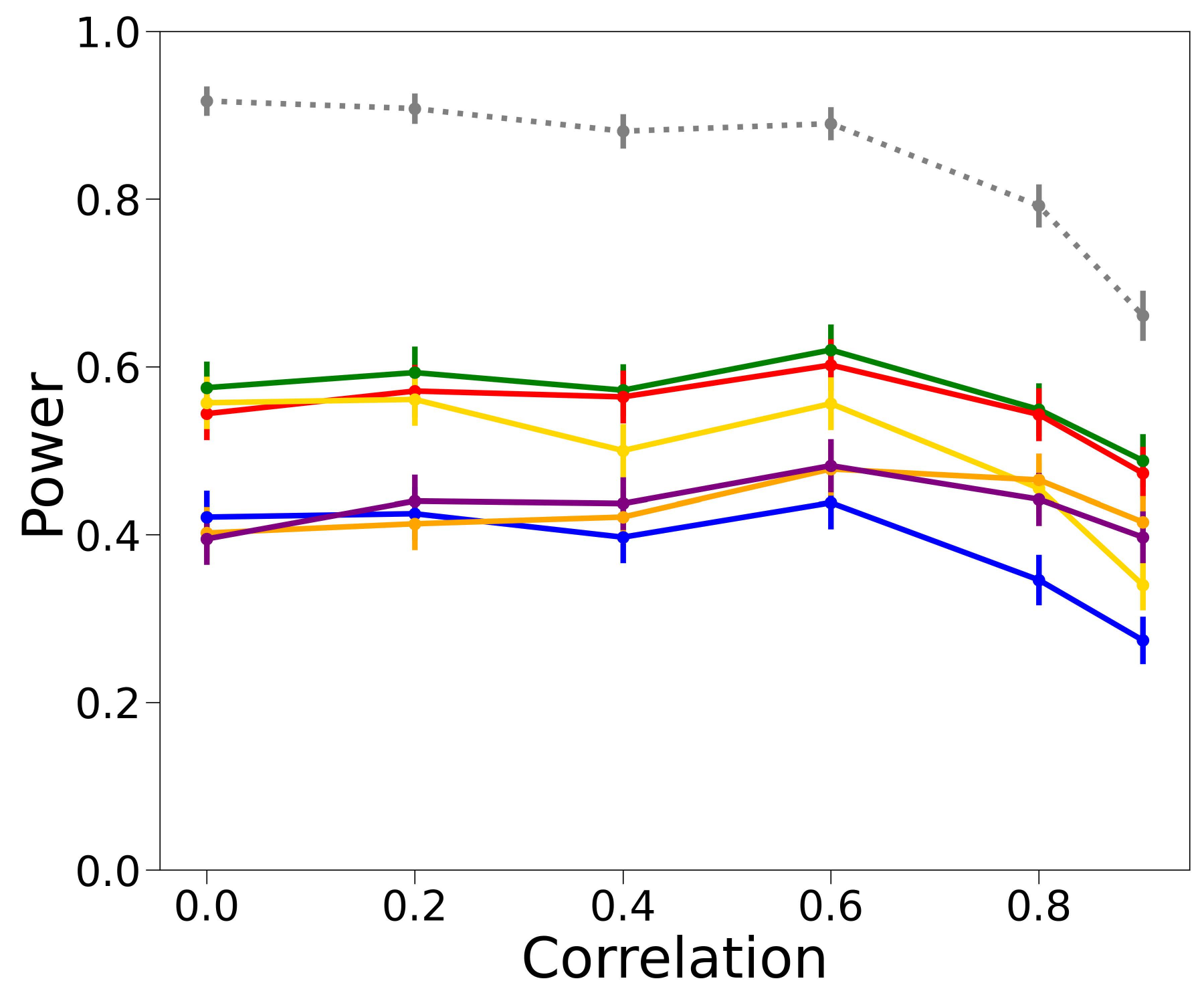}
    \end{subfigure} \\
    \begin{subfigure}[t]{0.4\linewidth}
        \centering
        \includegraphics[width=\linewidth]{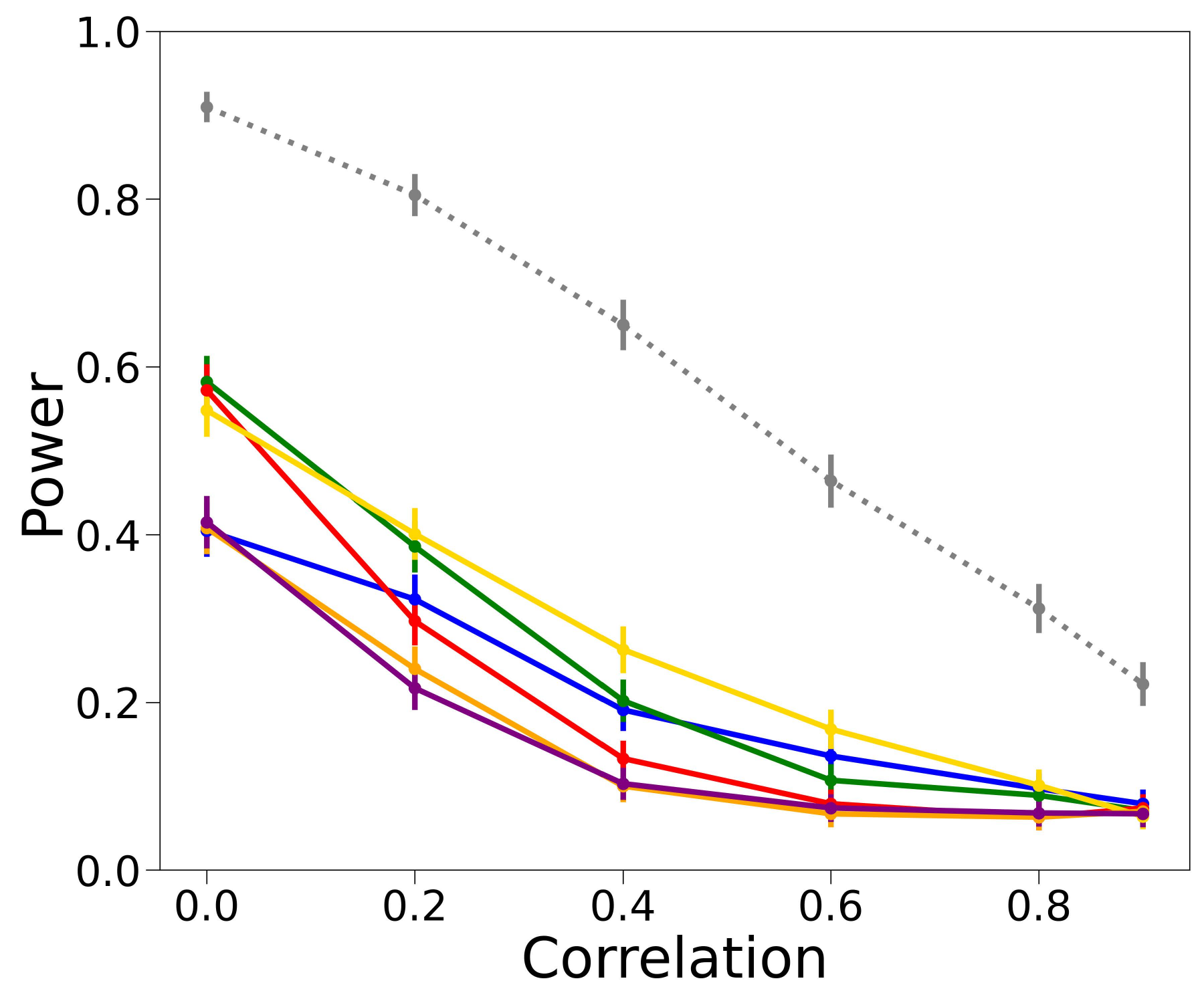}
    \end{subfigure}
     \caption{Top four panel settings are the same as those in Figure \ref{fig:powers_standard}. The bottom panel is the same as the middle panel of Figure \ref{fig:PC_vary_corr}, except using a non-orthogonal design.}
    \label{fig:PC_stand}
\end{figure}

When we first introduced the $L$-test, we mentioned that we could have considered constructing a test based on LASSO-type estimators other than the group LASSO. Figure \ref{fig:powers_lassos} compares the powers of the $L$-test and tests that use the p-value in \eqref{eq:mc_p} but with the group LASSO swapped out for the LASSO \cite{Tibshirani1996}, elastic net \cite{Zou&Hastie2005}, or sparse group LASSO \cite{Tibshirani2013}. For these tests, just as for the $L$-test, hyperparameter tuning is done on $(\tilde{\bm{y}}, \bm{X})$, where $\bm{\tilde y} \sim \bm{y}\mid \suffstat$ under $\hyp$, using 10-fold cross-validation. For the elastic net and sparse group LASSO, hyperparameter tuning is performed for both the penalty and the $\ell^1$-ratio.
\begin{figure}[ht]
    \centering
    \begin{subfigure}[t]{0.4\linewidth}
        \centering
        \includegraphics[width=\linewidth]{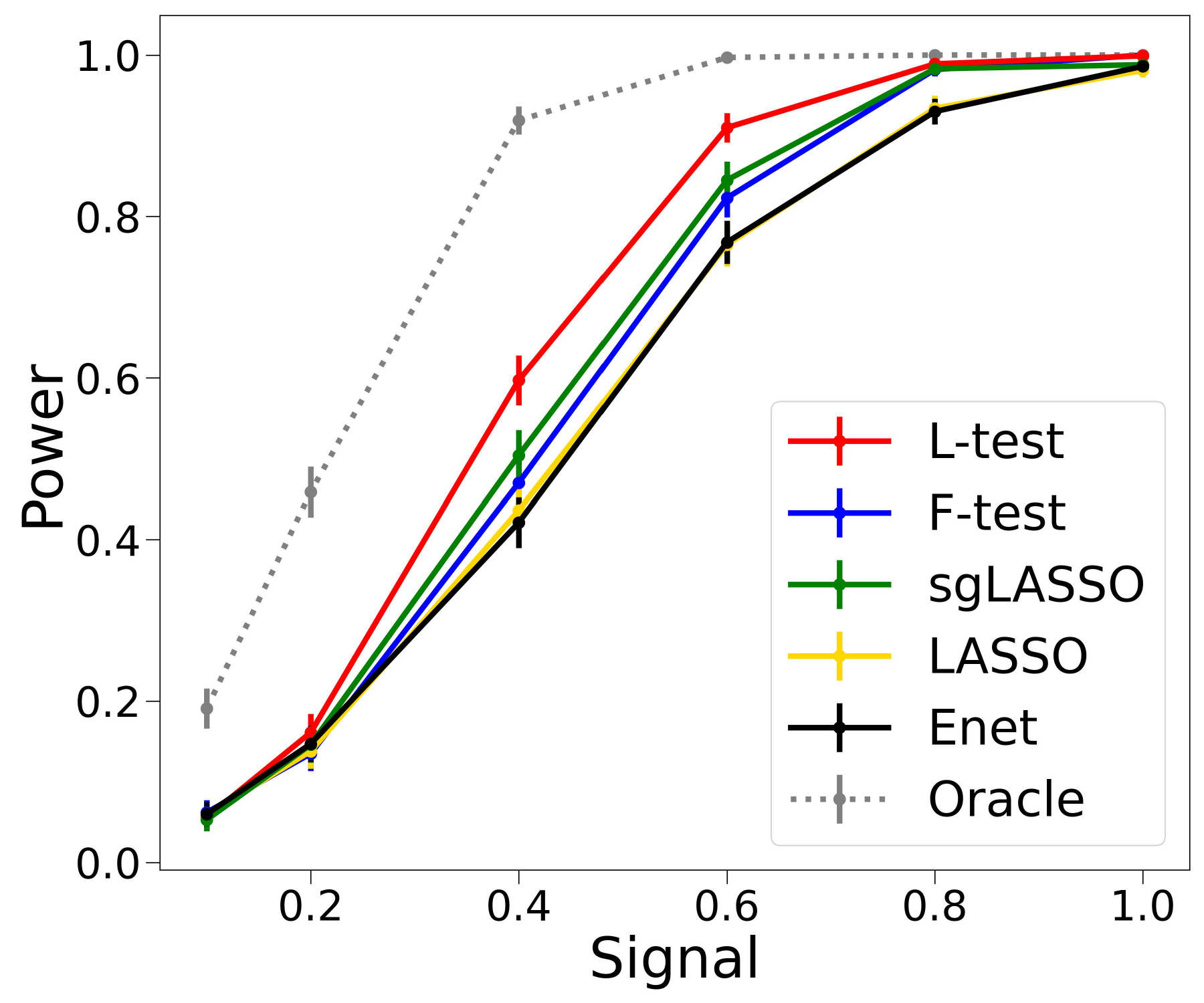}
    \end{subfigure}
    \hspace{0.2cm}
    \begin{subfigure}[t]{0.4\linewidth}
        \centering
        \includegraphics[width=\linewidth]{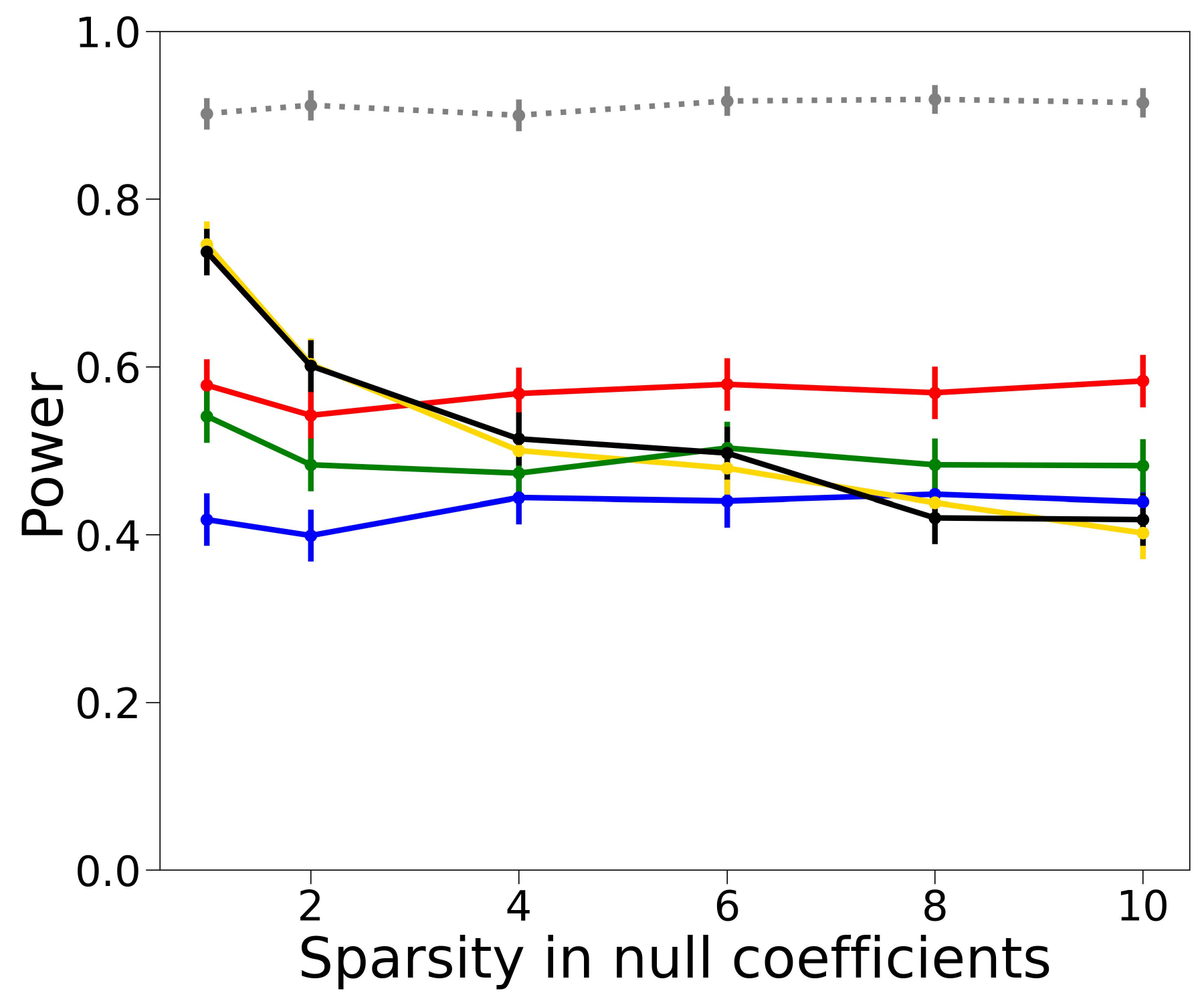}
    \end{subfigure} \\
    \begin{subfigure}[t]{0.4\linewidth}
        \centering
        \includegraphics[width=\linewidth]{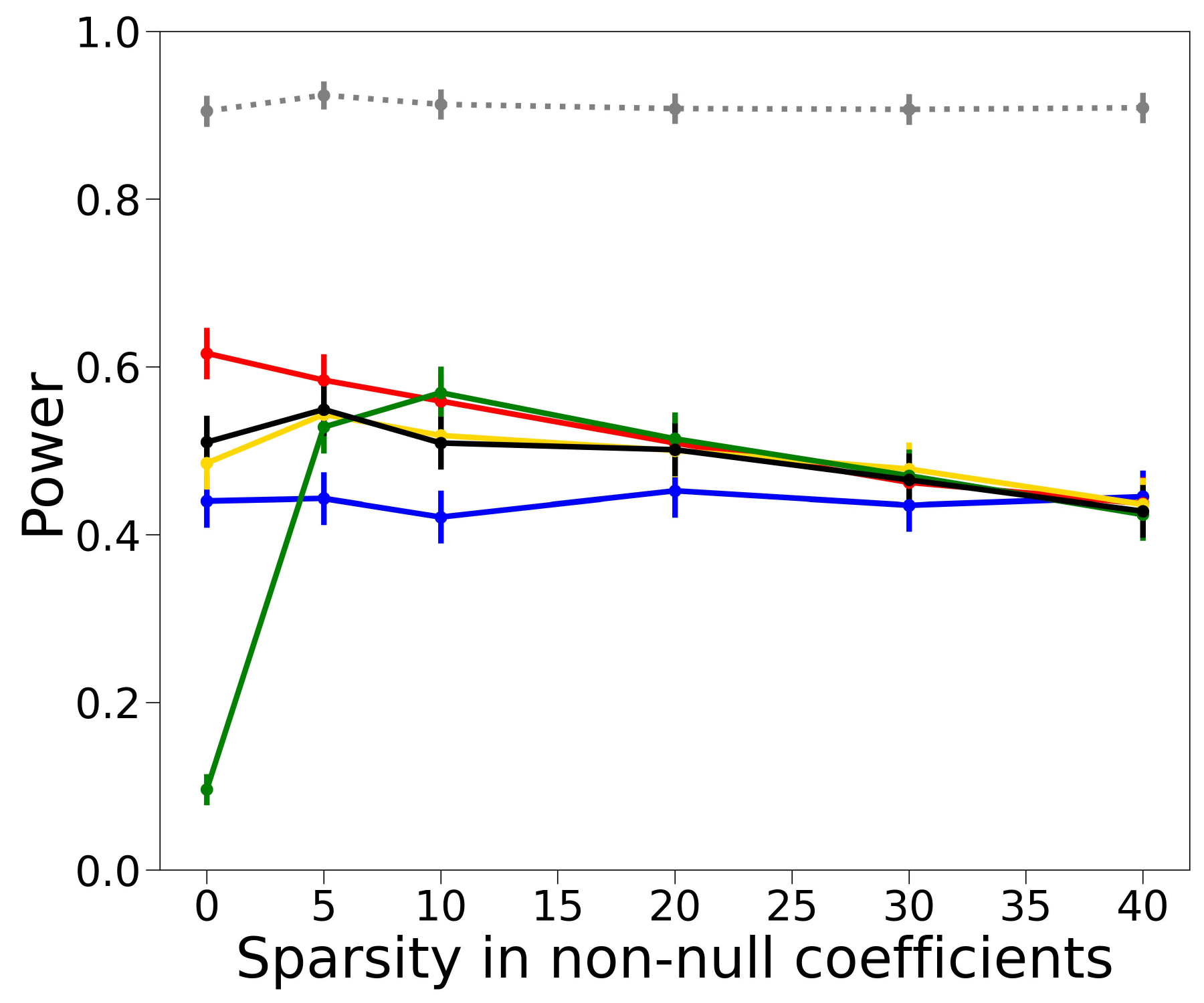}
    \end{subfigure}
    \hspace{0.2cm}
    \begin{subfigure}[t]{0.4\linewidth}
        \centering
        \includegraphics[width=\linewidth]{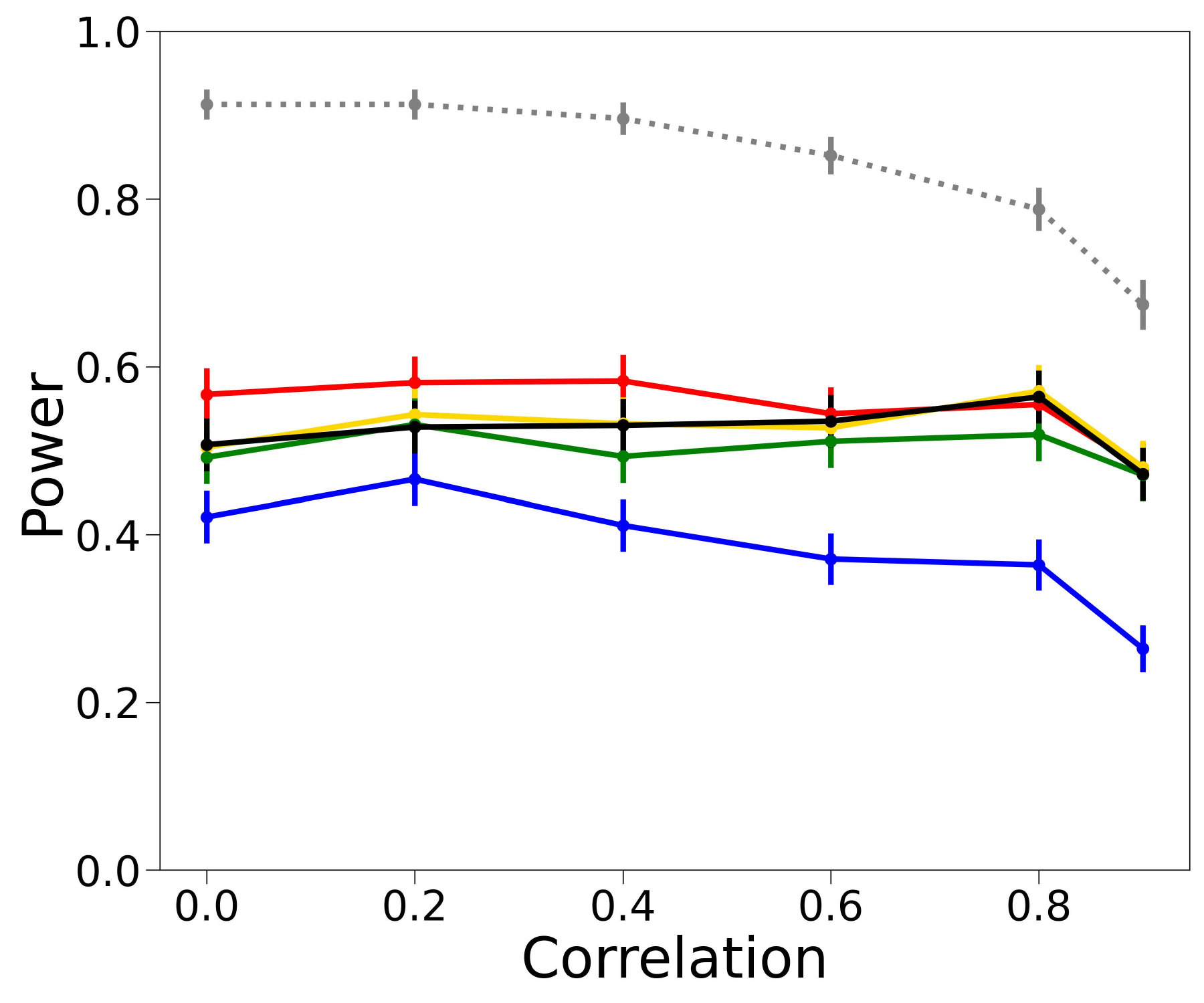}
    \end{subfigure}
    \caption{Comparison of $L$-test with other LASSO-based tests. Same settings as in Figure \ref{fig:powers_standard}.}
    \label{fig:powers_lassos}
\end{figure}

The $L$-test and sparse group LASSO test outperform the $F$-test and almost always have higher power than the LASSO and elastic net tests, which perform almost identically. The only exception is the setting where the null is 1-sparse, exactly when we would expect these latter estimates to perform well. While the $L$-test could in theory be modified to approximate the MC test that instead uses the statistic obtained by substituting the LASSO or elastic net for $\hat{\bm{\beta}}_{1:k, \mathrm{OLS}}$ in \eqref{eq:F-test_stat_cond} to achieve even higher power in this setting, it represents a relatively narrow scenario, unlikely to occur in practice when testing groups of coefficients, where one typically expects more than a single nonzero element.

Almost all of the tests show worse performance as sparsity among the non-null parameters increases. The only exception to this trend is the initial poor performance of the sparse group LASSO test when $k_2 = 0$. In this setting, the sparse group LASSO often returns the zero vector, resulting in large p-values. The other LASSO-based estimators do not exhibit this behavior because they penalize coefficients differently: the LASSO, with individual-level penalties, is able to identify the strong signals in $k_1$ out of the first $k$ coefficients, while the group LASSO, with group-level penalties, is able to select the first group. The sparse group LASSO combines both penalties, and their interaction drives all coefficients to zero in this setting.

\subsection{Robustness}
\label{sec:robustness_add}
This section extends the results from Section \ref{sec:robustness} by considering additional settings for each of the four types of model violation. These settings are described below. All settings are tested in both low ($\rho = 0$) and high ($\rho = 0.5$) correlation regimes.
\begin{enumerate}
    \item \textbf{Heavy-tailed errors}: For each $(n, d)$ pair, we take the errors to be independently and identically distributed from the $t_\nu$-distribution with $\nu$ degrees of freedom for $\nu \in \{2, 5, 10, 15, 20, 30\}$. Errors are standardized by the theoretical standard deviation of the t-distribution except when $\nu = 2$, since there is no finite second moment in this case. Smaller values of $\nu$ correspond to larger tails, and taking $\nu = \infty$ results in standard normal tails. 
    \item \textbf{Skewed errors}: Same as the heavy-tailed error experiments, except errors are taken to be independently and identically distributed from the Gamma distribution with scale parameter 1 and shape parameter $\alpha$ for $\alpha \in \{1, 2, 4, 6, 8, 10\}$. Errors are standardized. Larger values of $\alpha$ correspond to greater symmetry.  
    \item \textbf{Heteroskedastic errors}: Similar to the prior two settings, except errors are generated independently in the following manner. Letting $r_i$ denote the mean of the $i$-th row of $\bm{X}$,
    \begin{align*}
        \epsilon_i \sim 
        \begin{cases}
            \mathcal N(0, 1) \text{, if $r_i \leq 0$} \\
            \mathcal{N}(0, \eta^2) \text{, if $r_i > 0$}
        \end{cases}.
    \end{align*}
    We test $\eta^2 \in \{0.01, 0.25, 0.5, 1, 4, 8\}$, with larger departures from 1 corresponding to a greater degree of heteroskedasticity. 
    \item \textbf{Non-linearity}: Unlike the prior settings, errors are taken to be independently and identically distributed from the normal distribution with mean 0 and variance $\sigma^2$, but the design matrix is taken to be $\bm{X}^\delta$, where the $ij$-th entry is defined as $(X^\delta)_{ij} = \sign(X_{ij})|X_{ij}|^\delta$. We test $\delta \in \{0.3, 0.5, 1, 2, 3, 4\}$, where larger departures from 1 correspond to greater non-linearity in the model.
\end{enumerate}

\begin{figure}[ht]
    \centering
    \begin{subfigure}[t]{0.4\linewidth}
        \centering
        \includegraphics[width=\linewidth]{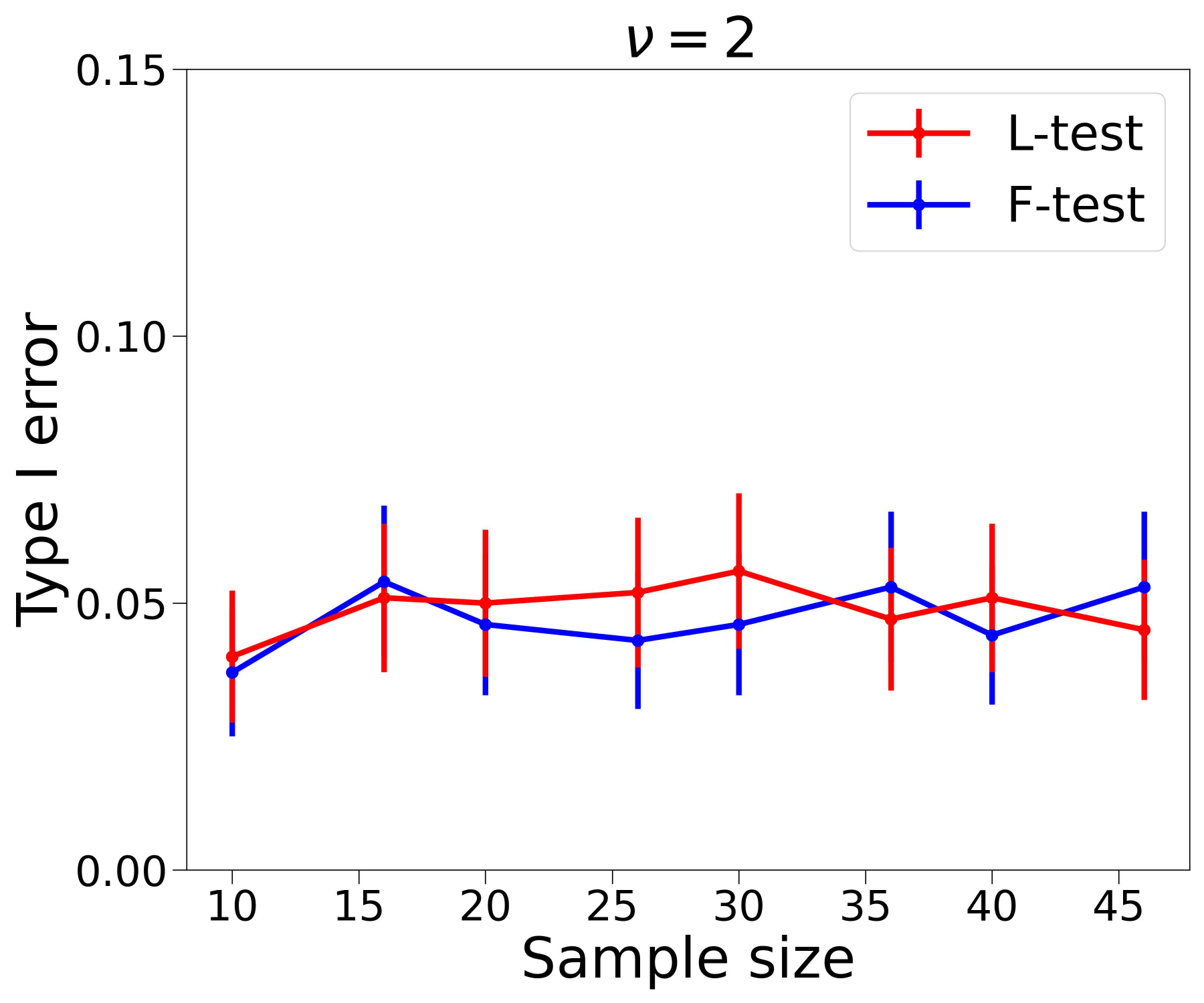}
    \end{subfigure}
    \hspace{0.2cm}
    \begin{subfigure}[t]{0.4\linewidth}
        \centering
        \includegraphics[width=\linewidth]{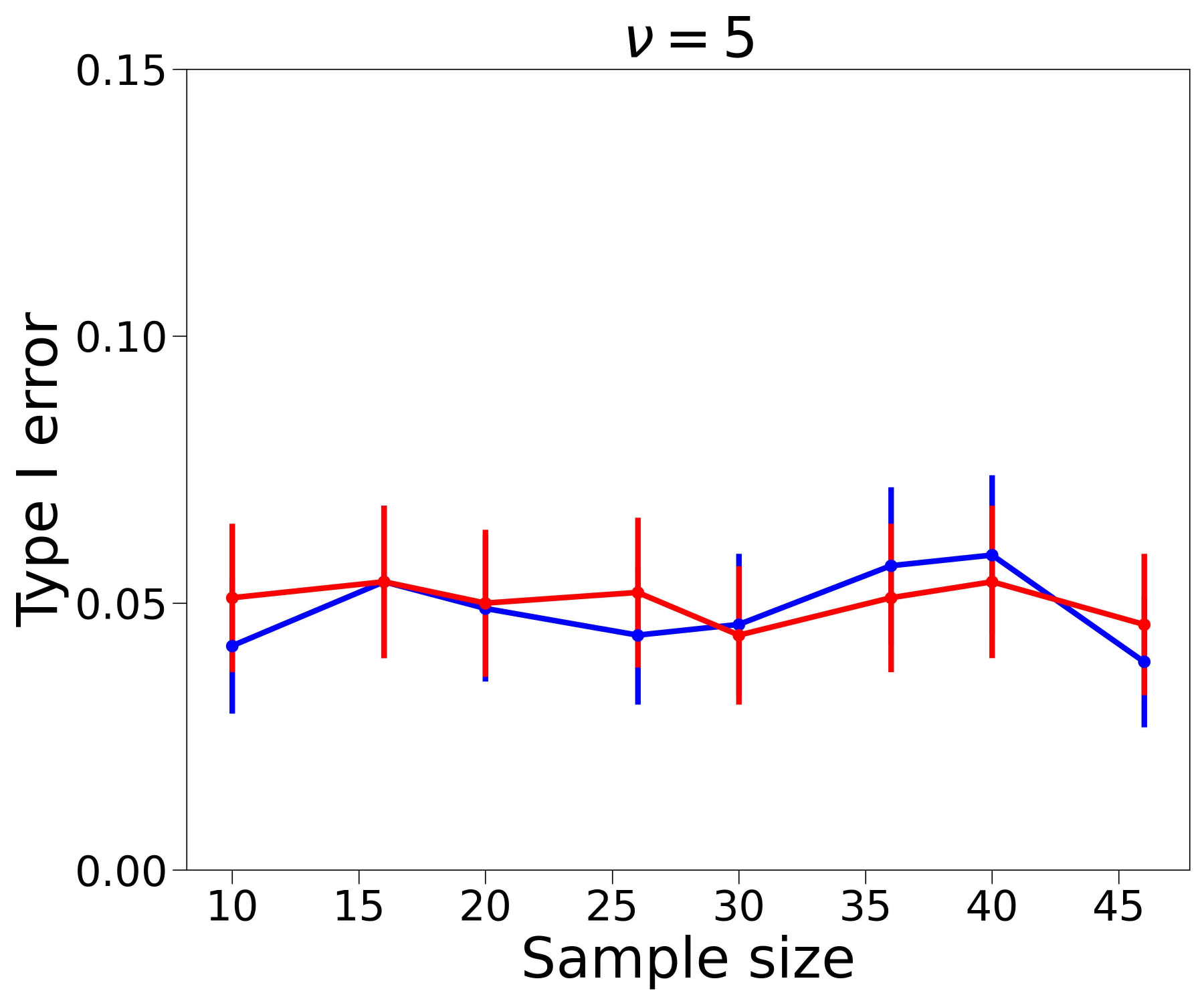}
    \end{subfigure} \\
    \begin{subfigure}[t]{0.4\linewidth}
        \centering
        \includegraphics[width=\linewidth]{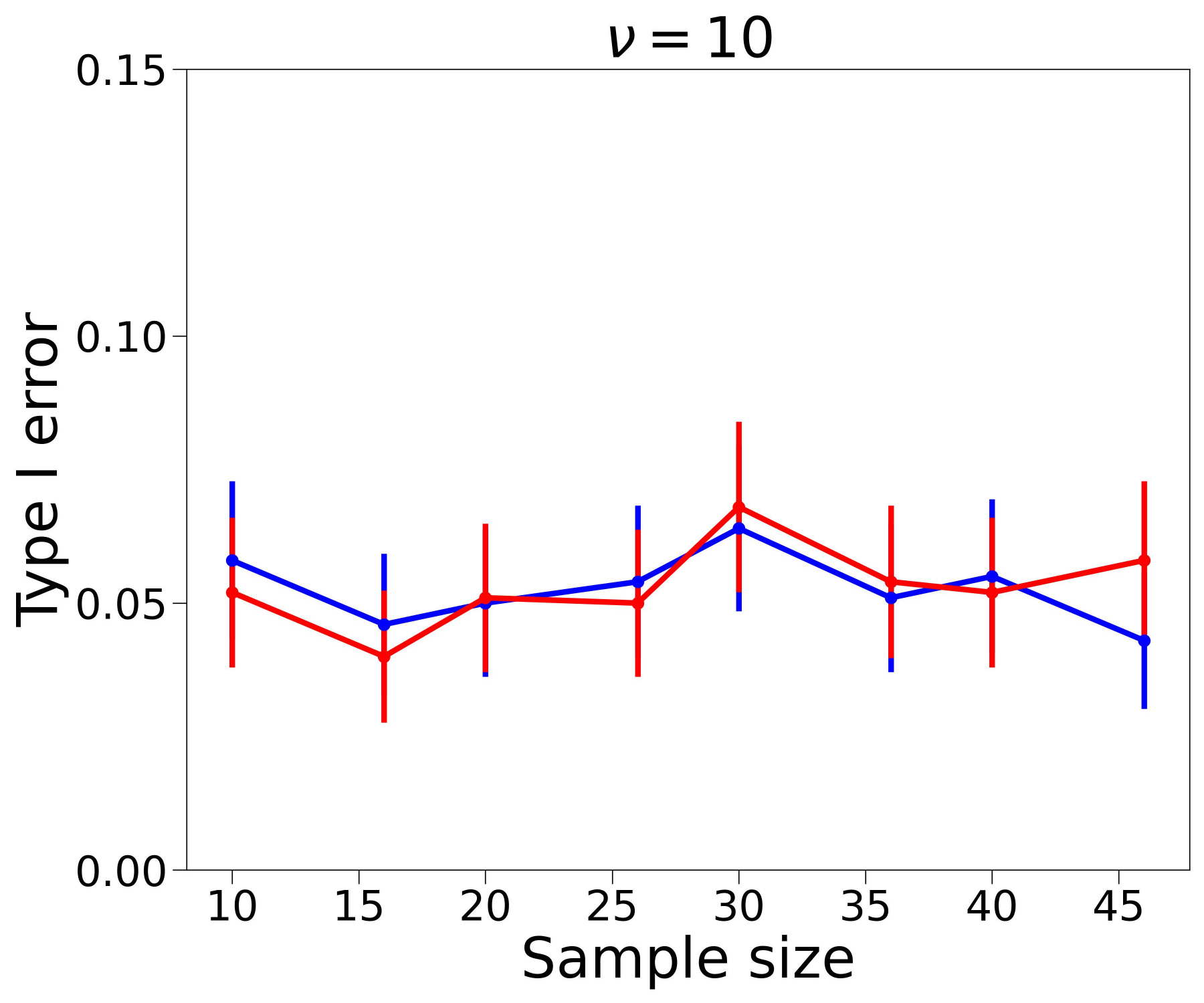}
    \end{subfigure}
    \hspace{0.2cm}
    \begin{subfigure}[t]{0.4\linewidth}
        \centering
        \includegraphics[width=\linewidth]{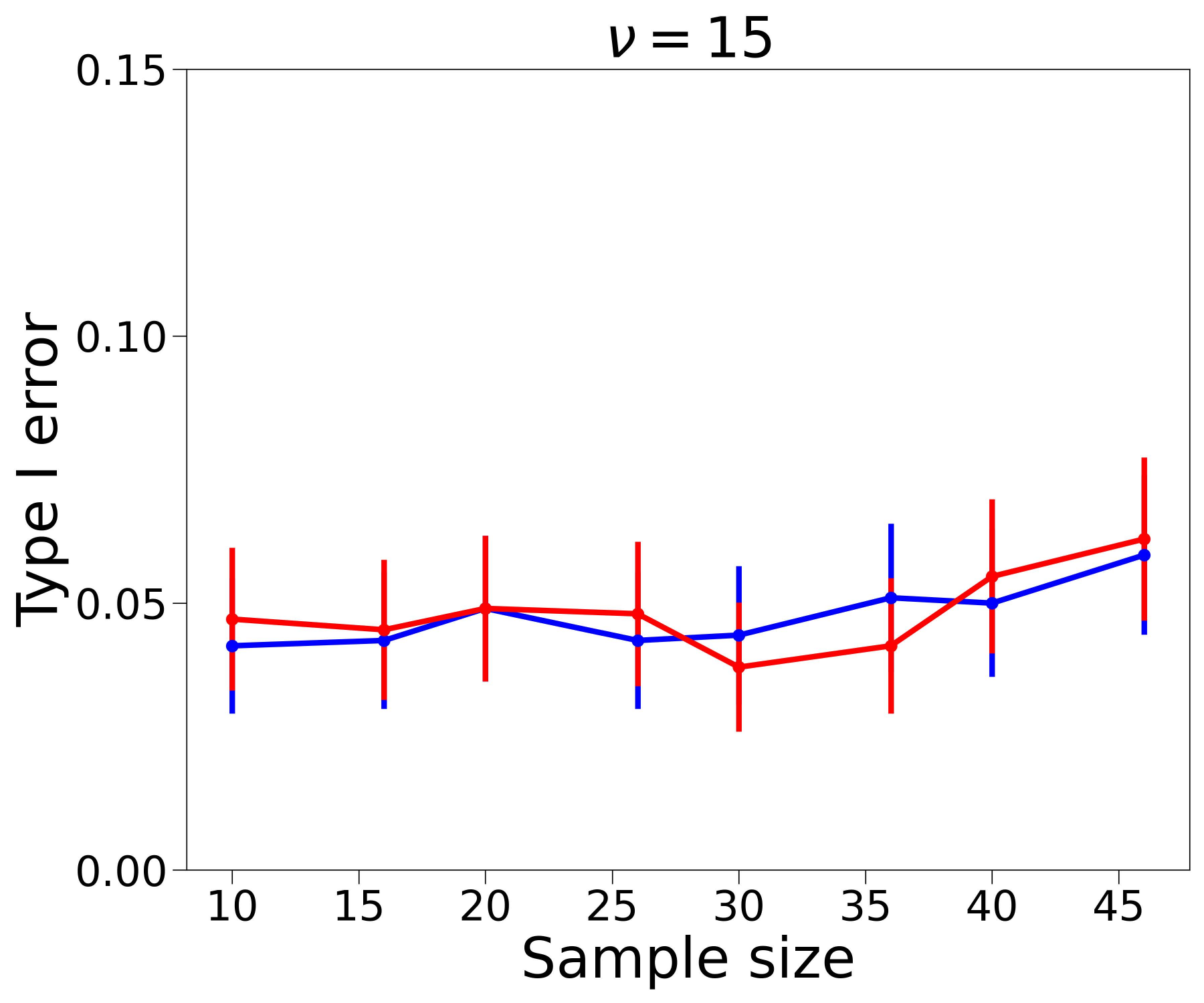}
    \end{subfigure} \\
    \begin{subfigure}[t]{0.4\linewidth}
        \centering
        \includegraphics[width=\linewidth]{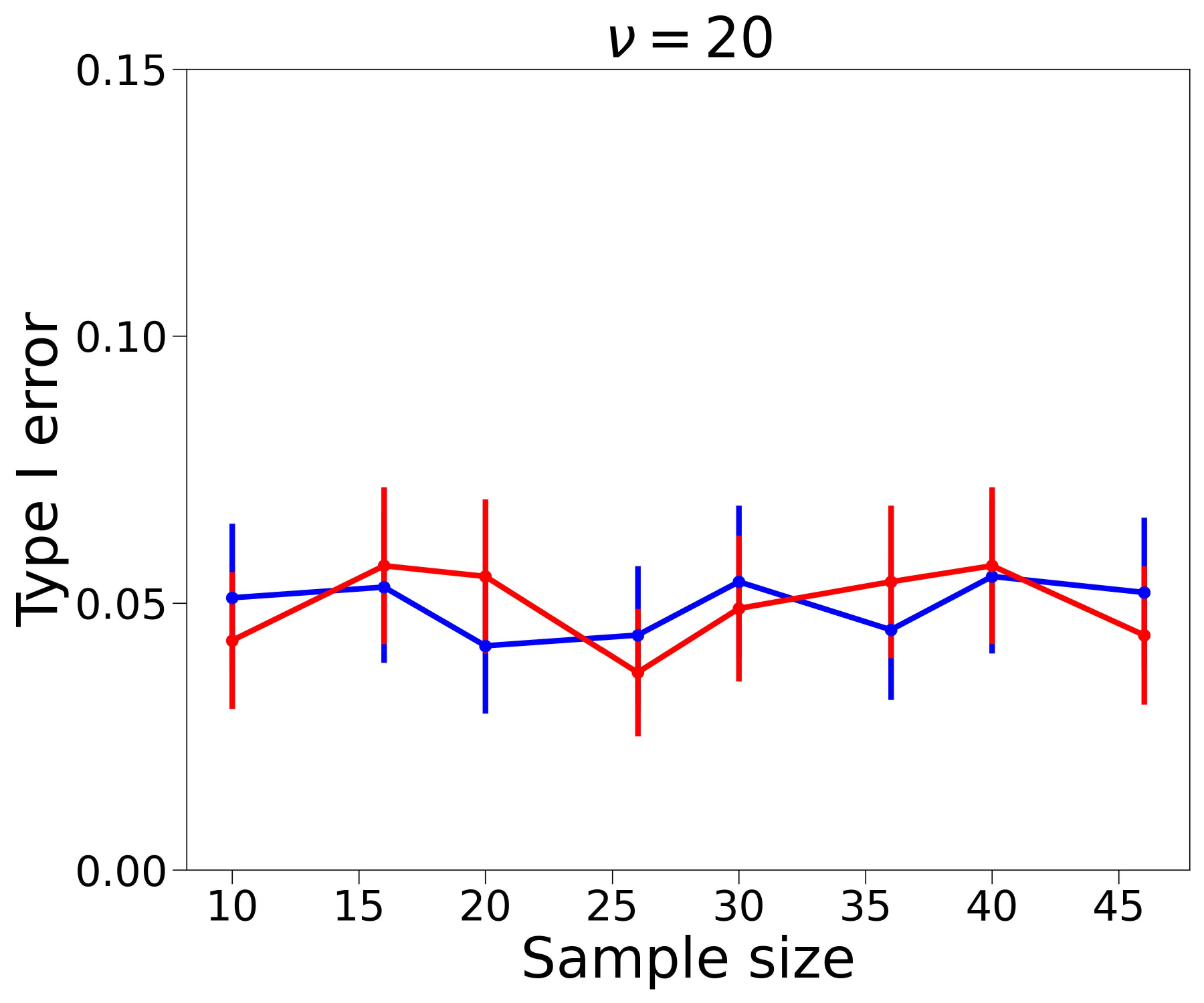}
    \end{subfigure}
    \hspace{0.2cm}
    \begin{subfigure}[t]{0.4\linewidth}
        \centering
        \includegraphics[width=\linewidth]{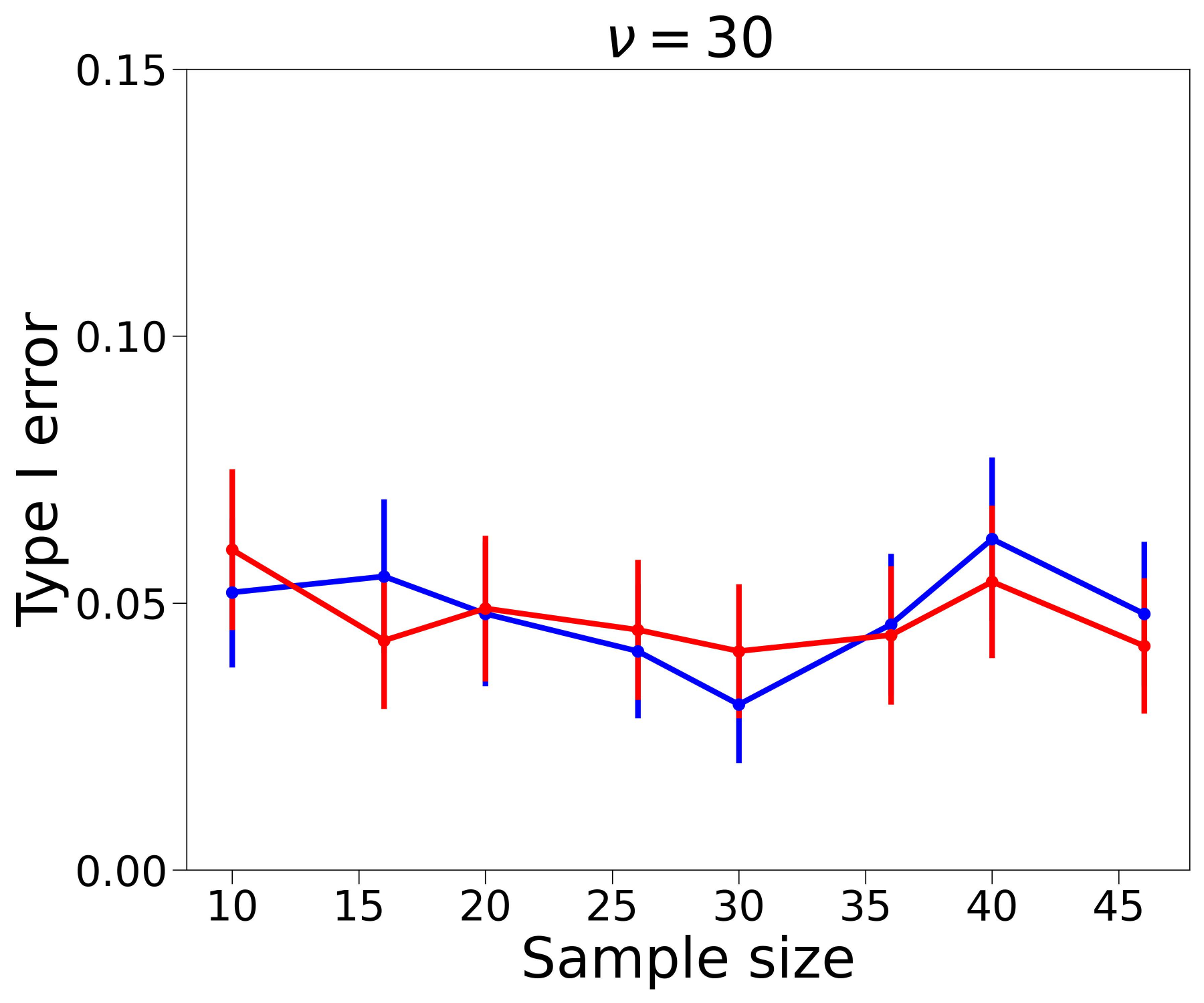}
    \end{subfigure}
    \caption{$\rho = 0$. $F$-test and $L$-test size for $t$-distributed errors with varying degrees of freedom.}
    \label{fig:rob_1}
\end{figure}

\begin{figure}[ht]
    \centering
    \begin{subfigure}[t]{0.4\linewidth}
        \centering
        \includegraphics[width=\linewidth]{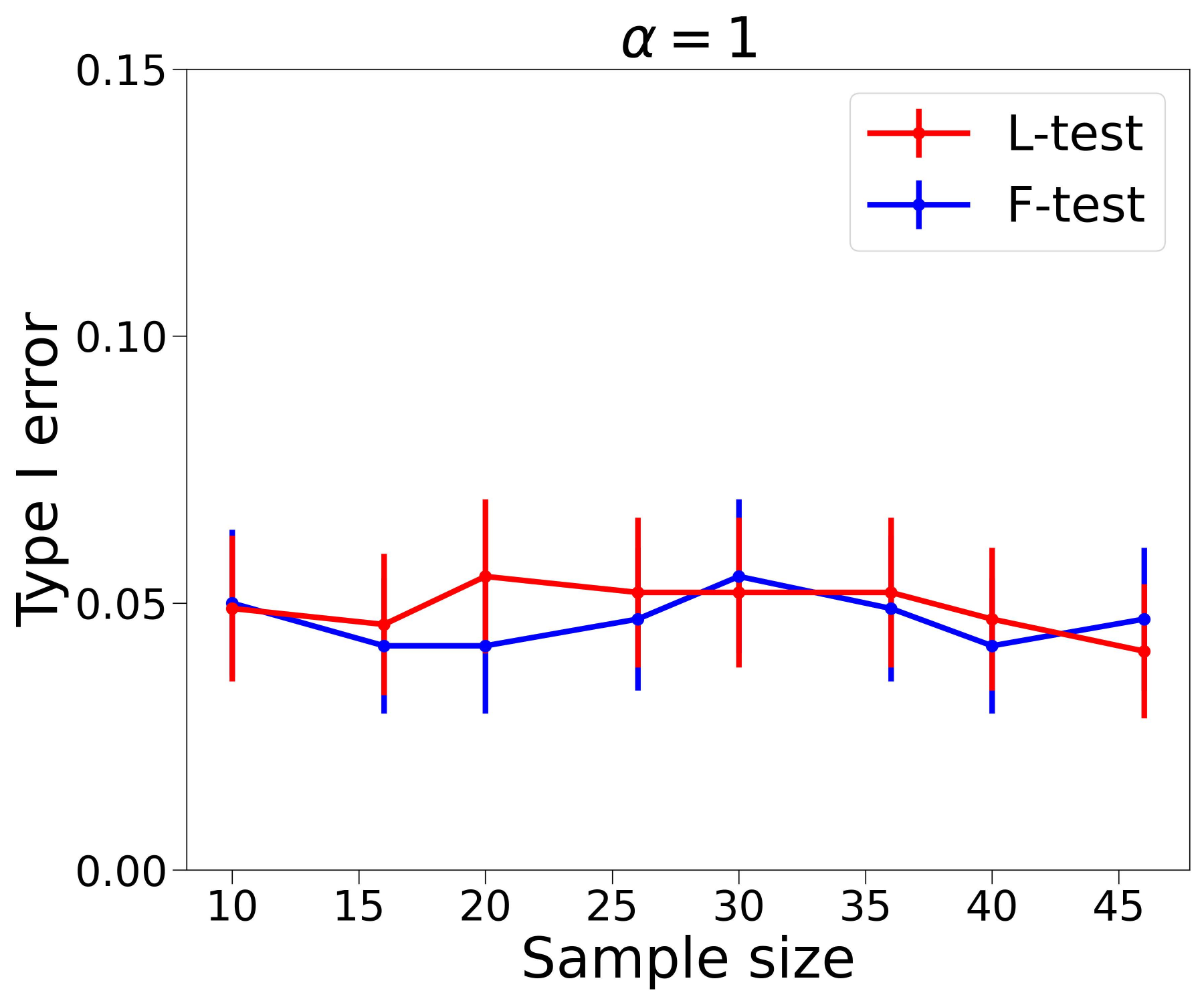}
    \end{subfigure}
    \hspace{0.2cm}
    \begin{subfigure}[t]{0.4\linewidth}
        \centering
        \includegraphics[width=\linewidth]{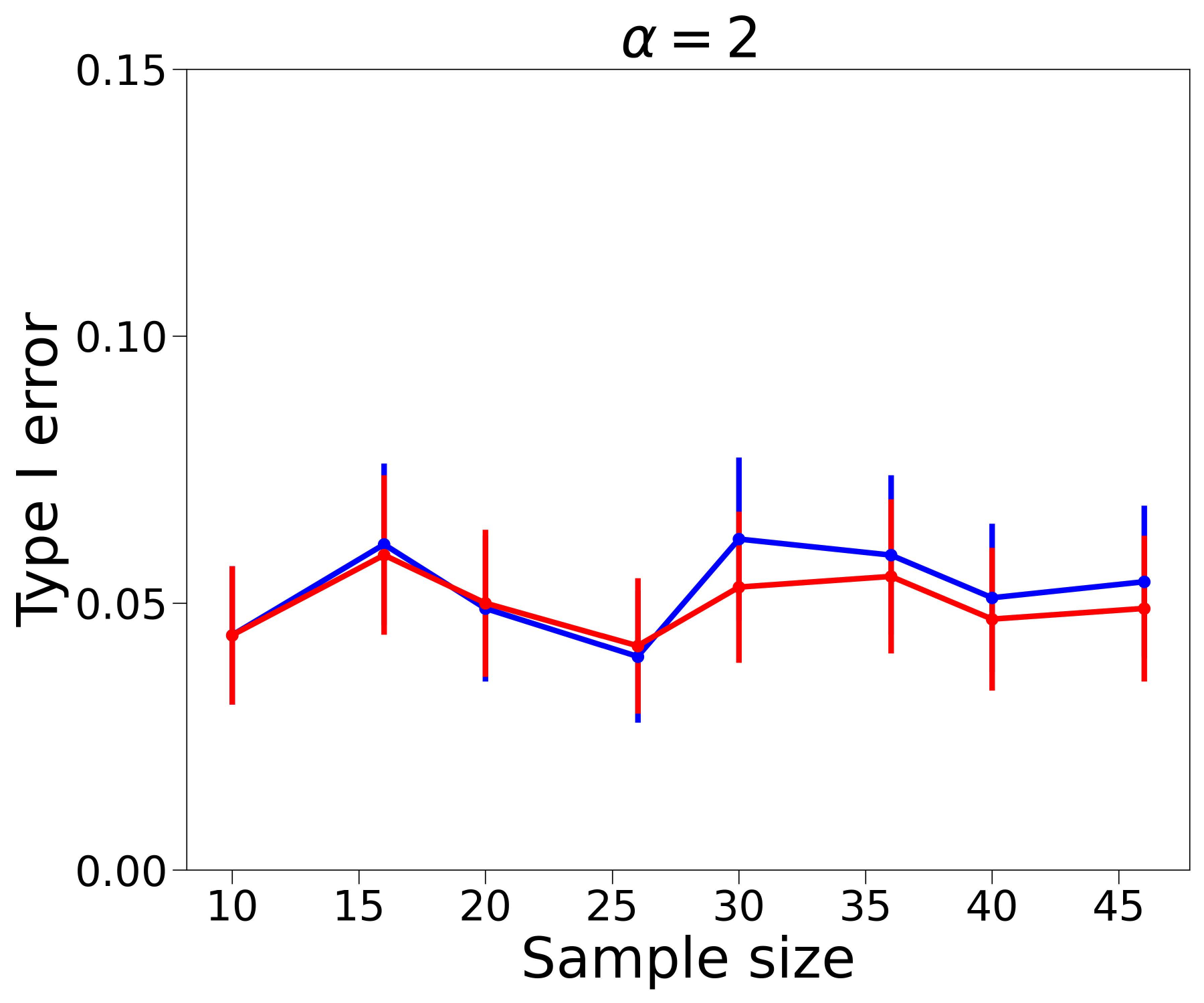}
    \end{subfigure} \\
    \begin{subfigure}[t]{0.4\linewidth}
        \centering
        \includegraphics[width=\linewidth]{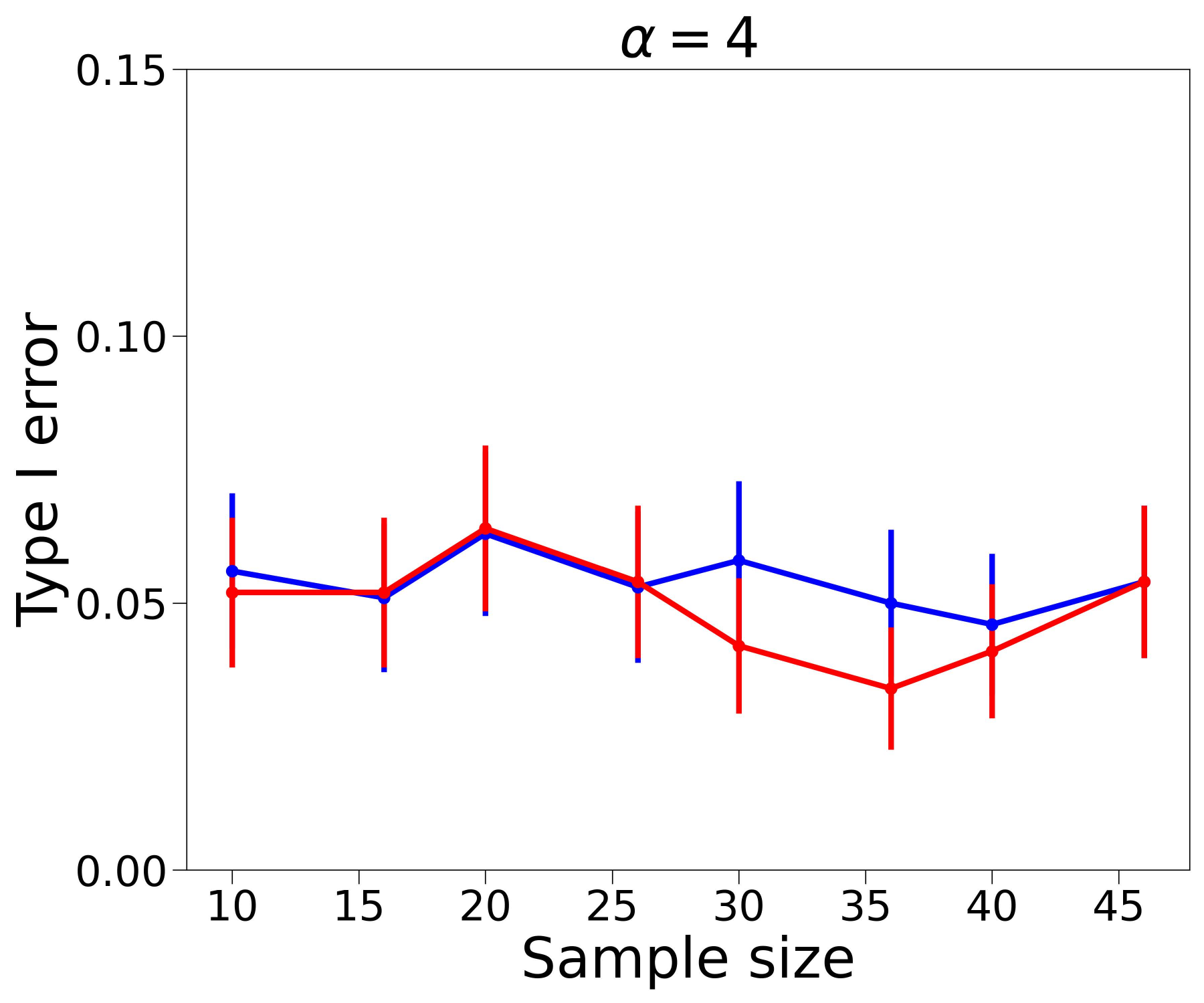}
    \end{subfigure}
    \hspace{0.2cm}
    \begin{subfigure}[t]{0.4\linewidth}
        \centering
        \includegraphics[width=\linewidth]{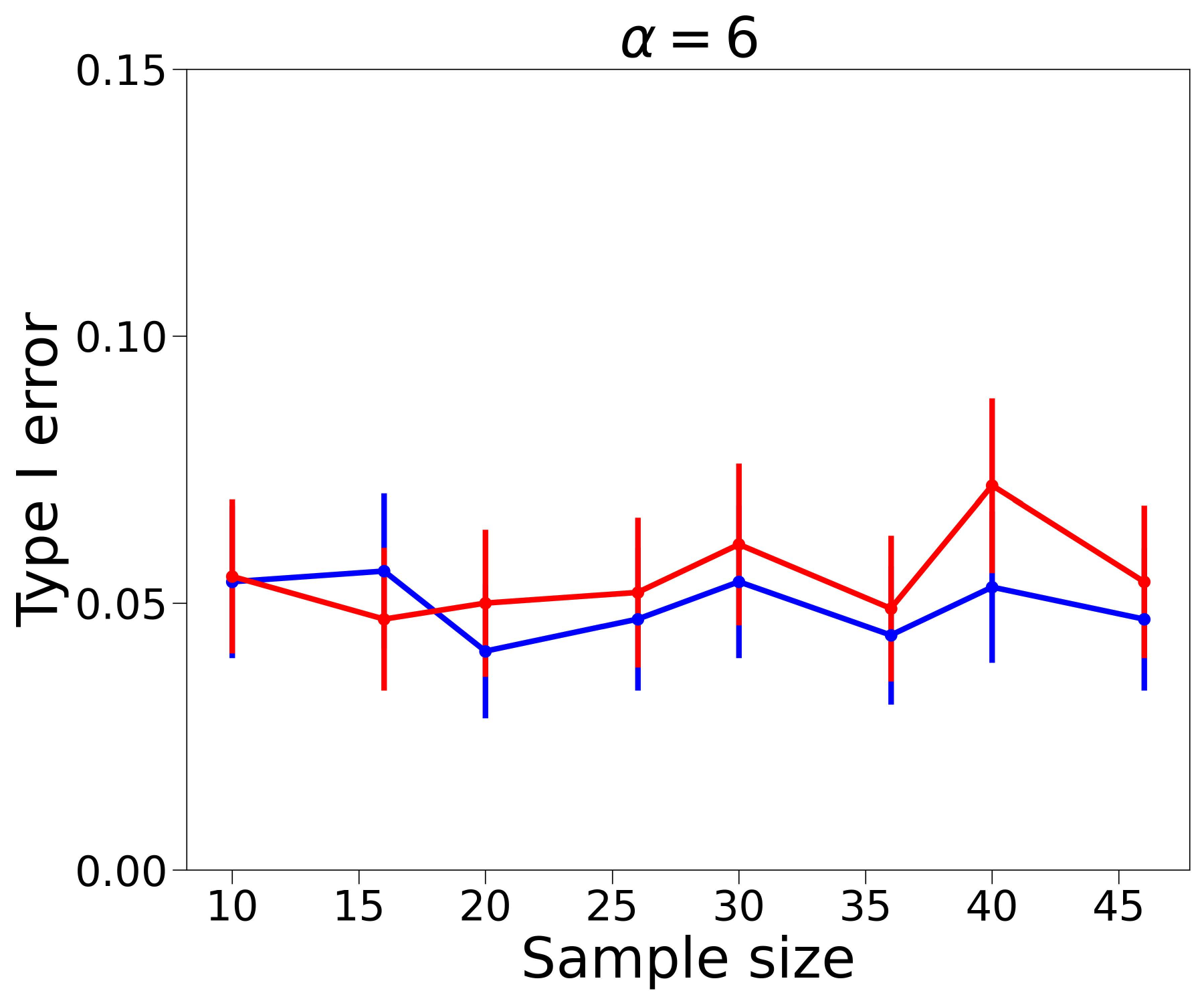}
    \end{subfigure} \\
    \begin{subfigure}[t]{0.4\linewidth}
        \centering
        \includegraphics[width=\linewidth]{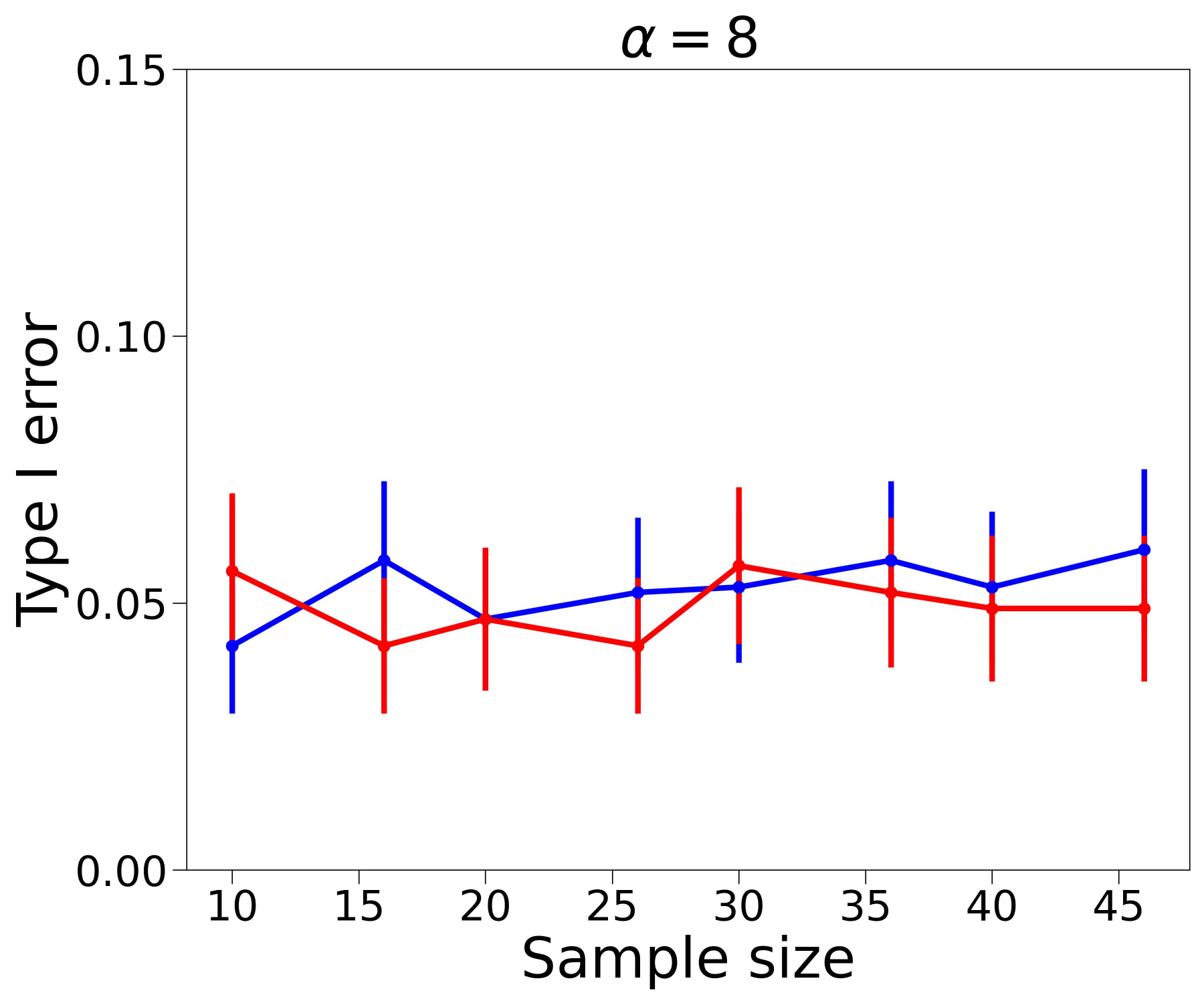}
    \end{subfigure}
    \hspace{0.2cm}
    \begin{subfigure}[t]{0.4\linewidth}
        \centering
        \includegraphics[width=\linewidth]{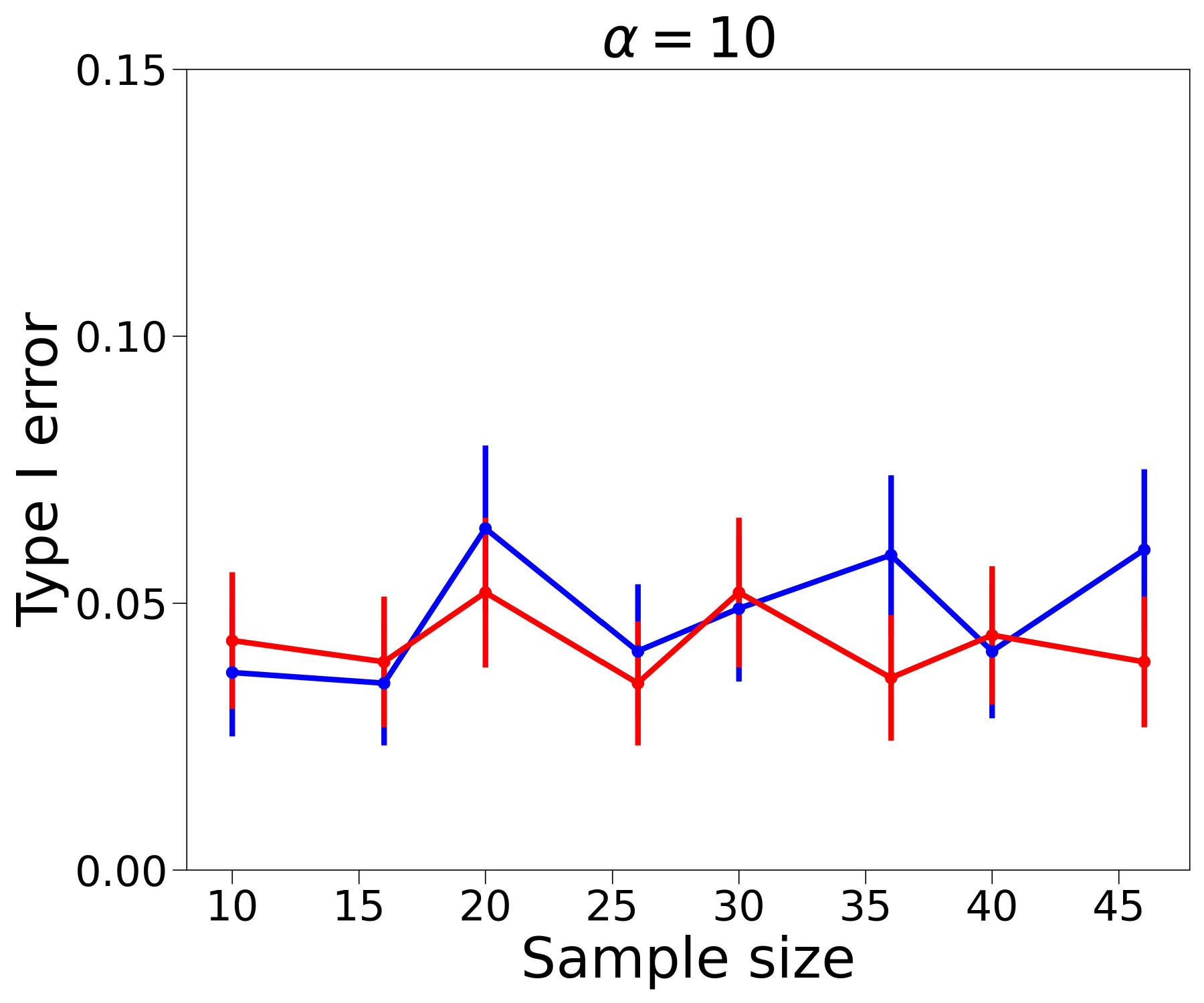}
    \end{subfigure}
    \caption{$\rho = 0$. $F$-test and $L$-test size for Gamma distributed errors with varying shape parameter (scale parameter = 1).}
    \label{fig:rob_2}
\end{figure}

\begin{figure}[ht]
    \centering
    \begin{subfigure}[t]{0.4\linewidth}
        \centering
        \includegraphics[width=\linewidth]{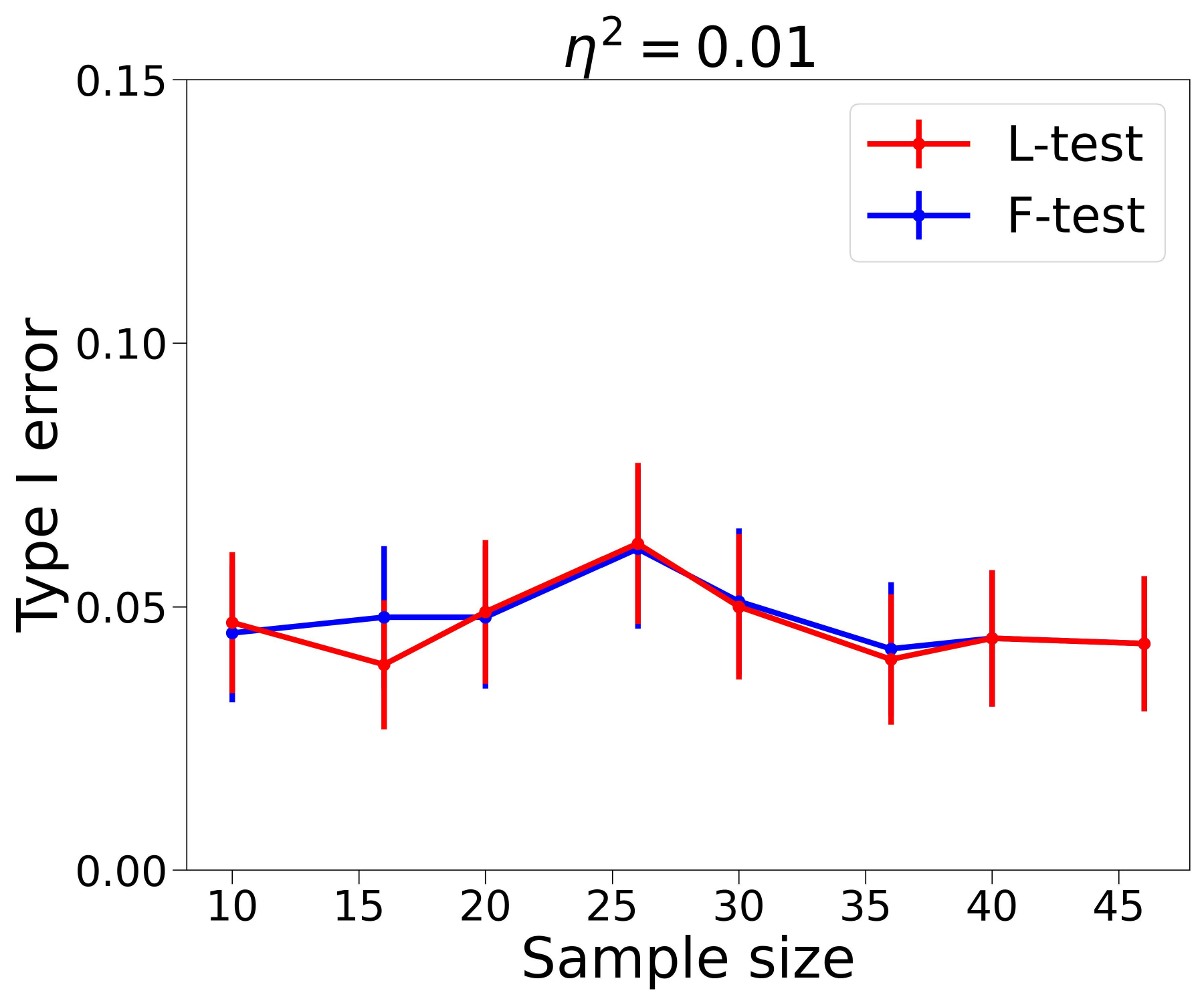}
    \end{subfigure}
    \hspace{0.2cm}
    \begin{subfigure}[t]{0.4\linewidth}
        \centering
        \includegraphics[width=\linewidth]{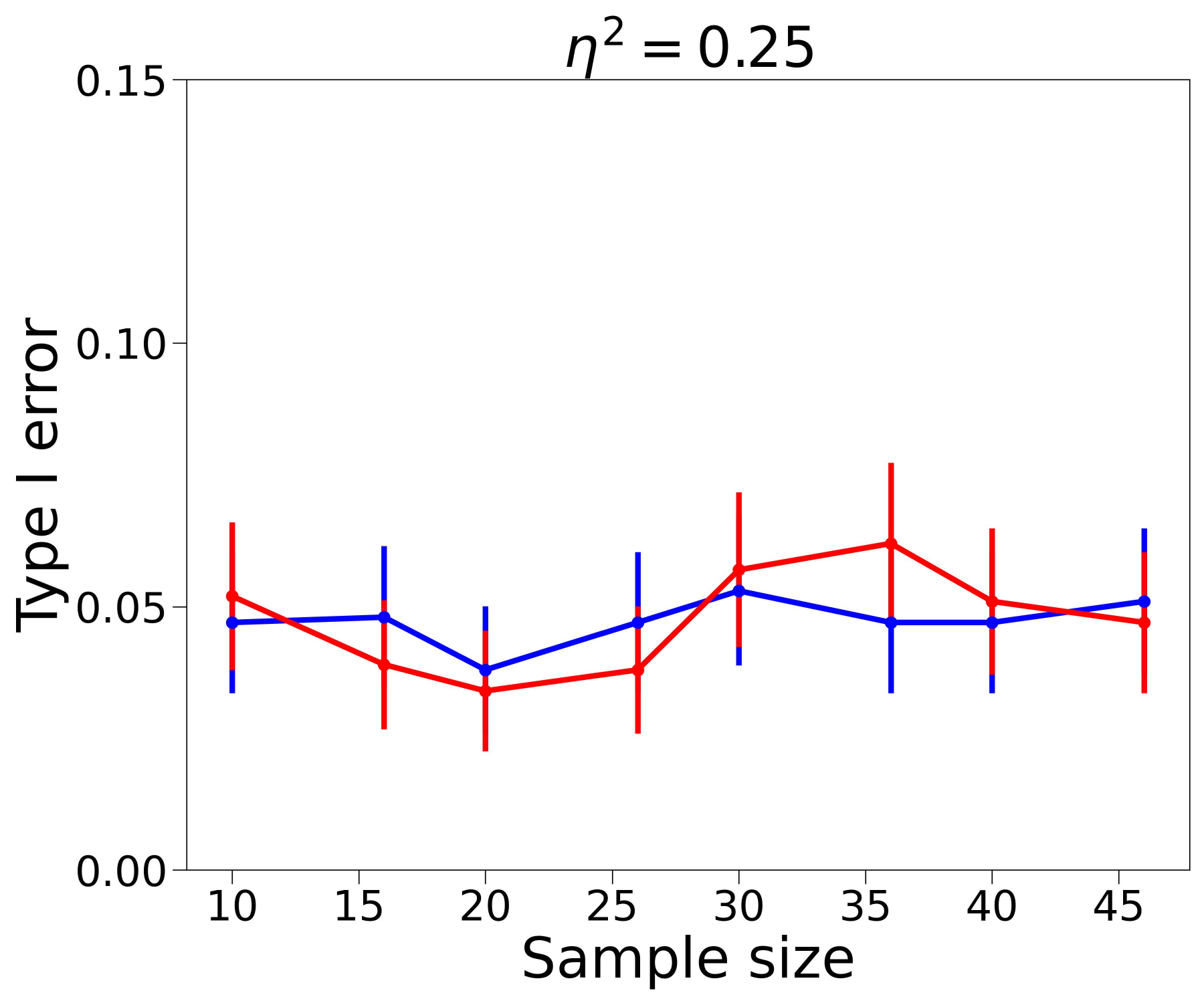}
    \end{subfigure} \\
    \begin{subfigure}[t]{0.4\linewidth}
        \centering
        \includegraphics[width=\linewidth]{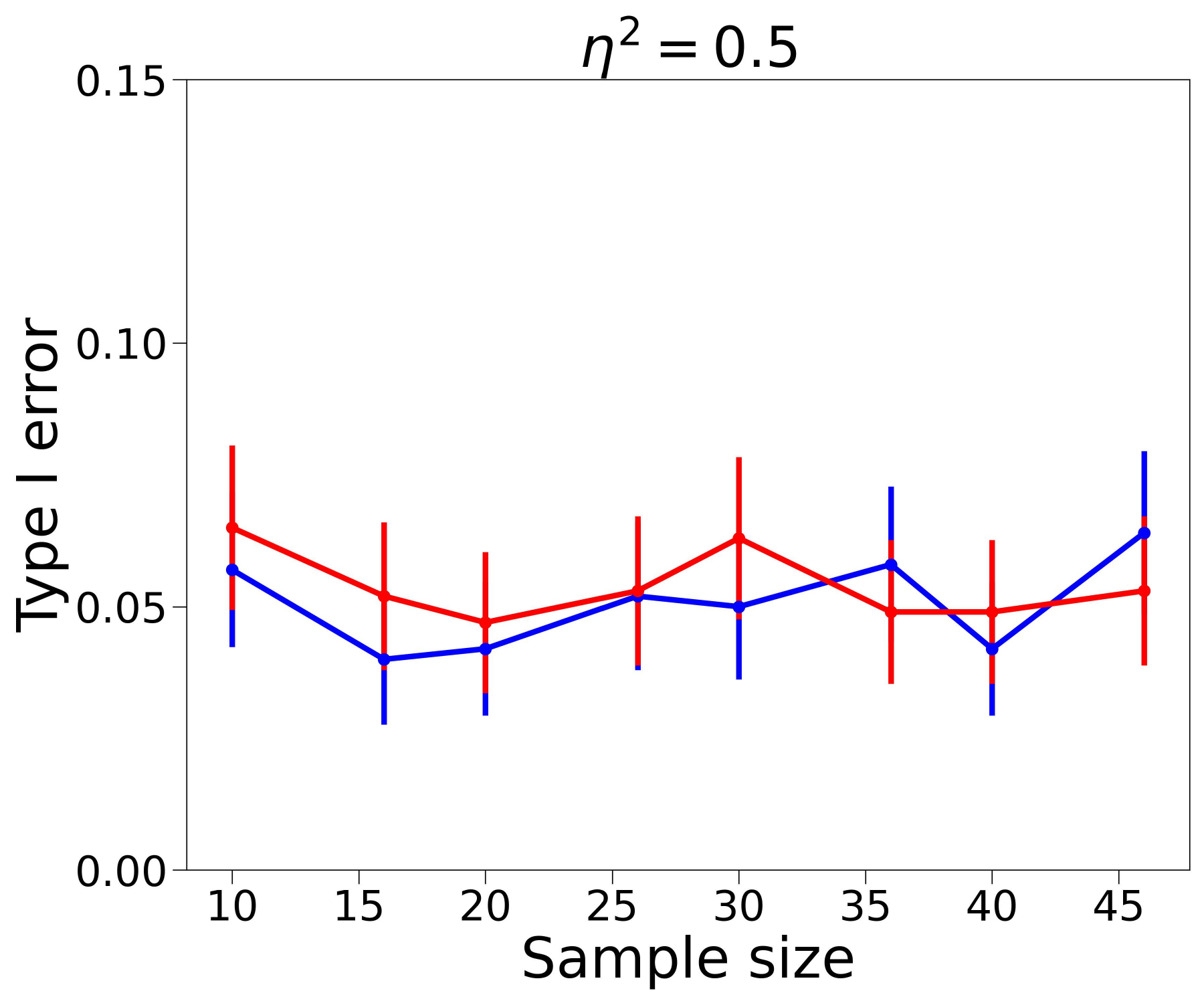}
    \end{subfigure}
    \hspace{0.2cm}
    \begin{subfigure}[t]{0.4\linewidth}
        \centering
        \includegraphics[width=\linewidth]{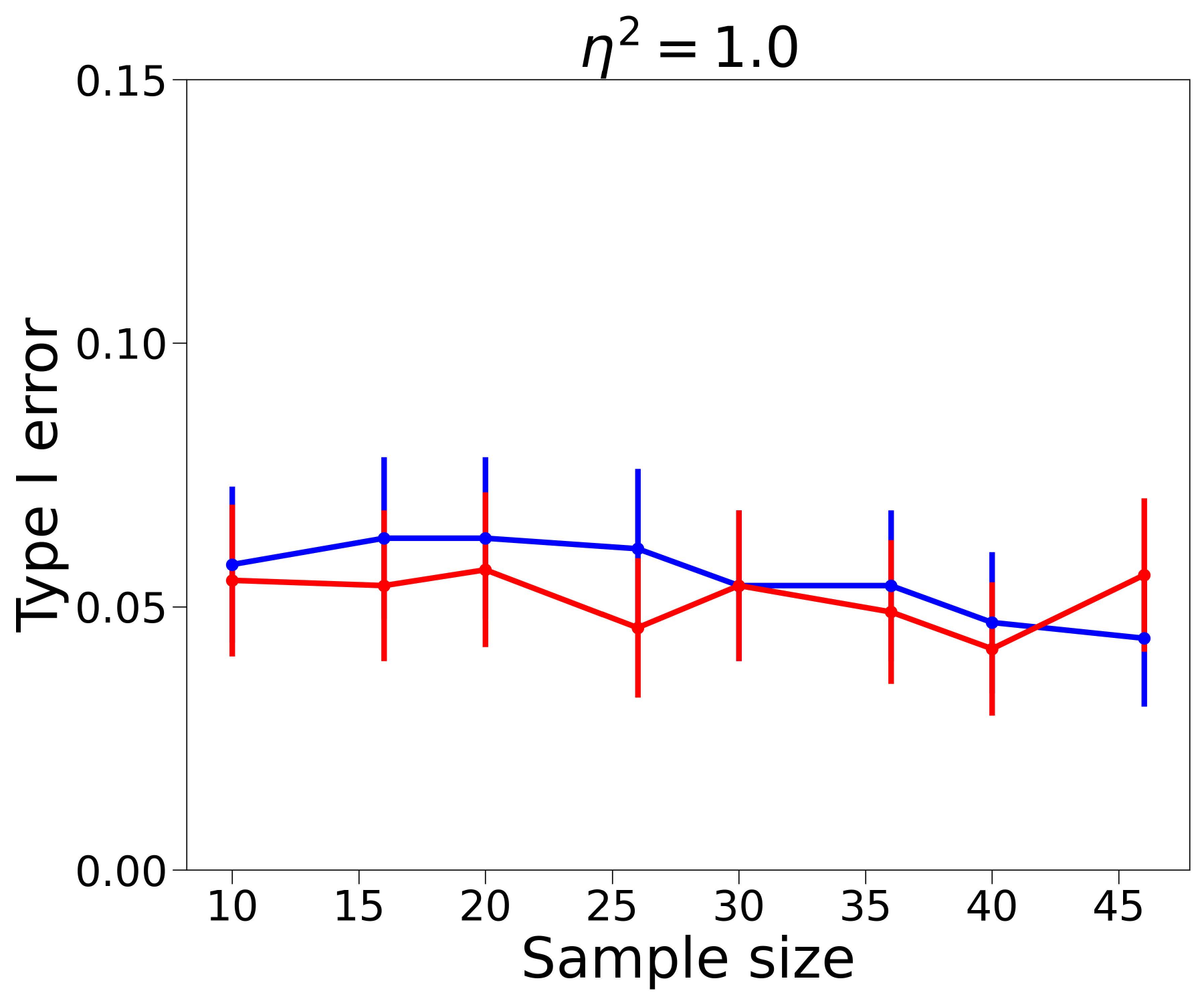}
    \end{subfigure} \\
    \begin{subfigure}[t]{0.4\linewidth}
        \centering
        \includegraphics[width=\linewidth]{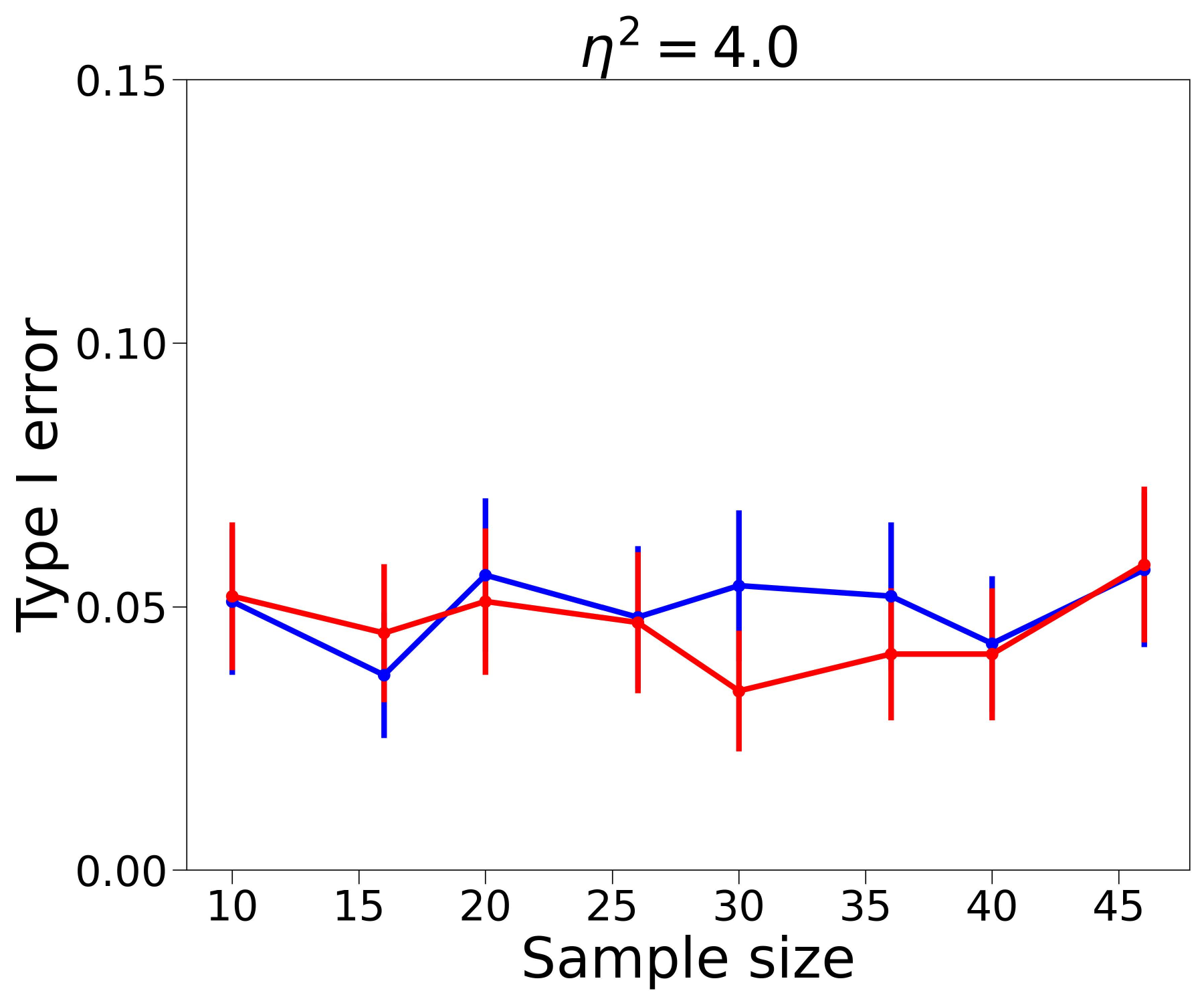}
    \end{subfigure}
    \hspace{0.2cm}
    \begin{subfigure}[t]{0.4\linewidth}
        \centering
        \includegraphics[width=\linewidth]{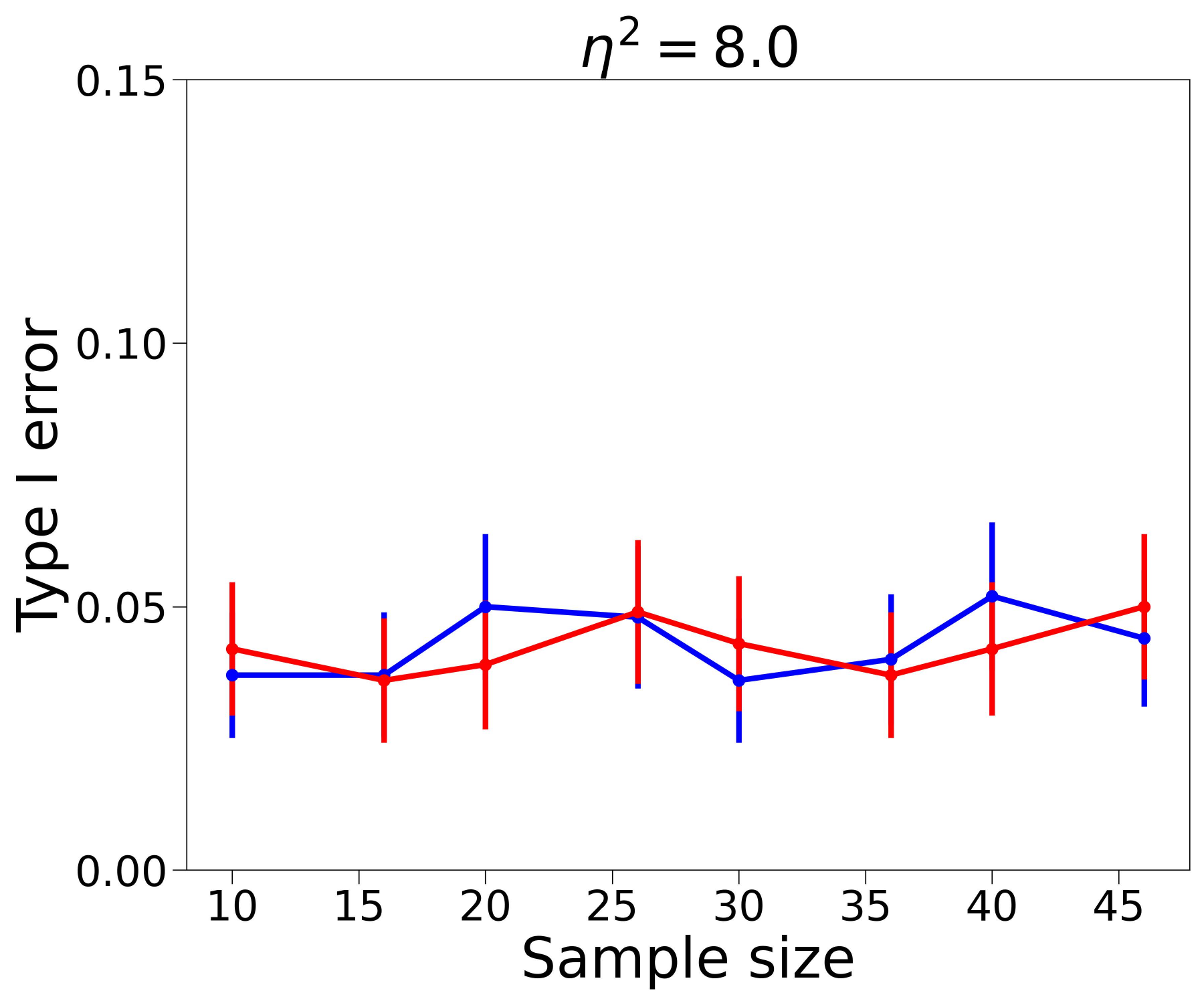}
    \end{subfigure}
    \caption{$\rho = 0$. $F$-test and $L$-test size for heteroskedastic errors with departures of $\eta^2$ from 1 indicating a greater degree of heteroskedasticity.}
    \label{fig:rob_3}
\end{figure}

\begin{figure}[ht]
    \centering
    \begin{subfigure}[t]{0.4\linewidth}
        \centering
        \includegraphics[width=\linewidth]{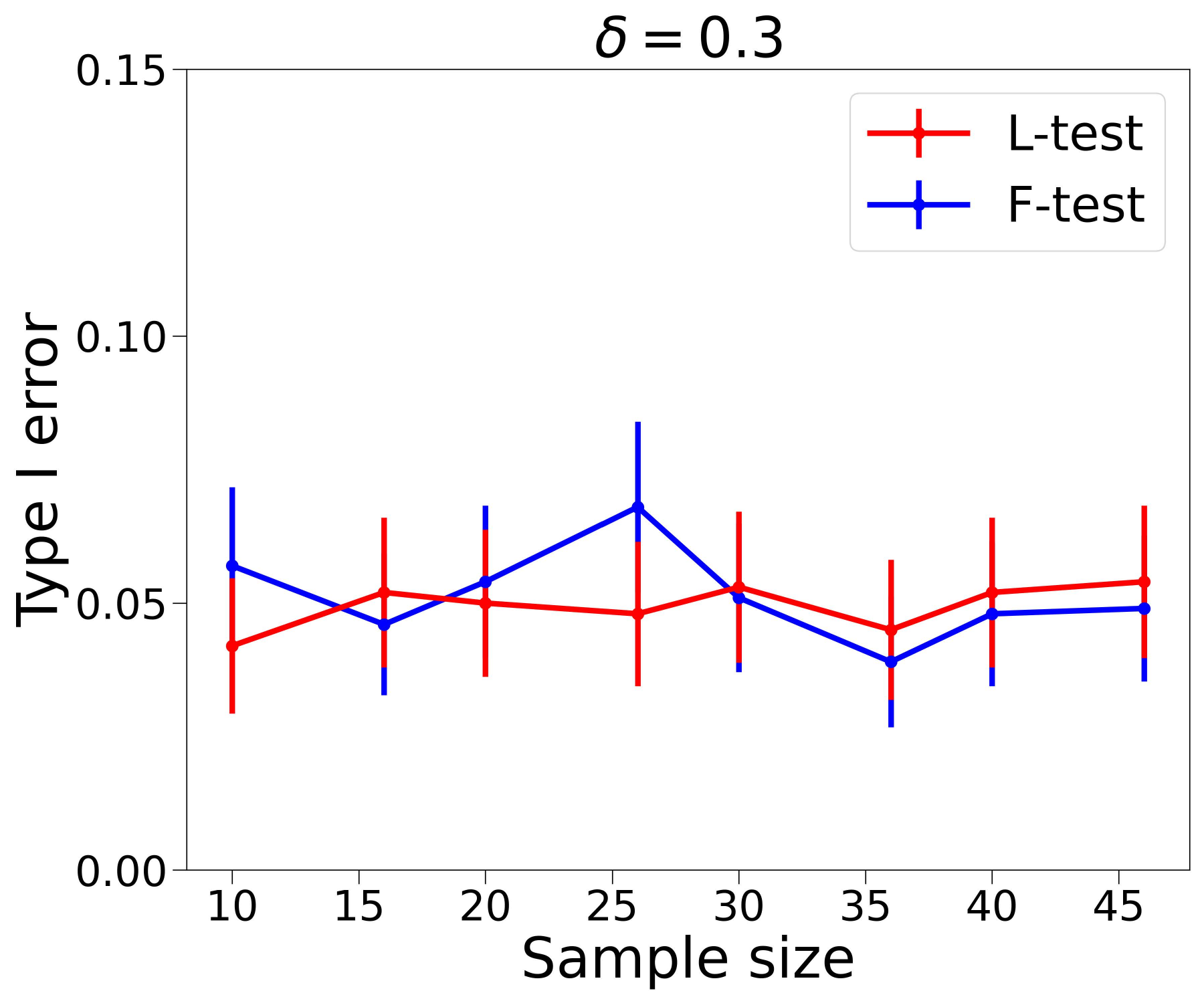}
    \end{subfigure}
    \hspace{0.2cm}
    \begin{subfigure}[t]{0.4\linewidth}
        \centering
        \includegraphics[width=\linewidth]{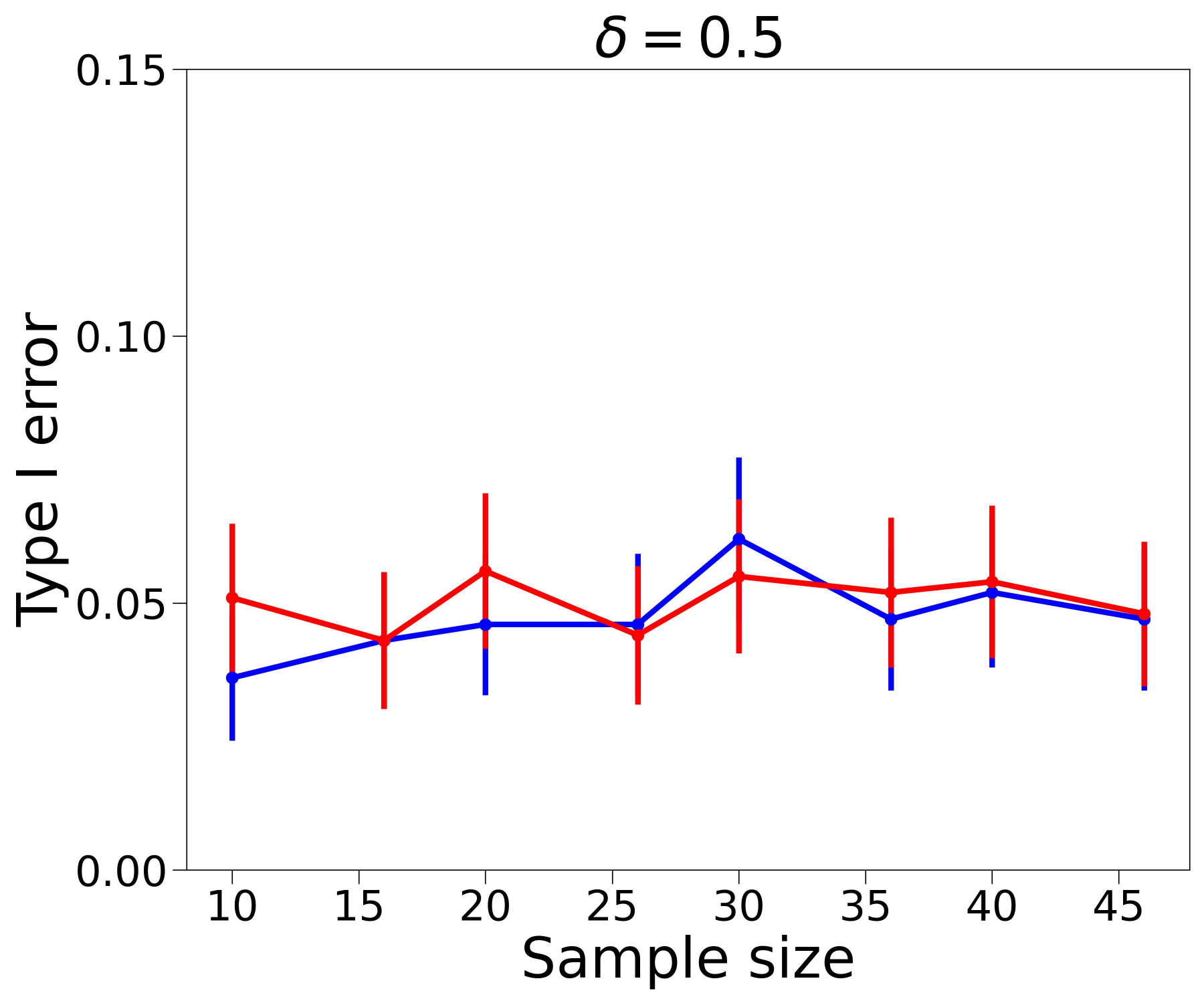}
    \end{subfigure} \\
    \begin{subfigure}[t]{0.4\linewidth}
        \centering
        \includegraphics[width=\linewidth]{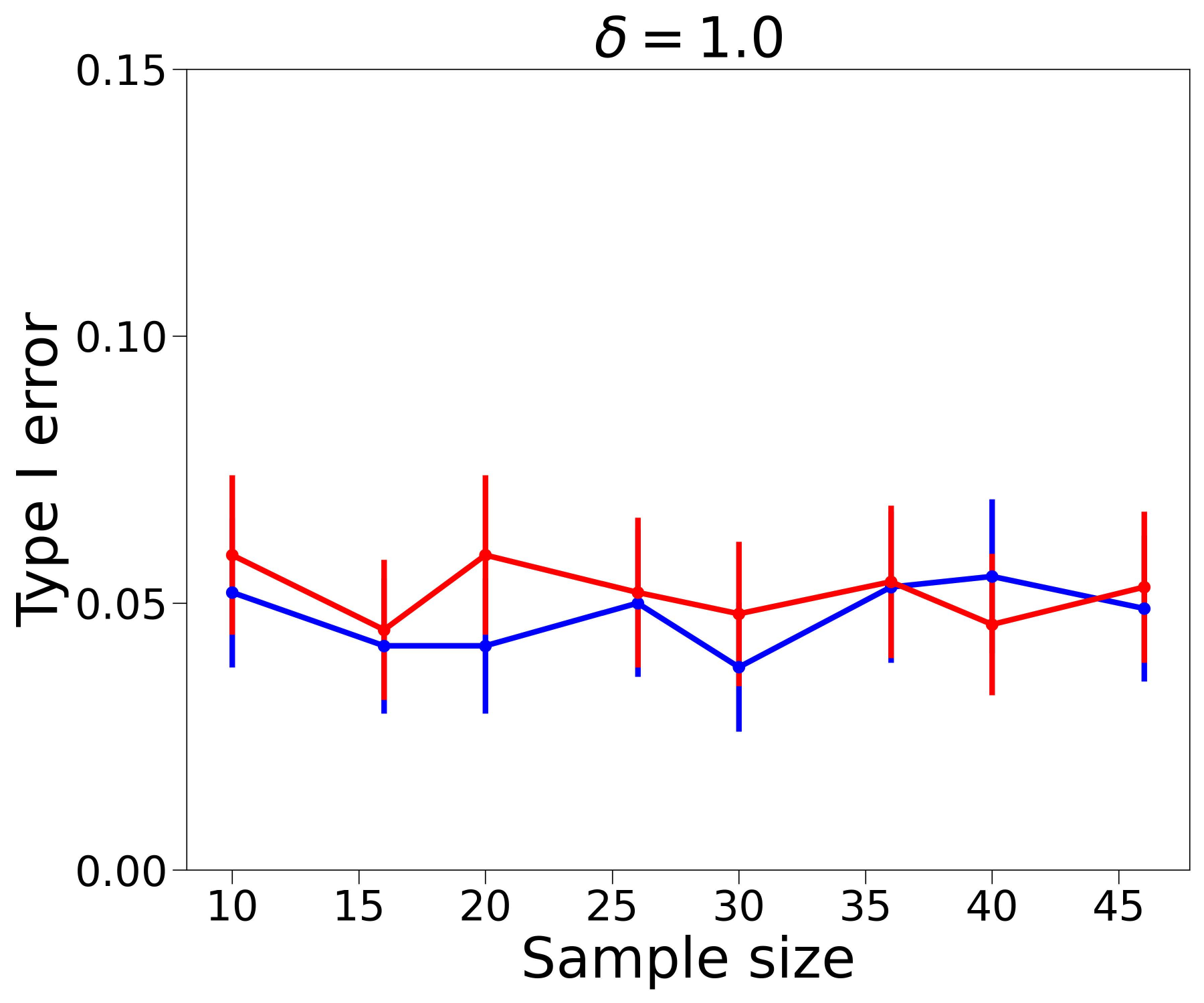}
    \end{subfigure}
    \hspace{0.2cm}
    \begin{subfigure}[t]{0.4\linewidth}
        \centering
        \includegraphics[width=\linewidth]{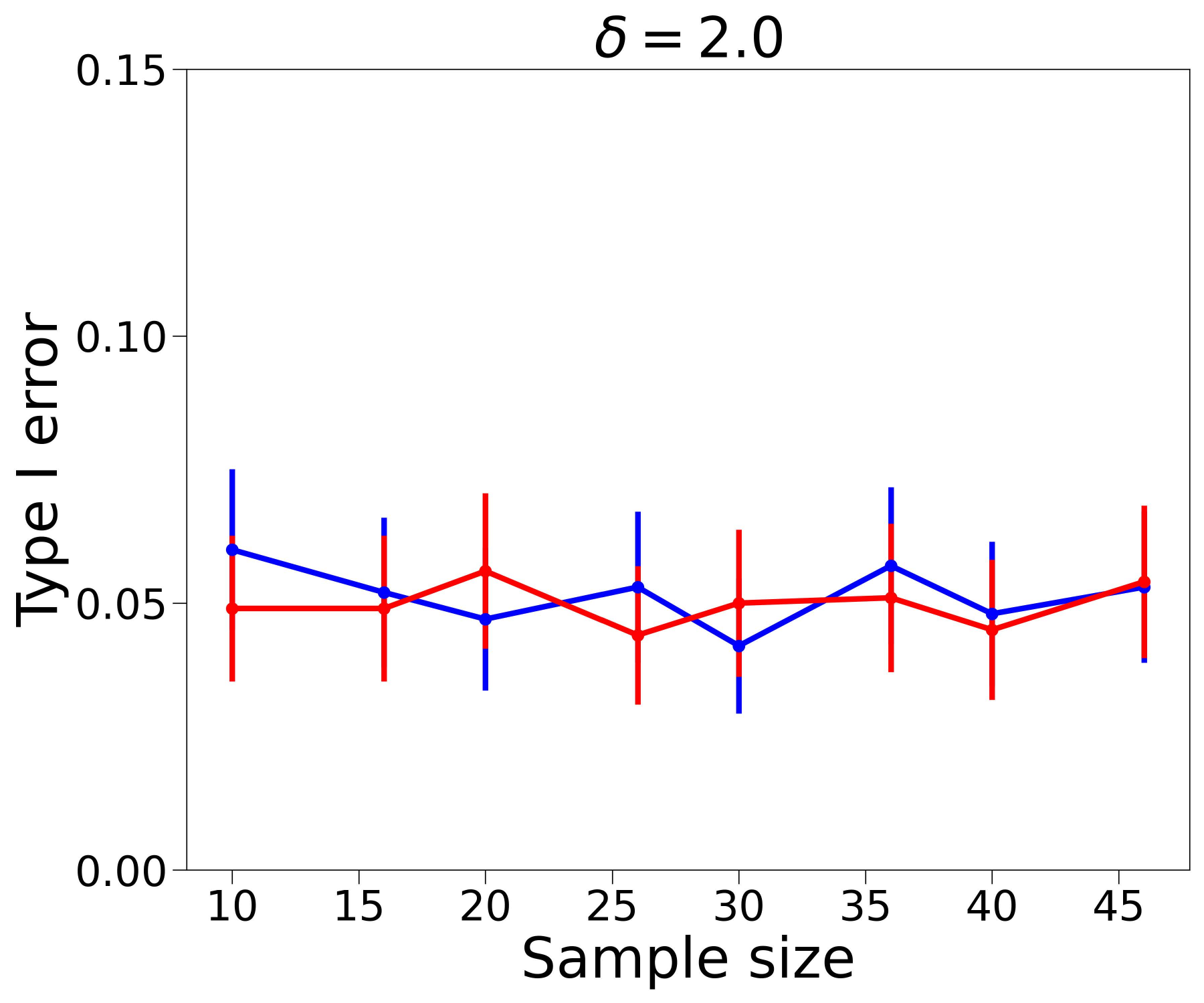}
    \end{subfigure} \\
    \begin{subfigure}[t]{0.4\linewidth}
        \centering
        \includegraphics[width=\linewidth]{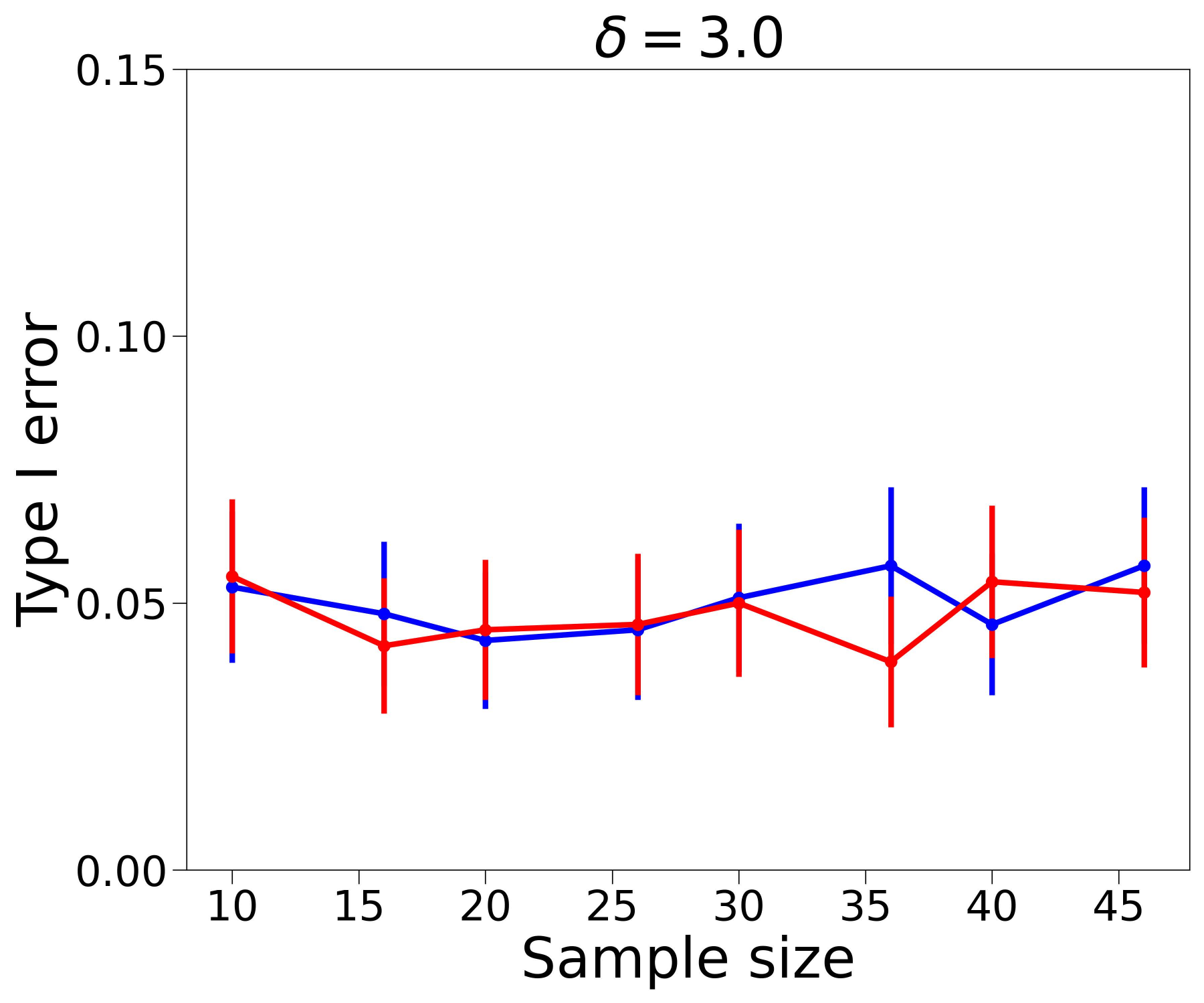}
    \end{subfigure}
    \hspace{0.2cm}
    \begin{subfigure}[t]{0.4\linewidth}
        \centering
        \includegraphics[width=\linewidth]{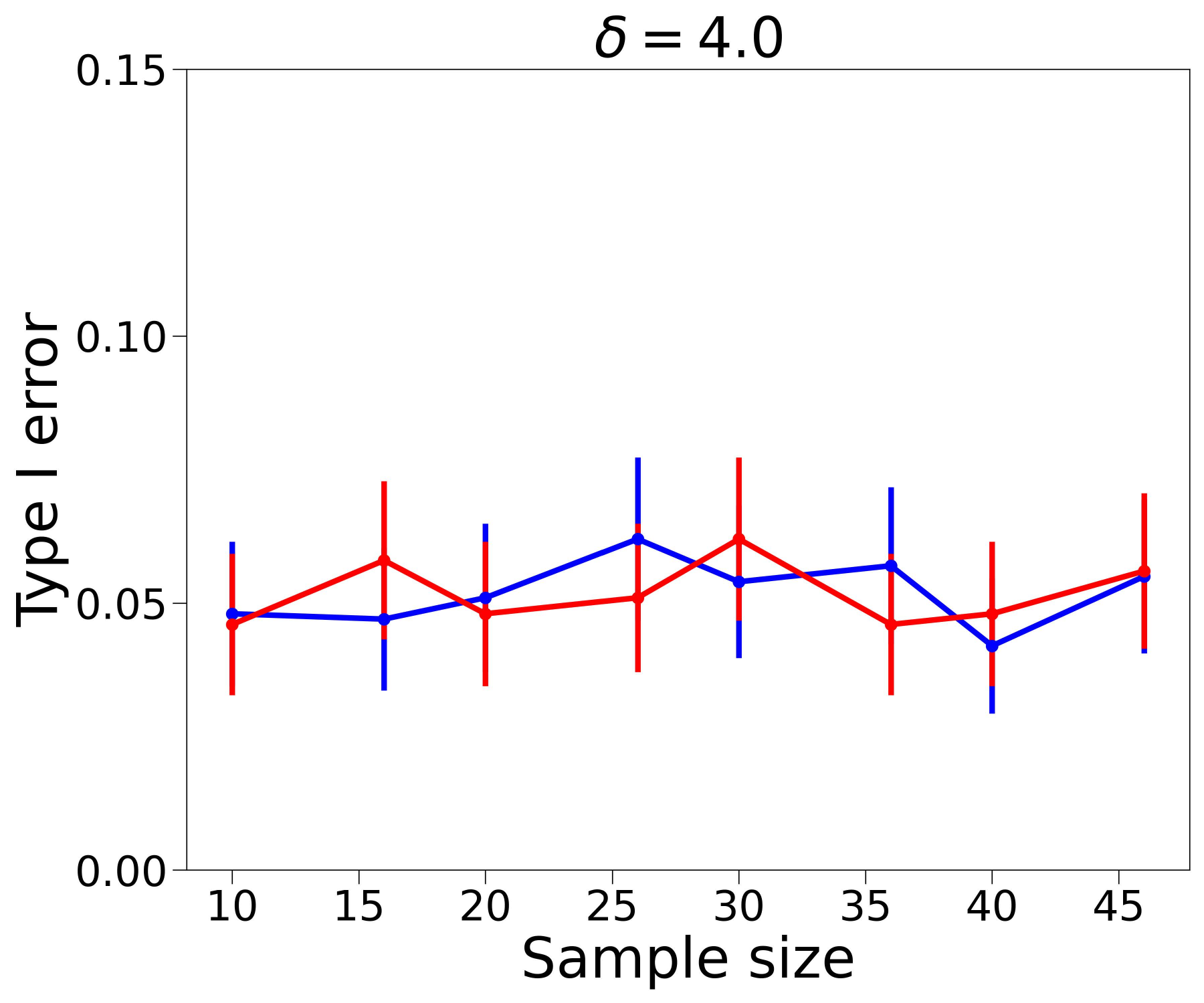}
    \end{subfigure}
    \caption{$\rho = 0$. $F$-test and $L$-test size under non-linearity, where departures of $\delta$ from 1 indicate a greater degree of non-linearity.}
    \label{fig:rob_4}
\end{figure}

\begin{figure}[ht]
    \centering
    \begin{subfigure}[t]{0.4\linewidth}
        \centering
        \includegraphics[width=\linewidth]{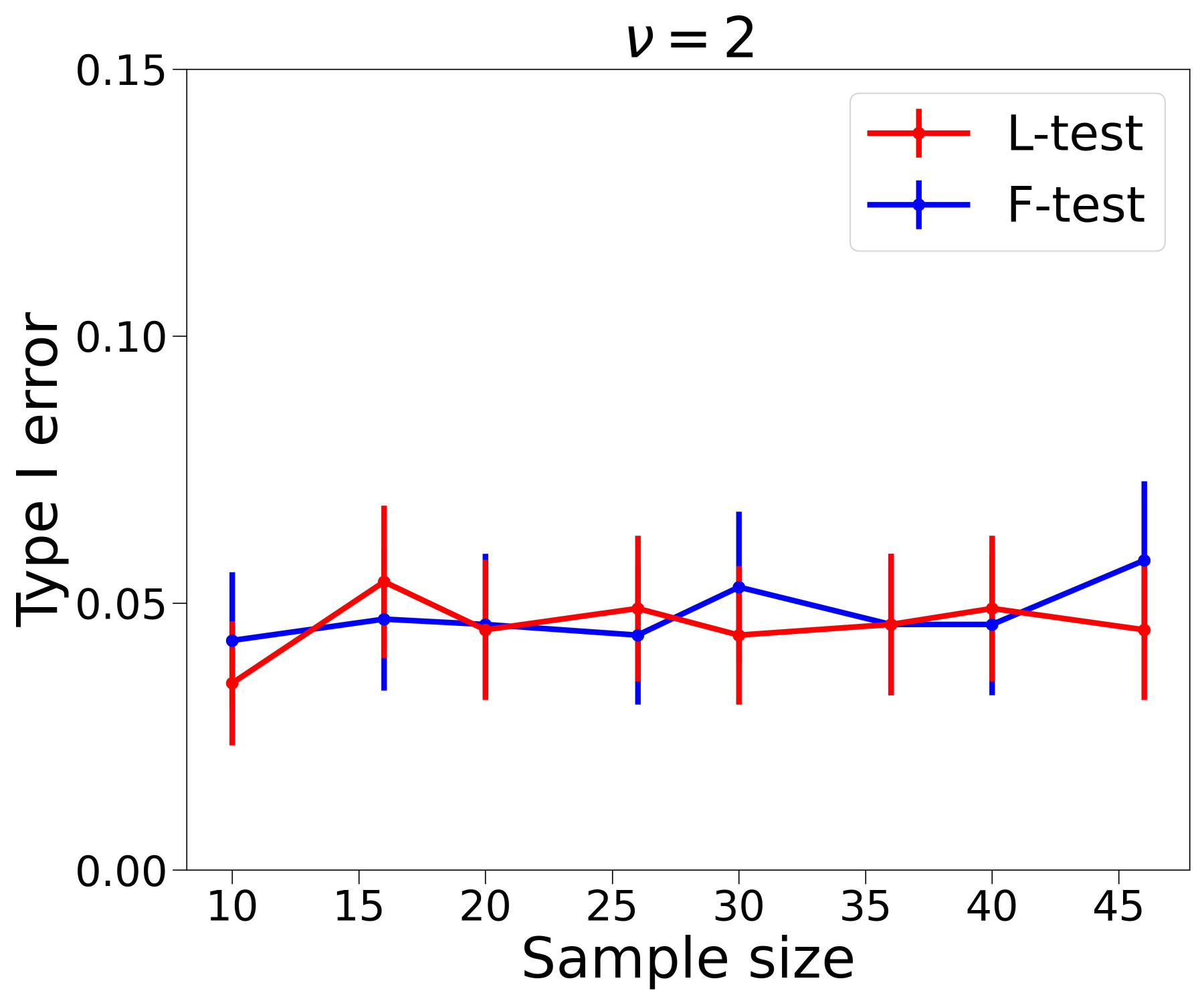}
    \end{subfigure}
    \hspace{0.2cm}
    \begin{subfigure}[t]{0.4\linewidth}
        \centering
        \includegraphics[width=\linewidth]{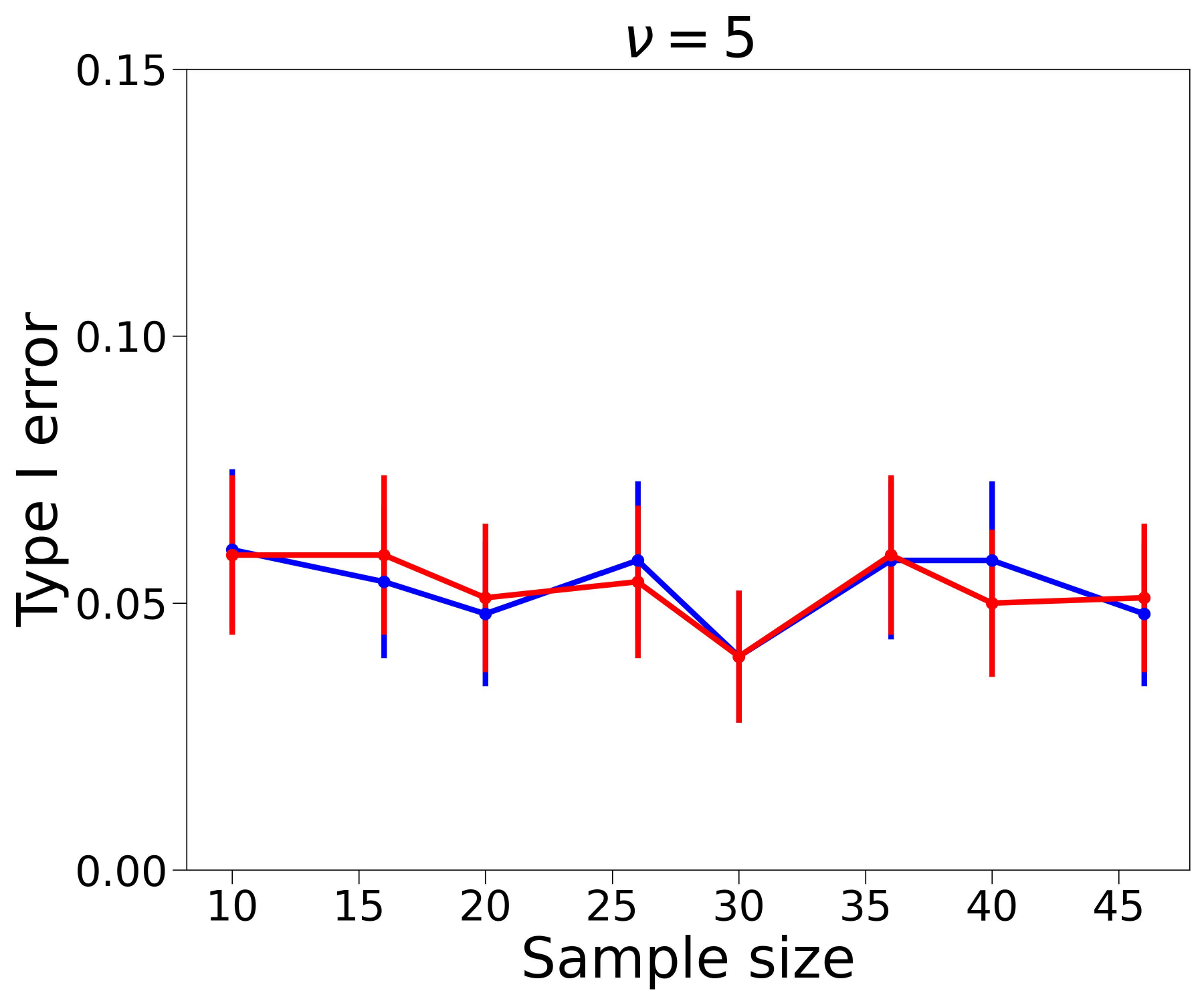}
    \end{subfigure} \\
    \begin{subfigure}[t]{0.4\linewidth}
        \centering
        \includegraphics[width=\linewidth]{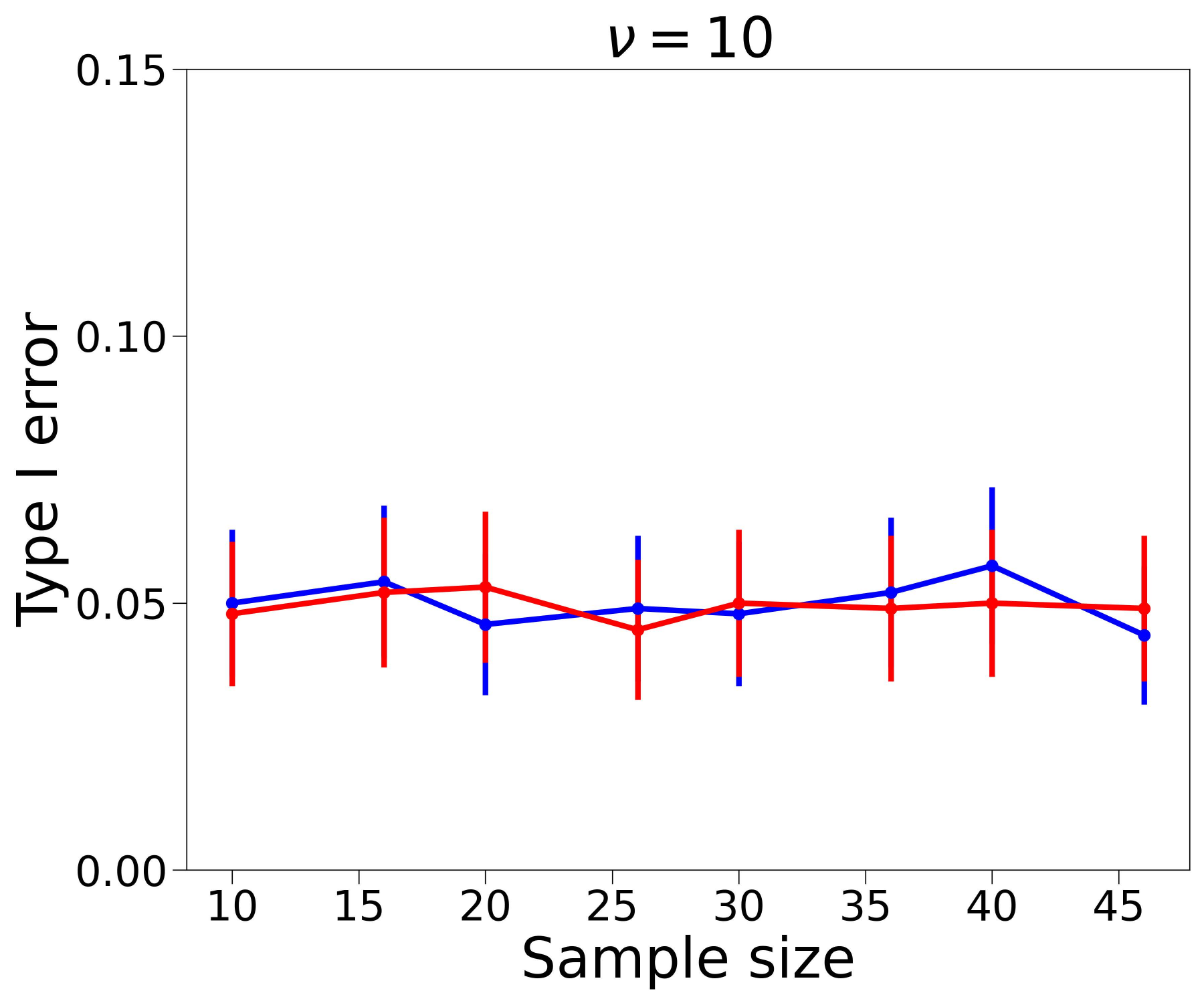}
    \end{subfigure}
    \hspace{0.2cm}
    \begin{subfigure}[t]{0.4\linewidth}
        \centering
        \includegraphics[width=\linewidth]{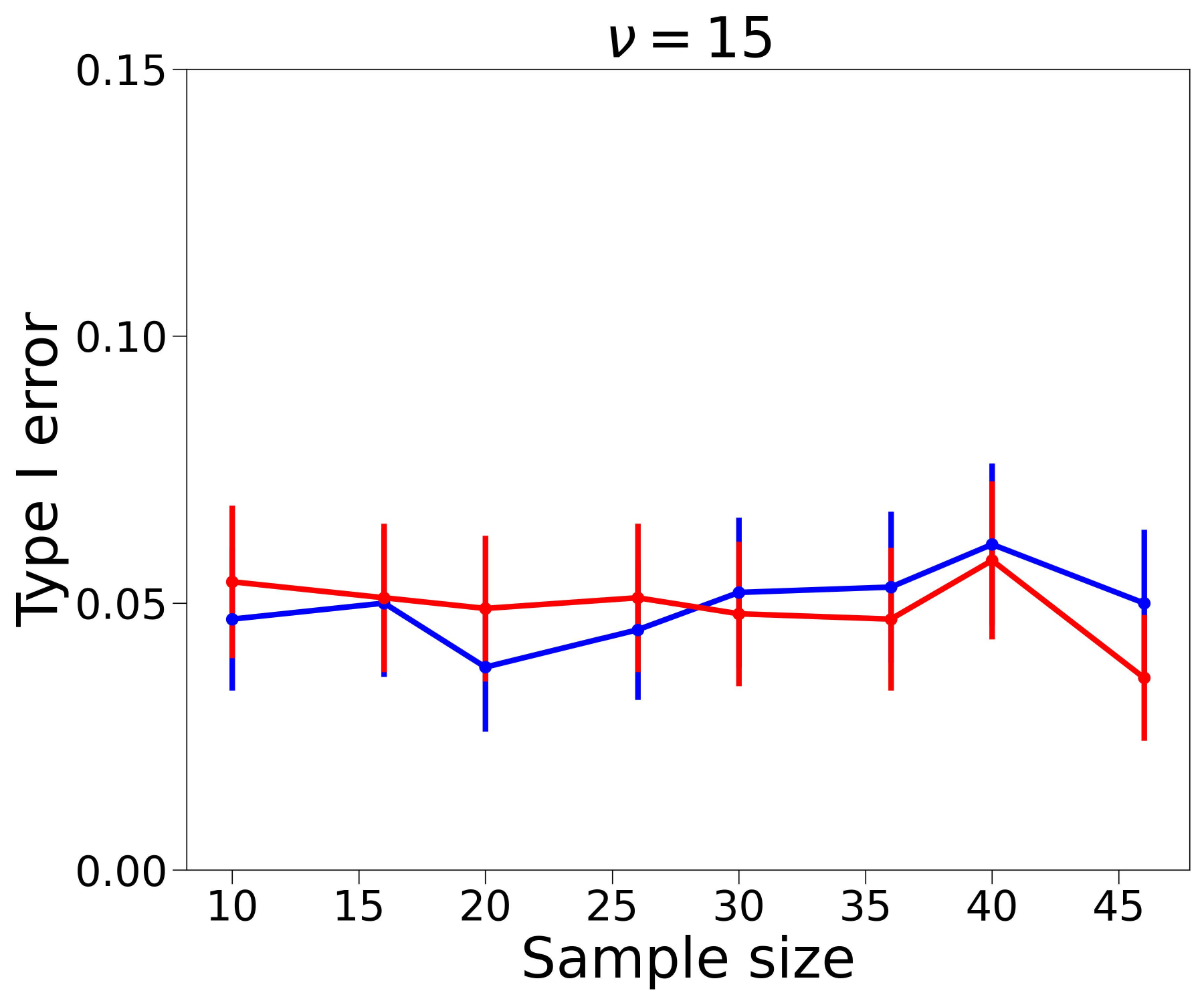}
    \end{subfigure} \\
    \begin{subfigure}[t]{0.4\linewidth}
        \centering
        \includegraphics[width=\linewidth]{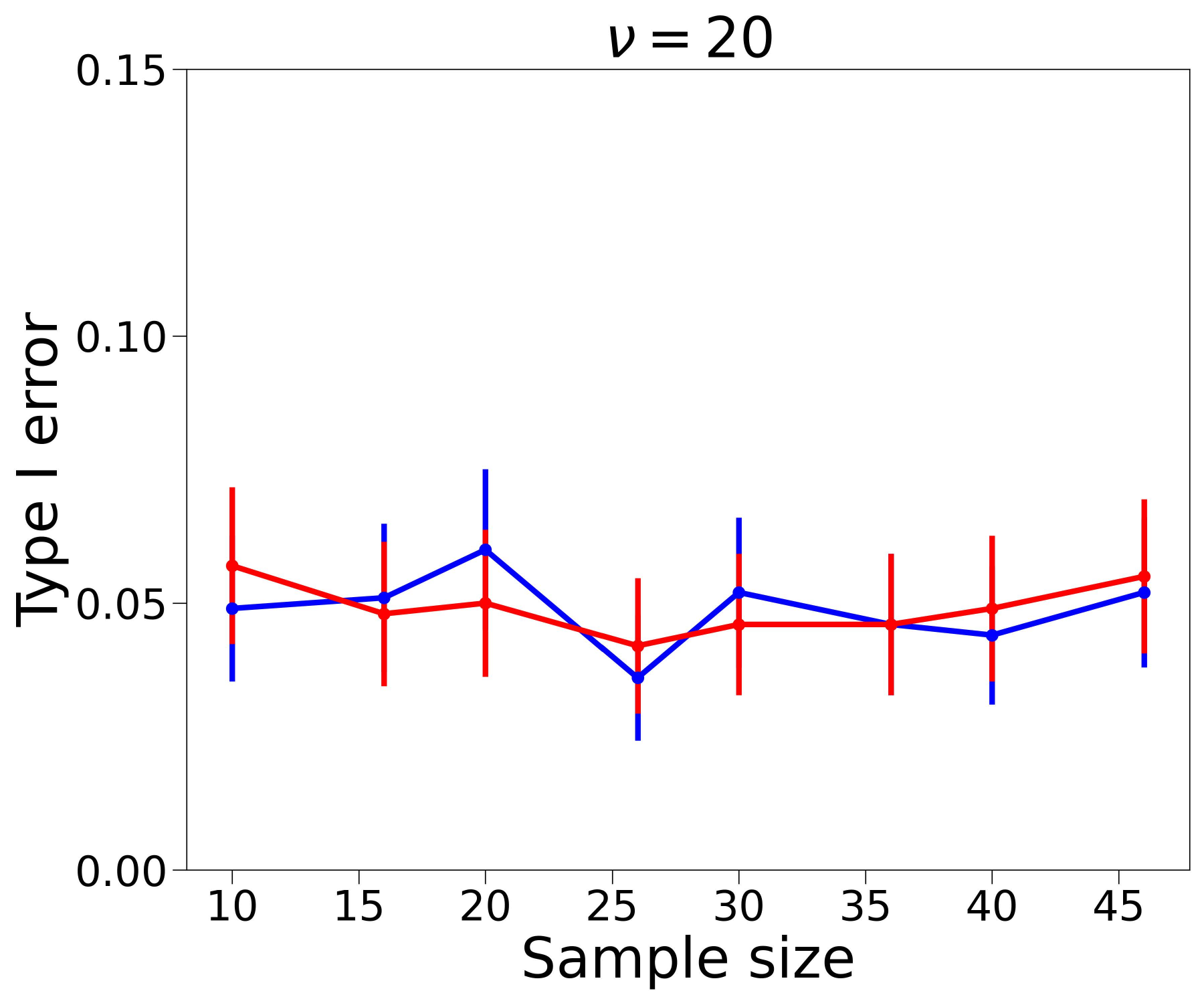}
    \end{subfigure}
    \hspace{0.2cm}
    \begin{subfigure}[t]{0.4\linewidth}
        \centering
        \includegraphics[width=\linewidth]{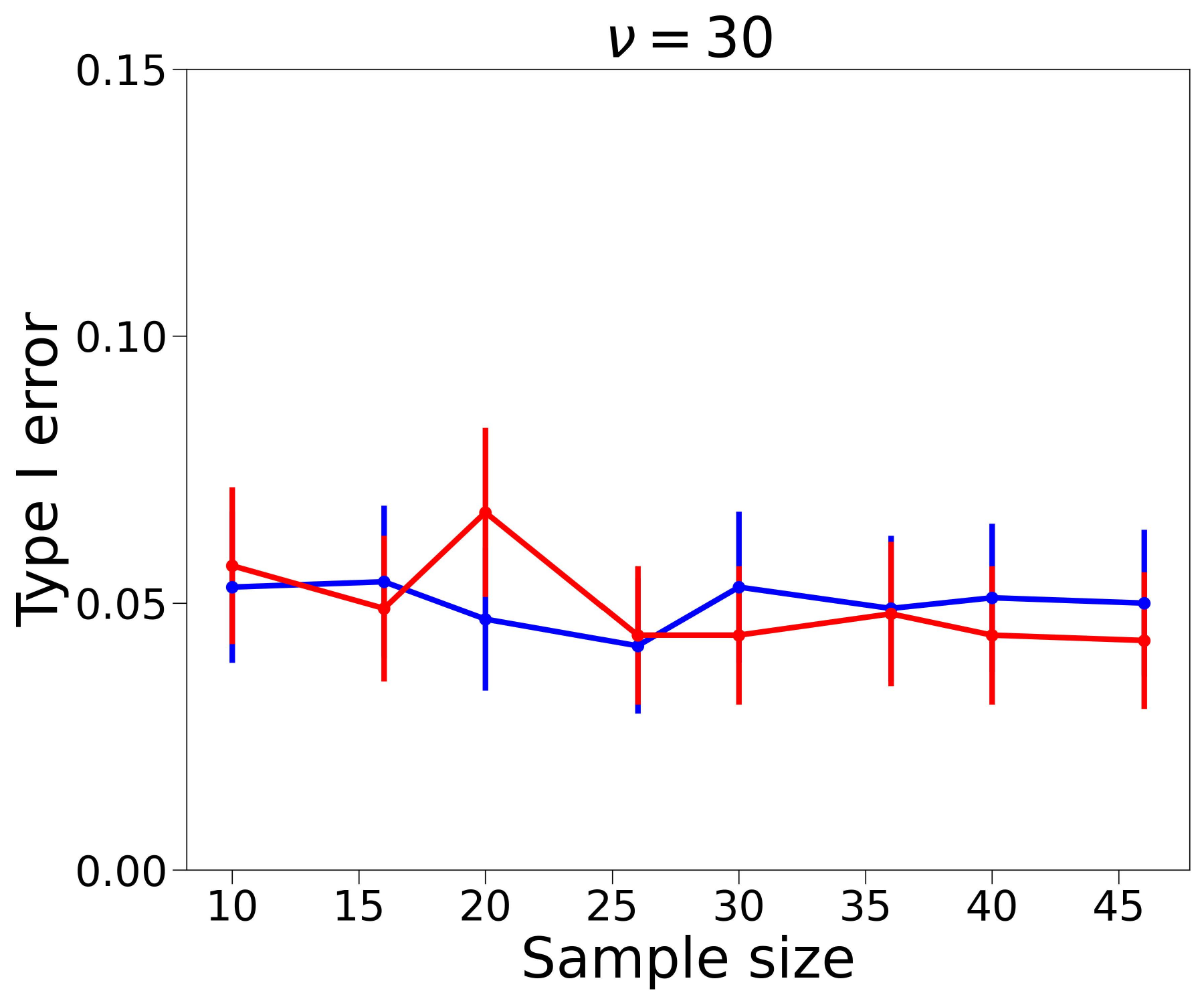}
    \end{subfigure}
    \caption{Same as Figure \ref{fig:rob_1} but with $\rho = 0.5$.}
    \label{fig:rob_5}
\end{figure}

\begin{figure}[ht]
    \centering
    \begin{subfigure}[t]{0.4\linewidth}
        \centering
        \includegraphics[width=\linewidth]{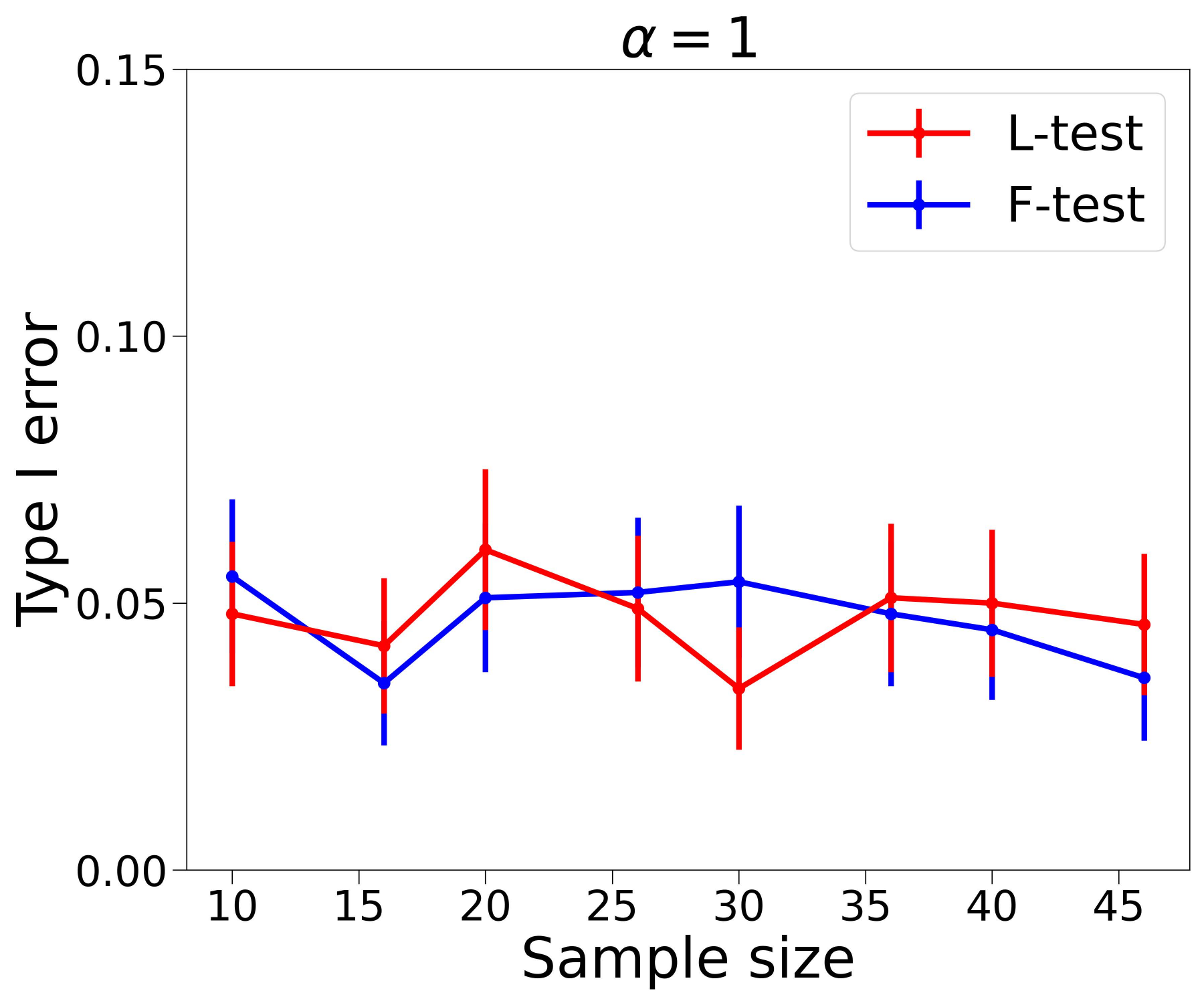}
    \end{subfigure}
    \hspace{0.2cm}
    \begin{subfigure}[t]{0.4\linewidth}
        \centering
        \includegraphics[width=\linewidth]{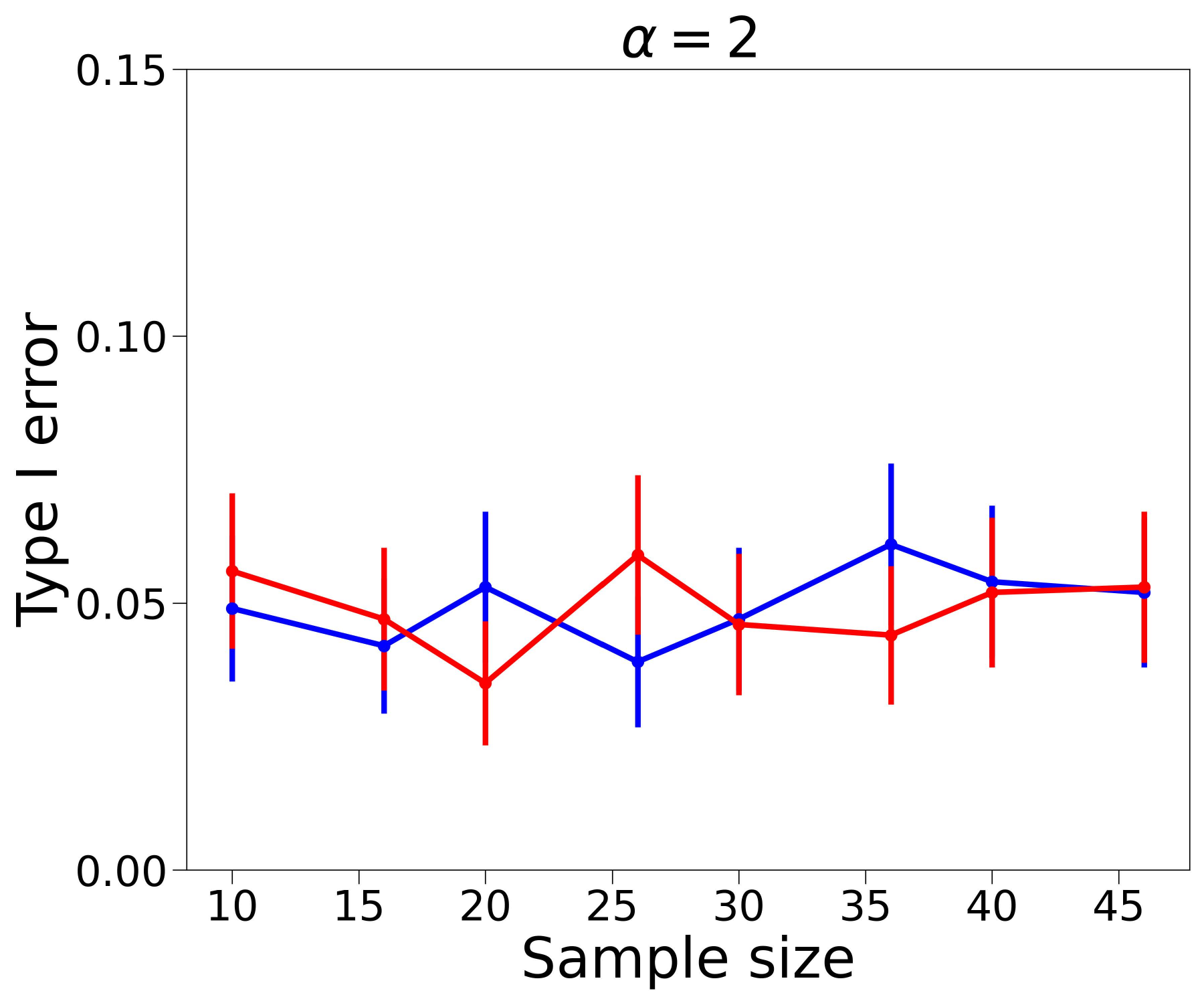}
    \end{subfigure} \\
    \begin{subfigure}[t]{0.4\linewidth}
        \centering
        \includegraphics[width=\linewidth]{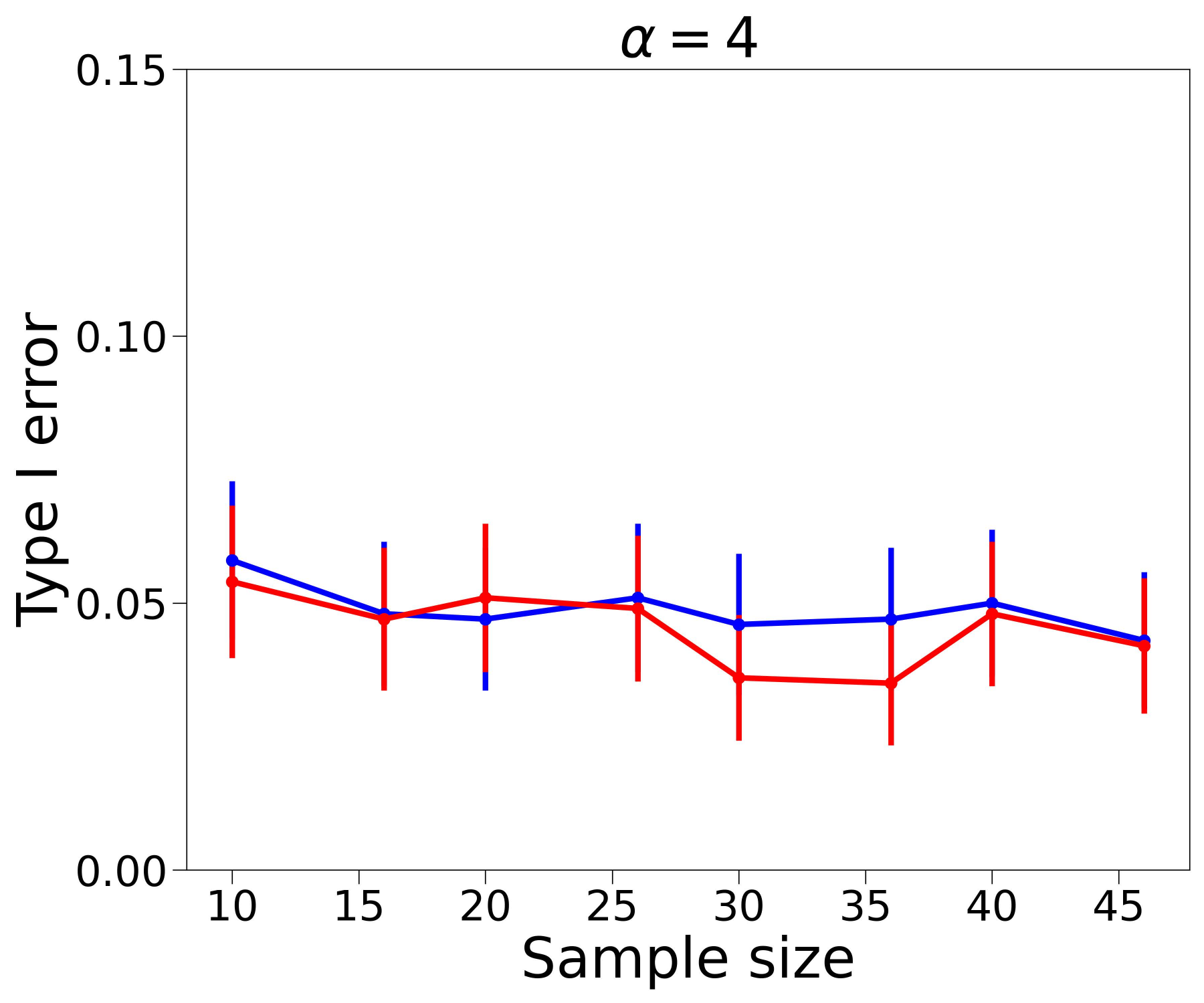}
    \end{subfigure}
    \hspace{0.2cm}
    \begin{subfigure}[t]{0.4\linewidth}
        \centering
        \includegraphics[width=\linewidth]{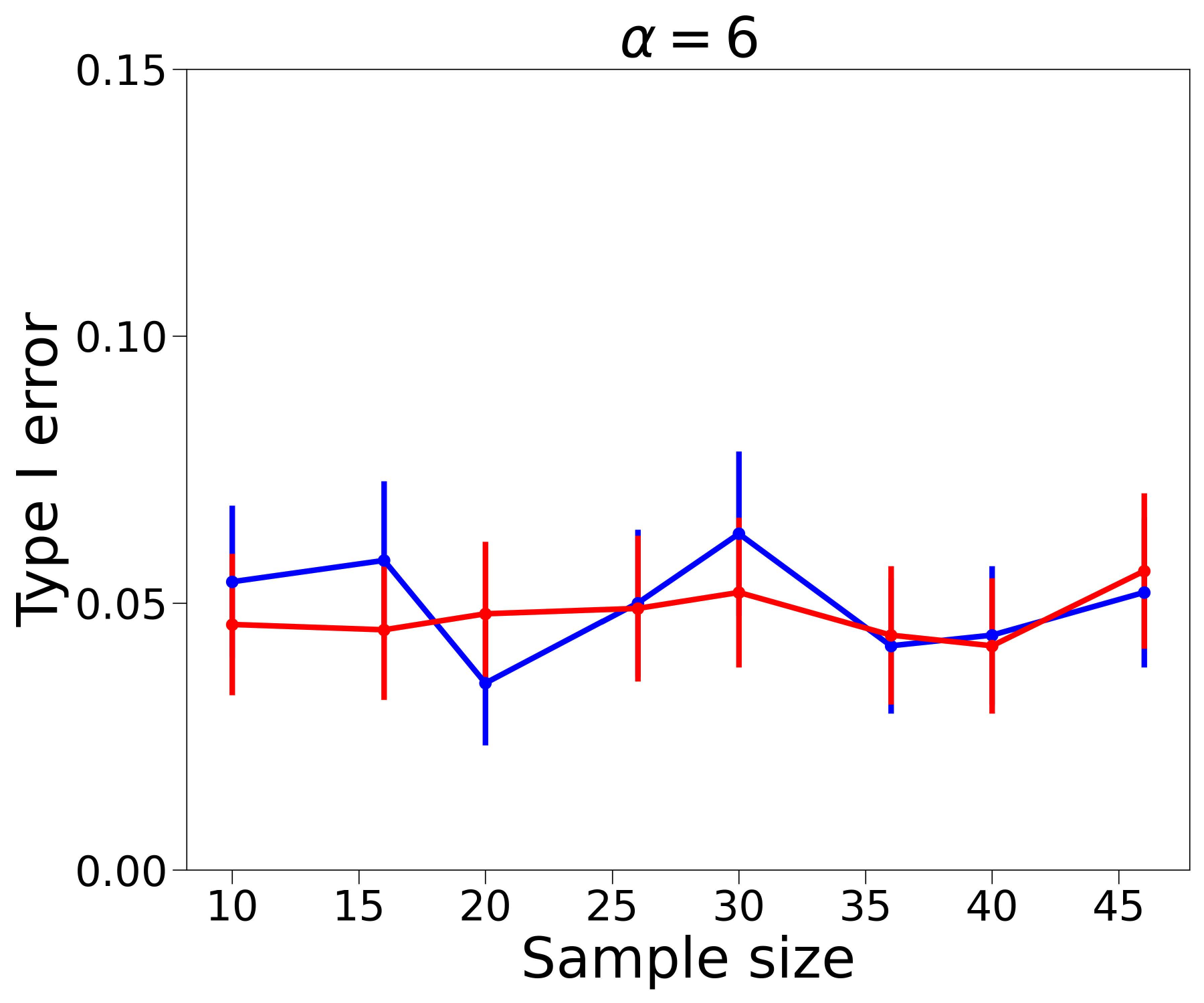}
    \end{subfigure} \\
    \begin{subfigure}[t]{0.4\linewidth}
        \centering
        \includegraphics[width=\linewidth]{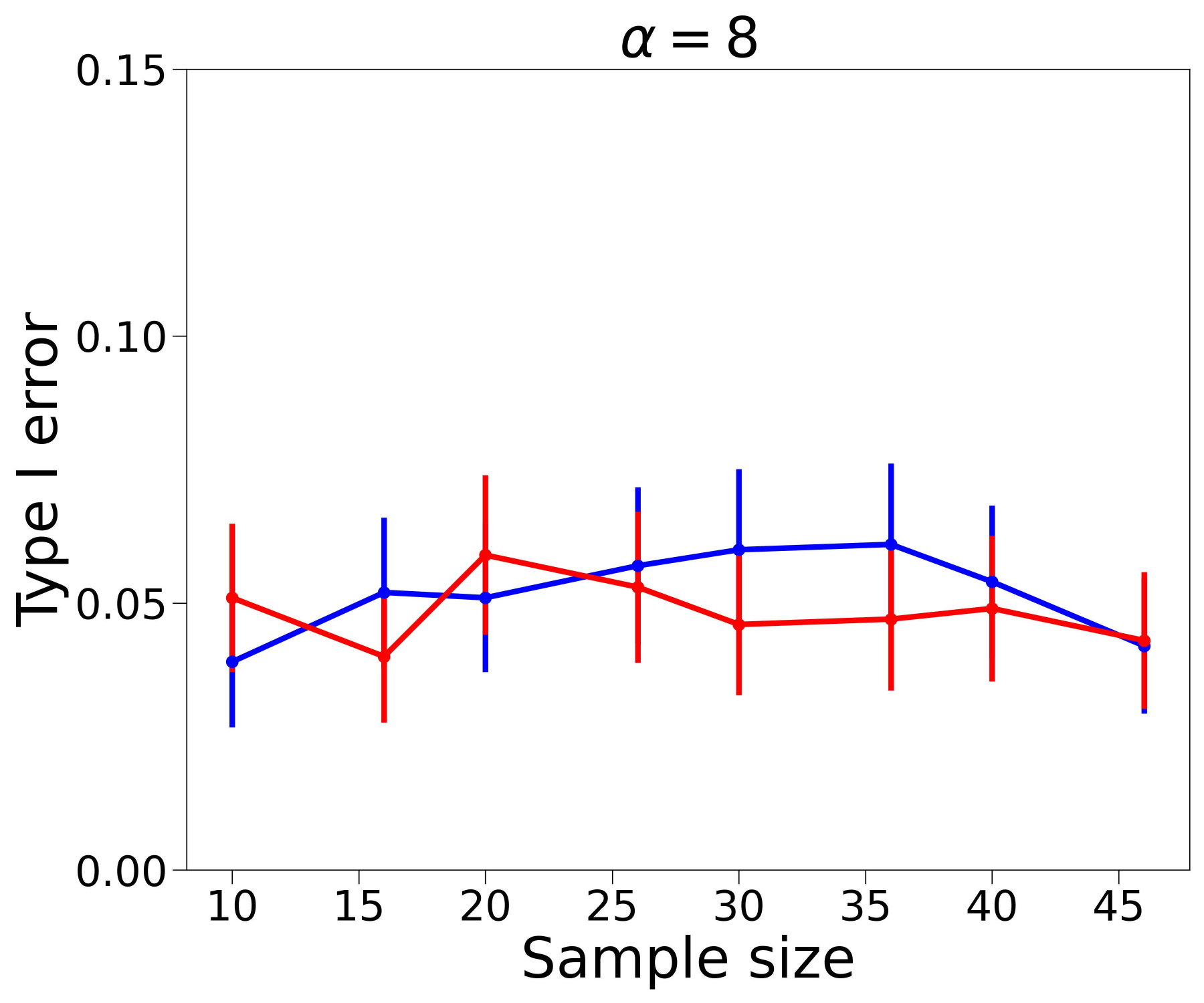}
    \end{subfigure}
    \hspace{0.2cm}
    \begin{subfigure}[t]{0.4\linewidth}
        \centering
        \includegraphics[width=\linewidth]{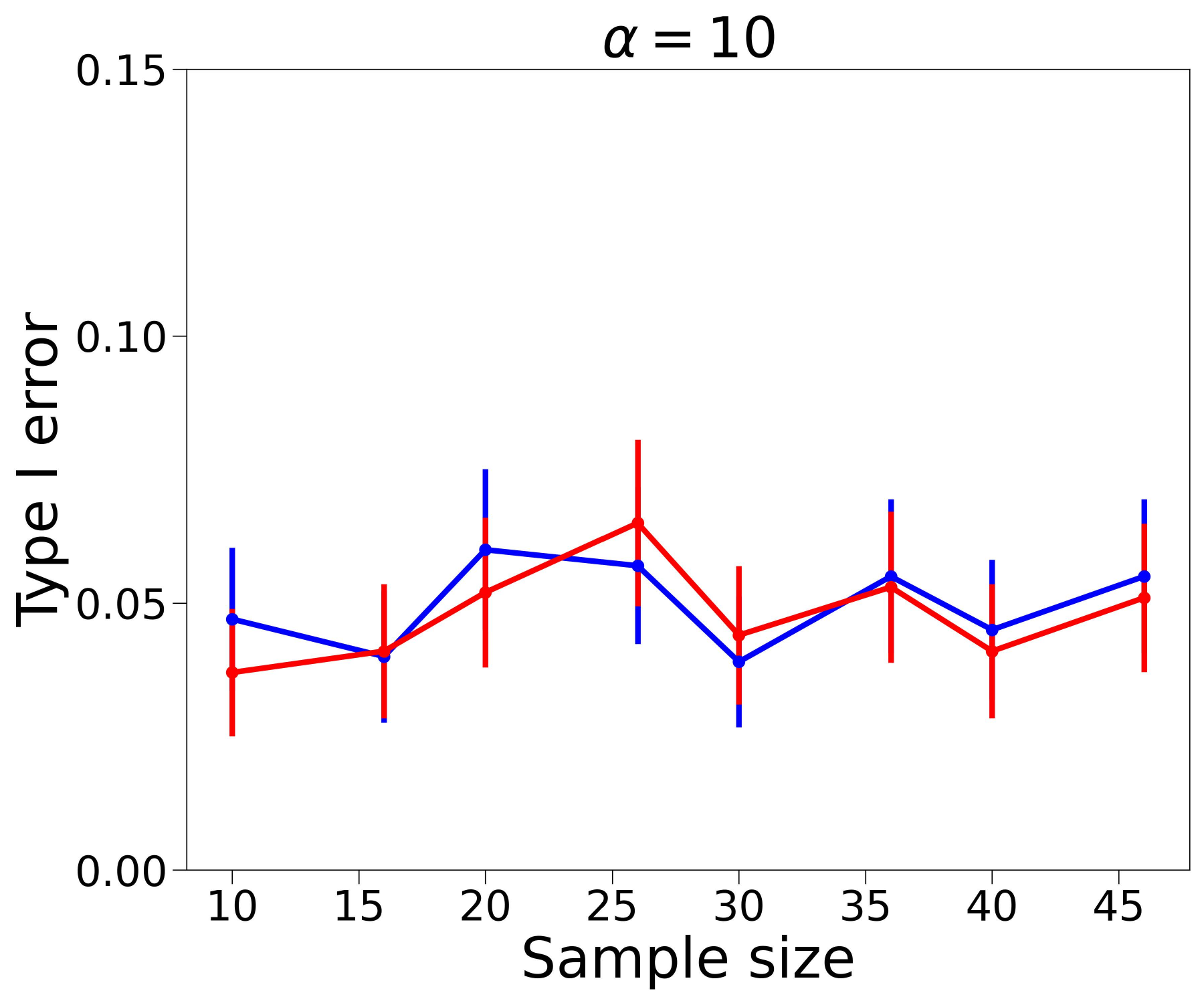}
    \end{subfigure}
    \caption{Same as Figure \ref{fig:rob_2} but with $\rho = 0.5$.}
    \label{fig:rob_6}
\end{figure}

\begin{figure}[ht]
    \centering
    \begin{subfigure}[t]{0.4\linewidth}
        \centering
        \includegraphics[width=\linewidth]{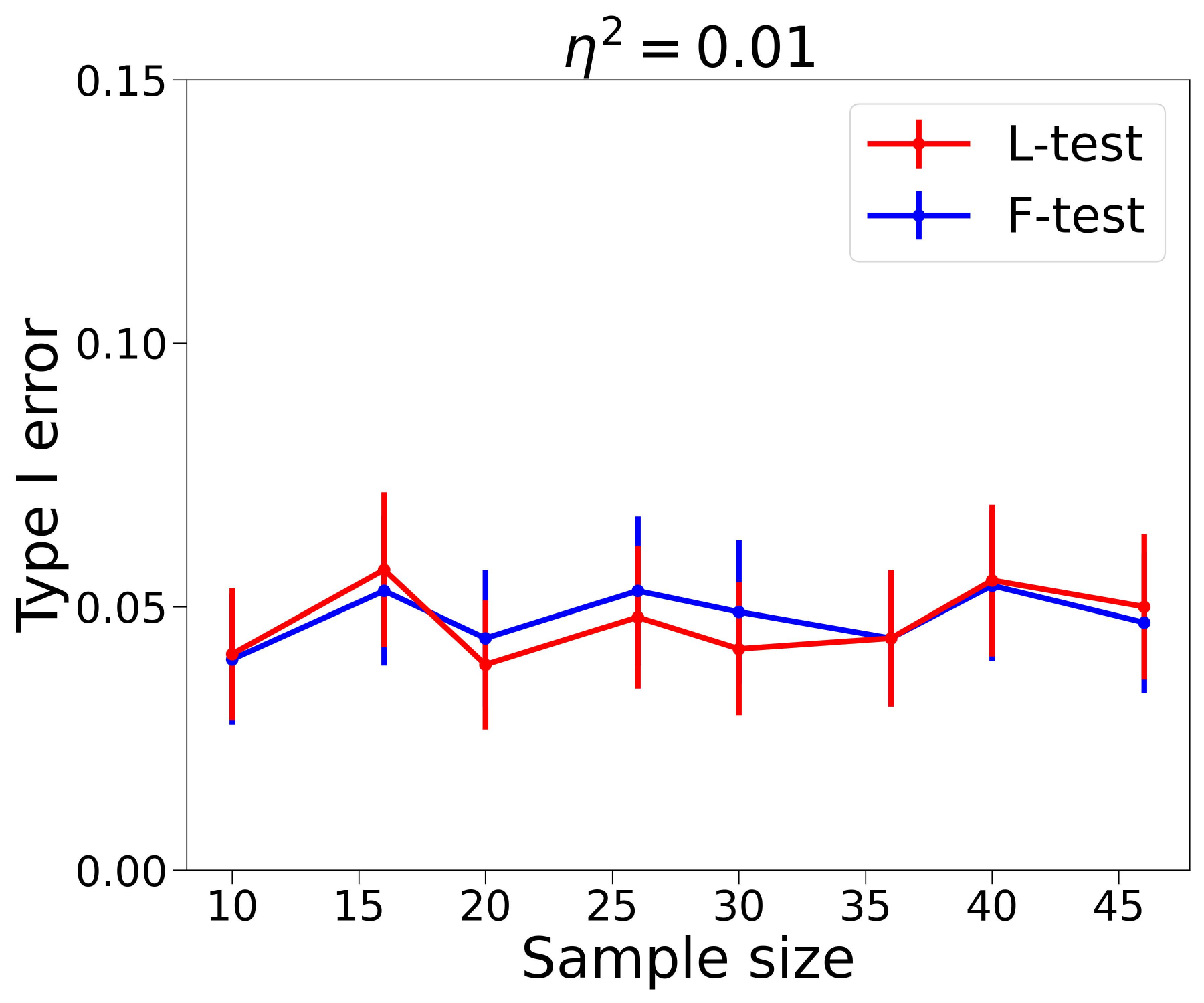}
    \end{subfigure}
    \hspace{0.2cm}
    \begin{subfigure}[t]{0.4\linewidth}
        \centering
        \includegraphics[width=\linewidth]{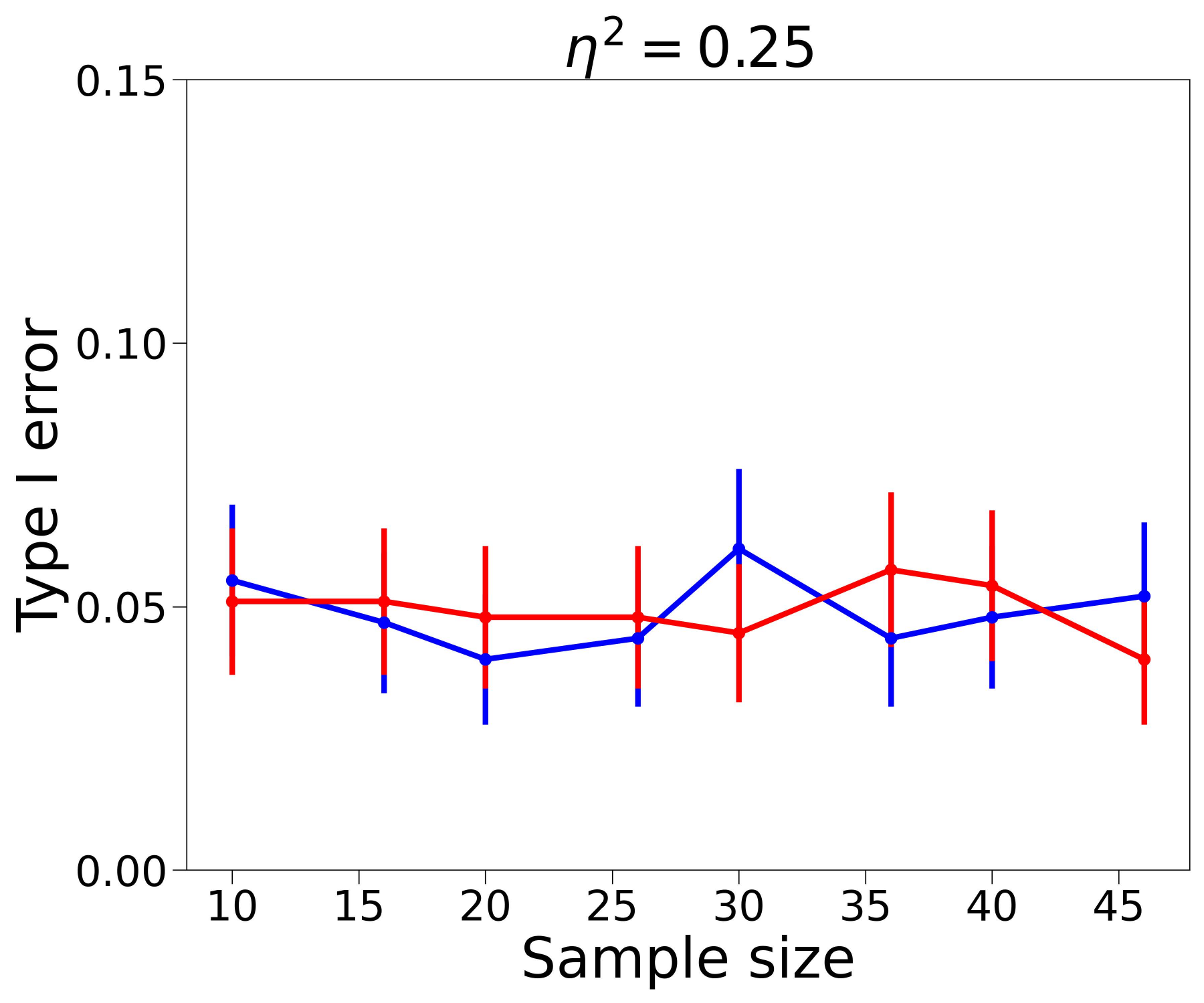}
    \end{subfigure} \\
    \begin{subfigure}[t]{0.4\linewidth}
        \centering
        \includegraphics[width=\linewidth]{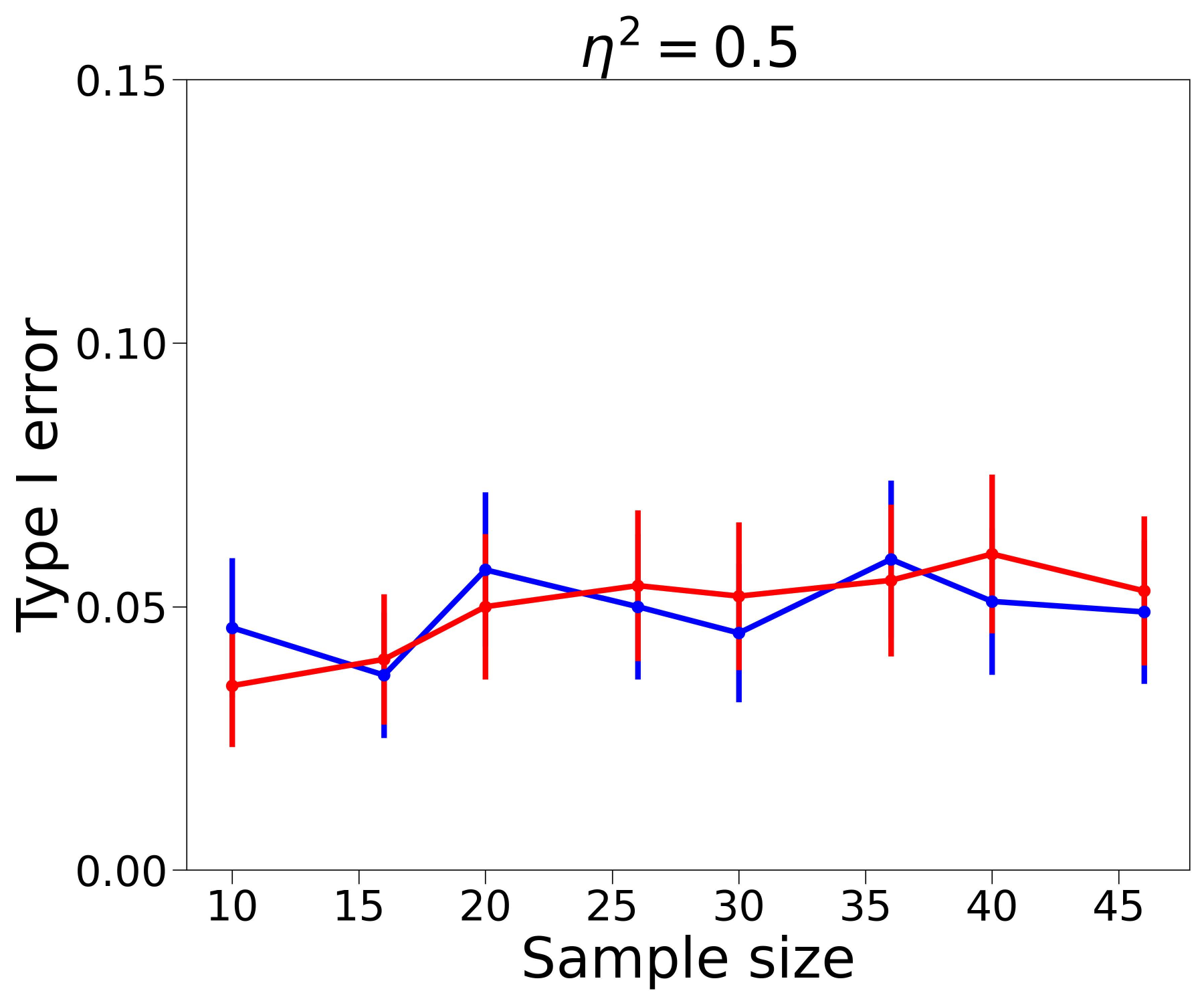}
    \end{subfigure}
    \hspace{0.2cm}
    \begin{subfigure}[t]{0.4\linewidth}
        \centering
        \includegraphics[width=\linewidth]{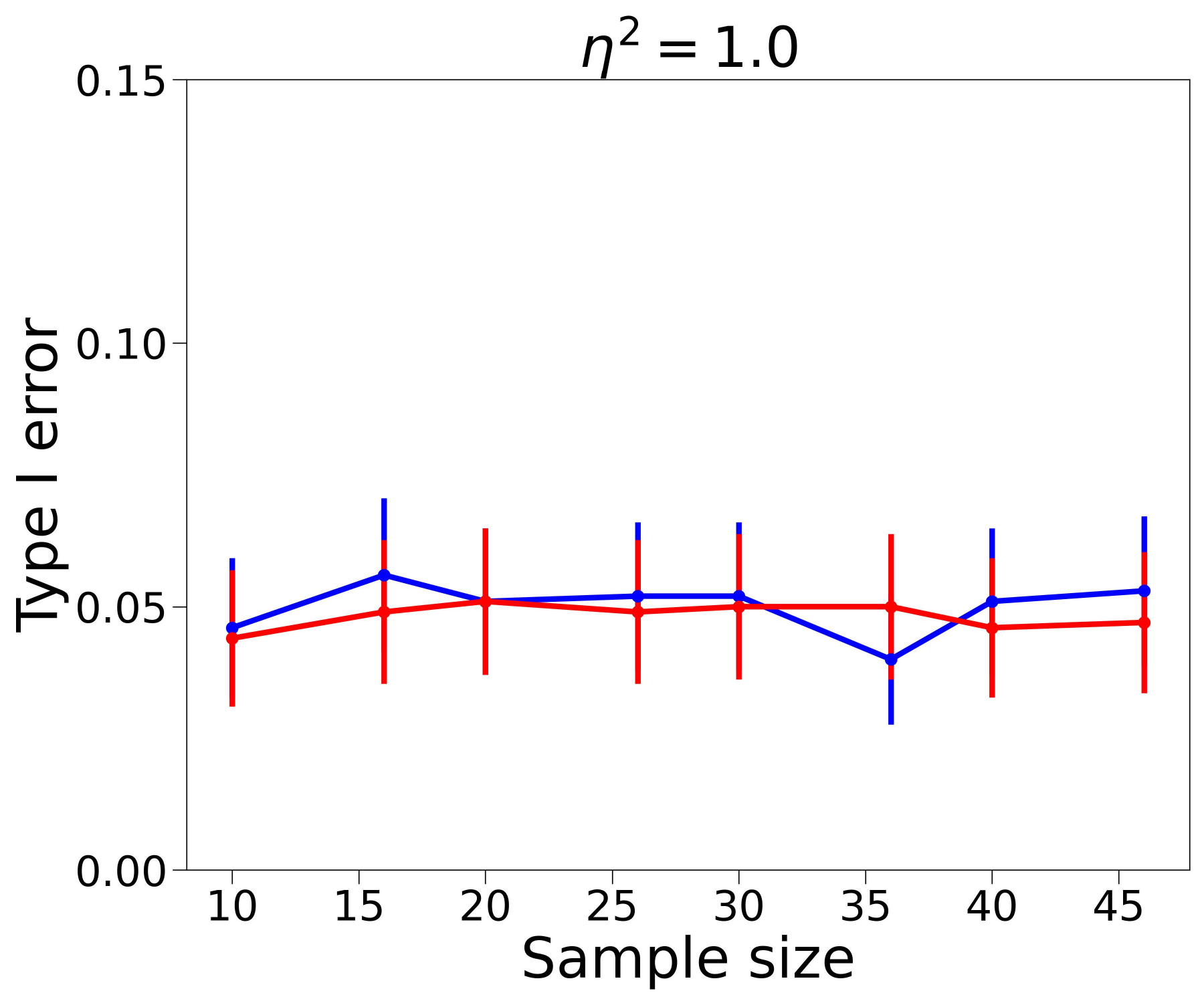}
    \end{subfigure} \\
    \begin{subfigure}[t]{0.4\linewidth}
        \centering
        \includegraphics[width=\linewidth]{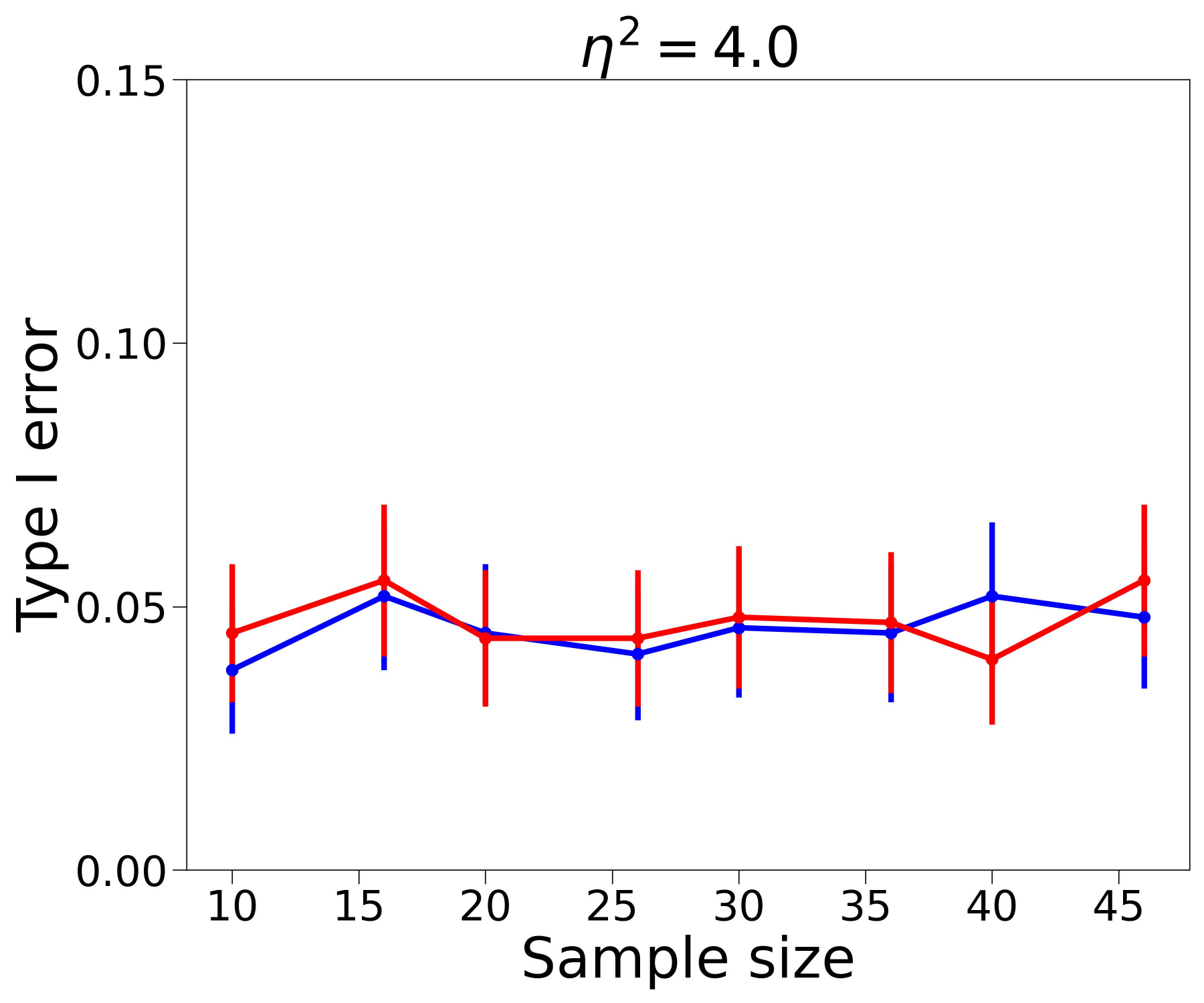}
    \end{subfigure}
    \hspace{0.2cm}
    \begin{subfigure}[t]{0.4\linewidth}
        \centering
        \includegraphics[width=\linewidth]{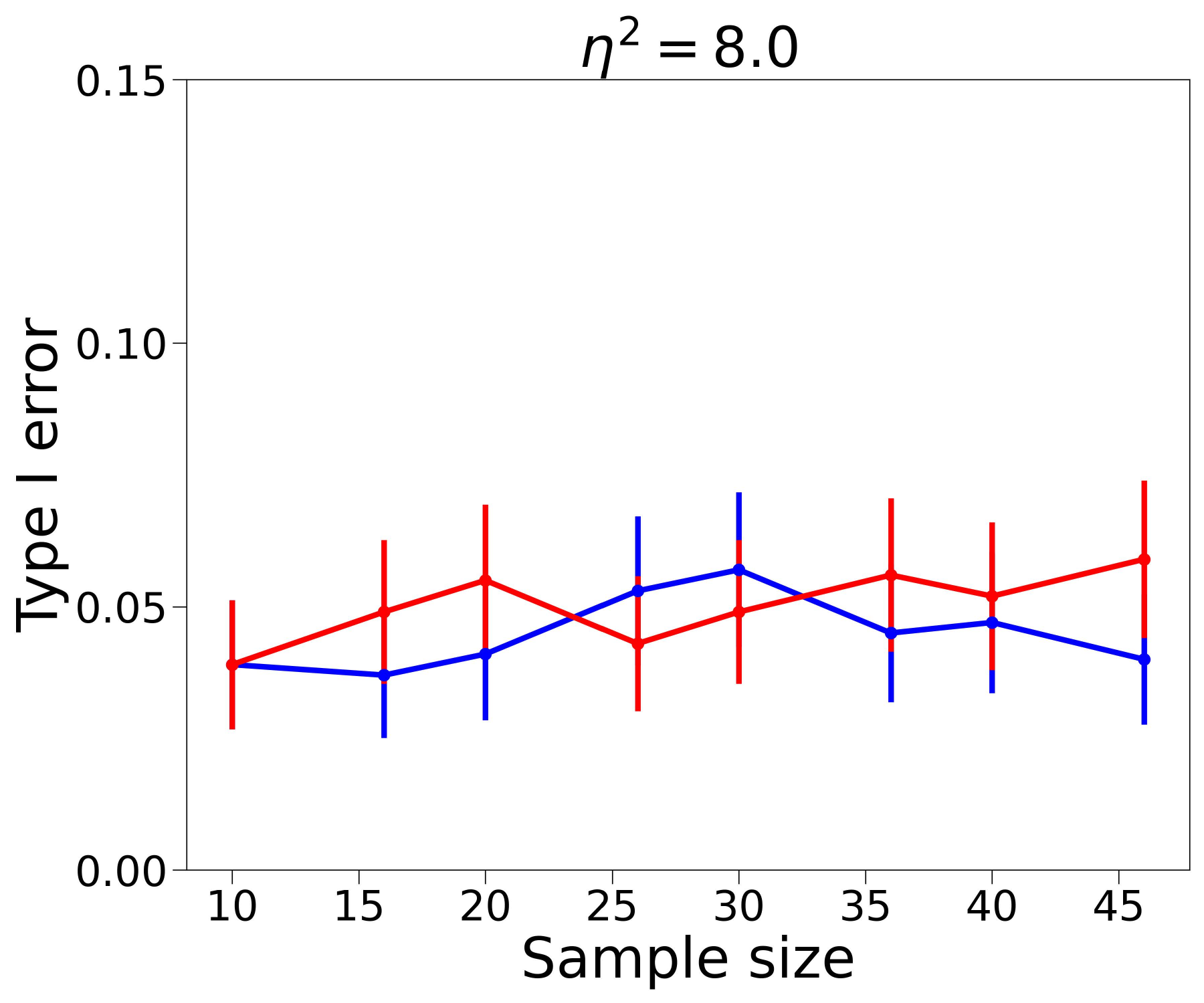}
    \end{subfigure}
    \caption{Same as Figure \ref{fig:rob_3} but with $\rho = 0.5$.}
    \label{fig:rob_7}
\end{figure}

\begin{figure}[ht]
    \centering
    \begin{subfigure}[t]{0.4\linewidth}
        \centering
        \includegraphics[width=\linewidth]{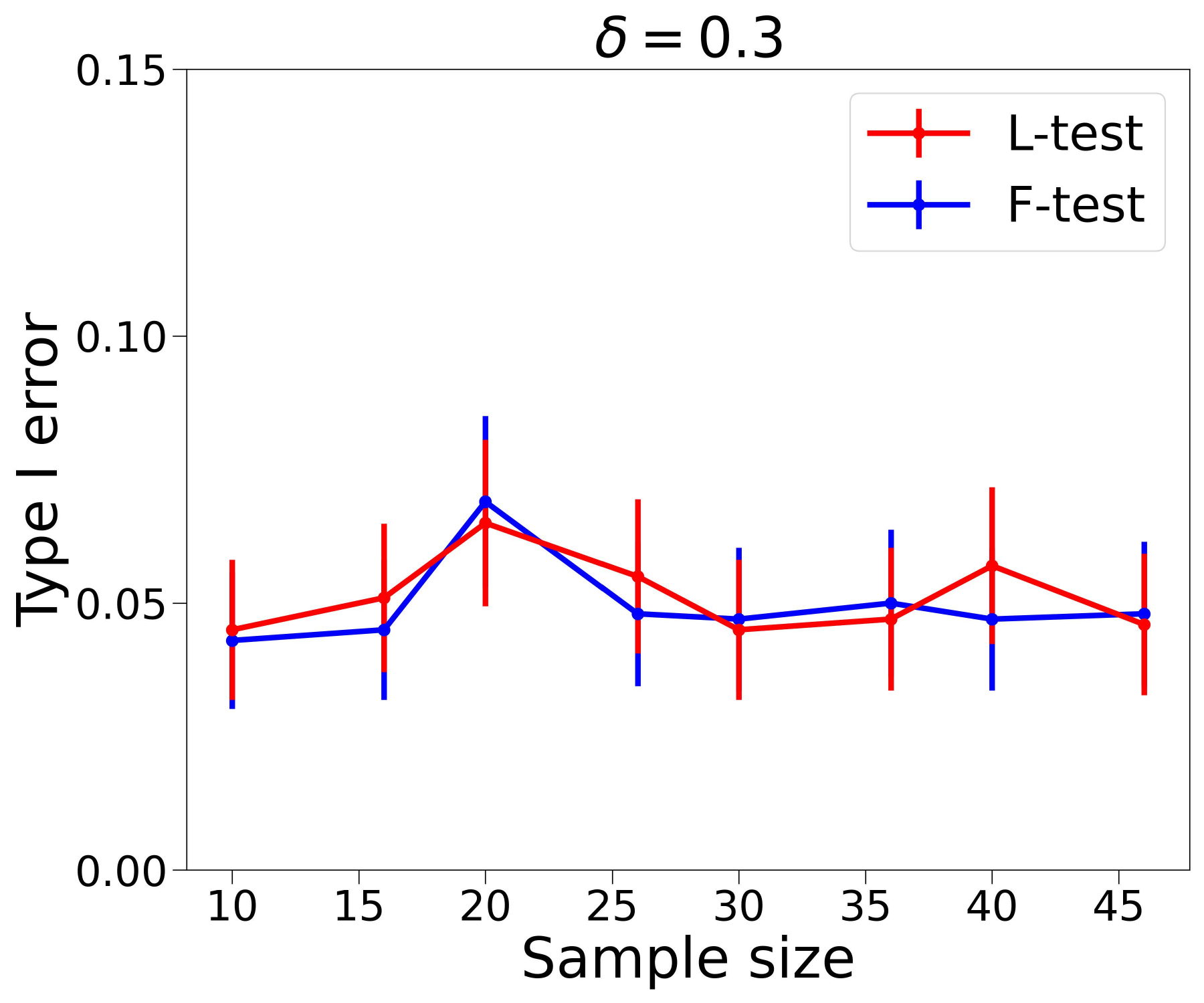}
    \end{subfigure}
    \hspace{0.2cm}
    \begin{subfigure}[t]{0.4\linewidth}
        \centering
        \includegraphics[width=\linewidth]{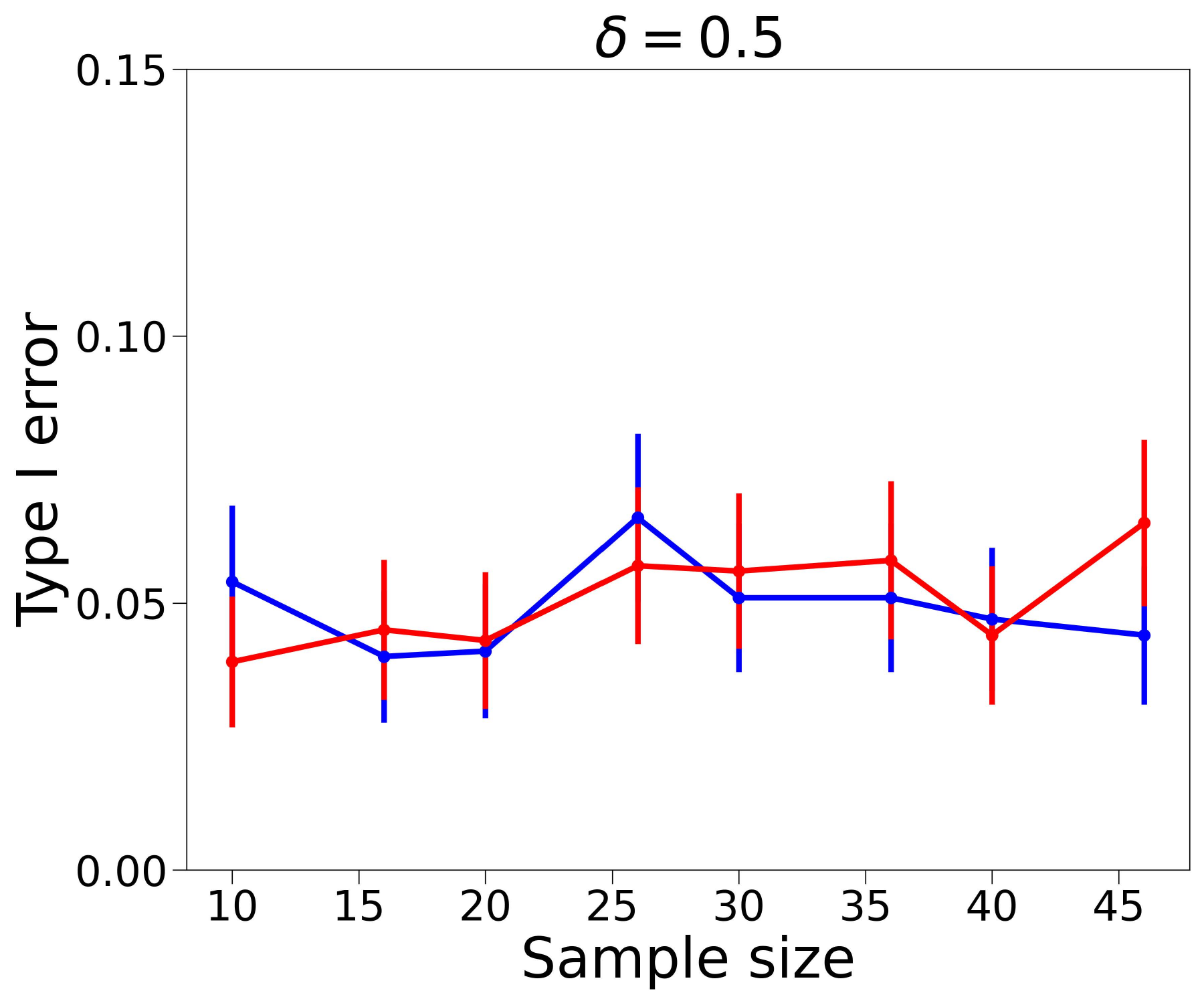}
    \end{subfigure} \\
    \begin{subfigure}[t]{0.4\linewidth}
        \centering
        \includegraphics[width=\linewidth]{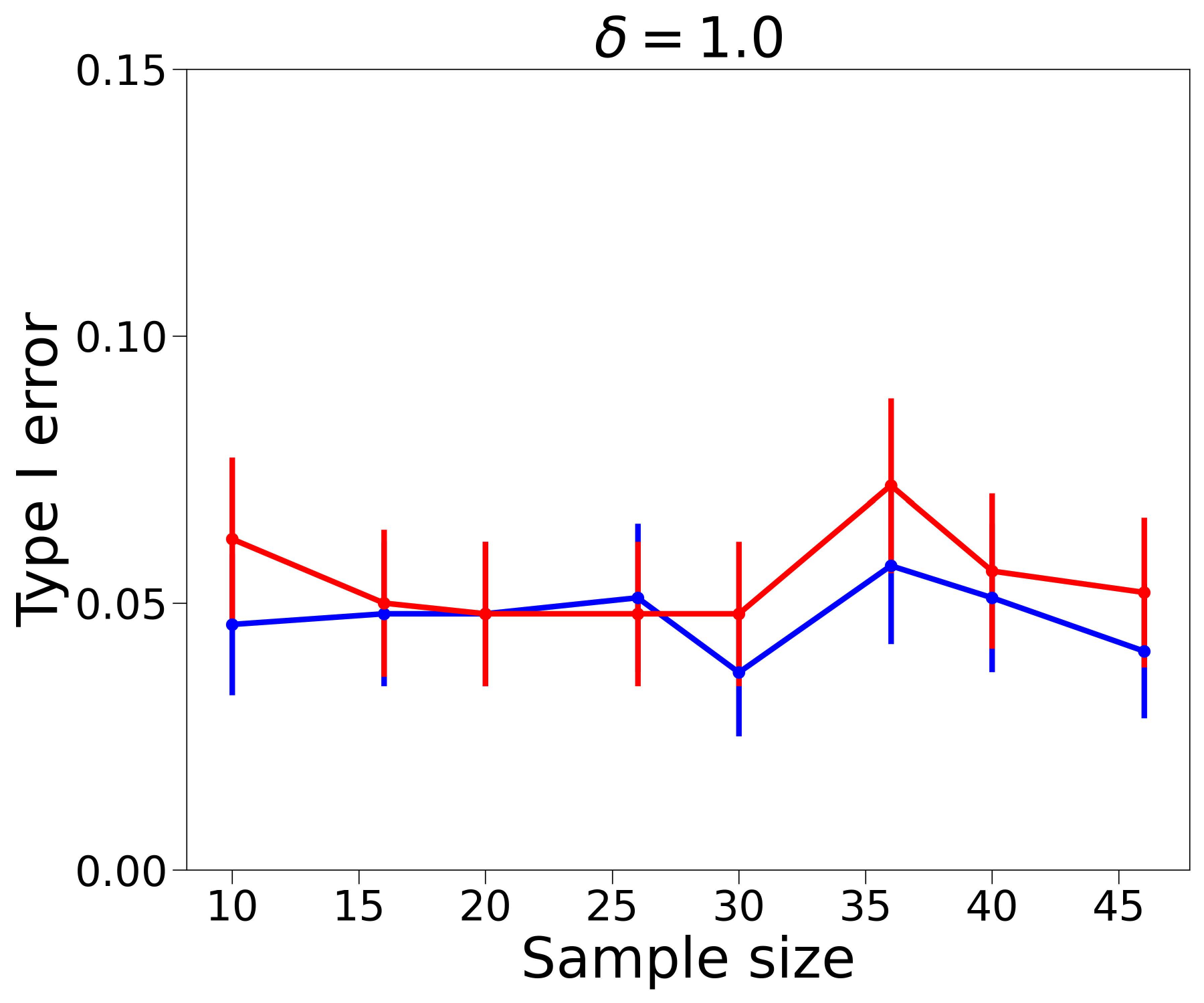}
    \end{subfigure}
    \hspace{0.2cm}
    \begin{subfigure}[t]{0.4\linewidth}
        \centering
        \includegraphics[width=\linewidth]{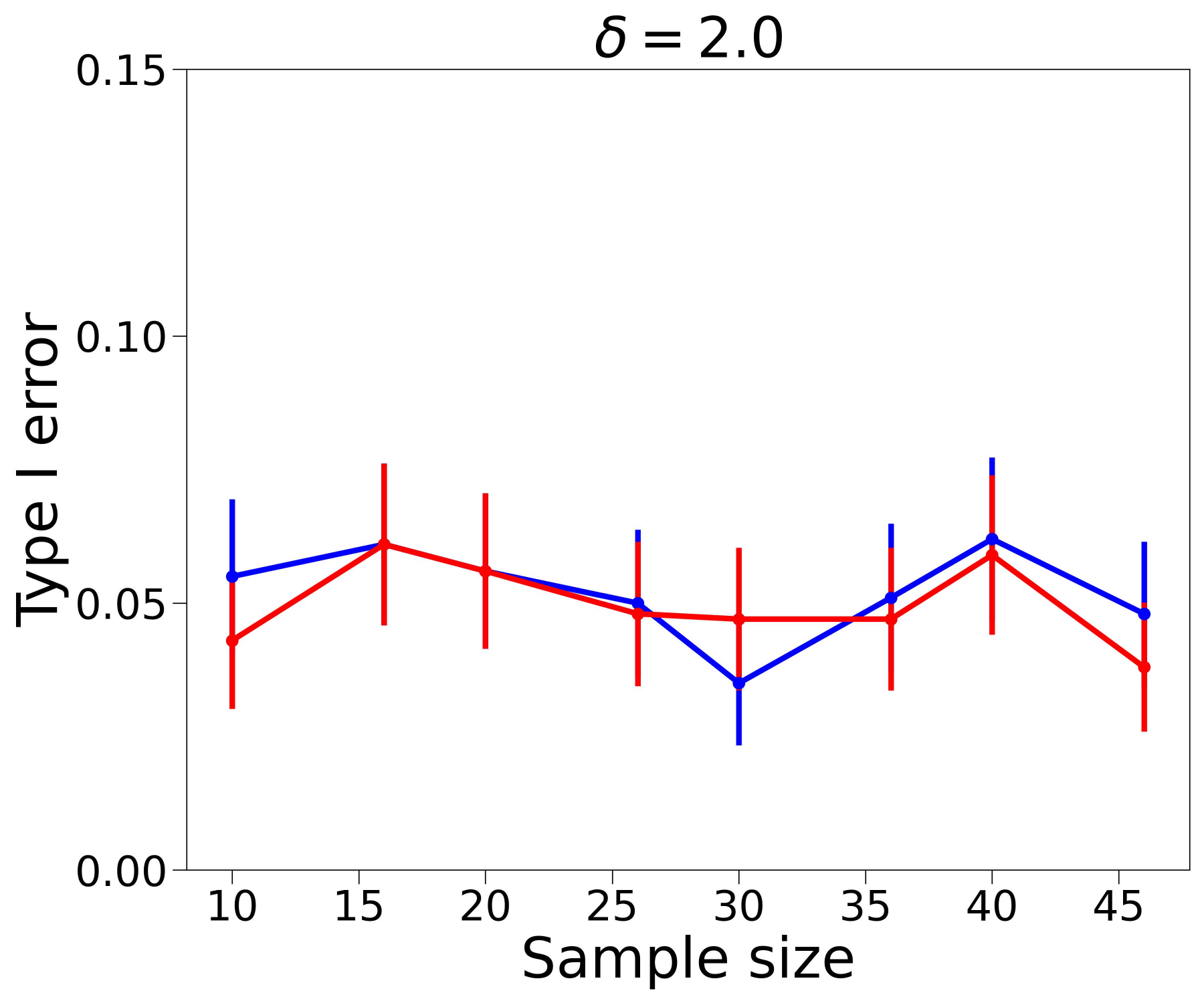}
    \end{subfigure} \\
    \begin{subfigure}[t]{0.4\linewidth}
        \centering
        \includegraphics[width=\linewidth]{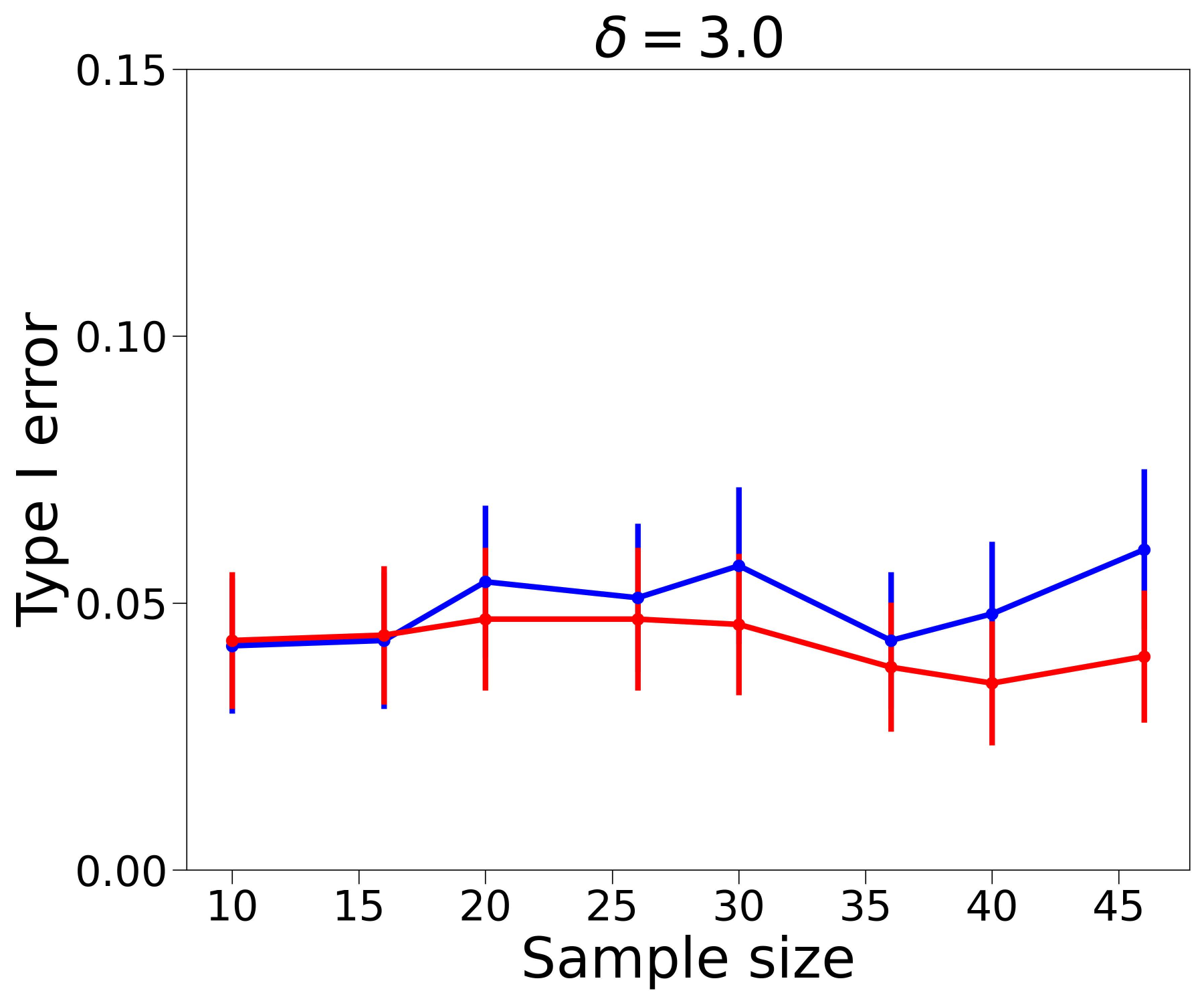}
    \end{subfigure}
    \hspace{0.2cm}
    \begin{subfigure}[t]{0.4\linewidth}
        \centering
        \includegraphics[width=\linewidth]{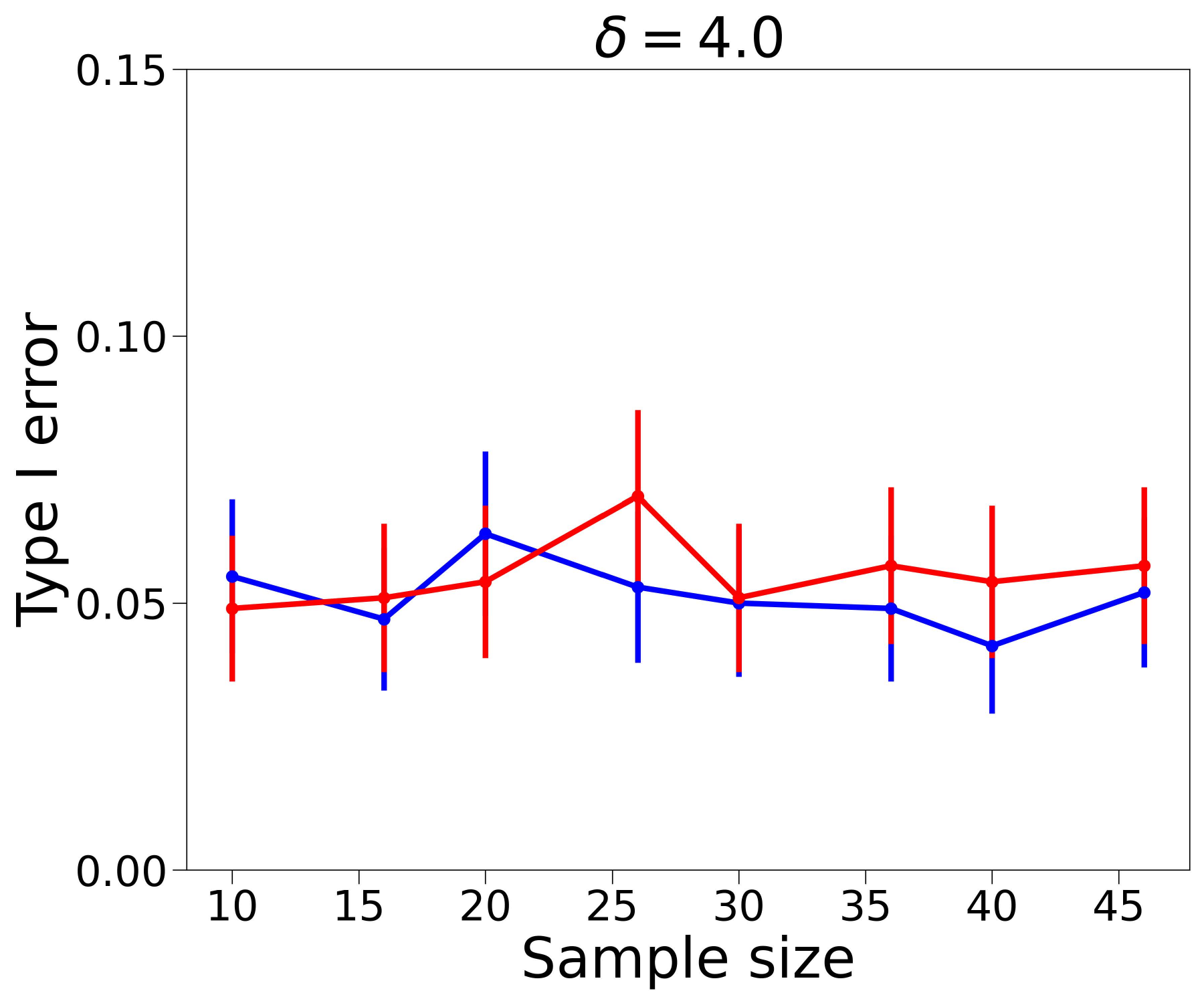}
    \end{subfigure}
    \caption{Same as Figure \ref{fig:rob_4} but with $\rho = 0.5$.}
    \label{fig:rob_8}
\end{figure}

\end{document}